\documentclass[10pt,reqno]{amsart}
\usepackage{hyperref}
\usepackage{latexsym,amsmath}
\usepackage{enumerate}
\usepackage{amsfonts}
\usepackage{amssymb}
\usepackage{latexsym}
\usepackage{fixmath}

\usepackage[T1]{fontenc}
\usepackage{fourier}
\usepackage{bbm}
\usepackage{color}
\usepackage{fullpage}

\renewcommand{\epsilon}{\eps}
\newcommand{\DeltaM}{\Sigma}

\newcommand{\cPcent}{\cP_*^{2}}

\newcommand{\GG}{\mathbb G}
\newcommand{\TT}{\mathbb T}
\newcommand{\HH}{\mathbb H}
\newcommand{\J}{\vec J}

\newcommand\MU{\vec\mu}
\newcommand\vY{\vec Y}
\newcommand\vm{\vec m}

\newcommand\GAMMA{{\vec\gamma}}

\newcommand\PSI{\vec\psi}
\newcommand\RHO{{\vec\rho}}
\newcommand\PHI{\vec\Phi}

\newcommand\nix{\,\cdot\,}

\newcommand\dd{{\mathrm d}}

\newcommand\G{\vec G}
\newcommand\T{\vec T}

\numberwithin{equation}{section}

\renewcommand{\vec}[1]{\boldsymbol{#1}}

\newcommand\SIGMA{\vec\sigma}
\newcommand\TAU{\vec\tau}

\newtheorem{definition}{Definition}[section]
\newtheorem{claim}[definition]{Claim}

\newtheorem{remark}[definition]{Remark}
\newtheorem{theorem}[definition]{Theorem}
\newtheorem{lemma}[definition]{Lemma}
\newtheorem{proposition}[definition]{Proposition}
\newtheorem{corollary}[definition]{Corollary}

\newtheorem{fact}[definition]{Fact}

\newcommand\fS{\mathfrak{S}}

\newcommand\fC{\mathfrak{C}}

\newcommand\fF{\mathfrak{F}}

\newcommand\cA{\mathcal{A}}
\newcommand\cB{\mathcal{B}}
\newcommand\cC{\mathcal{C}}

\newcommand\cF{\mathcal{F}}
\newcommand\cG{\mathcal{G}}
\newcommand\cE{\mathcal{E}}

\newcommand\cN{\mathcal{N}}

\newcommand\cH{\mathcal{H}}
\newcommand\cS{\mathcal{S}}

\newcommand\cK{\mathcal{K}}
\newcommand\cJ{\mathcal{J}}

\newcommand\cM{\mathcal{M}}
\newcommand\cO{\mathcal{O}}
\newcommand\cP{\mathcal{P}}
\newcommand\cX{\mathcal{X}}
\newcommand\cY{\mathcal{Y}}

\newcommand\cZ{\mathcal{Z}}
\def\cR{{\mathcal R}}
\def\cC{{\mathcal C}}
\def\cE{{\mathcal E}}

\newcommand{\beq}{\begin{equation}} \newcommand{\eeq}{\end{equation}}
\newcommand{\dc}{\dcond}
\newcommand{\dKS}{d_{\mathrm{KS}}}
\newcommand{\dr}{d_{\mathrm{rec}}}

\newcommand\thet{\vartheta}

\newcommand\eul{\mathrm{e}}
\newcommand\eps{\varepsilon}

\newcommand\Var{\mathrm{Var}}
\newcommand\Erw{\mathbb{E}}
\newcommand{\vecone}{\vec{1}}

\newcommand{\Po}{{\rm Po}}

\newcommand\TV[1]{\left\|{#1}\right\|_{\mathrm{TV}}}
\newcommand\tv[1]{\|{#1}\|_{\mathrm{TV}}}
\newcommand\dTV{d_{\mathrm{TV}}}

\newcommand{\bink}[2] {{\binom{#1}{#2}}}

\newcommand\bc[1]{\left({#1}\right)}
\newcommand\cbc[1]{\left\{{#1}\right\}}
\newcommand\bcfr[2]{\bc{\frac{#1}{#2}}}
\newcommand{\bck}[1]{\left\langle{#1}\right\rangle}
\newcommand\brk[1]{\left\lbrack{#1}\right\rbrack}
\newcommand\scal[2]{\bck{{#1},{#2}}}

\newcommand\abs[1]{\left|{#1}\right|}

\newcommand\RR{\mathbb{R}}

\newcommand{\stacksign}[2]{{\stackrel{\mbox{\scriptsize #1}}{#2}}}
\newcommand{\tensor}{\otimes}

\newcommand{\Erdos}{Erd\H{o}s}
\newcommand{\Renyi}{R\'enyi}

\newcommand\pr{\mathbb{P}} 
\renewcommand\Pr{\pr} 

\newcommand\Lem{Lemma}
\newcommand\Prop{Proposition}
\newcommand\Thm{Theorem}

\newcommand\Cor{Corollary}
\newcommand\Sec{Section}

\newcommand{\Faadi}{Fa\`a di Bruno}

\newcommand\id{\mathrm{id}}

\newcommand{\Phipsi}{\Phi_{\psi}}
\newcommand{\Phipsisub[1]}{\Phi_{\psi_{#1}}}

\newcommand{\Eig[1]}{\mathrm{Eig}^{\ast} (#1)}
\newcommand{\eig}{\mathrm{Eig}}
\newcommand{\Pomast}{\cP^2_\ast (\Omega)}

\newcommand{\dcond}{d_{\mathrm{cond}}}
\DeclareMathOperator{\Tr}{tr}

\begin{document}

\title{Charting the replica symmetric phase}

\author{Amin Coja-Oghlan$^{*}$, Charilaos Efthymiou$^{**}$, Nor Jaafari, Mihyun Kang$^{***}$, Tobias Kapetanopoulos$^{****}$}
\thanks{$^{*}$The research leading to these results has received funding from the European Research Council under the European Union's Seventh 
Framework Programme (FP7/2007-2013) / ERC Grant Agreement n.\ 278857--PTCC\\
$^{**}$ Supported by DFG grant  EF 103/1-1 \\
$^{***}$Supported by Austrian Science Fund (FWF): P26826.\\
$^{****}$Supported by Stiftung Polytechnische Gesellschaft PhD grant}

\address{Amin Coja-Oghlan, {\tt acoghlan@math.uni-frankfurt.de}, Goethe University, Mathematics Institute, 10 Robert Mayer St, Frankfurt 60325, Germany.}

\address{Charilaos Efthymiou, {\tt efthymiou@math.uni-frankfurt.de}, Goethe University, Mathematics Institute, 10 Robert Mayer St, Frankfurt 60325, Germany.}

\address{Nor Jaafari, {\tt jaafari@math.uni-frankfurt.de}, Goethe University, Mathematics Institute, 10 Robert Mayer St, Frankfurt 60325, Germany.}

\address{Mihyun Kang, {\tt kang@math.tugraz.at}, Technische Universit\"at Graz, Institute of Discrete Mathematics, Steyrergasse 30, 8010 Graz, Austria}

\address{Tobias Kapetanopoulos, {\tt kapetano@math.uni-frankfurt.de}, Goethe University, Mathematics Institute, 10 Robert Mayer St, Frankfurt 60325, Germany.}

\begin{abstract}
%\noindent
Diluted mean-field models are spin systems whose geometry of interactions is induced by a sparse random graph or hypergraph.
Such models play an eminent role in the statistical mechanics of disordered systems as well as in combinatorics and computer science.
In a path-breaking paper based on the non-rigorous `cavity method', physicists predicted not only the existence of a replica symmetry breaking phase transition in such models but also sketched a detailed picture of the evolution of the Gibbs measure within the replica symmetric phase and its impact on
important problems in combinatorics, computer science and physics [Krzakala et al.: PNAS 2007].
In this paper we rigorise this picture  completely for a broad class of models, encompassing the Potts antiferromagnet on the random graph, the $k$-XORSAT model
and the diluted $k$-spin model for even $k$.
We also prove a conjecture about the detection problem in the stochastic block model that has received considerable attention [Decelle et al.: Phys.\ Rev.\ E 2011].
\end{abstract}

\maketitle

\section{Introduction}\label{Sec_intro}

\subsection{The cavity method}
Contrasting the awe-inspiring arsenal of techniques at the disposal of modern combinatorics and probability with the
utter simplicity of terms in which, say, the \Erdos-\Renyi\ random graph model is defined, 
one might expect that after a half-century of study everything ought to be known about this and alike models.
Yet beneath the surface lurks a picture of mesmerizing complexity.
Its unexpected intricacy was brought out most clearly by a line of research that commenced in the statistical physics community
with the study of {\em diluted mean-field models}, 
spin systems whose geometry of interactions is induced by a sparse random graph or hypergraph.
Such models were put forward in physics as models of disordered systems~\cite{MM}.
Prominent examples include the diluted $k$-spin model or the Potts antiferromagnet on a random graph~\cite{CDGS,GT,PanchenkoTalagrand}.
The graph structure, convergent locally to the Bethe lattice or a Galton-Watson tree,
	induces a non-trivial metric, which is why such models have been argued to evince a closer semblance of physical reality than
	fully connected ones such as the Sherrington-Kirkpatrick model~\cite{MP1,mezard1990spin}.
But perhaps even more importantly, apart from and beyond the disordered systems thread,  in the course of the past half-century
models based on random graphs have come to play a role in combinatorics, probability, statistics and computer science that can hardly be overstated.
For example, the random $k$-SAT model is of fundamental interest in computer science~\cite{ANP},
	the stochastic block model has gained prominence in statistics~\cite{AbbeSurvey,Holland,Cris},
	low-density parity check codes are the bread and butter of modern coding theory~\cite{RichardsonUrbanke},
and problems such as random graph coloring have been the lodestars of probabilistic combinatorics
ever since the days of \Erdos\ and \Renyi~\cite{ANP,ECM,ER60}.

In the course of the past 20 years physicists developed an analytic but non-rigorous technique for the study of such models called the `cavity method'.
It has been brought to bear on all of the aforementioned and very many other models in an impressive and ongoing
 line of work that has led to numerous predictions that impact on an astounding variety of problems (e.g., \cite{Decelle,MM,MPZ,LF}).
The task of putting the cavity method on a rigorous foundation has therefore gained substantial importance, and despite
recent successes (e.g., \cite{CKPZ,DSS3,GMU,Cris}) much remains to be done.
In particular, while the cavity method can be applied to a given model almost mechanically, most rigorous arguments are still based
on ad hoc, model-specific delibarations.
This leads to the question of whether we can come up with abstract arguments that rigorise the cavity method wholesale, which is the thrust of the present paper.

One of the most important predictions of the cavity method is that the Gibbs measures induced by random graph models undergo a {\em replica symmetry breaking} or
{\em condensation} phase transition~\cite{pnas}.
Physically this phase transition resembles the Kauzmann transition from the study of glasses~\cite{Kauzmann48}.
The fact {\em that} a phase transition occurs at the location predicted by the cavity method was recently proved for a fairly broad family of models~\cite{CKPZ}.
However, that result fell short of establishing that the condensation phase transition does indeed mark
the point where the nature of correlations under the Gibbs measure changes as predicted by the cavity method.

Here we prove that this is indeed the case.
In fact, we rigorise the entire ``map'' of the replica symmetric phase as predicted in~\cite{Fede1,pnas,Fede2}, including its boundary,
	the evolution of the nature of correlations within and an important contiguity result.
More specifically, first and foremost we prove that the condensation phase transition does indeed separate a ``replica symmetric'' phase without extensive long-range correlations
from a phase where long-range correlations prevail, arguably the key feature of the physics picture.
Further, we verify the physics prediction on the threshold for the onset of point-to-set correlations, called the {reconstruction threshold}.
Additionally, we derive the precise limiting distribution of the free energy within the replica symmetric phase, thereby
vindicating a prediction that the free energy exhibits remarkably small fluctuations~\cite{Fede1,Fede2}.
Finally, verifying a prominent prediction from \cite{Decelle},
 we prove a contiguity statement that has an impact on statistical inference problems such as the stochastic block model.

The results of this paper cover a wide class of random graph models, 
 even broader than the family of models for which the condensation threshold was previously derived in~\cite{CKPZ}.
Indeed, as a testimony to the power of the present general approach we may point out that
even the specializations of the main results to prominent examples such as the Potts antiferromagnet on the random graph or the $k$-spin model were not previously known, even though these models received considerable attention in their own right.
Before presenting the general results in \Sec~\ref{Sec_results},
we illustrate their impact on three important examples:
the diluted $k$-spin model,  the Potts antiferromagnet on the random graph and the stochastic block model.

\subsection{The diluted $k$-spin model}\label{Sec_intro_kspin}
For integers $k\geq2$, $n\geq1$ and a real $p\in[0,1]$ let $\HH=\HH_k(n,p)$ be the random $k$-uniform hypergraph on $V_n=\{x_1,\ldots,x_n\}$
whose edge set $E(\HH)$ is obtained by including each of the $\bink nk$ possible $k$-subsets of $V_n$ with probability $p$ independently.
Additionally, let $\vec J = (\vec J_e)_{e\in E(\HH)}$ be a family of independent standard Gaussians.
The {\em $k$-spin model} on $\HH$ at inverse temperature $\beta>0$ is the distribution on the set $\{-1,1\}^{V_n}$ defined by
	\begin{align}\label{eqkSpin1}
	\mu_{\HH,\J,\beta}(\sigma)&=\frac{1}{Z_\beta(\HH,\J)}\prod_{e\in E(\HH)}\exp\bc{\beta\vec J_e\prod_{y\in e}\sigma(y)},
		\quad\mbox{ where }\quad
	Z_\beta(\HH,\J)=\sum_{\tau\in\{\pm1\}^{V_n}}\prod_{e\in E(\HH)}\exp\bc{\beta\vec J_e\prod_{y\in e}\tau(y)}.
	\end{align}
Arguably the most interesting and at the same time most challenging scenario arises in the case of a sparse random hypergraph \cite{MP1}.
Specifically, set $p=d/\bink{n-1}{k-1}$ for a fixed $d>0$ so that in the limit $n\to\infty$ the average vertex degree of $\HH$ converges to $d$ in probability.
How does the model change as we vary $d$?

According to the physics predictions for any $k$, $\beta$ there exists a {\em condensation threshold} $\dc(k,\beta)$ where the
function $d\mapsto\lim_{n\to\infty}\frac1n\Erw[\ln Z_\beta(\HH,\J)]$ is non-analytic~\cite{Silvio}.
This conjecture was proved in the case $k=2$ by Guerra and Toninelli~\cite{GT}.
However,  their technique does not give the precise condensation phase transition for $k>2$~\cite[\Sec~9]{GT}, nor does
he $k$-spin model belong to the class of models for which the condensation threshold was determined in~\cite{CKPZ}.
The following theorem pinpoints the precise condensation threshold for all $k\geq3$, proving the prediction from~\cite{Silvio}.

As is the case of most results inspired by the cavity method, the precise value  $\dc(k,\beta)$ comes in terms of a stochastic optimization problem.
Specifically, write $\cP(\cX)$ for the set of all probability distributions on a finite set $\cX$ and identify $\cP(\cX)$ with the standard simplex in $\RR^{\cX}$.
Moreover, let $\cP^2(\Omega)$ be the space of all probability measures on $\cP(\cX)$ and let $\cP^2_*(\cX)$ be the space of all
 $\pi\in\cP^2(\cX)$ whose barycenter $\int_{\cP(\cX)}\mu\dd\pi(\mu)$ is the uniform distribution on $\cX$.
Further, let $\Lambda(x)=x\ln x$.

\begin{theorem}\label{Thm_SK1}
Suppose that $d>0,\beta>0$ and that $k\geq3$.
Let $\GAMMA$ be a Poisson variable with mean $d$, let $\vec I_1,\vec I_2,\ldots$ be standard Gaussians 
and for $\pi\in\cPcent(\{\pm 1\})$  let $\RHO_1^\pi,\RHO_2^\pi,\ldots\in\cP(\{ \pm 1\})$ be random variables with distribution $\pi$, all mutually independent. Define
	\begin{align*}
	\cB_{k-\mathrm{spin}}(d,\beta,\pi)=\frac12\Erw&\brk{
			\Lambda\bc{\sum_{\sigma_k\in\{\pm1\}}\prod_{j=1}^{\vec\gamma}\sum_{\sigma_1,\ldots,\sigma_{k-1}\in\{\pm1\}}
				(1+\tanh(\beta\vec I_j\sigma_1\cdots\sigma_k))\prod_{h=1}^{k-1}\RHO_{kj+h}^{\pi}(\sigma_h)}}\\
			&-\frac dk
			\Erw\brk{\Lambda\bc{1+\sum_{\sigma_1,\ldots,\sigma_k\{\pm1\}}\tanh(\beta\vec I_1\sigma_1\cdots\sigma_k)
					\prod_{h=1}^k\RHO_h^{\pi}(\sigma_h)}}.		
	\end{align*}
and
	$\dc(k,\beta)=\inf\{d>0:\sup_{\pi\in\cPcent(\{1,-1\})}\cB_{k-\mathrm{spin}}(d,\beta,\pi)>\ln 2\}.$
Then $0<\dc(k,\beta)<\infty$ and
	$$\lim_{n\to\infty}\frac1n\Erw[\ln Z_\beta(\HH,\vec J)]\begin{cases}
		=\ln 2+\frac d{\sqrt{2\pi}k}\int_{-\infty}^\infty\ln(\cosh(z))\exp(-z^2/2)\dd z&\mbox{ if }d\leq \dc(k,\beta),\\
		<\ln 2+\frac d{\sqrt{2\pi}k}\int_{-\infty}^\infty\ln(\cosh(z))\exp(-z^2/2)\dd z&\mbox{ if }d> \dc(k,\beta).\end{cases}$$
\end{theorem}

From now on we assume that $k\geq4$ is even.
The regime $d<\dc(k,\beta)$ is called the {\em replica symmetric phase}.
According to the cavity method, its key feature is that with probability tending to $1$ in the limit $n\to\infty$,
two independent samples $\SIGMA_1,\SIGMA_2$ (`replicas') chosen from the Gibbs measure $\mu_{\HH,\J,\beta}$ are ``essentially perpendicular''.
To formalize this define for $\sigma,\tau:V_n\to\{\pm1\}$ the {\em overlap} as
	$\varrho_{\sigma,\tau}=\sum_{x\in V_n}\sigma(x)\tau(x)/n.$
We write  $\bck\nix_{\HH,\J,\beta}$ for the average on $\SIGMA_1,\SIGMA_2$ chosen independently from $\mu_{\HH,\J,\beta}$ and
denote the expectation over the choice of $\HH$ and $\J$ by $\Erw\brk\nix$.

\begin{theorem}\label{Thm_SK2}
For all $\beta>0$ and $k\geq4$ even we have
	$\dc(k,\beta)=\inf\cbc{d>0:\limsup_{n\to\infty}\Erw\bck{\varrho_{\SIGMA_1,\SIGMA_2}^2}_{\HH,\vec J,\beta}>0}.$
\end{theorem}

\noindent
The corresponding statement for $k=2$ was proved by Guerra and Toninelli, but as they point out their argument does not extend to larger $k$~\cite{GT}.

\Thm~\ref{Thm_SK2} implies the absence of extensive long-range correlations in the replica symmetric phase.
Indeed, for two vertices $x,y\in V_n$ and $s,t\in\{+1, -1\}$ let 
	$$\mu_{\HH,\J,\beta,x,y}(s,t)=\bck{\vecone\{\SIGMA_1(x)=s,\SIGMA_1(y)=t\}}_{\HH,\J,\beta}$$
be the joint distribution of the spins assigned to $x,y$.
Further, let $\bar\rho$ be the uniform distribution on $\{\pm 1\} \times \{\pm 1\}$.
Then the total variation distance $\|\mu_{\HH,\J,\beta,x,y}-\bar\rho\|_{\mathrm{TV}}$ is a measure of how correlated the spins of $x,y$ are.
Indeed, in the case that $k$ is even for every $x\in V_n$ the Gibbs marginals satisfy $\mu_{\HH,\J,\beta,x}(\pm1)=\bck{\vecone\{\SIGMA_1(x)=\pm1\}}_{\HH,\J,\beta}=1/2$ because $\mu_{\HH,\J,\beta}(\sigma)=\mu_{\HH,\J,\beta}(-\sigma)$ for every $\sigma\in\{-1, +1\}^n$.
Therefore, if the spins at $x,y$ were independent, then $\mu_{\HH,\J,\beta,x,y}=\mu_{\HH,\J,\beta,x}\tensor\mu_{\HH,\J,\beta,y}=\bar\rho$.
Furthermore, it is well known (e.g., \cite[\Sec~2]{Victor}) that
	\begin{equation}\label{eqVictor}
	\lim_{n\to\infty}\Erw\bck{\varrho_{\SIGMA_1,\SIGMA_2}^2}_{\HH,\J,\beta}=0\qquad\mathrm{iff}\qquad
		\lim_{n\to\infty}\frac1{n^2}\sum_{x,y\in V_n}\Erw\|\mu_{\HH,\J,\beta,x,y}-\bar\rho\|_{\mathrm{TV}}=0.
	\end{equation}
Thus, \Thm~\ref{Thm_SK2} implies that for $d<\dc(k,\beta)$, with probability tending to $1$, the spins assigned to two random vertices $x,y$ of $\HH$
are asymptotically independent.
By contrast, \Thm~\ref{Thm_SK2} and (\ref{eqVictor}) show that extensive long-range dependencies occur beyond but arbitrarily close to $\dc(k,\beta)$.
 
\subsection{The Potts antiferromagnet}\label{Sec_IntroPotts}
Let $q\geq2$ be an integer, let $\Omega=\{1,\ldots,q\}$ be a set of $q$ ``colors'' and let $\beta>0$.
The antiferromagnetic $q$-spin Potts model on a graph $G=(V(G),E(G))$ at inverse temperature $\beta$ is the probability distribution on $\Omega^{V_n}$ defined by
	\begin{align}\label{eqPottsAntiferro}
	\mu_{G,q,\beta}(\sigma)&=\frac1{Z_{q,\beta}(G)}\hspace{-1mm}\prod_{\{v,w\}\in E(G)}\hspace{-1mm}\exp(-\beta\vecone\{\sigma(v)=\sigma(w)\}),
	\ \ \mbox{where}\ \ 
	Z_{q,\beta}(G)=\hspace{-1mm}\sum_{\tau\in \Omega^{V(G)}}\prod_{\{v,w\}\in E(G)}\hspace{-2mm}\exp(-\beta\vecone\{\tau(v)=\tau(w)\}).
	\end{align}
The Potts model on the 
random graph $\GG=\GG(n,p)$ with vertex set $V_n=\{x_1,\ldots,x_n\}$ whose
edge set $E(\GG)$ is obtained by including each of the $\bink n2$ possible pairs $\{v,w\}$, $v,w\in V_n$, $v\neq w$,
with probability $p\in[0,1]$ independently, has received considerable attention (e.g.~\cite{Banks,Nor,CDGS}).
As in the $k$-spin model, the most challenging case is that $p=d/n$ for a fixed real $d>0$, so that the average degree converges to $d$ in probability.

The condensation phase transition in this model was pinpointed recently~\cite{CKPZ}.
As in the $k$-spin model, the answer comes as a stochastic optimization problem.
To be precise, let $\GAMMA$ be a $\Po(d)$-random variable, let $\RHO_1^\pi,\RHO_2^\pi,\ldots$
	denote samples from $\pi\in\cP^2_*(\Omega)$, mutually independent and independent of $\GAMMA$, and set
	\begin{align}
	\cB_{\mathrm{Potts}}(q,\beta,d)&=\sup_{\pi\in\cPcent(\Omega)}\Erw\brk{\frac{\Lambda(\sum_{\sigma=1}^q\prod_{i=1}^{\vec\gamma}1-(1-\eul^{-\beta})\RHO_{i}^{\pi}(\sigma))}{q(1-(1-\eul^{-\beta})/q)^{\vec\gamma}}
	-\frac d2\cdot\frac{\Lambda(1-(1-\eul^{-\beta})\sum_{\tau=1}^q\RHO_1^{\pi}(\tau)\RHO_2^{\pi}(\tau))}{1-(1-\eul^{-\beta})/q}}
		\label{eqBPotts},\\
	\dc(q,\beta)&=\inf\cbc{d>0:\cB_{\mathrm{Potts}}(q,\beta,d)>\ln q+d\ln(1-(1-\eul^{-\beta})/q)/2}.\label{eqdcPotts}
	\end{align}
Then~\cite[\Thm~1.1]{CKPZ} shows that $0<\dc(q,\beta)<\infty$ and
	\begin{align}\label{eqPottsFreeEnergy1}
	\lim_{n\to\infty}\frac1n\Erw[\ln Z_{q,\beta}(\GG)]&
		\begin{cases}=\ln q+d\ln(1-(1-\eul^{-\beta})/q)/2&\mbox{ if }d\leq\dc(q,\beta),\\
		<\ln q+d\ln(1-(1-\eul^{-\beta})/q)/2&\mbox{ if }d>\dc(q,\beta).\end{cases}
	\end{align}
While it may be difficult to calculate $\dc(q,\beta)$ numerically, there is the explicit {\em Kesten-Stigum bound}~\cite{abbe2015detection}
	\begin{align}\label{eqPottsKS}
	\dc(q,\beta)&\leq	d_{\mathrm{KS}}(q,\beta)=\bcfr{q-1+\eul^{-\beta}}{1-\eul^{-\beta}}^2,
	\end{align}
{which is known to be tight for $q=2$ for all $\beta$~\cite{massoulie2014community,mossel2013proof,Mossel}, conjectured to be tight for $q=3$ for all $\beta$
	\cite{Decelle,MMrec}, and known not to be tight for $q\geq5$~\cite{Sly}.}

What can we say about the nature of the Gibbs measure in the `replica symmetric phase' $0<d<\dc(q,\beta)$?
Azuma's inequality shows that $\frac1n\ln Z_{q,\beta}(\GG)$ converges to $\lim_{n\to\infty}\frac1n\Erw[\ln Z_{q,\beta}(\GG)]$ in probability, i.e.,
the free energy $\ln Z_{q,\beta}(\GG)$ has fluctuations of order $o(n)$.
On the other hand, given that key parameters such as the size of the largest connected component of $\GG$ exhibit fluctuations of order $\sqrt n$
even once we condition on the number $|E(\GG)|$ of edges, one might  expect that so does $\ln Z_{q,\beta}(\GG)$.
Yet remarkably, the following theorem shows that throughout the replica symmetric phase the free energy merely has {\em bounded} fluctuations given $|E(\GG)|$.
In fact, we know the precise limiting distribution.

\begin{theorem}\label{Thm_PottsNor}
Let $q\geq2$, $\beta>0$ and $0<d<\dc(q,\beta)$.
With $(K_l)_{l\geq3}$ a sequence of independent Poisson variables with mean $\Erw[K_l]=d^l/(2l)$, let
	$$\cK=\sum_{l=3}^\infty K_l\ln(1+\delta_l)-\frac{d^l\delta_l}{2l}\qquad\mbox{where}\quad\delta_l=
		(q-1)\left(\frac{\eul^{-\beta}-1}{q-1+\eul^{-\beta}}\right)^l.$$
Then $\Erw|\cK|<\infty$ and, in distribution,
	$$\ln Z_{q,\beta}(\GG)-\bc{n+\frac12}\ln q-|E(\GG)|\ln\bc{1-\frac{1-\eul^{-\beta}}q}
		+\frac{q-1}2\ln\bc{1+\frac{d(1-\eul^{-\beta})}{q-1+\eul^{-\beta}}}
			+\frac{d\delta_1}2+\frac{d^2\delta_2}4\quad\stacksign{$n\to\infty$}{\to}\quad \cK.$$
\end{theorem}

Further, as in the $k$-spin model the replica symmetric phase can be characterized in terms of the overlap.
Formally, define the {\em overlap} of two colorings $\sigma,\tau:V_n\to\Omega$ as the
	 probability distribution $\rho_{\sigma,\tau}=(\rho_{\sigma,\tau}(s,t))_{s,t\in\Omega}$ on $\Omega\times\Omega$ where
	$\rho_{\sigma,\tau}(s,t)=|\sigma^{-1}(s)\cap\tau^{-1}(t)|/n$ is the probability that a random vertex $v$ is colored $s$ under $\sigma$ and $t$ under $\tau$.
Let $\bar\rho$ denote the uniform distribution on $\Omega\times\Omega$,
write $\SIGMA_1,\SIGMA_2$ for two independent samples from $\mu_{\GG,q,\beta}$, denote the expectation
with respect to $\SIGMA_1,\SIGMA_2$ by $\bck\nix_{\GG,q,\beta}$ and
 the expectation over the choice of $\GG$ by $\Erw\brk\nix$.

\begin{theorem}\label{Thm_PottsOverlap}
For all $q\geq2,\beta>0$ we have
	$\dc(q,\beta)=\inf\cbc{d>0:\limsup_{n\to\infty}\Erw\bck{\|\rho_{\SIGMA_1,\SIGMA_2}-\bar\rho\|_{\mathrm{TV}}}_{\GG,q,\beta}>0}.$
\end{theorem}

\noindent
As in the case of the $k$-spin model
it is easy to see that $\Erw\langle\|\rho_{\SIGMA_1,\SIGMA_2}-\bar\rho\|_{\mathrm{TV}}\rangle_{\GG,q,\beta}=o(1)$ iff
the colors assigned to two randomly chosen vertices of $\GG$ are asymptotically independent with probability tending to one.
Hence,  $\dc(q,\beta)$ marks the onset of long-range correlations.

In many diluted models, and in particular in the Potts antiferromagnet, 
the condensation transition is conjectured to be preceded by another threshold where certain ``point-to-set correlations'' emerge~\cite{pnas}.
Intuitively, the {\em reconstruction threshold} is the point from where for a random vertex $y\in V_n$ correlations between 
the color assigned to $y$ and the colors assigned to {\em all} vertices at a large enough distance $\ell$ from $y$ persist.
Formally, with $\SIGMA$ chosen from $\mu_{\GG,q,\beta}$ let
	$\nabla_{\ell,q,\beta}(\GG,y)$ be the $\sigma$-algebra on $\Omega^{V_n}$ generated by the random variables
	$\SIGMA(z)$, where $z$ ranges over all vertices at distance at least $\ell$ from $y$.
Then 
	\begin{align}\label{eqGPottscorr}
	\mathrm{corr}_{q,\beta}(d)&=\lim_{\ell\to\infty}\limsup_{n\to\infty}\frac1n\sum_{y\in V_n}\sum_{s\in\Omega}
		\Erw\bck{\abs{\bck{\vecone\{\SIGMA(y)=s\}\big|\nabla_{\ell,q,\beta}(\GG,y)}_{\GG,q,\beta}-1/q}}_{\GG,q,\beta}
	\end{align}
measures the extent of correlations between $y$ and a random boundary condition in the limit $\ell,n\to\infty$
	(the outer limit exists due to mononicity).
Indeed, with the expectation $\Erw\brk\nix$ in (\ref{eqGPottscorr}) referring to the choice of $\GG$, the outer $\bck{\nix}_{\GG,q,\beta}$
 chooses a random coloring of the vertices at distance at least $\ell$ from $y$ and the inner $\langle\nix|\nabla_{\ell,q,\beta}(\GG,y)\rangle_{\GG,q,\beta}$
averages over the color of $y$ given the boundary condition.

The {\em reconstruction threshold} is defined as	$\dr(q,\beta)=\inf\{d>0:\mathrm{corr}_{q,\beta}(d)>0\}.$
A priori, calculating $\dr(q,\beta)$ appears to be rather challenging because we seem to have to control the joint distribution of the colors at distance $\ell$ from $y$.
However,  according to physics predictions $\dr(q,\beta)$ is identical to the corresponding threshold on a random tree \cite{pnas}, a conceptually {\em much} simpler object. 
Formally, let $\TT(d)$ be the Galton-Watson tree with offspring distribution $\Po(d)$.
Let $r$ be its root and for an integer $\ell\geq1$ let $\TT^\ell(d)$ be the finite tree obtained by deleting all vertices at distance greater than $\ell$ from $r$.
Then
	\begin{align*}
	\mathrm{corr}_{q,\beta}^\star(d)&=\lim_{\ell\to\infty}\sum_{s\in\Omega}
		\Erw\bck{\abs{\bck{\vecone\{\SIGMA(r)=s\}\big|\nabla_{\ell,q,\beta}(\TT^\ell(d),r)}_{\TT^\ell(d),q,\beta}-1/q}}_{\TT^\ell(d),q,\beta}
	\end{align*}
measures the extent of correlations between the color of the root and the colors at the boundary of the tree.
Accordingly, the {\em tree reconstruction threshold} is defined as	$\dr^\star(q,\beta)=\inf\{d>0:\mathrm{corr}_{q,\beta}^\star(d)>0\}.$
Combining \Thm~\ref{Thm_PottsOverlap} with a result of Gerschenfeld and Montanari~\cite{GM}, we obtain the following result.

\begin{corollary}\label{Cor_PottsRec}
For every $q\geq2$ and $\beta>0$ we have $1\leq\dr(q,\beta)=\dr^\star(q,\beta)\leq\dc(q,\beta).$
\end{corollary}

\noindent
Previously it was known that $\dr(q,\beta)=\dr^\star(q,\beta)$ for $q$ exceeding some (large but) undetermined constant $q_0$~\cite{montanari2011reconstruction}.
This assumption was required because the proof depended on model-specific combinatorial considerations.
A merit of the present approach is that we replace such combinatorial arguments by abstract probabilistic ones.

\subsection{The stochastic block model}\label{Sec_IntroSBM}
The disassortative {\em stochastic block model}, {originally introduced by Holland, Laskey, and Leinhardt~\cite{Holland}}, is
an intensely studied statistical inference problem associated with the Potts model~\cite{Cris}.
We first choose a random coloring $\SIGMA^*:V_n\to\Omega$ of $n$ vertices with $q\geq2$ colors.
Then, setting
	\begin{align*}
	d_{\mathrm{in}}&=\frac{dq\eul^{-\beta}}{q-1+\eul^{-\beta}},&d_{\mathrm{out}}&=\frac{dq}{q-1+\eul^{-\beta}}
	\end{align*}
we generate a random graph $\GG^*$ by connecting any two vertices $v,w$ of the same color $\SIGMA^*(v)=\SIGMA^*(w)$ with probability $d_{\mathrm{in}}/n$ and any two with distinct colors with probability $d_{\mathrm{out}}/n$ independently.
Thus, the average degree of $\GG^*$ converges to $d$ in probability.

Two fundamental statistical problems arise~\cite{Decelle}.
First, given $q,\beta$, for what values of $d$ is it possible to recover a non-trivial approximation of $\SIGMA^*$ given just the random graph $\GG^*$,
i.e., to do better than just a random guess (see \cite{Decelle} for a formal definition)?
A second, more modest task is the {\em detection problem}, which merely asks whether the random graph $\GG^*$ chosen from the stochastic block
can be told model apart from the natural ``null model'', namely the plain \Erdos-\Renyi\ random graph $\GG$.

Decelle, Krzakala, Moore and Zdeborov\'a~\cite{Decelle} predicted that for $d<\dc(q,\beta)$, i.e., 
below the Potts condensation threshold (\ref{eqdcPotts}), it is information-theoretically impossible to solve either problem.
That is, there is {\em no} test or algorithm that can infer with probability tending to $1$ as $n\to\infty$
whether its input was created via the stochastic block model or the \Erdos-\Renyi\ model, let alone obtain a non-trivial approximation to $\SIGMA^*$.
On the other hand, they predicted that there exist {\em efficient} algorithms to solve either problem if $d$ exceeds the Kesten-Stigum bound
	(\ref{eqPottsKS}).
Both of these conjectures were proved in the case $q=2$ by Mossel, Neeman and Sly~\cite{mossel2013proof,Mossel} and Massouli\'e~\cite{massoulie2014community}.
After advances by Bordanve, Lelarge and Massouli\'e~\cite{BLM},
the positive algorithmic conjecture was proved in full by Abbe and Sandon~\cite{abbe2015detection}.
On the negative side, \cite[\Thm~1.3]{CKPZ} shows that no algorithm can infer a non-trivial approximation to $\SIGMA^*$ if $d<\dc(q,\beta)$
for any $q\geq3$, $\beta>0$.
Additionally, Banks, Moore, Neeman, and Netrapalli~\cite{Banks} employed a second moment argument
based on Achlioptas and Naor~\cite{AchNaor} to determine an explicit range of $d$ 
where it is impossible to discern whether the graph was created via
the stochastic block model or the \Erdos-\Renyi\ model.
However, there has remained an extensive gap between their explicit bound and the actual condensation threshold.

Our next result closes this gap and thus settles the conjecture from~\cite{Decelle}.
Recall that the random graph models $\GG,\GG^*$ are {\em mutually contiguous} for $d>0$ if for any sequence $(\cA_n)_{n}$ of events we have
	$$\lim_{n\to\infty}\pr\brk{\GG\in\cA_n}=0\quad\mathrm{iff}\quad\lim_{n\to\infty}\pr\brk{\GG^*\in\cA_n}=0.$$
If so, then clearly no algorithm (efficient or not) can discern with probability $1-o(1)$
whether a given graph stems from the stochastic block model $\GG^*$ or the ``null model'' $\GG$.

\begin{theorem}\label{Thm_contigPotts}
For all $q\geq3$, $\beta>0$, $d<\dc(q,\beta)$ the random graph models $\GG$ and $\GG^*$ are mutually contiguous.
\end{theorem}

\noindent
This result is tight since~\cite[\Thm~2.6]{CKPZ} implies that $\GG,\GG^*$ fail to be mutually contiguous for $d>\dc(q,\beta)$.%

\Thm~\ref{Thm_contigPotts} deals with the disassortative version of the block model, which corresponds to the Potts antiferromagnet.
There is a contiguity conjecture in~\cite{Decelle} 
for the assortative (viz.\ ferromagnetic) version as well, and Banks, Moore, Neeman, and Netrapalli~\cite{Banks} obtained upper and lower bounds in that case too, but
the techniques of the present work do not apply to ferromagnetic models (see \Sec~\ref{Sec_results4}).

\section{Main results}\label{Sec_results}

\noindent
Factor graph models have emerged as a unifying framework for a multitude of concrete models arising in physics, combinatorics, 
and other disciplines~\cite{MM,RichardsonUrbanke}.
The main results of this paper, which we present in this section, therefore deal with a general class of random factor graph models,
subject merely to a few easy-to-check assumptions.
In \Sec~\ref{Sec_results1} we define this general notion.
Then we state the results for general random factor graph models in \Sec~\ref{Sec_results3}.
Moreover, in \Sec~\ref{Sec_examples} we indicate how the diluted $k$-spin model, the Potts antiferromagnet 
and the stochastic block model fit this framework.
\Sec~\ref{Sec_results4} contains a discussion of related work.

\subsection{Factor graphs}\label{Sec_results1}
The following definition encompasses most important examples of spin systems on graphs~\cite{MM}.

\begin{definition}
Let $\Omega$ be a finite set of {\em spins}, let $k\geq2$ be an integer and let $\Psi$ be a set of functions $\psi:\Omega^k\to(0,2)$ that we call {\em weight functions}.
A {\em $\Psi$-factor graph} $G=(V,F,(\partial a)_{a\in F},(\psi_a)_{a\in F})$ consists of
	\begin{itemize}
	\item a {finite} set $V$ of {\em variable nodes},
	\item a {finite} set $F$ of {\em constraint nodes},
	\item an ordered $k$-tuple $\partial a=(\partial_1a,\ldots,\partial_ka)\in V^k$ for each $a\in F$,
	\item a family $(\psi_a)_{a\in F}\in\Psi^F$ of weight functions.
	\end{itemize}
The {\em Gibbs distribution} of $G$ is the probability distribution on $\Omega^V$ defined by
	$\mu_G(\sigma)=\psi_G(\sigma)/{Z(G)}$ for $\sigma\in\Omega^V$, where
	\begin{align*}
	\psi_G(\sigma)&=\prod_{a\in F}\psi_a(\sigma(\partial_1a),\ldots,\sigma(\partial_ka))\quad\mbox{ and }\quad
	Z(G)=\sum_{\tau\in\Omega^V}\psi_G(\tau).
	\end{align*}
\end{definition}

A $\Psi$-factor graph $G$ induces a bipartite graph with vertex sets $V$ and $F$ where $a\in F$ is adjacent to $\partial_1a,\ldots,\partial_ka$.
We shall therefore use common graph-theoretic terminology and refer to, e.g., the vertices $\partial_1a,\ldots,\partial_ka$ as the {\em neighbors} of $a$.
Furthermore, the length of shortest paths in the bipartite graph induces a metric on the nodes of $G$.

Diluted mean-field models correspond to random factor graphs.
To define them formally, we observe that any weight function $\psi:\Omega^k\to(0,2)$ can be viewed as a point in $|\Omega|^k$-dimensional Euclidean space.
We thus endow the set of all possible weight functions with the $\sigma$-algebra induced by the Borel algebra.
Further, for a weight function $\psi:\Omega^k\to(0,2)$ and a permutation $\theta:\{1,\ldots,k\}\to\{1,\ldots,k\}$ we define
	$\psi^\theta:\Omega^k\to(0,2)$, $(\sigma_1,\ldots,\sigma_k)\mapsto\psi(\sigma_{\theta(1)},\ldots,\sigma_{\theta(k)})$.
{\em Throughout the paper we assume that $\Psi$ is a measurable set of weight functions 
such that for all $\psi\in\Psi$ and all permutations $\theta$ we have $\psi^\theta\in\Psi$.}
Moreover, we fix a probability distribution $P$ on $\Psi$.
We always denote by $\PSI$ an element of $\Psi$ chosen from $P$, and we set
	$$q=|\Omega|\quad\mbox{and}\quad\xi=q^{-k}\sum_{\sigma\in\Omega^k}\Erw[\PSI(\sigma)].$$
{\em Furthermore, we always assume that $P$ is such that the following three inequalities hold:}
			\begin{align}\label{eqBounded}
			\Erw[\ln^8(1-\max\{|1-\PSI(\tau)|:\tau\in\Omega^k\})]&<\infty,&
				\Erw[\max\{\PSI(\tau)^{-4}:\tau\in\Omega^k\}]&<\infty,&
				\sum_{\tau\in\Omega^k}\Erw[(\PSI(\tau)-\xi)^2]&>0.
			\end{align}
The first two inequalities bound the `tails' of $\PSI(\tau)$ for $\tau\in\Omega^k$.
The third one provides that $\PSI$ is non-constant.

With these conventions in mind suppose that $n,m>0$ are integers.
Then we define a random $\Psi$-factor graph $\G(n,m,P)$ as follows.
The set of variable nodes is $V_n=\{x_1,\ldots,x_n\}$, the set of constraint nodes is $F_m=\{a_1,\ldots,a_m\}$
and the neighborhoods $\partial a_i\in V_n^k$ are chosen uniformly and independently for $i=1,\ldots,m$.
Furthermore, the weight functions $\psi_{a_i}\in\Psi$ are chosen  from the distribution $P$ mutually independently and independently of the
neighborhoods $(\partial a_i)_{i=1,\ldots,m}$.
Where $P$ is apparent we just write $\G(n,m)$ rather than $\G(n,m,P)$.

Since we aim to study models on sparse random graphs such as the Potts model on the \Erdos-\Renyi\ graph we are concerned with the case that $m=O(n)$ as $n\to\infty$.
To express this elegantly and in order to be able to take the thermodynamic limit $n\to\infty$ easily,
	we fix a real $d>0$ that does not depend on $n$, let $\vec m=\vm_d(n)$ have distribution $\Po(dn/k)$
	and write $\G=\G(n,\vm,P)$ for brevity.
Then the expected degree of a variable node is equal to $d$.

While in $\G$ the neighborhoods $\partial a_i\in V_n^k$ are chosen uniformly, in order to accommodate certain applications such as the Potts model on the \Erdos-\Renyi\ graph we need to impose two conditions.
First, that for any constraint node $a_i$ the $k$ neighboring variable nodes $\partial_1a_i,\ldots,\partial_ka_i$ are distinct.
Second, that $\{\partial_1a_i,\ldots,\partial_ka_i\}\neq\{\partial_1a_j,\ldots,\partial_ka_j\}$ for all $i\neq j$.
Let us denote the event that these two conditions hold by $\fS$.
Combinatorially $\fS$ is the event that the hypergraph whose vertices are the variable nodes and whose edges are the neighborhoods of the contraint nodes
is simple and $k$-uniform.
We are going to state all results both for the unconstraint $\G$ and conditional on $\fS$.

Apart from the condition (\ref{eqBounded}), which we assume tacitly,
the main results require (some of) the following four assumptions.
Crucially, they {\em only} refer to the distribution $P$ on the set $\Psi$ of weight functions.
	\begin{description}
	\item[SYM] For all $i\in\{1,\ldots,k\}$, $\omega\in\Omega$ and $\psi\in\Psi$ we have 
		\begin{equation}\label{eqSYM}
		\sum_{\tau\in\Omega^k}\vecone\{\tau_i=\omega\}\psi(\tau)=q^{k-1}\xi
		\end{equation}
		and for every permutation $\theta$ and every measurable $\cA\subset\Psi$ we have $P(\cA)=P(\{\psi^\theta:\psi\in\cA\})$.
	\item[BAL]   The function
		$$\phi:\mu\in\cP(\Omega)\mapsto\sum_{\tau\in\Omega^k}\Erw[\PSI(\tau)]\prod_{i=1}^k\mu(\tau_i)$$
	is concave and attains its maximum at the uniform distribution on $\Omega$.
\item[MIN]
	Let $\cR(\Omega)$ be the set of all probability distribution $\rho=(\rho(s,t))_{s,t\in\Omega}$ on $\Omega\times\Omega$ such that
		$\sum_{s\in \Omega}\rho(s,t)=\sum_{s\in \Omega}\rho(t,s)=q^{-1}$ for all $t\in\Omega$.
	The function 
		$$\rho \in \cR(\Omega) \mapsto\sum_{\sigma, \tau\in \Omega^k}\Erw[\PSI(\sigma)\PSI(\tau)]  \prod_{i=1}^k\rho(\sigma_i,\tau_i)$$
	has the uniform distribution on $\Omega\times\Omega$ as its unique global minimizer.
\item [POS] 	For all $\pi,\pi'\in\cP_*^2(\Omega)$ the following is true. With $\RHO_1,\RHO_2,\ldots$ chosen from $\pi$,
			$\RHO_1',\RHO_2',\ldots$ chosen from $\pi'$ and $\PSI\in\Psi$ chosen from $P$, all mutually independent,  we have
	\begin{align*}
	\Erw\left[\Lambda\left(\sum_{\tau\in\Omega^k}\PSI(\tau)\prod_{i=1}^ k\RHO_i(\tau_i)\right)+(k-1)\Lambda\left(\sum_{\tau\in\Omega^k}\PSI(\tau)	\prod_{i=1}^k \RHO_i'(\tau_i)\right)
	-k\Lambda\left(\sum_{\tau\in\Omega^k}\PSI(\tau)\RHO_1(\tau_1)\prod_{i=2}^k\RHO_i'(\tau_i)\right)\right]\ge 0.
	\end{align*}
\end{description}

Conditions very similar to {\bf SYM}, {\bf BAL} and {\bf POS} appeared in~\cite{CKPZ} as well.
{\bf SYM} is a symmetry condition.%
	In the language of the cavity method~\cite{MM}, the condition ensures that the unique Belief Propagation fixed point
	on any acyclic $\Psi$-factor graph is such that all messages are identical to the uniform distribution on $\Omega$ (but we will not need this fact explicitly).%
			\footnote{The condition \eqref{eqSYM} emerged out of a discussion with Guilhem Semerjian.}
Condition {\bf BAL} is going to guarantee that for small enough values of $d$ the Gibbs measure $\mu_{\G}$ is typically
	concentrated on ``balanced'' $\sigma\in\Omega^{V_n}$, i.e., $|\sigma^{-1}(\omega)|\sim n/q$ for all $\omega\in\Omega$.
Further, {\bf MIN} is a technical condition that we need in order to study the overlap of two independent Gibbs samples.
Finally, {\bf POS} is required so that we can apply certain results from~\cite{CKPZ}.
As we shall see in \Sec~\ref{Sec_examples}, the conditions are easily verified in the models  from \Sec~\ref{Sec_intro} and several others.

\subsection{Results}\label{Sec_results3}
We proceed to state the results on the condensation phase transition,
the limiting distribution of the free energy, the overlap, the reconstruction and the detection thresholds for random factor graph models.

\subsubsection{The condensation phase transition}
The following theorem pins down the condensation phase transition in random factor graph models precisely
in terms of a stochastic optimization problem that encodes the ``1RSB cavity equations with Parisi parameter $1$'' from the cavity method~\cite{MM}.

\begin{theorem}\label{Thm_cond}
Assume that $P$ satisfies {\bf SYM}, {\bf BAL} and {\bf POS} and let $d>0$.
With $\vec\gamma$ a $\Po(d)$-random variable, $\RHO_1^{\pi},\RHO_2^{\pi},\ldots$
chosen from $\pi\in\cPcent(\Omega)$ and $\PSI_1,\PSI_2,\ldots\in\Psi$ chosen from $P$, all mutually independent, let
	\begin{align}\label{eqMyBethe}
	\cB(d,P,\pi)&=
	\Erw\brk{\frac{1}{q\xi^{\vec\gamma}}
			\Lambda\bc{\sum_{\sigma\in\Omega}\prod_{i=1}^{\vec\gamma}\sum_{\tau\in\Omega^k}\vecone\{\tau_k=\sigma\}\PSI_i(\tau)\prod_{j=1}^{k-1}\RHO_{ki+j}^{\pi}(\tau_j)}
	-\frac{d(k-1)}{k\xi}\Lambda\bc{\sum_{\tau\in\Omega^k}\PSI_1(\tau)\prod_{j=1}^k\RHO_j^{\pi}(\tau_j)}},\\
\dcond &=  \inf\left\{d>0\,:\, \sup_{\pi\in\Pomast} \cB(d,P,\pi) > \ln q + \frac{d}{k}\ln \xi\right\}.\label{eq:dcond}
\end{align}
Then {$1/(k-1)\leq\dcond<\infty$} and
	\begin{align*}
	\lim_{n\to\infty}\frac1n\Erw[\ln Z(\G)]&=\lim_{n\to\infty}\frac1n\Erw[\ln Z(\G)|\fS]=\ln q+\frac dk\ln\xi&\mbox{if }d<\dc,\\
	\limsup_{n\to\infty}\frac1n\Erw[\ln Z(\G)]&=\limsup_{n\to\infty}\frac1n\Erw[\ln Z(\G)|\fS]<\ln q+\frac dk\ln\xi&\mbox{if }d>\dc.
	\end{align*}
\end{theorem}

\noindent
\Thm~\ref{Thm_cond} generalizes~\cite[\Thm~2.7]{CKPZ}, which requires that the set $\Psi$ of weight functions be finite.

Admittedly the formula for $\dc$ provided by \Thm~\ref{Thm_cond} is neither very simple nor very explicit, but we are not aware of any reason why it ought to be.
Yet there is a natural generalization of the Kesten-Stigum bound for the Potts model from (\ref{eqPottsKS}) that provides an easy-to-compute upper bound on $\dc$
in terms of the spectrum of a certain linear operator.
The operator is constructed as follows.
For $\psi\in\Psi$ let  $\Phi_{\psi}\in\RR^{\Omega\times\Omega}$ be the matrix with entries
	\begin{equation}\label{eqPhiMatrices}
	 \Phi_{\psi} (\omega,\omega') =q^{1-k}\xi^{-1}\sum_{\tau\in\Omega^k}\vecone\{\tau_1=\omega,\tau_{2}=\omega'\}
 		\psi(\tau)\qquad\qquad\qquad (\omega,\omega' \in \Omega)
	\end{equation}
and let $\Xi=\Xi_P$ be the linear operator on the $q^2$-dimensional space $\RR^\Omega\tensor\RR^\Omega$ defined by
	\begin{align}\label{eqXi}
	\Xi&=\Xi_P=\Erw[\Phi_{\PSI}\tensor\Phi_{\PSI}].
	\end{align}
Further, with $\vecone$ denoting the vector with all entries equal to one, let
	\begin{equation}\label{eqSpaceE}
	\cE=\cbc{z\in\RR^q\tensor\RR^q:\forall y\in\RR^q:\scal{z}{\vecone\tensor y}=\scal{z}{y\tensor\vecone}=0}.
	\end{equation}
Finally, we introduce
	\begin{equation}\label{eqGenKS}
	\dKS=\bc{(k-1)\max_{x\in\cE:\|x\|=1}\scal{\Xi x}x}^{-1},
	\end{equation}
with the convention that $\dKS=\infty$ if $\max_{x\in\cE:\|x\|=1}\scal{\Xi x}x=0$.

\begin{theorem}\label{Thm_KS}
If $P$ satisfies {\bf SYM} and {\bf BAL}, then  $\dc\leq\dKS$.
\end{theorem}

\noindent
We shall see in \Sec~\ref{Sec_Proofs} that $\Xi$ is related to the ``broadcasting matrix'' of a suitable Galton-Watson tree, which justifies
referring to $\dKS$ as a generalized version of the classical {\em Kesten-Stigum bound} from~\cite{KSBoundBroadcasting}.
While the Kesten-Stigum bound is not generally tight, it plays a major conceptual role, as 
will emerge in due course.

\subsubsection{The free energy}
\Thm~\ref{Thm_cond} easily implies that $n^{-1}\ln Z(\G)$ converges to $\ln q+\frac dk\ln\xi$ in probability if $d<\dc$.
Yet due to the scaling factor of $1/n$ this is but a rough first order approximation.
The next theorem, arguably the principal achievement of this paper, yields the exact limiting distribution of the {\em unscaled} free energy $\ln Z(\G)$ in the entire replica symmetric phase.
Recalling (\ref{eqPhiMatrices}), we introduce the $\Omega\times\Omega$-matrix
	\begin{equation}\label{eqPhi}
	\Phi=\Phi_P=\Erw[ \Phi_{\PSI}].
	\end{equation}
Also recall that $\vm\stacksign{d}=\Po(dn/k)$ denotes the number of constraint nodes of $\G$ and let $\eig(\Phi)$ be the spectrum of $\Phi$.

\begin{theorem}\label{Thm_SSC}
Assume that $P$ satisfies {\bf SYM}, {\bf BAL}, {\bf POS} and {\bf MIN} and that  $0<d<\dcond$. 
Let $(K_{l})_{l\geq1}$ be a family of Poisson variables 
with means $\Erw[K_{l}]=\frac1{2l}(d(k-1))^l$ and let $(\PSI_{l,i,j})_{l,i,j\geq1}$ be a sequence of samples from $P$, all mutually independent.
Then the random variable
	\begin{align}\label{eqcK}
	\cK=\sum_{l=1}^\infty\brk{\frac{(d(k-1))^l}{2l}
	\bc{1-\Tr(\Phi^l)}+\sum_{i=1}^{K_{l}}
		\ln{\Tr\prod_{j=1}^l\Phi_{\PSI_{l,i,j}}}}
	\end{align}
satisfies $\Erw|\cK|<\infty$ and
	\begin{align}\label{eqThm_SSC}
	\ln Z(\G)-\bc{n+\frac12}\ln q-\vec m\ln(\xi)+{\frac12}\sum_{\lambda\in\eig(\Phi)\setminus\cbc1}\ln(1-d(k-1)\lambda)
		&\quad\stacksign{$n\to\infty$}\longrightarrow\quad\cK
	\end{align}
in distribution.
Further, given $\fS$ the random variable on the left hand side of (\ref{eqThm_SSC}) converges in distribution to
	\begin{align*}
	\cK'&=\frac{d(k-1)(1-\Tr(\Phi))}{2}+\vecone\{k=2\}\frac{d^2(1-\Tr(\Phi^2))}4+
		\sum_{l=2+\vecone\{k=2\}}^\infty\brk{\frac{(d(k-1))^l}{2l}\bc{1-\Tr(\Phi^l)}+\sum_{i=1}^{K_{l}}\ln{\Tr\prod_{j=1}^l\Phi_{\PSI_{l,i,j}}}},
	\end{align*}
which also satisfies $\Erw|\cK'|<\infty$.
\end{theorem}

Since key parameters of the random factor graph such as
	 the size of the largest connected component of $\G$ exhibit fluctuations of order $\sqrt n$ even once we condition on $\vm$, one might {\em a priori}
	expect that the same is true of the free energy $\ln Z(\G)$.
However, (\ref{eqThm_SSC}) shows that given $\vm$ the free energy has {\em bounded} fluctuations.

\subsubsection{The overlap}
For $\sigma,\tau\in\Omega^{V_n}$ we define the {\em overlap} $\rho_{\sigma,\tau}=(\rho_{\sigma,\tau}(\omega,\omega'))_{s,t\in\Omega}\in\cP(\Omega\times\Omega)$ by letting
	$$\rho_{\sigma,\tau}(\omega,\omega')=|\sigma^{-1}(\omega)\cap\tau^{-1}(\omega')|/n.$$
Let $\bar\rho$ be the uniform distribution on $\Omega\times\Omega$.
The following theorem confirms one of the core tenets of the cavity method,
namely the absence of extensive long-range correlations for $d<\dc$.
We write $\SIGMA,\TAU$ for two independent samples chosen from the Gibbs measure $\mu_{\G}$, $\bck{\nix}_{\G}$ for the
expectation with respect to the $\mu_{\G}$ and $\Erw\brk\nix$ for the expectation with respect to the choice of $\G$.

\begin{theorem}\label{Thm_overlap}
If $P$ satisfies {\bf SYM}, {\bf BAL}, {\bf POS} and {\bf MIN}, then 
	\begin{align*}
	\dc(q,\beta)&=\inf\cbc{d>0:\limsup_{n\to\infty}\Erw\bck{\|\rho_{\SIGMA,\TAU}-\bar\rho\|_{\mathrm{TV}}}_{\G}>0}
		=\inf\cbc{d>0:\limsup_{n\to\infty}\Erw\brk{\bck{\|\rho_{\SIGMA,\TAU}-\bar\rho\|_{\mathrm{TV}}}_{\G}|\fS}>0}.
	\end{align*}
\end{theorem}

If we let $\mu_{\G,y}(\nix)=\bck{\vecone\{\SIGMA(y)=\nix\}}_{\G}$ be the Gibbs marginal of $y\in V_n$ and
	$\mu_{\G,y_1,y_2}(\nix,\nix)=\bck{\vecone\{\SIGMA_1(y_1)=\nix,\SIGMA_2(y_2)=\nix\}}_{\G}$
	 the joint distribution of the spins at $y_1,y_2\in V_n$,
	then \Thm~\ref{Thm_overlap} implies together with standard arguments that
	\begin{align*}
	\lim_{n\to\infty}\frac1{n^2}\sum_{y_1,y_2\in V_n}\Erw\TV{\mu_{\G,y_1,y_2}-\mu_{\G,y_1}\tensor\mu_{\G,y_2}}&=0
	&\mbox{for all $d<\dc$.}
	\end{align*}
In other words, for $d<\dc$  with probability tending to $1$ as $n\to\infty$, 
the spins assigned to two randomly chosen variable nodes $y_1,y_2$ are asymptotically independent.

Conversely, \Thm~\ref{Thm_overlap} shows that for any $\eps>0$ there exists $\dc<d<\dc+\eps$ such that
	\begin{align}\label{eqLongRange}
	\limsup_{n\to\infty}\frac1{n^2}\sum_{y_1,y_2\in V_n}\Erw\TV{\mu_{\G,y_1,y_2}-\bar\rho}&>0.
	\end{align}
Hence, if we know that the Gibbs marginals $\mu_{\G,y}$ are uniform (e.g., due to the symmetry among colors in the Potts model
	or the inversion symmetry in the $k$-spin model for even $k$), then (\ref{eqLongRange}) becomes
	\begin{align}\label{eqLongRange'}
	\limsup_{n\to\infty}\frac1{n^2}\sum_{y_1,y_2\in V_n}\Erw\TV{\mu_{\G,y_1,y_2}-\mu_{\G,y_1}\tensor\mu_{\G,y_2}}&>0.
	\end{align}
Since two randomly chosen variable nodes $y_1,y_2$ of $\G$ have distance $\Omega(\ln n)$ with probability $1-o(1)$,
(\ref{eqLongRange'}) states that long range correlations persist for $d$ beyond but arbitrarily close to $\dc$.

\subsubsection{The teacher-student model}
Finally, there is a natural statistical inference version of the random factor graph model, the {\em teacher-student model}~\cite{LF}, a generalization of the stochastic block model from \Sec~\ref{Sec_IntroSBM}.
Suppose that $\sigma:V_n\to\Omega$ is an assignment of spins to variable nodes.
Then we introduce a random factor graph $\G^*(n,m,P,\sigma)$ with variable nodes $V_n$ and constraint nodes $F_m$
such that, independently for each $j\in[m]$, the neighborhood $\partial a_j$ and the weight function $\psi_{a_j}$ are chosen from the following joint distribution:
for any $y_1,\ldots,y_k\in V_n$ and for any measurable $\cA\subset\Psi$,
	\begin{align}\label{eqTeacher}
	\pr\brk{\partial a_j=(y_1,\ldots,y_k),\psi_{a_j}\in\cA}
			&=
		\frac{\Erw[\vecone\{\PSI\in\cA\}\PSI(\sigma(y_1),\ldots,\sigma(y_k))]}
			{\sum_{z_1,\ldots,z_k\in V_n}\Erw[\PSI(\sigma(z_1),\ldots,\sigma(z_k))]}.
	\end{align}
Thus, the probability of the outcome $(y_1,\ldots,y_k),\psi_{a_j}=\psi$ is the `prior' probability $P(\psi)$ of selecting $\psi$ times the
`posterior' weight $\psi(\sigma(y_1),\ldots,\sigma(y_k))$.

Further, given $d>0$ consider the following experiment where the initial assignment is chosen randomly as well.
\begin{description}
\item[TCH1] an assignment $\SIGMA^*:V_n\to\Omega$, the {\em ground truth}, is chosen uniformly at random.
\item[TCH2] independently of $\SIGMA^*$, draw $\vec m=\vm_d(n)$ from the Poisson distribution with mean $dn/k$.
\item[TCH3] generate $\G^*=\G^*(n,\vm,P,\SIGMA^*)$.
\end{description}
The intuition behind this model is that a ``teacher'', in possession of the ground truth $\SIGMA^*$, finds herself unable to communicate $\SIGMA^*$ to a student directly.
Instead the teacher utilizes $\SIGMA^*$ to set up a random factor graph $\G^*$ that the student gets to observe.
Given $\G^*$ the student aims to recover $\SIGMA^*$ as best as possible.
As in the case of the stochastic block model, two natural questions arise: given $\G^*$, is it information-theoretically possible to accomplish a better approximation to $\SIGMA^*$ than a mere independent random guess?
More modestly, there is the {\em detection problem}: given a factor graph $G$ is it possible to
	discern with probability $1-o(1)$ as $n\to\infty$ whether $G$ was chosen from the model $\G^*$  or from the ``null model'' $\G$?
As the imprint that the ground truth imbues on $\G^*$ increases with $d$, we should expect the existence of a threshold from where either problem turns solvable. 
Regarding the detection problem,  we recall that the random graph models $\G,\G^*$ are {\em mutually contiguous} if for any sequence $(\cA_n)_{n}$ of events
we have $\lim_{n\to\infty}\pr\brk{\G\in\cA_n}=0$ iff $\lim_{n\to\infty}\pr\brk{\G^*\in\cA_n}=0$.
The following theorem establishes a generalization of the conjectures put forward in~\cite{Decelle} for the stochastic block model
to the case of random factor graph models.

\begin{theorem}\label{Cor_contig}
If $P$ satisfies {\bf SYM}, {\bf BAL}, {\bf POS} and {\bf MIN}, then $\G,\G^*$ are mutually contiguous for all $d<\dc$,
{while $\G,\G^*$ fail to be mutually contiguous for $d>\dc$.}
The same holds given $\G,\G^*\in\fS$.
\end{theorem}

\noindent
Previously it was known that for $d<\dc$ it is impossible to recover an assignment that has a strictly greater overlap with $\SIGMA^*$~\cite[\Thm~2.6]{CKPZ}.
\Thm~\ref{Cor_contig} shows that, in fact, $\dc$ marks the threshold for the feasibility of the humble detection problem.

While Theorem ~\ref{Cor_contig} is bad news from a statistical inference point of view, the upshot is that throughout the replica symmetric phase typical properties
of Gibbs samples of $\G$ can be investigated accurately by way of the teacher-student model $(\G^*,\SIGMA^*)$,
a technique known as ``quiet planting''~\cite{Barriers,quiet}.
This idea has been used critically in rigorous work on specific examples of random factor graph models, e.g., \cite{Molloy}.
Formally, quiet planting applies if the factor graph/assignment pair
$(\G^*,\SIGMA^*)$ comprising the ground truth $\SIGMA^*$ and the outcome $\G^*$ of {\bf TCH1--TCH3} and the pair
$(\G,\SIGMA)$ consisting of the random factor graph $\G$ and a Gibbs sample $\SIGMA$ of $\G$ are mutually contiguous.
Previously this was known to be true for a few specific models (e.g., \cite{h2c,Nor}), albeit not generally in the entire replica symmetric phase.
The following corollary to \Thm~\ref{Cor_contig} shows that ``quiet planting'' is a universal phenomenon.

\begin{corollary}\label{Thm_contig}
Assume that $P$ satisfies {\bf SYM}, {\bf BAL}, {\bf POS} and {\bf MIN}.
For all $d<\dc$ the pairs $(\G,\SIGMA)$ and $(\G^*,\SIGMA^*)$ are mutually contiguous.
The same is true given $\G,\G^*\in\fS$.
\end{corollary}

\subsubsection{Reconstruction}

\noindent
According to the physics deliberations the condensation phase transition is generally preceded by another 
threshold where certain point-to-set  correlations emerge, the reconstruction threshold~\cite{pnas}.
Reconstruction plays a major role in the cavity formalism because it provides the conceptual underpinning for the notion that
the Gibbs measure decomposes into a multitude of ``clusters''~\cite{MM,MPZ}.
Formally, suppose that $G$ is a factor graph with variable nodes $V$, $y\in V$ and  that $\ell\geq0$.
Let $\nabla_{\ell}(G,y)$ be the $\sigma$-algebra on $\Omega^{V}$ generated by the random variables
	$\SIGMA(z)$ such that $z$ is a variable node whose distance from $y$ in $G$ is at least $2\ell$.
Further, define
	\begin{align}\label{eqGcorr}
	\mathrm{corr}(d)&=\lim_{\ell\to\infty}\limsup_{n\to\infty}\frac1n\sum_{y\in V_n}\sum_{s\in\Omega}
		\Erw\bck{\abs{\bck{\vecone\{\SIGMA(y)=s\}\big|\nabla_{\ell}(\G,y)}_{\G}-1/q}}_{\G}.
	\end{align}
Of course, the expectation $\Erw\brk\nix$ refers to the choice of $\G$,
the outer expectation $\bck\nix_{\G}$ averages over the ``boundary condition'', i.e., the spins of the variable nodes at distance
at least $2\ell$ from $y$, and the inner $\langle\nix|\nabla_{\ell}(\G,y)\rangle_{\G}$ is the conditional expectation given the boundary 
condition.
If $\mathrm{corr}(d)=0$, then the influence of a ``typical" boundary condition on the spin of $y$ decays with the radius $\ell$. 
Thus, the {\em reconstruction threshold} $\dr=\inf\{d>0:\mathrm{corr}(d)>0\}$
is the smallest degree where the influence of the boundary persists.

A priori determining $\dr$ appears to be challenging  because the joint distribution of the spins at distance $2\ell$ from $y$ is determined not merely by the ``local'' effects
within the radius-$2\ell$ neighborhood of $y$ but also by the graph beyond.
But according to physics predictions (e.g., \cite{pnas}), actually $\dr$ 	is equal to the corresponding threshold on a suitable Galton-Watson tree.  
Conceptually this amounts to an enormous simplification because the branches of the tree are mutually dependent only
through their being connected to the root, a situation amenable to precise treatment via the
Belief Propagation message passing scheme~\cite{MM}.

Formally, we introduce a multi-type Galton-Watson tree $\T(d,P)$ that mimics the local geometry of $\G$.
The types are either variable nodes or constraint nodes, each of the latter endowed with a weight function $\psi\in\Psi$.
The root of the Galton-Watson tree is a variable node $r$.
The offspring of a variable node is a $\Po(d)$ number of constraint nodes whose weight functions are chosen from $P$ independently.
Moreover, the offspring of a constraint node is $k-1$ variable nodes.
For an integer $\ell\geq0$ we let $\T^\ell(d,P)$ denote the (finite) tree obtained from  $\T(d,P)$ by deleting all variable or constraint nodes at distance
greater than $2\ell$ from $r$.
In analogy to (\ref{eqGcorr}) we set
\begin{equation}\label{def:PlantedCorr}
	\mathrm{corr}^\star(d) = \lim_{\ell\to\infty}\sum_{s\in\Omega}
		\Erw\bck{\abs{\bck{\vecone\{\SIGMA(r)=s\}\big|\nabla_{\ell}(\T^\ell(d,P),r)}_{\T^\ell(d,P)}-1/q}}_{\T^\ell(d,P)}
\end{equation}
The {\em tree reconstruction threshold} is defined as	$\dr^\star=\inf\{d>0:\mathrm{corr}^\star(d)>0\}$.

\begin{theorem}\label{thrm:TreeGraphEquivalence}
Suppose that $P$ satisfies {\bf SYM}, {\bf BAL}, {\bf POS} and {\bf MIN}. Then $0<\dr=\dr^\star\leq\dc$
and  $\mathrm{corr}(d)>0$ for all $d\in(\dr,\dc)$.
Moreover, 
	\begin{equation*}
\lim_{\ell\to\infty}\limsup_{n\to\infty}\frac1n\sum_{y\in V_n}\sum_{s\in\Omega}
		\Erw\brk{\bck{\abs{\bck{\vecone\{\SIGMA(y)=s\}\big|\nabla_{\ell}(\G,y)}_{\G}-1/q}}_{\G}|\fS} =0 
		\quad \textrm{if and only if}\quad 	\mathrm{corr}(d)=0.
	\end{equation*}
\end{theorem}

We prove \Thm~\ref{thrm:TreeGraphEquivalence} by way of the teacher-student model and the ``quiet planting'' result \Cor~\ref{Thm_contig}.
This argument provides a perspective on the reconstruction problem that has an impact on the statistical inference questions as well.
Specifically, we observe that the reconstruction problem on the random tree $\T(d,P)$ is equivalent to a natural ``Bayesian'' reconstruction problem in the teacher-student model.
Formally, let $\nabla^*_{\ell}(\G^*,\SIGMA^*,y)$ be the $\sigma$-algebra generated by the graph $\G^*$ and the random variables $\SIGMA^*(z)$ with $z$
at distance at least $2\ell$ from $y$.
Then
	\begin{align}\label{eqG*corr}
	\mathrm{corr}^*(d)&=\lim_{\ell\to\infty}\limsup_{n\to\infty}\frac1n\sum_{y\in V_n}\sum_{s\in\Omega}
		\Erw\brk{\abs{\pr\brk{\SIGMA^*(y)=s\big|\nabla_{\ell}^*(\G^*,\SIGMA^*,y)}-1/q}}
	\end{align}
measures the correlation between $\SIGMA^*(y)$, the spin at $y$ under the ground truth, and the spins that
 $\SIGMA^*$ assigns to the variables at distance at least $2\ell$.
The proof of \Thm~\ref{thrm:TreeGraphEquivalence}  is based on showing that $\mathrm{corr}^*(d)=\mathrm{corr}^\star(d)$ for all $d$.

\begin{theorem}\label{thrm:TreeGraphEquivalence*}
If $P$ satisfies {\bf SYM}, {\bf BAL}, {\bf POS} and {\bf MIN}, then for all $d>0$ we have
	$$\mathrm{corr}^\star(d)=\mathrm{corr}^*(d)
		=\lim_{\ell\to\infty}\limsup_{n\to\infty}\frac1n\sum_{y\in V_n}\sum_{s\in\Omega}
		\Erw\brk{\abs{\bck{\vecone\{\SIGMA(y)=s\}\big|\nabla_{\ell}(\G^*,y)}_{\G^*}(\SIGMA^*)-1/q}\bigg|\fS} .$$
\end{theorem}

Finally,  we highlight an immediate but  interesting consequence of \Thm s~\ref{Thm_KS} and~\ref{thrm:TreeGraphEquivalence}
that generalizes the classical Kesten-Stigum upper bound for reconstruction on trees \cite{KSBoundBroadcasting}.

\begin{corollary}
If $P$ satisfies {\bf SYM}, {\bf BAL}, {\bf POS} and {\bf MIN}, then $\textrm{corr}^{\star}(d)>0$ for all $d>\dKS$.
\end{corollary}

The reconstruction problem on a certain class of random factor graph models (that includes, e.g., the Potts antiferromagnet) was previously studied by Gerschenfeld and Montanari~\cite{GM}.
They observed that overlap concentration about $\bar\rho$ as provided by \Thm~\ref{Thm_overlap} for $d<\dc$ guarantees that
the reconstruction thresholds $\dr$ and $\dr^\star$ coincide.
Subsequently, with the condensation threshold well out of reach at the time,
 Montanari, Restrepo and Tetali~\cite{montanari2011reconstruction} attempted to verify the required overlap concentration
	at least for all $d$ up to the tree reconstruction threshold.
However, their combinatorial (essentially second moment) argument did not cover the entire range of parameters,
e.g., all $q$ and/or all $\beta$ in the Potts model.
By comparison to~\cite{GM,montanari2011reconstruction}, \Thm~\ref{thrm:TreeGraphEquivalence*} provides a different, perhaps more conceptual angle:
tree reconstruction is equivalent to reconstruction in the teacher-student model for {\em all} $d$, and up
to $\dc$ the equivalence extends to the random factor graph model $\G$ thanks to contiguity.

\subsection{Examples}\label{Sec_examples}
Here we show how the models from \Sec~\ref{Sec_intro} can be cast as random factor graph models that satisfy the assumptions {\bf SYM}, {\bf BAL}, {\bf POS} and {\bf MIN}.

\subsubsection{The Potts antiferromagnet}
For an integer $q\geq2$ and a real $\beta>0$ we let $\Omega=\{1,\ldots,q\}$ and
	\begin{equation}\label{eqPottsPsi}
	\psi_{q,\beta}:(\sigma_1,\sigma_2)\in\Omega^2\mapsto\exp(-\beta\vecone\{\sigma_1=\sigma_2\}).
	\end{equation}
Let $\Psi$ be the singleton $\{\psi_{q,\beta}\}$.
Then the Potts model on a given graph $G=(V,E)$ can be cast as a $\Psi$-factor graph:
we just set up the factor graph $G'=(V,E,(\partial e)_{e\in E},(\psi_e)_{e\in E})$ whose variable nodes are the vertices of the original graph $G$ and
whose constraint nodes are the edges of $G$.
For an edge $e=\{x,y\}\in E$ we let $\partial e=(x,y)$, where, say, the order of the neighbors is chosen randomly, and $\psi_e=\psi_{q,\beta}$, of course. Then $\mu_{G'}$ coincides with $\mu_{G,q,\beta}$ from (\ref{eqPottsAntiferro}).

To mimic the Potts model on the \Erdos-\Renyi\ graph $\GG=\GG(n,d/n)$ we let $P_{\mathrm{Potts}}=\delta_{\psi_{q,\beta}}$ be the atom on $\psi_{q,\beta}$.
Then the sole difference between the factor graph representation $\GG'$ of the \Erdos-\Renyi\ graph $\GG$ and $\G=\G(n,\vm,P)$ is that the latter may have factor nodes $a$ such that $\partial_1 a=\partial_2a$ (``self-loops'')
or pairs of distinct factor nodes $a,b$ such that $\{\partial_1a,\partial_2a\}=\{\partial_1b,\partial_2b\}$ (``double-edges'').
However, conditioning on the event $\fS$ rules out self-loops and double-edges.
Indeed, we have the following.

\begin{fact}[{\cite[\Lem~4.1]{CKPZ}}]\label{Lemma_PottsContig}
The random factor graph $\GG'$ and $\G$ given $\fS$ are mutually contiguous.
\end{fact}

\begin{lemma}\label{Lemma_Potts}
The assumptions {\bf SYM}, {\bf BAL}, {\bf POS} and {\bf MIN} hold for $P_{\mathrm{Potts}}$ for all $q\geq2$ and all $\beta>0$.
\end{lemma}
\begin{proof}
That {\bf SYM}, {\bf BAL} and {\bf POS} hold is known already~\cite[\Lem~4.3]{CKPZ}.
With respect to {\bf MIN}, we observe that for any distribution $\rho$ on $\Omega\times\Omega$ with uniform marginals,
	\begin{align*}
	\sum_{\sigma_1,\sigma_2,\tau_1,\tau_2\in\Omega}\psi_{q,\beta}(\sigma_1,\sigma_2)\psi_{q,\beta}(\tau_1,\tau_2)
				\rho(\sigma_1,\tau_1)\rho(\sigma_2,\tau_2)
		&=1-2(1-\eul^{-\beta})/q+(1-\eul^{-\beta})^2\sum_{\sigma,\tau\in\Omega}\rho(\sigma,\tau)^2.
	\end{align*}
The last expression is strictly convex as a function of $\rho$ with the minimum attained at the uniform distribution.
\end{proof}

Thus the results stated in \Sec~\ref{Sec_IntroPotts} follow from the results for general random factor graph models.
Indeed, to obtain \Thm~\ref{Thm_PottsNor} we observe that the matrices from (\ref{eqPhiMatrices}), (\ref{eqXi}) and (\ref{eqPhi}) satisfy
	\begin{equation}\label{eqPottsMatrices}
	\Phi=\Phi_{\psi_{q,\beta}}=(q-1+\eul^{-\beta})^{-1}(\vecone-(1-\eul^{-\beta})\id),\qquad
		\Xi=(q-1+\eul^{-\beta})^{-2}((\vecone-(1-\eul^{-\beta})\id)\tensor(\vecone-(1-\eul^{-\beta})\id)),
	\end{equation}
where $\vecone$ is the all-ones matrix and $\id$ is the identity matrix.
Clearly, the eigenvalues of $\Phi$ are $1$ and $(\eul^{-\beta}-1)/(q-1+\eul^{-\beta})$, the latter with multiplicity $q-1$.
Hence,
	\begin{align*}
	\Tr(\Phi^l)-1&=(q-1)\bcfr{\eul^{-\beta}-1}{q-1+\eul^{-\beta}}^l,&\ln\Tr(\Phi^l)&=\ln\bc{1+(q-1)\bcfr{\eul^{-\beta}-1}{q-1+\eul^{-\beta}}^l}.
	\end{align*}
Thus, \Thm~\ref{Thm_PottsNor} follows from \Thm~\ref{Thm_SSC} and \Thm~\ref{Thm_PottsOverlap} from \Thm~\ref{Thm_overlap}.
Finally, (\ref{eqPottsMatrices}) shows that $\max_{x\in\cE:\|x\|=1}\scal{\Xi x}x=(1-\eul^{-\beta})^2/(q-1+\eul^{-\beta})^2$ and thus (\ref{eqGenKS}) matches
the ``classical'' Kesten-Stigum bound (\ref{eqPottsKS}).

\subsubsection{The stochastic block model}
The teacher-student model $\G^*$ corresponding to $P_{\mathrm{Potts}}$ is very similar to the stochastic block model.
As in the case of the Potts model on the \Erdos-\Renyi\ graph, the only discrepancy is due to the possible occurrence of self-loops and double-edges.

\begin{lemma}[{\cite[\Lem~4.4]{CKPZ}}]\label{Lemma_SBMcontig}
For any $q\geq2$, $\beta>0$, $d>0$
the stochastic block model $\GG^*$ and the teacher-student model $\G^*$ given $\fS$ are mutually contiguous.
\end{lemma}

\noindent
\Thm~\ref{Thm_contigPotts} follows from \Thm~\ref{Cor_contig} and \Lem~\ref{Lemma_SBMcontig}.

\subsubsection{The $k$-spin model}
Let $\Omega=\{\pm1\}$.
For $J\in\RR,\beta>0$ we could define the weight function
	$\tilde\psi_{J,\beta}(\sigma_1,\ldots,\sigma_k)=\exp(\beta J\sigma_1\cdots\sigma_k)$ to match the definition (\ref{eqkSpin1}) of the $k$-spin model.
However, these functions do not necessarily take values in $(0,2)$.
To remedy this problem we introduce 
$\psi_{J,\beta}(\sigma_1,\ldots,\sigma_k)=1+\tanh(J\beta)\sigma_1\cdots\sigma_k$.
Then (cf.~\cite{PanchenkoTalagrand})
	\begin{align}\label{eqcosh}
	\tilde\psi_{J,\beta}(\sigma_1,\ldots,\sigma_k)=\cosh(J\beta)\psi_{J,\beta}(\sigma_1,\ldots,\sigma_k).
	\end{align}
Thus, let $\Psi=\{\psi_{J,\beta}:J\in\RR\}$, let $\PSI=\psi_{\vec J,\beta}$, where $\vec J$ is a standard Gaussian and let $P_{\vec J,\beta}$ be the law of $\PSI$.
Similarly as in the case of the Potts model we have the following.

\begin{fact}\label{Fact_kspin}
For all $k\geq2,d>0,\beta>0$ the random measure $\mu_{\HH,\vec J,\beta}$ from (\ref{eqkSpin1})
and the Gibbs measure $\mu_{\G(n,\vm,P_{\vec J,\beta})}$ of the random factor graph given $\fS$ are mutually contiguous.
Furthermore,
	\begin{align*}
	\Erw\brk{\ln Z_\beta(\HH,\J)-\sum_{e\in E(\HH)}\ln\cosh(\beta\vec J_e)}&=\Erw[\ln Z(\G(n,\vm,P_{\vec J,\beta}))|\fS]+o(n).
	\end{align*}
\end{fact}

\noindent
Instead of just verifying the conditions {\bf SYM}, {\bf BAL}, {\bf POS} and {\bf MIN} for the $k$-spin model with standard Gaussian
couplings $\vec J$, we will establish the following more general statement.
Recall that a random variable $\vec J$ is {\em symmetric} if $\vec J$ and $-\vec J$ have the same distribution.

\begin{lemma}\label{Lemma_pSpin}
For any $k\geq2$, $\beta>0$ and for any symmetric random variable $\vec J$ such that
$P_{\vec J,\beta}$ satisfies (\ref{eqBounded}) the three conditions {\bf SYM}, {\bf BAL} and {\bf POS} hold.
If $k$ is even, then {\bf MIN} holds as well .
\end{lemma}
\begin{proof}
It is immediate that $\xi=1$ and that $P_{\vec J,\beta}$ satisfies {\bf SYM}.
For {\bf BAL} observe that $\mu\mapsto \sum_{\tau\in\Omega^k}\Erw[\PSI(\tau)]\prod_{i=1}^k\mu(\tau_i)$ is constant because $\vec J$ is symmetric.
To verify {\bf POS} we generalize the argument from \cite[Section 4.4]{CKPZ} by observing that for any integer $l\geq1$, with the notation from {\bf POS},
	\begin{align*}
	\left(1-\sum_{\sigma\in\Omega^k}\psi_{\vec J,\beta}(\sigma)\prod_{i=1}^k\RHO_i(\sigma_i)\right)^l
		=\bc{\tanh(\vec J \beta)}^l\prod_{i=1}^k\left(\RHO_i(1)-\RHO_i(-1)\right)^l.
	\end{align*}
Hence, expanding $\Lambda(\nix)$ and using (\ref{eqBounded}) and Fubini's theorem to swap the sum and the expectation, we find
	\begin{align*}
	\Erw\left[\Lambda\left(\sum_{\tau\in\Omega^k}\PSI(\tau)\prod_{i=1}^ k\RHO_i(\tau_i)\right)\right]
		&=-1+\sum_{l=2}^\infty\frac{\Erw\brk{\tanh(\vec J \beta)^l}}{l(l-1)}\Erw\brk{(\RHO_1(1)-\RHO_1(-1))^l}^k.
	\end{align*}
Applying similarly manipulations to the other two terms from {\bf POS} and introducing $X_l=\Erw[(\RHO_1(1)-\RHO_1(-1))^l]$,
$Y_l=\Erw[(\RHO_1'(1)-\RHO_1'(-1))^l]$, we see that {\bf POS} comes down to showing that
	\begin{align}\label{eqkspinPOS}
	\sum_{l=2}^\infty\frac{1}{l(l-1)}\Erw\brk{\tanh(\vec J \beta)^l}\bc{X_l^k-kX_lY_l^{k-1}+(k-1)Y_l^k}&\geq0.
	\end{align}
Since $\vec J$ is symmetric we get $\Erw[\tanh(\vec J \beta)^l]=0$ for odd $l$, while
$\Erw[\tanh(\vec J \beta)^l]\geq0$ and $X_l,Y_l\geq0$ for even $l$.
Hence, (\ref{eqkspinPOS}) follows from the elementary fact that $x^k-kxy^{k-1}+(k-1)y^k\geq0$ for all $x,y\geq0$.

Moving on to {\bf MIN}, we assume that $k$ is even.
Suppose that $\rho\in\cR(\Omega)$ is a distribution on $\Omega\times\Omega$ with uniform marginals and let $\alpha=\rho(1,1)+\rho(-1,-1)$.
Then $\rho(1,1)=\rho(-1,-1)=\alpha/2$, $\rho(1,-1)=\rho(-1,1)=(1-\alpha)/2$ and because $\vec J$ is symmetric,
	\begin{align*}
	\sum_{\sigma,\tau\in\Omega^k}\Erw\brk{\psi_{\vec J,\beta}(\sigma)\psi_{\vec J,\beta}(\tau)}\prod_{i=1}^k\rho(\sigma_i,\tau_i)&=
		1+\Erw[\tanh(\beta\vec J)^2]\bc{\sum_{\sigma,\tau\in\Omega}\sigma\tau\rho(\sigma,\tau)}^k=
			1+\Erw[\tanh(\beta\vec J)^2](2\alpha-1)^k.
	\end{align*}
Because $k$ is even, the last expression is convex with the minimum attained at $\alpha=1/2$, viz.\ $\rho=\bar\rho$.
\end{proof}

\Lem~\ref{Lemma_pSpin} shows not only that the $k$-spin model from \Sec~\ref{Sec_intro_kspin} with a standard Gaussian $\vec J$ satisfies  {\bf SYM}, {\bf BAL},{\bf POS} and {\bf MIN}, but that the same is true if $\vec J$ is the uniform distribution on $\{\pm1\}$.
This model is known as the $k$-XORSAT model in computer science. It is intimately related to low-density generator matrix codes~\cite{Abbe}.

\begin{proof}[Proof of \Thm~\ref{Thm_SK1}]
Comparing (\ref{eqkSpin1}) and (\ref{eqcosh}), we see that
	\begin{align*}
	\frac1n\Erw[\ln Z_\beta(\HH,\J)] 
		&=\frac1n\Erw\brk{\sum_{e\in E(\HH)}\ln\cosh(\beta\vec J_e)}+
		\frac1n\Erw\brk{\ln\sum_{\tau\in\{\pm1\}^{V_n}}\prod_{e\in E(\HH)}1+\tanh\bc{\beta\vec J_e\prod_{y\in e}\tau(y)}}\\
		&=\frac d{\sqrt{2\pi}k}\int_{-\infty}^\infty\ln(\cosh(z))\exp(-z^2/2)\dd z
			+\frac1n\Erw\brk{\ln Z(\G)|\fS}.
	\end{align*}
Therefore, \Thm~\ref{Thm_SK1} follows from \Thm~\ref{Thm_cond} and \Lem~\ref{Lemma_pSpin}.
\end{proof}

\begin{proof}[Proof of \Thm~\ref{Thm_SK2}]
Equations (\ref{eqkSpin1}) and (\ref{eqcosh}) ensure that the Gibbs measures $\mu_{\HH,\vec J,\beta}$ and $\mu_{\G}$ given $\fS$ are identically distributed.
Hence, \Thm~\ref{Thm_SK2} follows from \Thm~\ref{Thm_overlap} and  \Lem~\ref{Lemma_pSpin}.
\end{proof}

\subsection{Discussion and related work}\label{Sec_results4}
The results in this section provide a map of the replica symmetric phase,
its boundary and the evolution of the Gibbs measure within it, thereby vindicating  for a universal class of models the predictions of the cavity method~\cite{pnas}.
The results extend, complement or generalize prior work on the condensation phase transition from~\cite{CKPZ}, which only dealt with the case that the support $\Psi$ of $P$ is finite, and on the reconstruction problem~\cite{GM,montanari2011reconstruction}.
Additionally, in the example of the Potts antiferromagnet and the stochastic block model prior work based on combinatorial methods
only gave approximate results~\cite{Banks,Nor}, whereas the present results are tight for all values of $q,\beta$.
Indeed, a merit of the present approach is that we perform fairly abstract arguments that do not require model-specific deliberations.

Beyond the examples treated explicitly in \Sec~\ref{Sec_examples} there are several other important and well-studied models that
also satisfy the assumptions of our main results.
For instance, Bapst, Coja-Oghlan and Ra\ss mann~\cite{h2c} obtained approximate results on the replica symmetry breaking phase transition in the random hypergraph $2$-coloring problem.
This model is easily seen to satisfy {\bf SYM}, {\bf BAL}, {\bf POS} and {\bf MIN} and thus the main results of the present paper clarify the structure
of the entire replica symmetric phase.
More generally, the hypergraph version of the Potts model satisfies our assumptions as well.
So does the random $k$-NAESAT model, a variant of Boolean satisfiability that resembles the hypergraph $2$-coloring model.

Apart from proving an upper bound on the condensation threshold, the Kesten-Stigum bound plays an important role with respect to statistical inference aspects
of random factor graph models.
Specifically, by extension of the predictions from~\cite{Decelle} for the stochastic block model, it seems natural to expect that there should be
efficient algorithms for both the detection problem and for recovering a non-trivial approximation to the ground truth in the teacher-student model for $d>\dKS$.
On the other hand, an intriguing question is whether for $\dc<d<\dKS$ these two problems may be soluble in exponential time but not efficiently, i.e.,
in polynomial time~\cite{Banks,Decelle}.
Indeed, while \Thm~\ref{Thm_cond} shows that $\dc$ is always finite, there are models where $\dKS=\infty$, e.g., the $k$-XORSAT model.
Thus, for such models there might be an enormous computational gap.
This question is intimately related to the $k$-SAT refutation problem, an important question in computer science~\cite{Feige,Feldman2015}.

There are a few models that fail to satisfy our assumptions. 
For instance, in the random $k$-SAT model~\cite{ANP}
and the hardcore model on the \Erdos-\Renyi\ random graph~\cite{Bandyopadhyay} condition {\bf SYM} is violated.
Indeed, in these two cases the Gibbs marginals are non-uniform in the replica symmetric phase.
In effect, we do not expect that the free energy is as tightly concentrated as \Thm~\ref{Thm_SSC} shows it is in the case of ``symmetric'' models.
Thus, it is not just that the present proof methods do not apply, but ``asymmetric'' models appear to be materially different.
Moreover, ferromagnetic models generally violate {\bf SYM}, {\bf BAL} and {\bf POS}.

A further class of models that we do not treat in this paper is models where the weight functions $\psi$ take values in $\{0,1\}$, thus imposing hard constraints.
An example of this is the ``zero-temperature'' version of the Potts antiferromagnet, better known as the random graph coloring problem~\cite{ANP}.
Certain specific models with hard constraints have received considerable attention in combinatorics.
For example, \cite{Cond,Silent, Feli2} established the precise condensation threshold, a contiguity result and the exact limiting distribution
of the number of $q$-colorings of the \Erdos-\Renyi\ random graph via combinatorial methods under the assumption that $q$ exceeds a large enough constant.
(Subsequently the condensation threshold in the random graph coloring problem was determined for all $q\geq3$~\cite{CKPZ}.)
Similar results, albeit not quite up to the precise condensation threshold, 
are know for the hypergraph $2$-coloring and the $k$-NAESAT problems~\cite{nae,AchMooreHyp2,Feli},
a version of the random $k$-SAT problem with regular literal degrees~\cite{COW} and the independent set problem in random regular graphs~\cite{Bhatnagar}.
Additionally, in zero temperature models the `satisfiability threshold' from where $Z(\G)$ is typically equal to $0$
plays a major role~\cite{Dimitris,yuval,DSS1,DSS3,Greenhill,CrisExact}.

\section{Proof strategy}\label{Sec_Proofs}

\noindent
\emph{Throughout this section we keep the notation from \Sec~\ref{Sec_results}.}

\medskip\noindent
The apex of the present work is \Thm~\ref{Thm_SSC} about the limiting distribution of the free energy;
all the other results either lead up to it or derive from it relatively easily.
The classical approach to proving such a result would be the second moment method, pioneered
in this context by Achlioptas and Moore~\cite{nae}, in combination with the small subgraph conditioning technique
of Robinson and Wormald~\cite{Janson,RobinsonWormald}.
This strategy  was applied to, e.g., the stochastic block model~\cite{Banks} 
and the $k$-spin model~\cite{GT}.
But only in the stochastic block model with two colors and the diluted $2$-spin model was it possible to obtain complete results~\cite{GT,mossel2013proof}.
Indeed, as noticed by Guerra and Toninelli~\cite{GT}, 
a combinatorial second moment computation generally appears to be too crude a device to cover the entire replica symmetric phase.

Therefore, here we pursue a different strategy. 
We craft a proof around the teacher-student model $\G^*$.
More specifically, the main achievement of the recent paper~\cite{CKPZ} was to verify the cavity formula for the leading order $\lim_{n\to\infty}\frac1n\Erw[\ln Z(\G^*)]$ of the free energy in the teacher-student model (in the case that the set $\Psi$ is finite).
We will replace the second moment calculation by that free energy formula, generalized to infinite $\Psi$,
and combine it with a suitably generalized small subgraph conditioning technique.
The challenge is to integrate these two components seamlessly.
We accomplish this by realizing that, remarkably, both arguments are inherently and rather elegantly tied together
via the spectrum of the linear operator $\Xi$ from (\ref{eqXi}).
But to develop this novel approach we first need to recall the classical second moment argument and understand why it founders.

\subsection{Two moments do not suffice}
For any second moment calculation it is crucial to fix the number of constraint nodes because its fluctuations would otherwise boost the variance.
Hence, we will work with a deterministic integer sequence $m=m(n)\geq0$.
More precisely, we will fix $d>0$ and consider specific integer sequences $m=m(n)\geq0$ is such that $|m(n)-dn/k|\leq n^{3/5}$ for all $n$.
Let $\cM(d)$ be the set of all such sequences.

The second moment method rests on showing that $\Erw[Z(\G(n,m))^2]$ is of the same order of magnitude as the square 
	$\Erw[Z(\G(n,m))]^2$ of the first moment.
If so, then standard concentration results
	can be used to show that $\lim_{n\to\infty}\frac1n\Erw[\ln Z(\G(n,m))]=\lim_{n\to\infty}\frac1n\ln\Erw[Z(\G(n,m))]$.
The second limit is easy to compute because the expectation sits inside the logarithm, and thus we  obtain
the leading order of the free energy.

In fact, if we can calculate the second moment $\Erw[Z(\G(n,m))^2]$ sufficiently accurately,
then it may be possible to determine the limiting distribution of $\ln Z(\G(n,m))$ precisely.
For suppose that there is a ``simple'' random variable $Q(\G(n,m))$ such that
	\begin{align}\label{eqSSCMotivation}
	\Var[Z(\G(n,m))]&=(1+o(1))\Var[\Erw[Z(\G(n,m))|Q(\G(n,m))]].
	\end{align}
Then the basic formula
	$\Var[Z(\G(n,m))]=\Var[\Erw[Z(\G(n,m))|Q(\G(n,m))]]+\Erw[\Var[Z(\G(n,m))|Q(\G(n,m))]]$
implies
	\begin{equation}\label{eqSSCMotivation2}
	\Erw[\Var[Z(\G(n,m))|Q(\G(n,m))]]=o(\Erw[Z(\G(n,m))]^2)
	\end{equation}
and typically it is not difficult to deduce from (\ref{eqSSCMotivation2})  that
$\ln Z(\G(n,m))-\ln\Erw[Z(\G(n,m))|Q(\G(n,m))]$ converges to $0$ in probability.
Hence, if $Q(\G(n,m))$ is ``reasonable enough'' so that the law of $\ln\Erw[Z(\G(n,m))|Q(\G(n,m))]$ is easy to express, 
then we have got the limiting distribution of $\ln Z(\G(n,m))$.
The basic insight behind the small subgraph conditioning technique is that (\ref{eqSSCMotivation}) sometimes holds with a variable $Q$ 
that is determined by the statistics of bounded-length cycles in $\G(n,m)$~\cite{Janson,RobinsonWormald}.

Anyhow, the crux of the entire argument is to calculate $\Erw[Z(\G(n,m))^2]$.
Of course, by the linearity of expectation and the independence of the constraint nodes, the second moment can be written in terms of the overlap $\rho_{\sigma,\tau}$ as
	\begin{align}\nonumber
	\Erw[Z(\G(n,m))^2]&=\sum_{\sigma,\tau\in\Omega^{V_n}}\Erw\brk{\prod_{i=1}^m
		\psi_{a_i}(\sigma(\partial_1a_i),\ldots,\sigma(\partial_ka_i))\psi_{a_i}(\tau(\partial_1a_i),\ldots,\tau(\partial_ka_i))}\\
		&=\sum_{\sigma,\tau\in\Omega^{V_n}}\bc{\sum_{s,t\in\Omega^k}\Erw[\PSI(s)\PSI(t)]\prod_{i=1}^k\rho_{\sigma,\tau}(s_i,t_i)}^m.
			\label{eqnotosmm1}
	\end{align}
Given a probability distribution $\rho=(\rho(s,t))_{s,t\in\Omega}$ on $\Omega^2$ such that $n\rho(s,t)$ is integral for all $s,t\in\Omega$,
the number of assignments $\sigma,\tau\in\Omega^{V_n}$ with $\rho_{\sigma,\tau}=\rho$ equals $\bink n{\rho n}$.
Therefore, Stirling's formula yields the approximation
	\begin{align}\label{eqnotosmm2}
	\ln\Erw[Z(\G(n,m))^2]&=\max_{\rho\in\cP(\Omega^2)}n\cH(\rho)+m\ln\bc{\sum_{s,t\in\Omega^k}\Erw[\PSI(s)\PSI(t)]\prod_{i=1}^k\rho(s_i,t_i)}
		+O(\ln n),
	\end{align}
where $\cH(\rho)$ denotes the entropy of $\rho$.
In other words, computing the second moment comes down to identifying the overlap $\rho$ that renders the dominant contribution to (\ref{eqnotosmm1}).
By comparison, under assumptions {\bf SYM} and {\bf BAL} it is not difficult to see (cf.\ \Lem~\ref{Cor_F} below) that the first moment satisfies 
	\begin{align}\label{eqnotosmm3}
	\ln\Erw[Z(\G(n,m))]=n\ln q+m\ln \xi+O(\ln n).
	\end{align}

But there are two major issues with the second moment argument.
First, actually solving the innocent-looking optimization problem (\ref{eqnotosmm2}) turns out to be daunting even in special cases.
For example, in the Potts antiferromagnet the task remains wide open, despite very serious attempts~\cite{AchNaor,Nor}.
The source of the trouble is that the entropy is concave while the second summand in (\ref{eqnotosmm2}) is convex (cf.\ {\bf MIN}), 
causing a proliferation of local maxima.
Second, and even worse, comparing (\ref{eqnotosmm2}) and (\ref{eqnotosmm3}) we can verify easily that the desired second moment bound
$\Erw[Z(\G(n,m,P)^2]=O(\Erw[Z(\G(n,m,P)]^2)$ can hold only if the maximizer $\rho_\star$ of (\ref{eqnotosmm2})  satisfies $\TV{\rho_\star-\bar\rho}=o(1)$.
However, this is not generally true for average degrees $d$ below but near the condensation threshold.
For instance, in the Potts antiferromagnet the second moment exceeds the square of the first moment by an exponential factor $\exp(\Omega(n))$ for
$d$ below the condensation threshold~\cite{Nor}. 

{The problem was noticed   and partly remedied in prior work by applying the second moment method to a suitably truncated random variable (e.g.~\cite{Cond,Nor}).
This method revealed, e.g., the condensation threshold in a few special cases such as
	the random graph $q$-coloring problem~\cite{Cond}, albeit only for $q$ exceeding some (astronomical) constant $q_0$, and
	in the random regular $k$-SAT model for large $k$~\cite{VictorSAT}.
Yet apart from introducing such extraneous conditions, 
ad-hoc arguments of this kind tend to require a meticulous combinatorial study of the specific model.}

\subsection{The condensation phase transition and the overlap}\label{Sec_outline_cond}
The merit of the present approach is that we avoid combinatorial deliberations altogether.
Rather than bothering with the second moment bound~(\ref{eqnotosmm2}) we will employ an asymptotic formula for the free energy
of the teacher-student model $\G^*$.
To be precise, it will be convenient to work with a slightly tweaked version $\hat\G$ of this model:
following~\cite[\Sec~3]{CKPZ}, we let $\hat\G(n,m,P)$ be the random factor graph chosen from the distribution
	\begin{align}\label{eq:NishimoriG}
	\pr\brk{\hat\G(n,m,P)\in\cA}&=\frac{\Erw[Z(\G(n,m,P))\vecone\{\G(n,m,P)\in\cA\}]}{\Erw[Z(\G(n,m,P))]}&\mbox{for any event }\cA.
	\end{align}
Recalling that $\vm=\vm_d(n)$ is a random variable with distribution $\Po(dn/k)$, we also introduce  $\hat\G=\hat\G(n,\vm,P)$.
As before we ease the notation by dropping $P$ where possible.

Loosely speaking $\hat\G(n,m)$ is a reweighted version of $\G(n,m)$ where the probability that $G$ comes up is proportional to $Z(G)$.
Intuitively, the construction of the teacher-student model $\G^*$ induces a similar reweighing as 
the probability that $\G^*=G$ depends on the number of assignments $\SIGMA^*$
that could plausibly be used to generate $G$ via (\ref{eqTeacher}).
In fact, as we shall see in \Sec~\ref{Sec_prelims} it is not difficult to verify the following.

\begin{lemma}\label{Prop_contig}
If $P$ satisfies conditions {\bf SYM} and {\bf BAL}, then $\G^*(n,m,\SIGMA^*)$ and $\hat\G(n,m)$ are mutually contiguous for all $d>0$, $m\in\cM(d)$.
\end{lemma}

The following theorem verifies the cavity formula for the free energy of $\hat\G$ and $\G^*$.

\begin{theorem}\label{Thm_plantedFreeEnergy}
Assume that $P$ satisfies {\bf SYM}, {\bf BAL} and {\bf POS} and let $d>0$.
Then with $\cB(d,P,\pi)$ from (\ref{eqMyBethe}) we have
	\begin{align*}
	\lim_{n\to\infty}\frac1n\Erw[\ln Z(\G^*)]&=\lim_{n\to\infty}\frac1n\Erw[\ln Z(\hat\G)]
		=\sup_{\pi\in\Pomast} \cB(d,P,\pi).
	\end{align*}
\end{theorem}

\noindent
\Thm~\ref{Thm_plantedFreeEnergy} was established in~\cite{CKPZ} for the case that the set $\Psi$ of weight functions is finite.
In \Sec~\ref{Sec_Thm_plantedFreeEnergy} we extend that results via a limiting argument
to prove \Thm~\ref{Thm_plantedFreeEnergy} for infinite $\Psi$.
Furthermore, in \Sec~\ref{sec:ProofPreCond} we deduce the following result from \Thm~\ref{Thm_plantedFreeEnergy}.

\begin{proposition}\label{prop:belowcond-unif}
Assume that  {\bf BAL}, {\bf SYM}, {\bf POS} and {\bf MIN} hold and that $d<\dc$.
{There exists a sequence $\zeta=\zeta(n)$, $\zeta(n)=o(1)$ but $n^{1/6}\zeta(n)\to\infty$ as $n\to\infty$, such that for all $m\in\cM(d)$ we have}
	\begin{equation}\label{eq:overlapunif}
	\Erw\bck{\TV{\rho_{\SIGMA_1,\SIGMA_2}-\bar\rho}}_{\hat\G(n,m)}\leq\zeta^2.
	\end{equation}
\end{proposition}

\Prop~\ref{prop:belowcond-unif} resolves our second moment troubles.
Indeed, the proposition enables a completely generic way of setting up a truncated second moment argument: 
with $\zeta$ from \Prop~\ref{prop:belowcond-unif} we define
	\begin{align}\label{eqcZeps}
	\cZ(G)&=Z(G)\vecone\cbc{\bck{\TV{\rho_{\SIGMA_1,\SIGMA_2}-\bar\rho}}_G\leq\zeta}.
	\end{align}
Hence, $\cZ(G)=Z(G)$ if ``most'' pairs $\SIGMA_1,\SIGMA_2$ drawn from $\mu_G$ have overlap close to $\bar\rho$, and $\cZ(G)=0$ otherwise.
\Prop~\ref{prop:belowcond-unif} shows immediately that the truncation does not diminish the first moment.

\begin{corollary}\label{cor:belowcond-unif}
If {\bf BAL}, {\bf SYM}, {\bf POS} and {\bf MIN} hold and $d<\dc$, then $\Erw[\cZ(\G(n,m))]\sim\Erw[Z(\G(n,m))]$ uniformly for all $m\in\cM(d)$.
\end{corollary}
\begin{proof}
Equation (\ref{eq:NishimoriG})  and \Prop~\ref{prop:belowcond-unif} yield
	\begin{align*}
	\Erw[\cZ(\G(n,m))]&=\Erw[Z(\G(n,m))]\cdot\pr\brk{\bck{\TV{\rho_{\SIGMA_1,\SIGMA_2}-\bar\rho}}_{\hat\G(n,m)}\leq\zeta}
		=(1+o(1))\Erw[Z(\G(n,m))],
	\end{align*}
as claimed.
\end{proof}

The second moment calculation for $\cZ$ is easy, too.
Indeed,  the very construction (\ref{eqcZeps}) of $\cZ$ guarantees that the dominant contribution to the second moment of $\cZ$ comes from pairs with an overlap close to $\bar\rho$.
Hence, computing the second moment comes down to expanding the right hand side of (\ref{eqnotosmm2}) around $\bar\rho$ via the Laplace method.
Yet in order to apply the Laplace method we need to verify that $\bar\rho$ is a local maximum of the function
	\begin{equation}\label{eqMyKS2}
	\rho\in\cP(\Omega^2)\mapsto\cH(\rho)+\frac dk\ln{\sum_{s,t\in\Omega^k}\Erw[\PSI(s)\PSI(t)]\prod_{i=1}^k\rho(s_i,t_i)}
	\end{equation}
 from (\ref{eqnotosmm2}).
{For the special case of the Potts antiferromagnet the overlap concentration (\ref{eq:overlapunif}) was established and the second moment argument for $\cZ$ was carried out in~\cite[\Sec~4.3]{CKPZ}.
While the generalization to random factor graph models is anything but straightforward, an even more important difference lies in the application of the Laplace method.
More specifically, in the case of the Potts antiferromagnet the fact that $\bar\rho$ is a local maximum of (\ref{eqMyKS2}) for all $d<\dc$ was derived extremely indirectly
by resorting to the statistical inference algorithm of Abbe and Sandon for the stochastic block model~\cite{abbe2015detection}.
But of course there ought to be a general, conceptual explanation.}
{As we shall see momentarily, there is one indeed, namely the generalized Kesten-Stigum bound.

\subsection{The Kesten-Stigum bound}\label{Sec_outline_KS}
To see the connection, we observe that the Hessian of (\ref{eqMyKS2}) at the point $\bar\rho$ is equal to $q(\id-d(k-1)\Xi)$
(with $\Xi$ the matrix from (\ref{eqXi})).
Hence, taking into account that the argument $\rho$ is a probability distribution on $\Omega\times\Omega$,
we find that  $\bar\rho$ is a local maximum of (\ref{eqMyKS2})  if and only if
	\begin{equation}\label{eqMyKS}
	\scal{(\id-d(k-1)\Xi)x}x>0\qquad\mbox{for all }x\in\RR^q\tensor\RR^q\mbox{ such that }x\perp\vecone\tensor\vecone.
	\end{equation}
In order to get a handle on the spectrum of the operator $\Xi$ from (\ref{eqXi})
we begin with the following observation about the matrices $\Phi_\psi$ and $\Phi$ from (\ref{eqPhiMatrices}) and (\ref{eqPhi}).}

\begin{lemma}\label{Lemma_Phi}
Assume that $P$ satisfies {\bf SYM}.
Then the matrix $\Phi_\psi$ is stochastic and thus $\Phi_\psi\vecone=\vecone$ for every $\psi\in\Psi$.
Moreover, $\Phi$ is symmetric and doubly-stochastic.
If, additionally, $P$ satisfies {\bf BAL}, then
	$\max_{x\perp\vecone}\scal{\Phi x}x\leq0.$
\end{lemma}

\noindent
Proceeding to the operator $\Xi$, we recall the definition of the space $\cE$ from (\ref{eqSpaceE}) and we introduce
	\begin{equation}\label{eqcE'}
	\cE'=\{x\in\RR^q\tensor\RR^q:\scal x{\vecone\tensor\vecone}=0\}\supset\cE.
	\end{equation}

\begin{lemma}\label{Lemma_Xi}
Assume that $P$ satisfies {\bf SYM} and {\bf BAL}.
The operator $\Xi$ is self-adjoint, $\Xi(\vecone\tensor\vecone)=\vecone\tensor\vecone$ and for every $x\in\RR^q$ we have
$\Xi (x\tensor\vecone)=(\Phi x)\tensor\vecone$, $\Xi (\vecone\tensor x)=\vecone\tensor(\Phi x)$ and
	\begin{align}\label{eqLemma_Xi}
	\scal{\Xi (x\tensor\vecone)}{x\tensor\vecone}&\leq0,&
		\scal{\Xi (\vecone\tensor x)}{\vecone\tensor x}&\leq0&\mbox{if }x\perp\vecone.
	\end{align}
Furthermore, $\Xi\cE\subset\cE$ and $\Xi\cE'\subset\cE'$.
\end{lemma}

\Lem~\ref{Lemma_Xi} shows that $\Xi$ induces a self-adjoint operator on the space $\cE$.
The following proposition yields a bound on the spectral radius of this operator.
Let
\begin{align}
\Eig[\Xi] = \left\{\lambda \in \RR : \exists x \in \cE\setminus\cbc0:\Xi x = \lambda x\right\}. \label{eq:IntroSetEigenvals}
\end{align}

\begin{proposition}\label{prop_KS}
If $P$ satisfies {\bf SYM} and {\bf BAL}, then  $\dc(k-1)\max_{\lambda\in\Eig[\Xi]}|\lambda| \leq 1.$
\end{proposition}

\noindent
The proof of \Prop~\ref{prop_KS}, which is based on highlighting an inherent connection between the spectrum of $\Xi$ and
the Bethe free energy functional $\cB$ from (\ref{eqMyBethe}), is the main technical achievement of this paper.
The details can be found in \Sec~\ref{Sec_prop_KS}.
Let us observe that \Thm~\ref{Thm_KS} is immediate from \Prop~\ref{prop_KS}.

\begin{proof}[Proof of \Thm~\ref{Thm_KS}]
We have $\max_{x\in\cE:\|x\|=1}\scal{\Xi x}x=\max_{\lambda\in\Eig[\Xi]}|\lambda|$
because \Lem\ \ref{Lemma_Xi} shows that $\Xi$ is self-adjoint.
Therefore, \Thm~\ref{Thm_KS} follows from \Prop~\ref{prop_KS}.
\end{proof}

\Lem~\ref{Lemma_Xi} and \Prop~\ref{prop_KS} show that (\ref{eqMyKS}) is satisfied,
and thus that $\bar\rho$ is a local maximum of (\ref{eqMyKS2}), for all $d<\dc$.
Indeed, it is immediate from (\ref{eqLemma_Xi}) that $\scal{(\id-d(k-1)\Xi)x}x>0$ if $x$ is of the form $\vecone\tensor y$ or $y\tensor\vecone$
for some $\vecone\perp y\in\RR^q$, and
\Thm~\ref{Thm_KS} shows that $\scal{(\id-d(k-1)\Xi)x}x>0$ for all $x\in\cE$.
Hence, \Prop~\ref{prop_KS} provides the link between the free energy calculation for the reweighted model $\hat\G$ and the second moment of $\cZ$.

\subsection{Second moment redux}\label{Sec_outline_smm}
We begin by deriving the following asymptotic formula for the first moment in \Sec~\ref{Sec_moments}.
Observe that by \Lem~\ref{Lemma_Phi} the set $\eig\bc\Phi$ of eigenvalues of $\Phi$ contains precisely one non-negative element, namely $1$.
Therefore, the following formula makes sense.

\begin{proposition} \label{lem:FirstMoment}
Suppose that $P$ satisfies {\bf SYM} and {\bf BAL} and let $0<d$.
Then uniformly for all $m\in\cM(d)$,
	\begin{equation}\label{eq:FirstMoment}
	\Erw[Z(\G(n,m))] \sim \frac{q^{n+\frac12}\xi^m} { \prod_{\lambda\in\eig(\Phi)\setminus\cbc 1}\sqrt{1-d(k-1)\lambda}}. 
	\end{equation}
\end{proposition}

\noindent
Proceeding to the second moment, we recall from \Lem~\ref{Lemma_Xi} that $\Xi$ induces an endomorphism on the subspace $\cE'$ from (\ref{eqcE'}) and 
we write
	$$\eig'(\Xi)=\{\lambda\in\RR:\exists x\in\cE'\setminus\cbc0:\Xi x=\lambda x\}$$
for the spectrum of $\Xi$ on $\cE'$.
\Lem~\ref{Lemma_Xi} and \Prop~\ref{prop_KS} imply that $\dc(k-1)\lambda\leq1$ for all $\lambda\in\eig'(\Xi)$.
Therefore, the following formula for the second moment, whose proof we defer to \Sec~\ref{Sec_moments}, makes sense as well.

\begin{proposition}\label{lem:SecondMoment}
Suppose that $P$ satisfies {\bf SYM} and {\bf BAL} and let $0<d<\dc$.
Then {uniformly} for all $m\in\cM(d)$,
	\begin{equation}\label{eq:SecondMoment}
	\Erw[\cZ(\G(n,m))^2] \leq \frac{(1+o(1))q^{2n+1}{\xi^{2m}}}
 	{ \prod_{{\lambda\in\eig'(\Xi)}}\sqrt{1-d(k-1)\lambda}}. 
	\end{equation}
\end{proposition}

\noindent
Combining \Cor~\ref{cor:belowcond-unif} with \Prop s~\ref{lem:FirstMoment} and~\ref{lem:SecondMoment} and applying \Lem~\ref{Lemma_Xi}, we obtain
for $m\in\cM(d)$,
	\begin{align}\label{eqAtTheEndOfTheDay}
	\frac{\Erw[\cZ(\G(n,m))^2]}{\Erw[\cZ(\G(n,m))]^2}&\sim\frac{\prod_{\lambda\in\eig(\Phi)\setminus\cbc 1}1-d(k-1)\lambda}
		{\prod_{\lambda\in\eig'(\Xi)}\sqrt{1-d(k-1)\lambda}}=\prod_{\lambda\in\Eig[\Xi]}\frac1{\sqrt{1-d(k-1)\lambda}}
		&\mbox{if }d<\dc.
	\end{align}
In particular, the ratio of the second moment and the square of the first is bounded as $n\to\infty$.

\subsection{Virtuous cycles}\label{Sec_outline_cycles}
In order to determine the limiting distribution of $\ln Z(\G(n,m))$ we are going to ``explain'' the remaining variance of $\cZ(\G(n,m))$ in terms 
of the statistics of the bounded-length cycles of $\G(n,m)$.
However, by comparison to prior applications of the small subgraph conditioning technique, here it does not suffice to merely record how many cycles of a given length occur.
We also need to take into account the specific weight functions along the cycle.
Yet this approach is complicated substantially by the fact that there may be infinitely many different weight functions.
To deal with this issue we are going to discretize the set of weight functions and perform a somewhat delicate limiting argument.

We need a few definitions.
A {\em signature of order $\ell$} is a family $$Y=(E_1,s_1,t_1,E_2,s_2,t_2,\ldots,E_\ell,s_\ell,t_\ell)$$ such that
 $E_1,\ldots,E_\ell\subset\Psi$ are events, $s_1,t_1,\ldots,s_\ell,t_\ell\in\{1,\ldots,k\}$ and $s_i\neq t_i$ for all $i\in\{1,\ldots,\ell\}$ and $s_1<t_1$ if $\ell=1$.
Let $\cY_\ell$ be the set of all signatures of order $\ell$, let $\cY_{\leq\ell}=\bigcup_{l\leq\ell}\cY_l$
	 and let $\cY=\bigcup_{\ell\geq1}\cY_\ell$ be the set of all signatures.
If $G$ is a factor graph with variable nodes $V_n$ and constraint nodes $F_m$, then we call a family
$(x_{i_1},a_{h_1},\ldots,x_{i_\ell},a_{h_\ell})$ a \emph{cycle of signature $Y$ in $G$} if the following conditions are satisfied.
	\begin{description}
	\item[CYC1] $i_1,\ldots,i_\ell\in\{1,\ldots,n\}$ are pairwise distinct and $i_1=\min\{i_1,\ldots,i_\ell\}$,
	\item[CYC2] $h_1,\ldots,h_\ell\in\{1,\ldots,m\}$ are pairwise distinct and $h_1<h_\ell$ if $\ell>1$,
	\item[CYC3] $\psi_{a_{h_j}}\in E_j$ for all $j\in\{1,\ldots,\ell\}$,
	\item[CYC4] $\partial_{s_j} a_{h_j}=x_{i_j}$ for all $j\in\{1,\ldots,\ell\}$,
		 $\partial_{t_j} a_{h_j}=x_{i_{j+1}}$ for all $j<\ell$ and $\partial_{t_\ell} a_{h_\ell}=x_{i_{1}}$.
	\end{description}
Conditions {\bf CYC1}-- {\bf CYC2} provide that the variable nodes that the cycle passes through are pairwise distinct.
Moreover, to avoid over-counting {\bf CYC1} specifies that the cycle starts at the variable node with the smallest index and
{\bf CYC2} that from there the cycle is oriented towards the constraint node with the smaller index if $\ell>1$, respectively that $s_1<t_1$ if $\ell=1$.
Further, {\bf CYC3} states that the weight functions along the cycle belong to $E_1,\ldots,E_\ell$.
Finally, {\bf CYC4} ensures that the cycle enters the $j$th constraint node in position $s_j$ and leaves in position $t_j$.

Let $C_Y(G)$ denote the number of cycles of signature $Y$.
Moreover, for an event $\mathcal{A}\subset\Psi$ with $\Pr(\mathcal{A})>0$ and $h,h'\in\{1,\ldots,k\}$ define the $q\times q$ matrix 
$\Phi_{\mathcal{A},h,h'}$ by letting
	\begin{equation}\label{eqCyclePhi}
	\Phi_{\mathcal{A} ,h,h'} (\omega,\omega') = q^{1-k}\xi^{-1}\sum_{\tau\in\Omega^k}\vecone\{\tau_h=\omega,\tau_{h'}=\omega'\}
 		\Erw[\PSI(\tau)| \mathcal{A} ]\qquad(\omega,\omega' \in \Omega).\end{equation}
In addition, for a signature $Y=(E_1,s_1,t_1,\ldots,E_\ell,s_\ell,t_\ell)$ define
\begin{align}\label{eqCycles}
	\kappa_{Y}&=\frac{1}{2\ell}\bcfr dk^\ell\prod_{i=1}^\ell P(E_i),&
		\Phi_Y&=\prod_{i=1}^\ell\Phi_{E_i,s_i,t_i},&\hat\kappa_Y&=\kappa_Y\Tr(\Phi_Y).
\end{align}
Further, two signatures $Y=(E_1,s_1,t_1,\ldots,E_\ell,s_\ell,t_\ell)$, $Y'=(E_1',s_1',t_1',\ldots,E_{\ell'}',s_{\ell'}',t_{\ell'}')$ are {\em disjoint} if either
$\ell\neq\ell'$, or $(s_i,t_i)\neq(s_i',t_i')$ for some $i$, or $E_i\cap E_i'=\emptyset$ for some $i$.
Finally, a {\em cycle of order $\ell$} is a family $(x_{i_1},a_{h_1},\ldots,x_{i_\ell},a_{h_\ell})$ that is a cycle of signature 
$(\Psi,s_1,t_1,\ldots,\Psi,s_\ell,t_\ell)$ for some sequence $s_1,t_1,\ldots,s_\ell,t_\ell$, and we let $C_\ell$ signify the number of such cycles.
The following is a basic fact from the theory of random graphs.

\begin{fact}[{\cite{BB}}]\label{lemma:cycles}
Let $\ell_1,\ldots,\ell_l\geq1$ be pairwise distinct integers and let $y_1,\ldots,y_l\geq0$ be integers.
Then for every $d>0$ uniformly for all $m\in\cM(d)$ we have
	$$\Pr\brk{\forall i\leq l: C_{\ell_i}(\G(n,m,P))=y_i}\sim \prod_{i=1}^l\Pr\brk{\Po\bc{\frac{((k-1)d)^{\ell_i}}{2\ell_i}}=y_i}$$
and the expected number of pairs of cycles of order at most $\ell_1+\cdots+\ell_l$ that share a common vertex is $O(1/n)$.
\end{fact}

\noindent
In \Sec~\ref{Sec_cycles} we establish the following enhancement that takes the weight functions along the cycles into account.

\begin{proposition}\label{prop:FirstCondOverFirst}
Suppose that $P$ satisfies {\bf SYM} and {\bf BAL}.
Let  $Y_1, Y_2, \ldots Y_l\in\cY$ be pairwise {disjoint} signatures and let $y_{1}, \ldots,y_l$ be non-negative integers.
Let $d>0$.
Then uniformly for all $m\in\cM(d)$,
	\begin{align}\label{eqCyclePoisson}
	\pr\brk{\forall t\le l:\ C_{Y_t}(\G(n,m))=y_{t}}&\sim\prod_{t=1}^l\pr\brk{\Po(\kappa_{Y_t})=y_t},&
	\pr\brk{\forall t\le l:\ C_{Y_t}(\hat\G(n,m))=y_{t}}&\sim\prod_{t=1}^l\pr\brk{\Po(\hat\kappa_{Y_t})=y_t}.
	\end{align}
Moreover,
	\begin{align*}
	\pr\brk{\G(n,m)\in\fS}&=\pr\brk{C_1(\G(n,m))+\vecone\{k=2\}C_2(\G(n,m))=0}+O(1/n)
			\sim\exp\bc{-d(k-1)/2-\vecone\{k=2\}d^2/4},\\
	\pr\brk{\hat\G(n,m)\in\fS}&=\pr\brk{C_1(\hat\G(n,m))+\vecone\{k=2\}C_2(\hat\G(n,m))=0}+O(1/n)
		\sim\exp\bc{-\frac{d(k-1)}2\Tr(\Phi)-\frac{\vecone\{k=2\}d^2}4\Tr(\Phi^2)}.
	\end{align*}
\end{proposition}

\noindent
Thus, for disjoint $Y_1,\ldots,Y_l$ the cycle counts $C_{Y_t}$ are asymptotically independent Poisson.

Equipped with \Prop s~\ref{lem:FirstMoment}, \ref{lem:SecondMoment} and~\ref{prop:FirstCondOverFirst},
in the case that the set $\Psi$ of weight functions is finite we could determine
the limiting distribution of $\ln Z(\G)$ and thus prove \Thm~\ref{Thm_SSC} by just applying Janson's version of the small subgraph conditioning theorem~\cite{Janson}.
However, to accommodate an infinite set of weight functions like in the $k$-spin model
a discretization of $\Psi$ and a limiting argument are required.
Specifically, recall that
	$$\Psi\subset[0,2]^{\Omega^k}$$
and for an integer $r\geq1$ let $\fC_r$ be the partition of $\Psi$ induced by slicing the cube $[0,2]^{\Omega^k}$ into pairwise disjoint sub-cubes of side length $1/r$.
Further, let $\cY_{\ell,r}$ denote the set of all signatures $(E_1,s_1,t_1,\ldots,E_\ell,s_\ell,t_\ell)$
such that $E_1,\ldots,E_\ell\in\fC_r$ and such that $\pr(E_i)>0$ for all $i\le\ell$, and define $\cY_{\leq\ell,r}=\bigcup_{l=1}^{\ell}\cY_{l,r}$.
Furthermore, if $\psi\in\Psi$ belongs to a sub-cube $C\in\fC_r$, then we let 
	\begin{align*}
	\psi^{(r)}(\tau)&=\Erw[\PSI(\tau)|C]&(\tau\in\Omega^k).
	\end{align*}
The following proposition, whose proof can be found in \Sec~\ref{Sec_totallyVexed}, establishes 
that the random variable $\cK$ from \Thm~\ref{Thm_SSC} is well-defined and that it
can be approximated arbitrarily well via the discretizations $\fC_r$.

\begin{proposition}\label{Lemma_totallyVexed}
Assume that $P$ satisfies {\bf SYM} and {\bf BAL} and let $0<d<\dc$.
Let $(K_l)_{l\geq1}$ be a family of independent Poisson variables with $\Erw[K_l]=(d(k-1))^l/(2l)$ and let $(\PSI_{l,i,j})_{l,i,j}$ be a
family of independent samples from $P$.
Furthermore, define
	\begin{align*}
	\cK_{\ell,r}&=
		\sum_{l=1}^\ell\brk{\frac{(d(k-1))^l}{2l}\bc{1-\Tr(\Phi^l)}+\sum_{i=1}^{K_{l}}\ln{\Tr\prod_{j=1}^l\Phi_{\PSI_{l,i,j}^{(r)}}}},&
	\cK_\ell&=\sum_{l=1}^\ell\brk{\frac{(d(k-1))^l}{2l}\bc{1-\Tr(\Phi^l)}+\sum_{i=1}^{K_{l}}\ln{\Tr\prod_{j=1}^l\Phi_{\PSI_{l,i,j}}}}
	\end{align*}
and 		$\cK=\sum_{\ell=1}^\infty\cK_\ell$.
Then all $\cK_{\ell,r}$ are uniformly bounded in the $L^1$-norm,
$\cK_{\ell,r}$ is $L^1$-convergent to $\cK_\ell$ as $r\to\infty$ and $\cK_\ell$ is $L^1$-convergent to $\cK$ as $\ell\to\infty$.
Furthermore,
	\begin{equation}\label{eqVarFormula}
	\lim_{\ell\to\infty}\lim_{r\to\infty}\exp{\sum_{Y\in\cY_{\leq\ell,r}}\frac{(\kappa_Y-\hat\kappa_Y)^2}{\kappa_Y}}
		=\prod_{\lambda\in\Eig[\Xi]}\frac1{\sqrt{1-d(k-1)\lambda}}.
	\end{equation}
\end{proposition}

\subsection{Small subgraph conditioning}
We have all the ingredients in place to prove \Thm~\ref{Thm_SSC}.
Thus, fix $0<d<\dc$ and let $m\in\cM(d)$.
Let $\fF_{\ell,r}=\fF_{\ell,r}{(n,m)}$ be the $\sigma$-algebra generated by the cycle counts $(C_Y)_{Y\in\cY_{\leq\ell,r}}$.
Following the small subgraph conditioning paradigm, we
intend to show that for sufficiently large $\ell,r$, with probability tending to $1$ as $n\to\infty$,
$Z(\G(n,m))$ is ``close'' to $\Erw[Z(\G(n,m))|\fF_{\ell,r}]$.
Since \Prop~\ref{lem:SecondMoment} shows that $\Erw[Z(\G(n,m))-\cZ(\G(n,m))]$ is small and that the second moment of $\cZ(\G(n,m))$ is under control,
we are going to argue via the truncated random variable.

More specifically, to show that $\cZ(\G(n,m))$ is ``close'' to $\Erw[\cZ(\G(n,m))|\fF_{\ell,r}]$ with probability $1-o(1)$ for sufficiently large $\ell,r$, we are going to prove that $\Erw[\Var(\Erw[\cZ(\G(n,m))|\cF_{\ell,r}])]$ is small.
Clearly,
	\begin{align}\label{eqTower}
	\Var[\cZ(\G(n,m))]=\Var(\Erw[\cZ(\G(n,m))|\fF_{\ell,r}])+\Erw[\Var(\Erw[\cZ(\G(n,m))|\fF_{\ell,r}])].
	\end{align}
Hence, to prove that $\Erw[\Var(\Erw[\cZ(\G(n,m))|\fF_{\ell,r}])]$ is small it suffices to show that
	\begin{equation}\label{eqTower'}
	\Var(\Erw[\cZ(\G(n,m))|\fF_{\ell,r}])=\Erw[\Erw[\cZ(\G(n,m))|\fF_{\ell,r}]^2]-\Erw[\cZ(\G(n,m))]^2
	\end{equation}
is nearly as big as $\Var[\cZ(\G(n,m))]$.
Given what we know at this point this is not particularly difficult.
Nonetheless, let us put the details off for just a little while to \Sec~\ref{Sec_tower}, where we prove the following.

\begin{lemma}\label{Claim_Tower2}
Suppose that $P$ satisfies {\bf SYM} and {\bf BAL} and let $0<d<\dc$.
For any $\eta>0$ there exists $\ell_0(\eta)$ such that for every $\ell>\ell_0(\eta)$ there exists $r_0(\eta,\ell)$ such that for all $r>r_0(\eta,\ell)$, 
uniformly for all $m\in\cM(d)$,
	$$\lim_{n\to\infty}\pr\brk{|\cZ(\G(n,m))-\Erw[\cZ(\G(n,m))|\fF_{\ell,r}]|>\eta \Erw[Z(\G(n,m))]}=0.$$
\end{lemma}

\begin{proof}[Proof of \Thm~\ref{Thm_SSC}]
Because $\cZ(\G(n,m))\leq Z(\G(n,m))$ and $\Erw[\cZ(\G(n,m))]\sim\Erw[Z(\G(n,m))]$ by \Cor~\ref{cor:belowcond-unif}, we have 
	$\Erw|\cZ(\G(n,m))-Z(\G(n,m))|=o(\Erw[Z(\G(n,m))]).$
Therefore,  \Lem~\ref{Claim_Tower2} implies that
	\begin{equation}\label{eq_Claim_Tower2}
	\lim_{n\to\infty}\pr\brk{|Z(\G(n,m))-\Erw[Z(\G(n,m))|\fF_{\ell,r}]|>\eta \Erw[Z(\G(n,m))]}=0.
	\end{equation}
Thus, we are left to determine the law of $\Erw[Z(\G(n,m))|\fF_{\ell,r}]$.
On this count, \Prop~\ref{prop:FirstCondOverFirst} shows that for any non-negative integer vector $(c_Y)_{Y\in\cY_{\le\ell,r}}$,
	\begin{align*}
	\frac{\Erw[Z(\G(n,m))|\forall Y\in\cY_{\leq\ell,r}:C_Y(\G(n,m))=c_Y]}{\Erw[Z(\G(n,m))]}&
		=\frac{\pr[\forall Y\in\cY_{\leq\ell,r}:C_Y(\hat\G(n,m))=c_Y]}{\pr\brk{\forall Y\in\cY_{\leq\ell,r}:C_Y(\G(n,m))=c_Y}}\\
		\sim\prod_{Y\in\cY_{\leq\ell,r}}\frac{\pr\brk{\Po(\hat\kappa_Y)=c_Y}}{\pr\brk{\Po(\kappa_Y)=c_Y}}
		&=\exp\bc{\sum_{Y\in\cY_{\leq\ell,r}}c_Y\ln(\Tr\Phi_Y)-(\hat\kappa_Y-\kappa_Y)}.
	\end{align*}
Hence, letting $K_{\ell,r}'(\G(n,m))=\sum_{Y\in\cY_{\leq\ell,r}} C_Y(\G(n,m))\ln(\Tr\Phi_Y)-(\hat\kappa_Y-\kappa_Y)$ we conclude that, in distribution,
	\begin{align}
	K_{\ell,r}(\G(n,m))&=\ln\Erw[Z(\G(n,m))|\fF_{\ell,r}]-\ln\Erw[Z(\G(n,m))
		]\ 
		\stacksign{$n\to\infty$}\to\  K_{\ell,r}'(\G(n,m))\label{eqProofThm_SSC1}.
	\end{align}
Further, by (\ref{eqCycles}) 
	\begin{align*}
	K_{\ell,r}'(\G(n,m))&=\sum_{l=1}^\ell\brk{\frac{(d(k-1))^l}{2l}(1-\Tr(\Phi^l))+\sum_{Y\in\cY_{l,r}}C_Y(\G(n,m))\ln\Tr\Phi_Y}.
	\end{align*}
Thus, combining \Prop s~\ref{prop:FirstCondOverFirst} and~\ref{Lemma_totallyVexed},
we conclude that $K_{\ell,r}'(\G(n,m))$ converges to $\cK_{\ell,r}$ in distribution as $n\to\infty$ for every $\ell,r$.
Hence, due to (\ref{eqProofThm_SSC1}) so does $K_{\ell,r}(\G(n,m))$.
Consequently, 
	\Prop~\ref{Lemma_totallyVexed}   and \eqref{eq_Claim_Tower2} show
	that for any bounded continuous function $g:\RR\to\RR$,
	\begin{align*}
	\forall\eps>0\exists\ell_0(\eps)\forall\ell\geq\ell_0(\eps)\exists r_0(\eps,\ell)\forall r>r_0(\eps,\ell)&:
		\limsup_{n\to\infty}\Erw[g(\cK)]-\Erw[g(K_{\ell,r}(\G(n,m)))]<\eps,\\
	\forall\eps>0\exists\ell_0'(\eps)\forall\ell\geq\ell_0'(\eps)\exists r_0'(\eps,\ell)\forall r>r_0'(\eps,\ell)&:
		\limsup_{n\to\infty}\Erw[g(K_{\ell,r}(\G(n,m)))]-\Erw\brk{g\bc{\ln\frac{Z(\G(n,m))}{\Erw[Z(\G(n,m))]}}}<\eps.
	\end{align*}
Combining these two statements and observing that the first and the last term are independent of $\ell,r$, we obtain
	$$\limsup_{n\to\infty}\Erw[g(\cK)]-\Erw[g(\ln Z(\G(n,m))-\ln\Erw[Z(\G(n,m))])]=0,$$
i.e., $\ln Z(\G(n,m))-\ln\Erw[Z(\G(n,m))]$ converges to $\cK$ in distribution.
Plugging in the formula for the first moment from (\ref{eq:FirstMoment}) yields (\ref{eqThm_SSC}).
Finally, because \Prop~\ref{prop:FirstCondOverFirst} shows that
	\begin{align*}
	\pr\brk{\G(n,m)\in\fS\triangle\{C_1(\G(n,m))+\vecone\{k=2\}C_2(\G(n,m))=0\}}=O(1/n),
	\end{align*}
the formula for the conditional free energy given $\fS$ follows from (\ref{eqThm_SSC}) and \Lem~\ref{Claim_Tower2}.
\end{proof}

\subsection*{Organization}
The paper is organized as follows.
After proving \Lem~\ref{Claim_Tower2} in \Sec~\ref{Sec_tower},
in \Sec~\ref{Sec_prelims} we collect some preliminaries, introduce notation, supply the proofs of \Lem s~\ref{Lemma_Phi} and~\ref{Lemma_Xi}
and show how \Thm~\ref{Thm_overlap}, \Thm~\ref{Cor_contig} and \Cor~\ref{Thm_contig} follow from \Thm~\ref{Thm_SSC}.
Because we consider the proof of \Prop~\ref{prop_KS} the main technical achievement of this work, the proof is self-contained, and as we deem the argument rather interesting, that proof follows in \Sec~\ref{Sec_prop_KS}.
Further, \Sec~\ref{sec:ProofPreCond} contains the proof of Proposition \ref{prop:belowcond-unif}, which is by way of a (substantial) generalization
of an argument from~\cite{CKPZ} for the Potts antiferromagnet.
Subsequently \Sec~\ref{Sec_moments} contains the proofs of \Prop~\ref{lem:FirstMoment} and \Prop~\ref{lem:SecondMoment} about
the moments of the truncated variable $\cZ$.
Moreover, \Sec~\ref{Sec_cycles} deals with the proof of \Prop~\ref{prop:FirstCondOverFirst}.
The somewhat delicate proof of  \Prop~\ref{Lemma_totallyVexed} can be found in \Sec~\ref{Sec_totallyVexed}.
\Sec~\ref{Sec_Thm_plantedFreeEnergy} contains the rather technical proofs of \Thm~\ref{Thm_cond} and \Thm~\ref{Thm_plantedFreeEnergy}.
Finally, the proof of \Thm~\ref{thrm:TreeGraphEquivalence} about the reconstruction problem can be found in \Sec~\ref{sec:thrm:TreeGraphEquivalence}.

\subsection{Proof of \Lem~\ref{Claim_Tower2}}\label{Sec_tower}
The proof is by generalization of the argument from~\cite[\Sec~2]{COW} for the random regular $k$-SAT model to the current setting of random factor graph models.
We begin with the following lower bound on the second moment of the conditional expectation.
Let $\delta_Y=\Tr(\Phi_Y)-1=(\hat\kappa_Y-\kappa_Y)/\kappa_Y$.

\begin{lemma}\label{Claim_Tower1}
Suppose that $P$ satisfies {\bf SYM} and {\bf BAL} and let $0<d<\dc$, $\ell,r>0$.
Then uniformly for
all $m\in\cM(d)$,
	$$\Erw[\Erw[\cZ(\G(n,m))|\fF_{\ell,r}]^2]\geq\Erw[Z(\G(n,m))]^2\exp\bc{o(1)+\sum_{Y\in\cY_{\leq\ell,r}}\delta_Y^2\kappa_Y}.$$
\end{lemma}
\begin{proof}
Fix a number $\alpha>0$, choose $B=B(\alpha,\ell,r)$ sufficiently large and let
$\Gamma=\Gamma(\ell,r,B)$ be the set of all families $(c_{Y})_{Y\in\cY_{\leq\ell,r}}$ of non-negative integers such that $\sum_{Y\in\cY_{\leq\ell,r}}c_Y\leq B$.
Moreover, let $\cC=\cC(\ell,r,B)$ be the event that $(C_Y(\G(n,m)))_{Y\in\cY_{\leq\ell,r}}\in\Gamma$.
Then~\eqref{eq:NishimoriG} and \Prop~\ref{prop:FirstCondOverFirst} yield
	\begin{align}
	\frac{\Erw[\vecone\{\cC\}\Erw[Z(\G(n,m))|\fF_{\ell,r}]^2]}{\Erw[Z(\G(n,m))]^2}& 
			=\sum_{c\in\Gamma}\frac{\pr[\forall Y\in\cY_{\leq\ell,r}:C_Y(\hat\G(n,m))=c_Y]^2}
					{\pr\brk{\forall Y\in\cY_{\leq\ell,r}:C_Y(\G(n,m))=c_Y}}
			\sim\sum_{c\in\Gamma}\prod_{Y\in\cY_{\leq\ell,r}}\frac{\pr\brk{\Po((1+\delta_Y)\kappa_Y)=c_Y}^2}
					{\pr\brk{\Po(\kappa_Y)=c_Y}}\nonumber\\
			&=\exp\bc{-\sum_{Y\in\cY_{\leq\ell,r}}(1+2\delta_Y)\kappa_Y}
				\sum_{c\in\Gamma}\prod_{Y\in\cY_{\leq\ell,r}}\frac{((1+\delta_Y)^2\kappa_Y)^{c_Y}}{c_Y!}.		\label{eqTower0}
	\end{align}
Let $S=\sum_{Y\in\cY_{\leq\ell,r}}(1+\delta_Y)^2\kappa_Y$.
Since the matrices $\Phi_\psi$ are stochastic, (\ref{eqCycles}) shows that there is a number $T(\ell)$ such that $S\le T(\ell)$.
Therefore, choosing $B=B(\alpha,\ell,d)$ sufficiently large, we can ensure that $\exp(S)\leq\exp(\alpha)\sum_{L\leq B}S^L/L!$.
Hence,
	\begin{align}\label{eqTower1}
	\exp\bc{S-\alpha}&\leq\sum_{L\leq B}\frac{S^L}{L!}=
		\sum_{c\in\Gamma}\prod_{Y\in\cY_{\leq\ell,r}}\frac{((1+\delta_Y)^2\kappa_Y)^{c_Y}}{c_Y!}.
	\end{align}
Combining (\ref{eqTower0}) and (\ref{eqTower1}), we find
	\begin{equation}\label{eqTower2}
	\Erw[\vecone\{\cC\}\Erw[Z(\G(n,m))|\fF_{\ell,r}]^2]\geq\Erw[Z(\G(n,m))]^2\exp\bc{-\alpha+\sum_{Y\in\cY_{\ell,r}}\delta_Y^2\kappa_Y}.
	\end{equation}

Finally, we need to show that $Z(\G(n,m))$ can be replaced by $\cZ(\G(n,m))$ on the l.h.s.\ of (\ref{eqTower2}).
Since $Z(\G(n,m))\geq\cZ(\G(n,m))$ but $\Erw[Z(\G(n,m))]\sim\Erw[\cZ(\G(n,m))]$, we have
	\begin{align}\nonumber
	\Erw&\brk{\vecone\{\cC\}(\Erw[Z(\G(n,m))|\fF_{\ell,r}]^2-\Erw[\cZ(\G(n,m))|\fF_{\ell,r}]^2)}\\
	&=\Erw\brk{\vecone\{\cC\}(\Erw[Z(\G(n,m))|\fF_{\ell,r}]+\Erw[\cZ(\G(n,m))|\fF_{\ell,r}])(\Erw[Z(\G(n,m))|\fF_{\ell,r}]-\Erw[\cZ(\G(n,m))|\fF_{\ell,r}])}\nonumber\\
	&\leq	2\|\vecone\{\cC\}\Erw[Z(\G(n,m))|\fF_{\ell,r}]\|_\infty\Erw\brk{\Erw[Z(\G(n,m))|\fF_{\ell,r}]-\Erw[\cZ(\G(n,m))|\fF_{\ell,r}]}\nonumber\\
	&=o(\Erw[Z(\G(n,m))])\|\vecone\{\cC\}\Erw[Z(\G(n,m))|\fF_{\ell,r}]\|_\infty.\label{eqTower3}
	\end{align}
To bound $\|\vecone\{\cC\}\Erw[Z(\G(n,m))|\fF_{\ell,r}]\|_\infty$ we observe that for all $(c_Y)_Y\in\Gamma$,
	\begin{align*}
	\frac{\Erw[Z(\G(n,m))|\forall Y:C_Y=c_Y]}{\Erw[Z(\G(n,m))]}
	&=\frac{\pr[\forall Y\in\cY_{\leq\ell,r}:C_Y(\hat\G(n,m))=c_Y]}
					{\pr\brk{\forall Y\in\cY_{\leq\ell,r}:C_Y(\G(n,m))=c_Y}}
						&&\mbox{[by~\eqref{eq:NishimoriG}]}
						\nonumber\\
			&\sim\prod_{Y\in\cY_{\leq\ell,r}}\frac{\pr\brk{\Po((1+\delta_Y)\kappa_Y)=c_Y}}
					{\pr\brk{\Po(\kappa_Y)=c_Y}}&&\mbox{[by \Prop~\ref{prop:FirstCondOverFirst}]}\nonumber\\
			&=\prod_{Y\in\cY_{\leq\ell,r}}(1+\delta_Y)^{c_Y}\exp(-\delta_Y\kappa_Y)=O(1)&&
				\mbox{[as $\delta_Y=O(1)$ and $\sum_Yc_Y\leq B$].}
	\end{align*}
Hence, $\|\vecone\{\cC\}\Erw[Z(\G(n,m))|\fF_{\ell,r}]\|_\infty=O(\Erw[Z(\G(n,m))])$ and the assertion follows from (\ref{eqTower2}) and (\ref{eqTower3})
by taking $\alpha\to0$ sufficiently slowly as $n\to\infty$.
\end{proof}

\begin{proof}[Proof of \Lem~\ref{Claim_Tower2}]
We use a similar trick as in the proof of~\cite[\Cor~2.6]{COW}.
Recall that aim to show that
	\begin{equation}\label{eqClaim_Tower2}
	\pr\brk{|\cZ(\G(n,m))-\Erw[\cZ(\G(n,m))|\fF_{\ell,r}]|>\eta \Erw[Z(\G(n,m))]}=0.
	\end{equation}
Given $\eta>0$ choose $\alpha=\alpha(\eta)>0$ small enough.
Then by (\ref{eqTower}), (\ref{eqTower'}) and \Lem~\ref{Claim_Tower1} and (\ref{eqVarFormula}), for sufficiently $\ell,r,n$ we have
	\begin{align}\label{eqCondVarBound}
	\Erw\brk{\Var[\cZ(\G(n,m))|\fF_{\ell,r}]}<\alpha{\Erw[Z(\G(n,m))]^2}.
	\end{align}
Now  define
	$$X(\G(n,m))=|\cZ(\G(n,m))- \Erw[\cZ(\G(n,m))|\fF_{\ell,r}]|\vecone\cbc{
		\frac{|\cZ(\G(n,m))- \Erw[\cZ(\G(n,m))|\fF_{\ell,r}]|}{\Erw[Z(\G(n,m))]}>\alpha^{1/3}}.$$
Then
	\begin{align}\label{eqXNotTooSmall}
	X(\G(n,m))<\alpha^{1/3}\Erw[Z(\G(n,m))]&\Rightarrow\abs{\cZ(\G(n,m))-\Erw[\cZ(\G(n,m))|\fF_{\ell,\eps}]}
		\leq\alpha^{1/3}\Erw[Z(\G(n,m))].
	\end{align}
Furthermore,  by Chebyshev's inequality
	\begin{align}\nonumber
	\Erw[X(\G(n,m))|\fF_{\ell,r}]&\leq\alpha^{1/3}\Erw[Z(\G(n,m))]\sum_{j\geq0}2^{j+1}\pr\brk{X(\G(n,m))>2^j\alpha^{1/3}\Erw[Z(\G(n,m))]\big|\fF_{\ell,r}}\\
		& \leq4\alpha^{-1/3}\Erw[Z(\G(n,m))]\cdot\frac{\Var[\cZ(\G(n,m))|\fF_{\ell,r}]}{\Erw[Z(\G(n,m))]^2}.\label{eqErwX}
	\end{align}
Combining	(\ref{eqCondVarBound}) and (\ref{eqErwX}), we obtain
	\begin{align}\label{eqErwX'}
	\Erw[X(\G(n,m))]=\Erw[\Erw[X(\G(n,m))|\fF_{\ell,r}]]\leq\alpha^{1/2}\Erw[Z(\G(n,m))].
	\end{align}
Finally, (\ref{eqClaim_Tower2}) follows from (\ref{eqXNotTooSmall}), (\ref{eqErwX'}) 
and Markov's inequality.
\end{proof}

\section{Getting started}\label{Sec_prelims}

\subsection{Basics}
Throughout the paper we continue to use the notation introduced in \Sec s~\ref{Sec_results} and~\ref{Sec_Proofs}.
In particular, we write $V_n=\{x_1,\ldots,x_n\}$ for a set of $n$ variable nodes and $F_m=\{a_1,\ldots,a_m\}$ for a set of $m$ constraint nodes.
Further,  $\vm_d(n)$ is a random variable with distribution $\Po(dn/k)$ and we just write $\vm_d$ or $\vm$ if $n$ and/or $d$ are apparent.
Moreover, for an integer $l\geq1$ we let $[l]=\{1,\ldots,l\}$.

For a finite set $\cX$ we denote the set of probability distributions on $\cX$ by $\cP(\cX)$.
We identify $\cP(\cX)$ with the standard simplex in $\RR^\cX$ and endow $\cP(\cX)$ accordingly with the Borel $\sigma$-algebra.
By $\cP^2(\cX)$ we denote the set of probability measures on $\cP(\cX)$ and by $\cP^2_*(\cX)$ the set of all $\pi\in\cP^2(\cX)$
whose mean $\int_{\cP(\cX)}\mu\dd\pi(\mu)$ is the uniform distribution on $\cX$.
In addition, for a point $x$ in a measurable space we write $\delta_x$ for the Dirac measure on $x$.
The entropy of a probability distribution $\mu$ on a finite set $\cX$ is always denoted by $\cH(\mu)$.
Thus, recalling that $\Lambda(z)=z\ln z$ for $z>0$ and setting $\Lambda(0)=0$, we have
	$\cH(\mu)=-\sum_{x\in\cX}\Lambda(\mu(x)).$

Further, if $\mu\in\cP(\Omega^{V_n})$ is a probability measure on the discrete cube $\Omega^{V_n}$, then 
	$\SIGMA_{\mu},\TAU_{\mu},\SIGMA_{1,\mu},\SIGMA_{2,\mu},\ldots\in\Omega^{V_n}$ denote mutually independent samples from $\mu$.
If $\mu=\mu_G$ is the Gibbs measure induced by a factor graph $G$, we write $\SIGMA_G$ etc.\ instead of $\SIGMA_{\mu_G}$.
Where $\mu$ or $G$ are apparent from the context we omit the index and just write $\SIGMA,\TAU$, etc.
If $X:(\Omega^{V_n})^l\to\RR$ is a random variable, then we use the notation
	\begin{align*}
	\bck{X}_\mu=\bck{X(\SIGMA_1,\ldots,\SIGMA_l)}_\mu&
		=\sum_{\sigma_1,\ldots,\sigma_l\in\Omega^{V_n}}X(\sigma_1,\ldots,\sigma_l)\prod_{j=1}^l\mu(\sigma_j).
	\end{align*}
Thus, $\bck{X}_\mu$ is the mean of $X$ over independent samples from $\mu$.
If $\mu=\mu_G$ for a factor graph $G$, then we simplify the notation by writing $\bck\nix_G$ rather than $\bck\nix_{\mu_G}$.
We use this notation to distinguish averages over $\mu_G$ from other sources of randomness (e.g., the choice of  the random factor graph),
for which we reserve the symbols $\Erw\brk\nix$ and $\Var\brk\nix$.

Finally, we need a few facts about probability distributions on sets of the form $\Omega^l$.
For  $\sigma_1,\ldots,\sigma_l:V\to\Omega$ let $\rho_{\sigma_1,\ldots,\sigma_l}\in\cP(\Omega^l)$ denote the {\em $l$-wise overlap}, defined by 
	\begin{equation}\label{eqMultiOverlap}
	\rho_{\sigma_1,\ldots,\sigma_l}(\omega_1,\ldots,\omega_l)=|\sigma_1^{-1}(\omega_1)\cap\cdots\cap\sigma_l^{-1}(\omega_l)|/|V|.
	\end{equation}
We use this notation also in the case $l=1$ and observe that $\rho_{\sigma_1}$ is nothing but the empirical distribution of the spins under $\sigma_1$.
Further, we let $\bar\rho_l$ signify the uniform distribution on $\Omega^l$; we usually omit the index $l$ to ease the notation.
For two spin assignments $\sigma,\tau:V\to\Omega$ we let $\sigma\triangle\tau=\{v\in V:\sigma(v)\neq\tau(v)\}.$

\begin{lemma}[\cite{Victor}]\label{Lemma_multiOverlap}
For any finite set $\Omega$, any $\eps>0$ and any $l\geq3$ there exist $\delta=\delta(\Omega,\eps,l)$ and $n_0=n_0(\Omega,\eps,l)$
	such that for all $n>n_0$ and all $\mu\in\cP(\Omega^{V_n})$ the following is true:
	if $\bck{\TV{\rho_{\SIGMA_1,\SIGMA_2}-\bar\rho}}<\delta$, then $\bck{\TV{\rho_{\SIGMA_1,\ldots,\SIGMA_l}-\bar\rho_l}}<\eps$.
\end{lemma}

\noindent Call $\sigma\in\Omega^{V_n}$ \emph{nearly balanced} if $\TV{\rho_\sigma-\bar\rho}\leq n^{-2/5}$.

\begin{lemma}[{\cite[\Lem~4.7]{CKPZ}}]\label{Lemma_nbalanced_CKPZ}
For any $\eps>0$ there is $\delta>0$ such that for all sufficiently large $n$ the following is true.
If $\mu\in\cP(\Omega^n)$ satisfies
	$\bck{\TV{\rho_{\SIGMA,\TAU}-\bar\rho}}_{\mu}<\delta$, 
then for all nearly balanced $\tau$ we have $\bck{\TV{\rho_{\SIGMA,\tau}-\bar\rho}}_{\mu}<\eps$.
\end{lemma}

\noindent
Finally, we need the following elementary observation.

\begin{fact}\label{Lemma_balanceRowCols}
For any finite set $\Omega$ and any $\eps>0$ there is $\delta>0$ such that the following holds.
If $\rho=(\rho(s,t))_{s,t\in\Omega}\in\cP(\Omega^2)$ satisfies
	\begin{align*}
	\sum_{s\in\Omega}\abs{\frac1q-\sum_{t\in\Omega}\rho(s,t)}+\abs{\frac1q-\sum_{t\in\Omega}\rho(t,s)}<\delta,
	\end{align*}
then there exists $\rho'\in\cP(\Omega^2)$ such that $\TV{\rho-\rho'}<\eps$ and
	$\sum_{t\in\Omega}\rho'(s,t)=\sum_{t\in\Omega}\rho'(t,s)=1/q\mbox{ for all }s\in\Omega.$
\end{fact}
 
\subsection{The Nishimori identity}
There exists an important distributional relationship between the teacher-student model $\G^*(n,m,P,\sigma)$ and the
reweighted random graph model $\hat\G(n,m,P)$ from~(\ref{eq:NishimoriG}) (cf.~\cite{LF} for a discussion from the physics viewpoint).
To state this connection, we need to define an appropriately reweighted distribution on the set $\Omega^{V_n}$ of spin assignments.
Specifically, we let $\hat\SIGMA_{n,m,P}\in\Omega^{V_n}$ be a random assignment chosen from the distribution
	\begin{align}\label{eq:NishimoriS}
	\pr[\hat\SIGMA_{n,m,P}=\sigma]&=\frac{\Erw[\psi_{\G(n,m,P)}(\sigma)]}{\Erw[Z(\G(n,m,P))]}&(\sigma\in\Omega^{V_n}).
	\end{align}
As before we skip the index $P$ where possible.
We refer to the following statement as the {\em Nishimori identity}.

\begin{lemma}[{\cite[\Prop~3.10]{CKPZ}}]\label{lem:nishimori}
For every distribution $P$ on weight functions $\Omega^k\to(0,2)$, for all integers $n,m$, for every $\sigma\in\Omega^{V_n}$ and for every event $\cA$ we have
	\begin{align}\label{eq:nishimori}
	\pr\brk{\hat\SIGMA_{n,m,P}=\sigma}\cdot\pr\brk{\G^*(n,m,P,\sigma)\in\cA}&=
		\Erw\brk{\vecone\{\hat\G(n,m,P)\in\cA\}\mu_{\hat\G(n,m,P)}(\sigma)}.
	\end{align}
\end{lemma}

 \noindent
A useful consequence of this result is that
$\Erw[\cX(\G^*(n,\vec m,\hat\SIGMA_{n,m}),\hat\SIGMA_{n,m})]=\Erw\bck{\cX(\hat\G,\SIGMA)}_{\hat\G}$
for every   $L^1$-function $\cX$ .
	
\subsection{Eigenvalues}
The vector or matrix with all entries equal to one (in any dimension) is signified by $\vecone$.
The transpose of a matrix $A$ we denote by $A^\ast$.
Additionally, $\id$ denotes the identity matrix (in any dimension).
Further, the standard basis vectors on $\RR^\Omega$ are denoted by $e_\omega$, $\omega\in\Omega$.
For the entries of a matrix $A\in\RR^{\Omega\times\Omega}$ we use the notation $A(\sigma,\tau)$;
thus, $A(\sigma,\tau)=\scal{Ae_\tau}{e_\sigma}$ for all $\sigma,\tau\in\Omega$.
The spectrum of a linear operator $X:E\to E'$ is denoted by $\eig(X)$.

The following simple observation will be used several times. Recall $\Phi$ from (\ref{eqPhi}).

\begin{lemma}\label{Lemma_F}
Assume that $P$ satisfies {\bf SYM}.
Then the function
	\begin{equation}\label{eqF}
	\phi:\RR^\Omega\to(0,2),\qquad \rho\mapsto\sum_{\tau\in\Omega^k}\Erw[\PSI(\tau)]\prod_{i=1}^k\rho(\tau_i)
	\end{equation}
satisfies
	$D\phi(\bar\rho)=k\xi\vecone$, $D^2\phi(\bar\rho)=qk(k-1)\xi\Phi$ and $\phi$ is bounded away from $0$.
\end{lemma}
\begin{proof}
Since 
	$\frac{\partial\phi}{\partial\rho(\omega)}=\sum_{j=1}^k\sum_{\tau\in\Omega^k}\vecone\{\tau_j=\omega\}\Erw[\PSI(\tau)]\prod_{i\neq j}\rho(\tau_i)$
	for every $\omega\in\Omega$,
{\bf SYM} immediately yields $D\phi(\bar\rho)=k\xi\vecone$.
Proceeding to the second derivatives, we find
\begin{align}\nonumber
\frac{\partial^2 \phi}{\partial\rho(\omega)\partial\rho(\omega')}&=
		\sum_{\tau \in \Omega^k}\sum_{j,l\in[k]:j\neq l}
			 \mathbf{1}\{ \tau_j=\omega, \ \tau_l=\omega'\}\Erw[\PSI(\tau)]\prod_{i \in [k]\setminus \{j, l \} }\rho\bc{\tau_i} .
\end{align}
Consequently, {\bf SYM} yields $D^2\phi(\bar\rho)=qk(k-1)\xi\Phi$.
Finally, the fact that $\inf_{\rho\in\cP(\Omega)}\phi(\rho)>0$ follows from (\ref{eqBounded}).
\end{proof}

\noindent
As an immediate application we prove \Lem s~\ref{Lemma_Phi} and~\ref{Lemma_Xi}.

\begin{proof}[Proof of \Lem~\ref{Lemma_Phi}]
Condition {\bf SYM} readily implies that $\Phi_\psi$ is stochastic for every $\psi\in\Psi$.
Hence, $\Phi_{\psi}\vecone=\vecone$ for all $\psi\in\Psi$ and consequently $\Phi\vecone=\vecone$.
To see that $\Phi$ is symmetric let $\theta$ be the permutation on $\{1,\ldots,k\}$ such that $\theta(1)=2$, $\theta(2)=1$ and $\theta(i)=i$ for all $i>2$.
Since {\bf SYM} implies that $\PSI$ and $\PSI^\theta$ are identically distributed, we obtain
	\begin{align*}
	\Phi(\omega,\omega')&=q^{1-k}\xi^{-1}\sum_{\tau\in\Omega^k}\vecone\{\tau_1=\omega,\tau_2=\omega'\}\Erw[\PSI(\tau)]\\
		&=q^{1-k}\xi^{-1}\sum_{\tau\in\Omega^k}\vecone\{\tau_1=\omega,\tau_2=\omega'\}\Erw[\PSI^\theta(\tau)]=
			q^{1-k}\xi^{-1}\sum_{\tau\in\Omega^k}\vecone\{\tau_1=\omega',\tau_2=\omega\}\Erw[\PSI(\tau)]=\Phi(\omega',\omega).
	\end{align*}
To verify the last assertion,  consider the function $\phi$ from (\ref{eqF}).
Condition {\bf BAL} ensures that $\phi$ is concave on the set $\cP(\Omega)$ of probability measures on $\Omega$.
Since by \Lem~\ref{Lemma_F} the Hessian  satisfies $D^2\phi(\bar\rho)=qk(k-1)\xi\Phi$,
we see that $\Phi$ induces a negative semidefinite endomorphism of the subspace $\{x\in\RR^q:x\perp\vecone\}$.
Hence, $\max_{x\perp\vecone}\scal{\Phi x}x\leq0$.
\end{proof}

\begin{proof}[Proof of \Lem~\ref{Lemma_Xi}]
To see that $\Xi$ is self-adjoint
let $(e_\omega)_{\omega\in\Omega}$ be the canonical basis of $\RR^\Omega$ and
let $\theta$ be the permutation on $\{1,\ldots,k\}$ such that $\theta(1)=2$, $\theta(2)=1$ and $\theta(i)=i$ for all $i>2$.
Then for all $s,t,\sigma,\tau\in\Omega$ we have
	\begin{align}\nonumber
	\scal{\Xi e_{\sigma}\tensor e_{\tau}}{e_s\tensor e_t}&=\Erw\brk{\scal{\Phi_{\PSI}e_{\sigma}}{e_s}\scal{\Phi_{\PSI}e_{\tau}}{e_t}}
		=\Erw\brk{\Phi_{\PSI}(s,\sigma)\Phi_{\PSI}(t,\tau)}=\Erw\big[\Phi_{\PSI^\theta}(s,\sigma)\Phi_{\PSI^\theta}(t,\tau)\big]&\mbox{[due to {\bf SYM}]}
		\\&=\Erw\brk{\Phi_{\PSI}(\sigma,s)\Phi_{\PSI}(\tau,t)}=
		\Erw\brk{\scal{e_{\sigma}}{\Phi_{\PSI}e_s}\scal{e_{\tau}}{\Phi_{\PSI}e_t}}=\scal{ e_{\sigma}\tensor e_{\tau}}{\Xi e_s\tensor e_t}.
			\label{eqSelfAdj}
	\end{align}
Since $(e_s\tensor e_t)_{s,t\in\Omega}$ is a basis of $\RR^\Omega\tensor\RR^\Omega$, (\ref{eqSelfAdj}) shows that $\Xi$ is self-adjoint.

Furthermore, since $\Phi_\psi\vecone=\vecone$ for all $\psi\in\Psi$ by \Lem~\ref{Lemma_Phi},
we see that
	$\Xi( x\tensor\vecone)=\Erw[\Phi_{\PSI}x\tensor\Phi_{\PSI}\vecone]=(\Phi x)\tensor\vecone.$
Similarly, $\Xi(\vecone\tensor x)=\vecone\tensor(\Phi x)$ and thus (\ref{eqLemma_Xi}) follows from \Lem~\ref{Lemma_Phi}.
In particular, since $\Phi\vecone=\vecone$ by \Lem~\ref{Lemma_Phi} we obtain $\Xi(\vecone\tensor\vecone)=\vecone\tensor\vecone$.
Because $\Xi$ is self-adjoint, this implies that $\Xi\cE'\subset\cE'$.
Finally, assume that $z\in\cE$.
Then for all $y\in\RR^q$ we have
	$\scal{\Xi z}{y\tensor\vecone}=\scal{z}{\Xi(y\tensor\vecone)}=\scal{z}{(\Phi y)\tensor\vecone}=0,$
and analogously $\scal{\Xi z}{\vecone\tensor y}=0$.
Hence, $\Xi\cE\subset\cE$.
\end{proof}

\subsection{Contiguity}
Throughout the paper we apply contiguity between several probability spaces.
Some of these contiguity results derive from the following first moment calculation, which also delivers the proof of (\ref{eqnotosmm3}).

\begin{lemma}\label{Cor_F}
Suppose that $P$ satisfies {\bf SYM} and {\bf BAL}.
For any $D>0$ there exists $c>0$ such that for all $m\leq Dn/k$,
	$$cq^n\xi^m\leq\Erw[Z(\G(n,m))]\leq q^n\xi^m.$$
Moreover, for any $\sigma\in\Omega^{V_n}$ we have, uniformly for all $m\leq Dn/k$,
	\begin{align}\label{eqQuenchedAnnealed}
	\Erw|\ln Z(\G(n,m))|&\leq O(n),&
		\Erw|\ln Z(\G^*(n,m,\sigma))|&\leq 
			O(n).
	\end{align}
\end{lemma}
\begin{proof}
By the linearity of expectation and because the constraint nodes of $\G(n,m)$ are chosen independently,
	\begin{align*}
	\Erw[Z(\G(n,m))]=\sum_{\sigma\in\Omega^{V_n}}\phi(\rho_\sigma)^m.
	\end{align*}
Since {\bf SYM} and {\bf BAL} provide that $\phi(\rho_\sigma)\leq\xi$
	for every $\sigma$, the upper bound $\Erw[Z(\G(n,m))]\leq q^n\xi^m$ is immediate.
With respect to the lower bound, recall that the number of $\sigma:V_n\to\Omega$ such that $\TV{\rho_\sigma-\bar\rho}\leq n^{-1/2}$ is of order $\Omega(q^n)$.
Hence, applying \Lem~\ref{Lemma_F}, we see that for such $\sigma$,
	\begin{align}\nonumber
	\phi(\rho_\sigma)&=\phi(\bar\rho)+k\xi\scal{\vecone}{\rho_\sigma-\bar\rho}+qk(k-1)\xi\scal{\Phi(\rho_\sigma-\bar\rho)}{\rho_\sigma-\bar\rho}/2
		+O(\TV{\rho_\sigma-\bar\rho}^3)\\
		&=\phi(\bar\rho)+O(\TV{\rho_\sigma-\bar\rho}^2)=\phi(\bar\rho)+O(1/n).\label{eqDirtyPhi}
	\end{align}
Thus, $\Erw[Z(\G(n,m))]\geq\Omega(q^n)(\phi(\bar\rho)+O(1/n))^m=\Omega(q^n\xi^m)$, uniformly for all $m\leq Dn/k$.
Finally, (\ref{eqQuenchedAnnealed}) follows from because
	$\Erw|\ln Z(\G(n,m))|\leq m\Erw[\max_{\tau\in\Omega^k}|\ln\PSI(\tau)|]=O(n)$ due to (\ref{eqBounded}) and the independence of the
	constraint nodes, and similarly
	$\Erw|\ln Z(\G^*(n,m,P,\sigma))|\leq2m\Erw\brk{\max_{\tau\in\Omega^k}|\ln\PSI(\tau)|}/\phi(\rho_\sigma)=O(n)$
	by \Lem~\ref{Lemma_F} and (\ref{eqBounded}).
\end{proof}

\begin{corollary}\label{lem:conc_coloring}
Assume that $P$ satisfies {\bf SYM} and {\bf BAL} and let $D>0$.
Then uniformly for all $m\leq Dn/k$,
	\begin{align}\label{eq:conc_col}
	\pr\left[\TV{\rho_{\hat\SIGMA_{n,m}}-\bar\rho}>n^{-\frac12}\ln n\right]\le O(n^{-\ln\ln n})
	\end{align}
and the distribution of $\hat\SIGMA_{n,m}$ and that of $\mathbold{\sigma}^*$  are mutually contiguous.
Additionally, for any $\eps>0$ there exists $c=c(\eps,D)>0$ such that
	\begin{equation}\label{eq:conc_col2}
	\limsup_{n\to\infty}\max_{m\leq Dn}\pr\brk{\TV{\rho_{\hat\SIGMA_{n,m}}-\bar\rho}>cn^{-1/2}}\leq \eps.
	\end{equation}
\end{corollary}
\begin{proof}
The bound (\ref{eq:conc_col}) and the mutual contiguity of $\hat\SIGMA_{n,m}$ and the uniformly random $\SIGMA^*$ follow from \cite[\Cor~3.27]{CKPZ}.
With respect to (\ref{eq:conc_col2}) {\bf BAL}, {\bf SYM} and \Lem~\ref{Cor_F} ensure there is $c'=c'(D)>0$ such that for every $c>0$,
	\begin{align*}
	\pr\brk{\TV{\rho_{\hat\SIGMA_{n,m}}-\bar\rho}>cn^{-\frac12}}&=\sum_{\sigma\in\Omega^{V_n}}
		\vecone\cbc{\TV{\rho_{\hat\SIGMA_{n,m}}-\bar\rho}>cn^{-\frac12}}\frac{\Erw[\psi_{\G(n,m)}(\sigma)]}{\Erw[Z[\G(n,m)]}\\
		&\leq\frac{q^n\xi^m}{\Erw[Z(\G(n,m))]}\pr\brk{\TV{\rho_{\SIGMA^*}-\bar\rho}>cn^{-\frac12}}\leq 
			c'\cdot\pr\brk{\TV{\rho_{\SIGMA^*}-\bar\rho}>cn^{-1/2}}.
	\end{align*}
By Stirling we can choose $c=c(\eps)>0$ large enough so that the last expression is smaller than $\eps>0$.
\end{proof}

\begin{corollary}\label{Cor_strCntg}
Assume that $P$ satisfies {\bf SYM} and {\bf BAL}, let $d>0$ and let $(\cS_n)_n$ be a sequence of events.
Then the following two statements are true.
	\begin{align}\label{eqCor_strCntg1}
	\forall\eps>0\,\exists\delta>0:\limsup_{n\to\infty}\pr\brk{(\G^*,\SIGMA^*)\in\cS_n}
		<\delta\Rightarrow\limsup_{n\to\infty}\pr\brk{(\hat\G,\hat\SIGMA)\in\cS_n}<\eps,\\
	\forall\eps>0\,\exists\delta>0:\limsup_{n\to\infty}\pr\brk{(\hat\G,\SIGMA_{\hat\G})\in\cS_n}
			<\delta\Rightarrow\limsup_{n\to\infty}\pr\brk{(\G^*,\SIGMA^*)\in\cS_n}<\eps.		\label{eqCor_strCntg2}
	\end{align}
\end{corollary}
\begin{proof}
Fix $m\in\cM(d)$.
By \Lem~\ref{lem:nishimori}, {\bf BAL} and \Lem~\ref{Cor_F},
	\begin{align}\nonumber
	\pr\brk{(\hat\G(n,m),\SIGMA_{\hat\G(n,m)})\in\cS_n}&=\pr\brk{(\G^*(n,m,\hat\SIGMA_{n,m}),\hat\SIGMA_{n,m})\in\cS_n}
		=\sum_{\sigma\in\Omega^{V_n}}\pr\brk{(\G^*(n,m,\sigma),\sigma)\in\cS_n}\pr\brk{\hat\SIGMA_{n,m}=\sigma}\\
		&\hspace{-2cm}\leq\frac{\xi^m}{\Erw[Z(\G(n,m))]}\sum_{\sigma\in\Omega^{V_n}}\pr\brk{(\G^*(n,m,\sigma),\sigma)\in\cS_n}
		\leq c^{-1}\pr\brk{(\G^*(n,m,\SIGMA^*),\SIGMA^*)\in\cS_n}\label{eqCor_strCntg666}
	\end{align}
which implies (\ref{eqCor_strCntg1}).
To prove (\ref{eqCor_strCntg2}) pick $L=L(\eps)>0$ large enough so
	that $\pr\brk{\tv{\rho_{\SIGMA^*}-\bar\rho}>Ln^{-1/2}}<\eps/2$.
Then \Lem~\ref{Lemma_F} shows that there exists $\eta=\eta(L)>0$ such that
	$\Erw[\psi_{\G(n,m)}(\sigma)]=\phi(\rho_\sigma)^m\geq\eta\xi^m$ for all $\sigma\in\Omega^{V_n}$
		such that $\tv{\rho_{\sigma}-\bar\rho}\leq Ln^{-1/2}$.
Hence, by \Lem s~\ref{lem:nishimori} and~\ref{Cor_F},
	\begin{align*}
	\pr\brk{(\G^*(n,m,\SIGMA^*),\SIGMA^*)\in\cS_n}&\leq\frac\eps2+
		\pr\brk{(\G^*(n,m,\SIGMA^*),\SIGMA^*)\in\cS_n,\,\tv{\rho_{\SIGMA^*}-\bar\rho}\leq Ln^{-1/2}}\\
	&\leq\frac\eps2+\sum_{\sigma:\tv{\rho_{\sigma}-\bar\rho}\leq Ln^{-1/2}}
		\pr\brk{(\G^*(n,m,\sigma),\sigma)\in\cS_n}\frac{\Erw[\psi_{\G(n,m)}(\sigma)]}{\eta q^n\xi^m}\\
		&\leq\frac\eps2+\frac{\Erw[Z(\G(n,m))]}{\eta q^n\xi^m}\pr\brk{(\G^*(n,m,\hat\SIGMA_{n,m}),\hat\SIGMA_{n,m})\in\cS_n}
		\leq\frac\eps2+\frac{\pr[(\hat\G(n,m),\SIGMA_{\hat\G(n,m)})\in\cS_n]}\eta.
	\end{align*}
Thus, setting $\delta=\eps\eta/3$, we obtain (\ref{eqCor_strCntg2}).
\end{proof}

\begin{proof}[Proof of \Lem~\ref{Prop_contig}]
By construction, the mutual contiguity of $\G^*(n,m,\SIGMA^*)$ and $\G^*(n,m,\hat\SIGMA_{n,m})$ 
is immediate from the mutual contiguity of $\SIGMA^*$ and $\hat\SIGMA_{n,m}$ furnished by \Cor~\ref{lem:conc_coloring}.
Moreover, $\hat\G(n,m)$ and $\G^*(n,m,\hat\SIGMA_{n,m})$ are identically distributed by the Nishimori identity.
\end{proof}

\noindent
Finally, we derive \Thm~\ref{Cor_contig}, \Cor~\ref{Thm_contig} and \Thm~\ref{Thm_overlap} from \Thm~\ref{Thm_SSC}.

\begin{proof}[Proof of \Thm~\ref{Cor_contig}]
Suppose that $d<\dc$ and that $(\cS_n)_n$ is a sequence of events.
We will prove the following two statements, from which the mutual contiguity of $\G$ and $\hat\G$ is immediate.
	\begin{align}\label{eqStrongCor_contig1}
	\forall \eps>0\,\exists\alpha>0:\limsup_{n\to\infty}\pr\brk{\hat\G\in\cS_n}<\alpha\Rightarrow\limsup_{n\to\infty}\pr\brk{\G\in\cS_n}<\eps,\\
	\forall \eps>0\,\exists\alpha>0:\limsup_{n\to\infty}\pr\brk{\G\in\cS_n}<\alpha\Rightarrow\limsup_{n\to\infty}\pr\brk{\hat\G\in\cS_n}<\eps.
		\label{eqStrongCor_contig2}
	\end{align}
Since $\hat\G$ and $\G^*$ are mutually contiguous by \Lem~\ref{Prop_contig}, mutual contiguity of $\G$ and $\G^*$ follows from (\ref{eqStrongCor_contig1}) and (\ref{eqStrongCor_contig2}).
Moreover, the conditional mutual contiguity given $\fS$ follows by applying the unconditional result to $\cS_n\cap\fS$, because 
\Lem~\ref{Prop_contig} and \Prop~\ref{prop:FirstCondOverFirst} show that the probability of $\fS$ is bounded away from $0$ in either model.

We proceed to prove (\ref{eqStrongCor_contig1}).
Because the random variable $\cK$ from \Thm~\ref{Thm_SSC} satisfies $\Erw\abs\cK<\infty$, there exists $\delta>0$ such that 
	$\Erw[\pr\brk{Z(\G)<\delta\Erw[Z(\G)|\vm]|\vm}]<\eps/2.$
Hence,
	\begin{align*}
	\pr\brk{\G\in\cS_n}&=\Erw[\pr\brk{\G\in\cS_n|\vm}]\leq\eps+\Erw[\pr\brk{\G\in\cS_n,\,Z(\G)\geq\delta\Erw[Z(\G)|\vm]|\vm}]\\
		&\leq\eps+\delta^{-1}\Erw\brk{\frac{\Erw[Z(\G)\vecone\{\G\in\cS_n\}|\vm]}{\Erw[Z(\G)|\vm]}}
		=\eps/2+\delta^{-1}\Erw[\pr[\hat\G\in\cS_n|\vm]]=\eps/2+\delta^{-1}\pr\brk{\hat\G\in\cS_n}.
	\end{align*}
Thus, setting $\alpha=\delta\eps/2$, we obtain (\ref{eqStrongCor_contig1}).

Let us move on to the proof of (\ref{eqStrongCor_contig2}).
\Prop~\ref{lem:SecondMoment} shows that  for every $d<\dc$ there is $c(d)>0$ such that uniformly for all $m\in\cM(d)$,
	\begin{align*}
	\Erw[\cZ(\hat\G(n,m))]&=\frac{\Erw[\cZ(\hat\G(n,m))Z(\G(n,m))]}{\Erw[Z(\G(n,m))]}=\frac{\Erw[\cZ(\G(n,m))^2]}{\Erw[Z(\G(n,m))]}\leq c(d)\Erw[Z(\G(n,m))].
	\end{align*}
Hence, by Markov's inequality for any $\eps>0$ there is $L>0$ such that 
	$\pr[\cZ(\hat\G(n,m))>L\cdot\Erw[Z(\G(n,m))]]<\eps/2.$
Moreover, $\pr[\cZ(\hat\G(n,m))=Z(\hat\G(n,m))]=1-o(1)$ by \Prop~\ref{prop:belowcond-unif}.
As a consequence,
	\begin{align*}
	\pr&\brk{\hat\G(n,m)\in\cS_n}=o(1)+\pr\brk{\hat\G(n,m)\in\cS_n,\,\cZ(\hat\G(n,m))=Z(\hat\G(n,m))}\\
		&\leq\eps/2+o(1)+\pr\brk{\hat\G(n,m)\in\cS_n,\,\cZ(\hat\G(n,m))=Z(\hat\G(n,m)),\,\cZ(\hat\G(n,m))\leq L\cdot\Erw[Z(\G(n,m))]}\\
		&\leq\eps/2+o(1)+\pr\brk{\hat\G(n,m)\in\cS_n,\,Z(\hat\G(n,m))\leq L\cdot\Erw[Z(\G(n,m))]}\\
		&=\frac\eps2+o(1)+\frac{\Erw[Z(\G(n,m))\vecone\{\G\in\cS_n,\,Z(\G(n,m))\leq L\cdot\Erw[Z(\G(n,m)]\}]}{\Erw[Z(\G(n,m))]}
			\leq\frac\eps2+o(1)+L\cdot\pr\brk{\G(n,m)\in\cS_n}.
	\end{align*}
Thus, choosing $\alpha<\eps/(3L)$, say, we obtain (\ref{eqStrongCor_contig2}).
\end{proof}

\begin{proof}[Proof of \Cor~\ref{Thm_contig}]
The corollary is immediate from \Thm~\ref{Cor_contig},  \Lem~\ref{lem:nishimori} and \Cor~\ref{lem:conc_coloring}.
\end{proof}

\begin{proof}[Proof of \Thm~\ref{Thm_overlap}]
\Thm~\ref{Cor_contig} and \Prop~\ref{prop:belowcond-unif} imply that 
	$\lim_{n\to\infty}\Erw\bck{\|\rho_{\SIGMA_1,\SIGMA_2}-\bar\rho\|_{\mathrm{TV}}}_{\G}=0$ for all $d<\dc$.
To prove that this fails to hold for $d$ beyond but arbitrarily close to $\dc$, we calculate the derivative $\frac{\partial}{\partial d}\Erw[\ln Z(\G)]$
(for the random graph coloring problem a similar argument was used in~\cite{ECM}).
It is well known that
	\begin{align}
	\frac1n\frac{\partial}{\partial d}&\Erw[\ln Z(\G)]=\frac1n\sum_{m=0}^\infty\brk{\frac{\partial}{\partial d}\pr\brk{\Po(dn/k)=m}}
		\Erw[\ln Z(\G)|\vm=m]\nonumber\\
		&=\frac1k\sum_{m=0}^\infty\brk{\vecone\{m\geq1\}\pr\brk{\Po(dn/k)=m-1}+\pr\brk{\Po(dn/k)=m}}\Erw[\ln Z(\G)|\vm=m]\nonumber\\
		&=\frac1k[\Erw[\ln Z(\G(n,\vec m+1))]-\Erw[\ln Z(\G(n,\vec m)]]
		=\Erw[\ln\langle \psi_{a_{\vm+1}}(\SIGMA(\partial_1 a_{\vm+1}),\ldots,\SIGMA(\partial_k a_{\vm+1}))\rangle_{\G(n,\vm)}].
\label{eq:nullmodelderiv1}
	\end{align}
Expanding the logarithm using Fubini and (\ref{eqBounded}), we find
	\begin{align}
	\frac{1}{n}\frac{\partial}{\partial d}\Erw[\ln Z(\G)]&=
	-\sum_{l=1}^\infty\sum_{h_1,\ldots,h_k\in[n]}\frac{1}{lkn^k}\Erw\langle 1-\PSI(\SIGMA(x_{h_1},\ldots,x_{h_k})\rangle^l_{\G}.
			\label{eq:nullmodelderiv2}
\end{align}
{Further with $\rho_{\SIGMA_1,\ldots,\SIGMA_l}$ denoting the overlap of $l$ independent samples from $\mu_{\G}$ as in (\ref{eqMultiOverlap}),  we can cast (\ref{eq:nullmodelderiv2}) as }
	\begin{align*}
	\frac{1}{n}\frac{\partial}{\partial d}\Erw[\ln Z(\G)]&=
		-\sum_{l=1}^\infty\sum_{h_1,\ldots,h_k\in[n]}\frac{1}{lkn^k}\Erw\bck{\prod_{i=1}^l1-\PSI(\SIGMA_i(x_{h_i}))}_{\G}\\
		&	=-\sum_{l=1}^\infty\frac{1}{kl}\Erw\brk{\sum_{\tau\in\Omega^{k\times l}}
				\bck{\prod_{j=1}^k\rho_{\SIGMA_1,\ldots,\SIGMA_l}(\tau_{j,1},\ldots,\tau_{j,l})}_{\G}\prod_{i=1}^l1-\PSI(\tau_{1,i},\ldots,\tau_{k,i})}.
	\end{align*}
Hence, if $\lim_{n\to\infty}\Erw\bck{\|\rho_{\SIGMA_1,\SIGMA_2}-\bar\rho\|_{\mathrm{TV}}}_{\G}=0$, then
due to (\ref{eqBounded}), dominated convergence and \Lem~\ref{Lemma_multiOverlap}
	\begin{align}			\label{eq:nullmodelderiv3}
	\lim_{n\to\infty}\frac{1}{n}\frac{\partial}{\partial d}\Erw[\ln Z(\G)]&=-\sum_{l=1}^\infty\frac{(1-\xi)^l}{kl}=k^{-1}\ln\xi.
	\end{align}
Now, suppose that $D>0$ is such that $\Erw\bck{\|\rho_{\SIGMA_1,\SIGMA_2}-\bar\rho\|_{\mathrm{TV}}}_{\G}=o(1)$ for all $d<D$.
Then (\ref{eqBounded}), dominated convergence and (\ref{eq:nullmodelderiv3}) yield
	\begin{align*}
	\ln q+\frac Dk\ln\xi&=\ln q+\int_0^D\lim_{n\to\infty}\frac{1}{n}\frac{\partial}{\partial d}\Erw[\ln Z(\G)]\dd d
		=\ln q+\lim_{n\to\infty}\frac{1}{n}\int_0^D\frac{\partial}{\partial d}\Erw[\ln Z(\G)]\dd d
		=\lim_{n\to\infty}\frac1n\Erw[\ln Z(\G(n,\vm_D))].
	\end{align*}
Thus, \Thm~\ref{Thm_cond} shows that $D\leq\dc$.
Consequently, for any $D>\dc$ there exists an average degree $d<D$ such that 
	$\limsup_{n\to\infty}\Erw\bck{\|\rho_{\SIGMA_1,\SIGMA_2}-\bar\rho\|_{\mathrm{TV}}}_{\G}>0,$
as claimed.
The very same argument applies given $\fS$.
\end{proof}

\noindent
As a preparation for \Sec~\ref{sec:thrm:TreeGraphEquivalence} we put the following on record.

\begin{corollary}\label{Cor_strongCntig}
Assume that $P$ satisfies {\bf SYM} and {\bf BAL} and that $d<\dc$.
Then for any sequence $(\cS_n)_n$ of events the following two statements hold.
	\begin{align}\label{eqLemma_strCntg1}
	\forall\eps>0\,\exists\delta>0:\limsup_{n\to\infty}\pr\brk{(\G^*,\SIGMA^*)\in\cS_n}<\delta
		\Rightarrow\limsup_{n\to\infty}\pr\brk{(\G,\SIGMA)\in\cS_n}<\eps,\\
	\forall\eps>0\,\exists\delta>0:\limsup_{n\to\infty}\pr\brk{(\G,\SIGMA)\in\cS_n}
		<\delta\Rightarrow\limsup_{n\to\infty}\pr\brk{(\G^*,\SIGMA^*)\in\cS_n}<\eps.
		\label{eqLemma_strCntg2}
	\end{align}
\end{corollary}
\begin{proof}
To prove (\ref{eqLemma_strCntg1}) pick a small enough $\eta=\eta(\eps)>0$ and a smaller $\delta=\delta(\eta)>0$.
Then \Cor~\ref{Cor_strCntg} shows that 
$\limsup_{n\to\infty}\pr\brk{(\G^*,\SIGMA^*)\in\cS_n}<\delta$ implies
$\limsup_{n\to\infty}\pr\brk{(\hat\G,\SIGMA_{\hat\G})\in\cS_n}<\eta$.
Hence, 
	$$\limsup_{n\to\infty}\pr\brk{\bck{\vecone\{(\hat\G,\SIGMA_{\hat\G})\in\cS_n\}}_{\hat\G}\geq\sqrt\eta}<\sqrt\eta$$
and thus (\ref{eqStrongCor_contig1}) implies
	$\limsup_{n\to\infty}\pr\brk{\bck{\vecone\{(\G,\SIGMA)\in\cS_n\}}_{\G}\geq\eps}<\eps$, which proves (\ref{eqLemma_strCntg1}).

Similarly, to obtain (\ref{eqLemma_strCntg2}) choose $\eta=\eta(\eps)>0$ and $\delta=\delta(\eta)>0$ sufficiently small.
If $\limsup\pr\brk{(\G,\SIGMA)\in\cS_n}<\delta$, 
then (\ref{eqStrongCor_contig2}) yields
	$\limsup_{n\to\infty}\pr\brk{(\hat\G,\SIGMA_{\hat\G})\in\cS_n}<\eta.$
Hence, (\ref{eqCor_strCntg1}) implies $\limsup_{n\to\infty}\pr\brk{(\G^*,\SIGMA^*)\in\cS_n}<\eps$.
\end{proof}

\section{The Kesten-Stigum bound}\label{Sec_prop_KS}

\noindent
{\em Throughout this section we assume that $P$ satisfies {\bf SYM} and {\bf BAL}.}

\subsection{Outline}
In this section we prove \Prop~\ref{prop_KS}.
The key insight is that the dominant eigenvector of $\Xi$ restricted to the space $\cE$ gives rise to a natural family of
	 probability distributions $\pi_{\epsilon}\in\cP^2_*(\Omega)$, $\epsilon>0$.
Up to an error term that decays as $\eps\to0$, the Bethe free energy $\cB(d,P,\pi_{\epsilon})$ of this distribution is given by a quadratic function of the corresponding eigenvalue.
Ultimately, the desired bound on $\max\{\abs\lambda:\lambda\in\Eig[\Xi]\}$ follows because
the definition (\ref{eq:dcond}) of $\dc$ ensures that $\cB(d,P,\pi_\epsilon)\leq\ln q+\frac dk\ln\xi$ for all $d<\dc$, $\epsilon>0$.
To implement this programme we need to show that the dominant eigenvector of $\Xi$ has a particular form.
More precisely, in \Sec~\ref{Sec_Cor:MaxOpt} we prove

\begin{lemma}\label{Cor:MaxOpt}
Let $\hat\lambda=\max\Eig[\Xi]$.
Then $\hat\lambda\geq-\min\Eig[\Xi]$ and there exists
an orthonormal basis  $u_1,\ldots,u_{q-1}\in\RR^\Omega$  of the space $\{x\in\RR^\Omega:x\perp\vecone\}$ and $\bar\lambda_1,\ldots,\bar\lambda_{q-1}\geq0$
such that
	\begin{equation}\label{eqDeltaVector}
	\DeltaM=\sum_{i=1}^{q-1}\bar\lambda_i u_i\tensor u_i\in\RR^\Omega\tensor\RR^\Omega
	\end{equation}
is a unit vector and $\Xi\DeltaM=\hat\lambda\DeltaM$.
\end{lemma}

Throughout this section we denote the eigenvector promised by \Lem~\ref{Cor:MaxOpt} by $\DeltaM$ and the corresponding eigenvalue by $\hat\lambda$.
The particular structure of $\DeltaM$ ensures that
	\begin{align}\label{eqDeltaSymmetric}
	\scal{\DeltaM}{e_\sigma\tensor e_\tau}=\scal{\DeltaM}{e_\tau\tensor e_\sigma}.
	\end{align}
Further, 
because the coefficients $\bar\lambda_i$ in (\ref{eqDeltaVector}) are non-negative and $u_1,\ldots,u_{q-1}\perp\vecone$, we obtain
	\begin{equation}\label{eqetaVector}
	\eta=\sum_{i=1}^{q-1}\sqrt{\bar\lambda_i}\ u_i\tensor u_i\in\cE.
	\end{equation}
Recalling that $(e_\omega)_{\omega\in\Omega}$ is the canonical basis of $\RR^\Omega$, for each $\omega\in\Omega$ we define $\pi_{\epsilon,\omega}\in\RR^\Omega$ by letting
	\begin{align}\label{eqCreative}
	\pi_{\epsilon,\omega}(\sigma)&=\frac1q+\epsilon\scal{\eta}{e_\omega\tensor e_\sigma}.
	\end{align}
Finally, let $\pi_\epsilon=\frac1q\sum_{\omega\in\Omega}\delta_{\pi_{\epsilon,\omega}}$ (with $\delta_z$ the Dirac measure on $z\in\RR^\Omega$).

\begin{lemma}\label{Lemma_sound}
There exists $\epsilon_0>0$ such that for all
 $0<\epsilon<\epsilon_0$ we have $\pi_{\epsilon,\omega}\in\cP(\Omega)$ for all $\omega\in\Omega$ and $\pi_\epsilon\in\cP^2_*(\Omega)$.
\end{lemma}
\begin{proof}
Clearly,  $\pi_{\epsilon,\omega}(\sigma)\geq0$ for all $\sigma,\omega\in\Omega$ for small enough $\epsilon>0$.
Moreover, since $\eta\in\cE$ by (\ref{eqetaVector}),
	\begin{align*}
	\sum_{\sigma\in\Omega}\pi_{\epsilon,\omega}(\sigma)&=1+\epsilon\sum_{\sigma\in\Omega}\scal{\eta}{e_\omega\tensor e_\sigma}
		=1+\epsilon\scal{\eta}{e_\omega\tensor \vecone}=1\qquad\mbox{for all }\omega\in\Omega.
	\end{align*}
Hence, $\pi_{\epsilon,\omega}\in\cP(\Omega)$ and $\pi_\epsilon\in\cP^2(\Omega)$.
Similarly, once more because $\eta\in\cE$, for each $\sigma\in\Omega$ we have
	\begin{align*}
	\frac1q\sum_{\omega\in\Omega}\pi_{\epsilon,\omega}(\sigma)&=\frac1q\sum_{\omega\in\Omega}
			\bc{\frac1q+\epsilon\scal{\eta}{e_\omega\tensor e_\sigma}}=\frac1q+\epsilon\scal{\eta}{\vecone\tensor e_\sigma}=\frac1q,
	\end{align*}
whence $\pi_{\epsilon}\in\cP^2_*(\Omega)$.
\end{proof}

Our next goal is to calculate $\cB (d,P, \pi_{\epsilon})$.
More precisely, we aim to expand $\cB (d,P, \pi_{\epsilon})$ to the fourth order in the limit $\epsilon\to0$.
The key tool for this expansion is the following elementary lemma, whose proof can be found in \Sec~\ref{Sec_lem:DerivsVanish}.

\begin{lemma}\label{lem:DerivsVanish}
Suppose $\ell \geq1$ and that $F : \cP(\Omega)^\ell \to (0,\infty)$, $(\rho_1,\ldots,\rho_\ell)\mapsto F(\rho_1,\ldots,\rho_\ell)$ has four continuous derivatives.
Moreover, setting $\bar a = (\bar\rho, \dots , \bar\rho)\in \cP(\Omega)^{\ell}$, assume that $F$ satisfies the following conditions.
\begin{description}
\item[T1] for all $a=(a_1,\ldots,a_\ell)\in\cP(\Omega)^\ell$, all $r\in[\ell]$ and all $c_1,c_2\in\Omega$ we have
	$$ \frac{\partial^2  F(a)}{\partial\rho_{r} (c_1) \partial\rho_{r}(c_2)} = 0.$$
\item[T2] there is $C_0\in\RR$ such that the gradient of $F$ at $\bar a$ satisfies
	$DF(\bar a)=C_0\vecone$.
\end{description}
Further, suppose that $\pi\in\cP^2_*(\Omega)$, let $\RHO,\RHO_1,\RHO_2,\ldots$ be mutually independent samples from $\pi$ and define 
	\begin{align}\label{eqJ}
	J&:\cP(\Omega)^\ell\to\RR,&
	(\rho_1,\ldots,\rho_\ell)&\mapsto
		\sum_{j=1}^4\sum_{r\in[\ell]^j,c\in\Omega^j}
		\frac1{j!}\frac{\partial^j\,\Lambda\circ F}{\partial\rho_{r_1}(c_1)\cdots\partial\rho_{r_j}(c_j)}(\bar a)
			\cdot\prod_{h=1}^j(\rho_{r_h}(c_h)-1/q).
	\end{align}
Then
	\begin{align*}
	\Erw[J(\RHO_1,\ldots,\RHO_\ell)]&=\frac{1}{24F(\bar a)}\sum_{r_1\neq r_2\in[\ell],c\in\Omega^4}
		\bc{\frac{\partial^2F(\bar a)}{\partial\rho_{r_1}(c_1)\partial\rho_{r_2}(c_3)}\frac{\partial^2F(\bar a)}{\partial\rho_{r_1}(c_2)\partial\rho_{r_2}(c_4)}
		+\frac{\partial^2F(\bar a)}{\partial\rho_{r_1}(c_2)\partial\rho_{r_2}(c_3)}\frac{\partial^2F(\bar a)}{\partial\rho_{r_1}(c_1)\partial\rho_{r_2}(c_4)}}\\
		&\qquad\qquad\qquad\qquad\qquad\qquad\qquad\cdot
			\left(\Erw\left[ \RHO(c_1) \RHO (c_2) \right]- q^{-2}\right) \left( \Erw\left[ \RHO (c_3) \RHO (c_4) \right] - q^{-2}\right).
	\end{align*}
\end{lemma}

\noindent
Equipped with \Lem~\ref{lem:DerivsVanish} we will derive the following asymptotic formula in \Sec~\ref{Sec_prop:TaylorBdpi}.

\begin{lemma} \label{prop:TaylorBdpi}
We have
$\cB (d, P,\pi_{\epsilon} )= \cB (d,P,\pi_0)  +\frac{ d (k-1)}{12}   \left( (k-1) d\hat\lambda^2 -\hat\lambda \right)\epsilon^4 + O (\epsilon^5)$ as $\epsilon\to0$.

\end{lemma}

\noindent
Finally, Proposition \ref{prop_KS} is immediate from \Lem~\ref{prop:TaylorBdpi}.

\begin{proof}[Proof of Proposition \ref{prop_KS}]
Due to {\bf SYM} it is straightforward to verify that	$\cB(d,P,\pi_0)=\ln q+\frac dk\ln\xi.$
Hence, if $0<d<\dc$, then
 $\cB (d,P, \pi_{\epsilon} )\leq \cB (d,P, \pi_0)$ for all small enough $\epsilon>0$ because $\pi_{\epsilon}\in\cP^2_*(\Omega)$ by \Lem~\ref{Lemma_sound}.
Therefore, \Lem~\ref{prop:TaylorBdpi} implies that
	$(k-1) d\hat\lambda^2 -\hat\lambda\leq0$.
As this bound holds for all $d<\dc$, we conclude that $(k-1)\dc\hat\lambda\leq1$, and thus the assertion follows from \Lem~\ref{Cor:MaxOpt}.
\end{proof}

\begin{remark}
A local expansion of the Bethe functional around the atom $\pi=\delta_{\bar\rho}$ on the uniform distribution was performed independently
by Guilhem Semerjian (manuscript in preparation), albeit with a different objective and without the realization that the
eigenvectors of $\Xi$ can be used to construct an explicit family of perturbations, cf.~(\ref{eqCreative}).
\end{remark}

\subsection{Proof of \Lem~\ref{Cor:MaxOpt}}\label{Sec_Cor:MaxOpt}
The canonical basis $(e_\omega)_{\omega\in\Omega}$ gives rise to the basis $(e_\sigma\tensor e_\tau)_{\sigma,\tau\in\Omega}$
of the $q^2$-dimensional space $\RR^\Omega\tensor\RR^\Omega$.
Hence, we can identify $\RR^\Omega\tensor\RR^\Omega$ with the space $\RR^{\Omega\times \Omega}$ of $q\times q$-matrices via the linear map
	$$\iota:\RR^\Omega\tensor\RR^\Omega\to\RR^{\Omega\times \Omega},\qquad
		\sum_{\sigma,\tau\in\Omega} a_{\sigma,\tau}\,e_\sigma\tensor e_\tau\mapsto
			\sum_{\sigma,\tau\in\Omega} a_{\sigma,\tau}e_\sigma e_\tau^\ast
			\qquad(a_{\sigma,\tau}\in\RR).$$
Since $\ker\iota=\{0\}$, $\iota$ is an isomorphism.
Moreover, if we equip the space $\RR^{\Omega\times\Omega}$ with the Frobenius inner product $\scal\nix\nix$, then
$\scal xy=\scal{\iota(x)}{\iota(y)}$ for all $x,y\in\RR^\Omega\tensor\RR^\Omega$.

By \Lem~\ref{Lemma_Xi} the linear operator $\Xi$ is self-adjoint and $\Xi\cE\subset\cE$.
Therefore, $\cE$ admits an orthogonal decomposition into eigenspaces of $\Xi$.
Suppose that $\lambda=\max\{|L|:L\in\Eig[\Xi]\}$ and let $\cE_\lambda\subset\cE$ be the corresponding eigenspace.
Moreover, consider the linear map defined by
	$\thet:\cE\to\cE$, $e_\sigma\tensor e_\tau\mapsto e_\tau\tensor e_\sigma$ for $\sigma,\tau\in\Omega$.
Due to the particular form (\ref{eqXi}) of $\Xi$ we have $\Xi\thet y=\thet\Xi y$ for all $y\in\cE$.
Consequently, $\thet\cE_\lambda\subset \cE_\lambda$.
Therefore, for any $z\in\cE_\lambda$ we have $\frac12(z+\thet(z))\in\cE_\lambda$.
Because $\thet^2=\id$, this means that there exists a unit vector $z\in\cE_\lambda$ such that $\thet z=z$.
Further, $\iota(z)$ is a symmetric matrix as $\thet z=z$ and $\iota(z)$ satisfies $\iota(z)\vecone=0$ and $\iota(z)x\perp\vecone$ for all $x\in\RR^\Omega$
because $z\in\cE$.
Thus,  there exists an orthonormal basis $u_1,\ldots,u_{q-1}$ of the space $\{x\in\RR^\Omega:x\perp\vecone\}$ and $w_1,\ldots,w_{q-1}\in\RR$
such that
	\begin{equation}\label{eqlem:MaxOpt1}
	\iota(z)=\sum_{i=1}^{q-1} w_i u_iu_i^*.
	\end{equation}
Since $\iota$ is an isomorphism, (\ref{eqlem:MaxOpt1}) yields the representation
	\begin{equation}\label{eqlem:MaxOpt2}
	z=\sum_{i=1}^{q-1}w_i u_i\tensor u_i.
	\end{equation}
Further, if we define $\DeltaM=\sum_{i=1}^{q-1}|w_i| u_i\tensor u_i$, then $\DeltaM\in \cE$ because $u_i\perp\vecone$ for all $i$.
Moreover, because $z$ is a unit vector and $u_1,\ldots,u_{q-1}$ are orthonormal,
	\begin{align}\label{eqlem:MaxOpt3}
	\|\DeltaM\|^2=\scal\DeltaM\DeltaM=\sum_{i,j=1}^{q-1}|w_iw_j|\scal{u_i}{u_j}^2=\sum_{i=1}^{q-1}w_i^2=\|z\|^2=1.
	\end{align}
Finally, once more due to the particular form (\ref{eqXi}) of $\Xi$, (\ref{eqlem:MaxOpt1}) yields
	\begin{align}\nonumber
	\lambda&=|\scal{\Xi z}z|=\abs{\sum_{i,j=1}^{q-1}w_iw_j\scal{\Xi u_i\tensor u_i}{u_j\tensor u_j}}
		=\abs{\sum_{i,j=1}^{q-1}w_iw_j\Erw\brk{\scal{\Phi_{\PSI} u_i}{u_j}^2}}\\
		&\leq\sum_{i,j=1}^{q-1}\abs{w_iw_j}\Erw\brk{\scal{\Phi_{\PSI} u_i}{u_j}^2}=
			\sum_{i,j=1}^{q-1}\abs{w_iw_j}\scal{\Xi u_i\tensor u_i}{u_j\tensor u_j}=\scal{\Xi\DeltaM}\DeltaM.
				\label{eqlem:MaxOpt4}
	\end{align}
Combining (\ref{eqlem:MaxOpt3}) and (\ref{eqlem:MaxOpt4}), we thus see that $\DeltaM$ is a unit vector with 
 $\scal{\Xi\DeltaM}\DeltaM=\lambda=\max\{|\scal{\Xi y}y|:y\in\cE,\|y\|=1\}$, as desired.
 
\subsection{Proof of \Lem~\ref{lem:DerivsVanish}}\label{Sec_lem:DerivsVanish}
We recall the following well-known generalization of the chain rule.

\begin{fact}[\Faadi 's formula]\label{Fact_Faadi}Suppose that $F:(\RR^\Omega)^j\to\infty$ has $j\ge1$ continuous derivatives. Let $\Pi(j)$ be the set of all partitions of $[j]$, denote by $|\Upsilon|$ the cardinality of a partition $\Upsilon\in\Pi(j)$ and similarly let $|B|$ denote the cardinality of a set $B\in \Upsilon$ in the partition $\Upsilon$. Then
\begin{equation}
\frac{\partial^j \Lambda (F(x_1,\ldots,x_j))}{\partial x_1 \dots \partial x_j} = \sum_{\Upsilon \in \Pi (j)} \Lambda^{(|\Upsilon|)}(F(x_1,\ldots,x_j)) \prod_{B \in \Upsilon} \frac{\partial ^{|B|}F(x_1,\ldots,x_j)}{\prod_{i \in B}\partial x_i}. \label{eq:FaadiBrunoGeneral}
\end{equation}
\end{fact}

\noindent
For $r\in[\ell]^j$ and $c\in\Omega^j$ let
	\begin{align*}
	\cJ_{r,c}&=\frac{\partial^j\,\Lambda\circ F}{\partial\rho_{r_1}(c_1)\cdots\partial\rho_{r_j}(c_j)}(\bar a)
			\cdot\Erw\brk{\prod_{h=1}^j(\RHO_{r_h}(c_h)-1/q)}.
	\end{align*}
Because $\RHO_1,\ldots,\RHO_\ell$ are mutually independent with mean $\bar\rho$, we have 
$\cJ_{r,c}=0$ unless for each $i\in[j]$ there is $h\in[j]\setminus\{i\}$ such that $r_i=r_h$.
Hence, setting $R_j =\{ r \in [\ell]^j \;:\; \forall i\in[j]\exists h\in[j]\setminus\{i\}:r_i=r_h\}$, we see that
	\begin{equation}\label{eqMeanZero}
	\cJ_j=\sum_{r\in[\ell]^j,c\in\Omega^j}\cJ_{r,c}=\sum_{r\in R_j,c\in\Omega^j}\cJ_{r,c}.
	\end{equation}
In particular, (\ref{eqMeanZero}) implies
	\begin{equation}\label{eqMeanZero1}
	\cJ_1=0.
	\end{equation}

Proceeding to $j=2$, we apply Fact~\ref{Fact_Faadi} to obtain
	\begin{align}\label{eqMeanZero2}
	\frac{\partial^2\Lambda\circ F}{\partial\rho_{r_1}(c_1)\partial\rho_{r_2}(c_2)}(\bar a)
		&=\Lambda''(F(\bar a))\frac{\partial F}{\partial\rho_{r_1}(c_1)}\frac{\partial F}{\partial\rho_{r_2}(c_2)}(\bar a)
			+\Lambda'(F(\bar a))\frac{\partial^2 F}{\partial\rho_{r_1}(c_1)\partial\rho_{r_2}(c_2)}(\bar a).
	\end{align}
Since $R_2=\{(r,r):r\in[\ell]\}$, {\bf T1} and (\ref{eqMeanZero2}) entail that
	\begin{align}
	\cJ_2&=\Lambda''(F(\bar a))\sum_{r=1}^\ell\sum_{c_1,c_2\in\Omega}
			\frac{\partial F(\bar a)}{\partial\rho_{r}(c_1)}\frac{\partial F(\bar a)}{\partial\rho_{r}(c_2)}
				\Erw\brk{(\RHO_r(c_1)-1/q)(\RHO_r(c_2)-1/q)}\nonumber\\
		&=C_0^2\Lambda''(F(\bar a))\ell\cdot\Erw\brk{\sum_{c_1,c_2\in\Omega}(\RHO(c_1)-1/q)(\RHO(c_2)-1/q)}
			&&\mbox{[due to {\bf T2}]}\nonumber\\
		&=C_0^2\Lambda''(F(\bar a))\ell\cdot\Erw\brk{\bc{\sum_{c\in\Omega}(\RHO(c)-1/q)}^2}
			=0&&\mbox{[as $\textstyle \sum_{c\in\Omega}\RHO(c)=1$]}.
			\label{eqMeanZero3}
	\end{align}

Moving on to $\cJ_3$, we observe that $R_3=\{(r,r,r):r\in[\ell]\}$.
Moreover, Fact~\ref{Fact_Faadi} yields
	\begin{align*}
	\frac{\partial^3 \Lambda\circ F}{\partial \rho_r(c_1) \partial \rho_r(c_2) \partial \rho_r(c_3)}& =
		 \Lambda'(F(\bar a)) \frac{\partial^3 F}{\partial \rho_r(c_1) \partial \rho_r(c_2) \partial \rho_r(c_3)}\\
		 	&\quad + \Lambda''(F(\bar a))\left( \frac{\partial F}{\partial \rho_r(c_1)} \frac{\partial^2 F}{ \partial \rho_r(c_2) \partial \rho_r(c_3)} + 
				\frac{\partial F}{\partial \rho_r(c_2)} \frac{\partial^2 F}{ \partial \rho_r(c_1) \partial \rho_r(c_3)}+
				\frac{\partial F}{\partial \rho_r(c_3)} \frac{\partial^2 F}{ \partial \rho_r(c_1) \partial \rho_r(c_2)} \right)\\
				&\quad+\Lambda'''(F(\bar a)) \frac{\partial F}{\partial \rho_r(c_1)} \frac{\partial F}{ \partial \rho_r(c_2)}\frac{\partial F}{ \partial \rho_r(c_3)}\\
		&=\Lambda'''(F(\bar a)) \frac{\partial F}{\partial \rho_r(c_1)} \frac{\partial F}{ \partial \rho_r(c_2)}\frac{\partial F}{ \partial \rho_r(c_3)}
		\qquad\mbox{[due to {\bf T1}]}.
	\end{align*}
Hence, {\bf T2} yields
	\begin{align}\nonumber
	\cJ_3&=\Lambda'''(F(\bar a))\sum_{r\in[\ell],c_1,c_2,c_3\in\Omega}
		\frac{\partial F(\bar a)}{\partial \rho_r(c_1)} \frac{\partial F(\bar a)}{ \partial \rho_r(c_2)}\frac{\partial F(\bar a)}{ \partial \rho_r(c_3)}
			\Erw\brk{\prod_{h=1}^3(\RHO(c_h)-1/q)}\\
		&=\ell C_0^3\Lambda'''(F(\bar a))\cdot\Erw\brk{
		\bc{\sum_{c\in\Omega}(\RHO(c)-1/q)}^3}=0	&		\mbox{[as $\textstyle \sum_{c\in\Omega}\RHO(c)=1$]}.
			\label{eqMeanZero4}
	\end{align}

Finally, we come to $\cJ_4$.
Fact \ref{Fact_Faadi} yields
	\begin{align}
	\frac{\partial^4 \Lambda\circ F}{\partial \rho_{r_1}(c_1)\cdots\partial \rho_{r_4}(c_4)}
		&=\Lambda'(F(\bar a))\frac{\partial^4 F}{\partial \rho_{r_1}(c_1)\cdots\partial \rho_{r_4}(c_4)}\nonumber\\
		&\quad+\Lambda''(F(\bar a))\sum_{i\in [4]}\frac{\partial F}{\partial\rho_{r_{i}}(c_{i})}
			\frac{\partial^3 F}{\prod_{j \in [4]\setminus \{i \}}\partial\rho_{r_{j}}(c_{j})}\nonumber\\
		&\quad+\Lambda''(F(\bar a))\sum_{i,j \in [4], i<j}
			\frac{\partial^2 F}{\partial\rho_{r_{i}}(c_{i})\partial\rho_{r_{j}}(c_{j})}
			\frac{\partial^2 F}{\prod_{\ell \in [4]\setminus \{i,j\} }\partial\rho_{r_{\ell}}(c_{\ell})}\nonumber\\
		&\quad+\Lambda'''(F(\bar a))
			\sum_{i,j \in [4], i<j}
			\frac{\partial F}{\partial\rho_{r_{i}}(c_{i})}
			\frac{\partial F}{\partial\rho_{r_{j}}(c_{j})}
			\frac{\partial^2 F}{\prod_{\ell \in [4]\setminus \{i,j\} }\partial\rho_{r_{\ell}}(c_{\ell})}\nonumber\\
		&\quad+\Lambda''''(F(\bar a))			
			\frac{\partial F}{\partial\rho_{r_{1}}(c_{1})}
			\frac{\partial F}{\partial\rho_{r_{2}}(c_{2})}
			\frac{\partial F}{\partial\rho_{r_{3}}(c_{3})}
			\frac{\partial F}{\partial\rho_{r_{4}}(c_{4})}.\label{eqMeanZero14}
	\end{align}
Since $R_4=\{(r_1,r_2,r_3,r_4)\in[\ell]^4:|\{r_1,r_2,r_3,r_4\}|\leq2\}$, {\bf T1} implies that
	\begin{align}
	\frac{\partial^4 F}{\partial \rho_{r_1}(c_1)\cdots\partial \rho_{r_4}(c_4)} &= 0 \qquad\mbox{and}\qquad
		\frac{\partial^3 F}{\prod_{j \in [4]\setminus \{i \}}\partial\rho_{r_{j}}(c_{j})}=0
	& \mbox{for all }r\in R_4, i \in [4].\label{eqMeanZero15}
	\end{align}
Moreover, similarly as before {\bf T2} implies
	\begin{align}\nonumber
	\Lambda''''(F(\bar a))&\sum_{r\in R_4,c\in\Omega^4}
			\frac{\partial F}{\partial\rho_{r_{1}}(c_{1})}
			\frac{\partial F}{\partial\rho_{r_{2}}(c_{2})}
			\frac{\partial F}{\partial\rho_{r_{3}}(c_{3})}
			\frac{\partial F}{\partial\rho_{r_{4}}(c_{4})}\Erw\brk{\prod_{h=1}^4(\RHO_{r_h}(c_h)-1/q)}\\
			&=C_0^4\Lambda''''(F(\bar a))\sum_{r\in R_4}\Erw\brk{\prod_{h=1}^4\bc{\sum_{c\in\Omega}\RHO_{r_h}(c)-1/q)}}=0 
			& \mbox{[as $\textstyle \sum_{c\in\Omega}\RHO(c)=1$]}.\label{eqMeanZero16}
	\end{align}
Analogously, once more by {\bf T2}
	\begin{align}
	\Lambda'''(F(\bar a))&
			\sum_{r\in R_4,c\in\Omega^4}
			\sum_{i,j \in [4], i<j}
			\frac{\partial F}{\partial\rho_{r_{i}}(c_{i})}
			\frac{\partial F}{\partial\rho_{r_{j}}(c_{j})}
			\frac{\partial^2 F}{\prod_{\ell \in [4]\setminus \{i,j\} }\partial\rho_{r_{\ell}}(c_{\ell})}
				\Erw\brk{\prod_{h=1}^4(\RHO_{r_h}(c_h)-1/q)}\nonumber\\
			&=C_0^2\Lambda'''(F(\bar a))
					\sum_{i,j \in [4], i<j}\;\sum_{r\in R_4,c\in\Omega^4}
			\frac{\partial^2 F}{\prod_{\ell \in [4]\setminus \{i,j\} }\partial\rho_{r_{\ell}}(c_{\ell})}
				\Erw\brk{\prod_{h=1}^4(\RHO_{r_h}(c_h)-1/q)}\nonumber\\
			&=C_0^2\Lambda'''(F(\bar a))\sum_{i,j \in [4], i<j }\;\sum_{r\in R_4,c_{i_3},c_{i_4}\in\Omega} 
				\frac{\partial^2 F}{\prod_{\ell \in [4]\setminus \{i,j\} }\partial\rho_{r_{\ell}}(c_{\ell})}
				\Erw\brk{\prod_{h=1}^2\bc{\sum_{c\in \Omega}(\RHO_{r_h}(c)-1/q)}\prod_{h=3}^4(\RHO_{r_h}(c_h)-1/q)}\nonumber\\&=0.
					\label{eqMeanZero16}
	\end{align}
Thus, combining (\ref{eqMeanZero14})--(\ref{eqMeanZero16}), we obtain
	\begin{align}\nonumber
	\cJ_4&=\Lambda''(F(\bar a))\sum_{r\in R_4,c\in\Omega^4}
			\sum_{i,j \in [4], i<j}
			\frac{\partial^2 F}{\partial\rho_{r_{i}}(c_{i})\partial\rho_{r_{j}}(c_{j})}
			\frac{\partial^2 F}{\prod_{\ell \in [4]\setminus \{i,j\} }\partial\rho_{r_{\ell}}(c_{\ell})}\Erw\brk{\prod_{h=1}^4(\RHO_{r_h}(c_h)-1/q)}\\
		&=\Lambda''(F(\bar a))\sum_{r_1\neq r_2\in[\ell],c\in\Omega^4}
		\bc{\frac{\partial^2F}{\partial\rho_{r_1}(c_1)\partial\rho_{r_2}(c_3)}\frac{\partial^2F}{\partial\rho_{r_1}(c_2)\partial\rho_{r_2}(c_4)}
		+\frac{\partial^2F}{\partial\rho_{r_1}(c_2)\partial\rho_{r_2}(c_3)}\frac{\partial^2F}{\partial\rho_{r_1}(c_1)\partial\rho_{r_2}(c_4)}}\nonumber\\
		&\qquad\qquad\qquad\qquad\qquad\qquad\qquad\cdot
			\left(\Erw\left[ \RHO(c_1) \RHO (c_2) \right]- q^{-2}\right) \left( \Erw\left[ \RHO (c_3) \RHO (c_4) \right] - q^{-2}\right)
				\qquad\qquad\qquad\mbox{[due to {\bf T1}]}.
					\label{eqMeanZero20}
	\end{align}
Since $\Erw[J(\RHO_1,\ldots,\RHO_\ell)]=\sum_{j=1}^4\frac1{j!}\cJ_j$ and $\Lambda''(x)=1/x$, the assertion follows from (\ref{eqMeanZero1}), (\ref{eqMeanZero3}), (\ref{eqMeanZero4}) and~(\ref{eqMeanZero20}).

\subsection{Proof of \Lem~\ref{prop:TaylorBdpi}}\label{Sec_prop:TaylorBdpi}
Recall that $\hat \lambda= \max_{\lambda\in\Eig[\Xi]}|\lambda|$, that $\DeltaM\in\cE$ is an eigenvector of $\Xi$ with eigenvalue $\hat\lambda$,
and that $\eta$ is the vector defined by \eqref{eqetaVector}.
We tacitly assume that $\epsilon$ is small enough so that $\pi_\epsilon\in\cP_2^*(\Omega)$ (cf.\ \Lem~\ref{Lemma_sound})
and we denote by $\RHO,\RHO_1,\RHO_2,\ldots$ independent samples from $\pi_\epsilon$.
Hence, for any function $X:(\RR^{\Omega})^{\ell}\to\RR$ the expectation $\Erw[X(\RHO_1,\ldots,\RHO_\ell)]$ can be viewed as a function of $\epsilon$.
Further, since $\pi_\epsilon$ is the uniform distribution on the distributions $\pi_{\epsilon,\omega}$ from (\ref{eqCreative}), which are atoms,
the function $\epsilon\mapsto\Erw[X(\RHO_1,\ldots,\RHO_\ell)]$ has the same continuity as $X$.

Ultimately we are going to expand the function $\epsilon\mapsto\cB(d,P,\pi_\epsilon)$ to the fourth order.
But first we need a few preparations.
First we observe that $\DeltaM$ encodes the covariance matrix of the random vector $(\RHO(\omega))_{\omega\in\Omega}$.

\begin{claim}\label{Fact_EtaEtaT}
We have $\Erw[ \mathbold{\rho}(c_1) -q^{-1}]=0$ and
	$\Erw[(\mathbold{\rho}(c_1) - q^{-1})( \mathbold{\rho} (c_2) -q^{-1})]
		=q^{-1}  \epsilon^2\scal{\DeltaM}{e_{c_1}\tensor e_{c_2}}$ 		 for all $c_1,c_2 \in \Omega$. 
\end{claim}
\begin{proof}
The first assertion follows from \Lem~\ref{Lemma_sound}, which shows that $\pi_\epsilon\in\cP^2_*(\Omega)$.
Moreover, because the vectors $u_1,\ldots,u_{q-1}\in\cE$ from (\ref{eqetaVector}) are orthonormal, (\ref{eqDeltaVector}) and (\ref{eqCreative}) yield
	  \begin{align*}
	q\epsilon^{-2}\Erw&\left[ \left(\mathbold{\rho}(c_1) - q^{-1} \right) \left( \mathbold{\rho} (c_2) -q^{-1} \right) \right]
	 =\sum_{\omega \in \Omega} \scal{\eta}{e_\omega\tensor e_{c_1}}\scal{\eta}{e_\omega\tensor e_{c_2}}\\
	 &=\sum_{i,j=1}^{q-1}\sqrt{\bar\lambda_i\bar\lambda_j}\sum_{\omega\in\Omega}
	 	\scal{u_i\tensor u_i}{e_\omega\tensor e_{c_1}}	\scal{u_j\tensor u_j}{e_\omega\tensor e_{c_2}}
	=\sum_{i,j=1}^{q-1}\sqrt{\bar\lambda_i\bar\lambda_j}\scal{u_i}{e_{c_1}}\scal{u_j}{e_{c_2}}\sum_{\omega\in\Omega}
	 	\scal{u_i}{e_\omega}	\scal{u_j}{e_\omega}\\
	&=\sum_{i,j=1}^{q-1}\sqrt{\bar\lambda_i\bar\lambda_j}\scal{u_i}{e_{c_1}}\scal{u_j}{e_{c_2}}
	 	\scal{u_i}{u_j}
		=\sum_{i=1}^{q-1}\lambda_i\scal{u_i\tensor u_i}{e_{c_1}\tensor e_{c_2}}
		=\scal\DeltaM{e_{c_1}\tensor e_{c_2}},
\end{align*}
as claimed.
\end{proof}

\noindent
Additionally, we need the following algebraic relation.

\begin{claim}\label{Claim_deconstruct}
For any $\psi\in\Psi$ we have
$\langle \left( \Phipsi \otimes \Phipsi \right)  \DeltaM, \DeltaM \rangle = \sum_{c \in \Omega^4}  \Phipsi(c_1,c_3)  \Phipsi(c_2,c_4) 
	\scal{\DeltaM}{e_{c_1}\tensor e_{c_2}}\scal{\DeltaM}{e_{c_3}\tensor e_{c_4}}$. 
\end{claim}
\begin{proof}
Since $\DeltaM=\sum_{i\in\Omega}\bar\lambda_iu_i\tensor u_i$  we have
	\begin{align*}
	\scal{(\Phipsi\tensor\Phipsi)\DeltaM}\DeltaM&=\sum_{i,j\in\Omega}\bar\lambda_i\bar\lambda_j\scal{(\Phipsi\tensor\Phipsi)(u_i\tensor u_i)}{(u_j\tensor u_j)}
			=\sum_{i,j\in\Omega}\bar\lambda_i\bar\lambda_j\scal{\Phipsi u_i}{u_j}^2\\
		&=\sum_{i,j\in\Omega}\bar\lambda_i\bar\lambda_j\bc{\sum_{c\in\Omega}\scal{\Phipsi u_i}{e_c}\scal{u_{j}}{e_c}}^2
			=\sum_{i,j\in\Omega}\bar\lambda_i\bar\lambda_j\bc{\sum_{c,c'\in\Omega}\Phi_{\psi}(c,c')\scal{u_{i}}{e_{c'}}\scal{u_{j}}{e_c}}^2\\
		&=\sum_{i,j\in\Omega}\sum_{c\in\Omega^4}\bar\lambda_i\bar\lambda_j
			\Phi_{\psi}(c_1,c_3)\Phi_{\psi}(c_2,c_4)\scal{u_j}{e_{c_1}}\scal{u_j}{e_{c_2}}\scal{u_i}{e_{c_3}}\scal{u_i}{e_{c_4}}\\
		&=\sum_{c\in\Omega^4}\Phi_{\psi}(c_1,c_3)\Phi_{\psi}(c_2,c_4)
			\bc{\sum_{j\in\Omega}\bar\lambda_j\scal{u_j\tensor u_j}{e_{c_1}\tensor e_{c_2}}}
			\bc{\sum_{i\in\Omega}\bar\lambda_i\scal{u_i\tensor u_i}{e_{c_3}\tensor e_{c_4}}}\\
	&=\sum_{c\in\Omega^4}\Phi_{\psi}(c_1,c_3)\Phi_{\psi}(c_2,c_4)\scal{\DeltaM}{e_{c_1}\tensor e_{c_2}}\scal{\DeltaM}{e_{c_3}\tensor e_{c_4}},
	\end{align*}
as claimed.
\end{proof}

We proceed to expand $\epsilon\mapsto\cB(d,P,\pi_\epsilon)$.
For $\psi,\psi_1,\ldots,\psi_\gamma\in\Psi$ let
	\begin{align*}
	B_1(\psi_1,\ldots,\psi_\gamma)&=
		\Erw\left[ \Lambda\left(\sum_{h\in[q]}\prod_{j=1}^\gamma\sum_{\tau\in\Omega^k}\vecone\{\tau_k=h\}\psi_j(\tau)\prod_{i=1}^{k-1}\vec{\rho}_{k(j-1)+i}(\tau_i)\right) \right],&
	B_2(\psi)&=\Erw\left[  \Lambda\left(\sum_{\tau\in\Omega^k}\psi(\tau)\prod_{i=1}^k \vec\rho_i(\tau_i)\right)\right].
	\end{align*}
Then with $\PSI,\PSI_1,\PSI_2,\ldots$ chosen independently from $P$,
	\begin{align}\label{eqBfinal}
	\cB(d,P,\pi_\epsilon)&=\frac1q\Erw\brk{\xi^{-\GAMMA}B_1(\PSI_1,\ldots,\PSI_\GAMMA)}
		-\frac{d(k-1)}{k\xi}\Erw\brk{B_2(\PSI)}
	\end{align}
and we shall derive the approximations to both summands separately, using \Lem~\ref{lem:DerivsVanish} in either case.

\begin{claim}\label{Claim_eqB1final}
We have
	\begin{align}\label{eqB1final}
	\Erw\brk{q^{-1}\xi^{-\GAMMA}B_1(\PSI_1,\ldots,\PSI_\GAMMA)}&=
		\ln q+d\ln\xi+\frac{\eps^4d(k-1)}{12}\brk{(k-2)\scal{\Xi\DeltaM}\DeltaM+d(k-1)\scal{\Xi^2\DeltaM}\DeltaM}+O(\eps^5).
	\end{align}
\end{claim}
\begin{proof}
Fixing $\gamma$ and $\psi_1,\ldots,\psi_\gamma$ for the moment, we consider the function
	\begin{align*}
	F_{\psi_1,\ldots,\psi_\gamma}&: \cP (\Omega)^{(k-1)\gamma}\to (0,\infty),&
		(\rho_{1,1}, \dots , \rho_{\gamma,k-1}) &\mapsto 
	\sum_{h \in \Omega} \prod_{j=1}^{\gamma} \sum_{\tau \in \Omega^k} \vecone \{\tau_k = h\} \psi_j (\tau) \prod_{i=1}^{k-1} \rho_{j,i}(\tau_i) .
	\end{align*}
Then with $J_{\psi_1,\ldots,\psi_\gamma}$ denoting the fourth Taylor polynomial of $\Lambda\circ F_{\psi_1,\ldots,\psi_{\gamma}}$ as in equation (\ref{eqJ}), we can write $\Lambda\circ F_{\psi_1,\ldots,\psi_\gamma}=J_{\psi_1,\ldots,\psi_\gamma}+R_{\psi_1,\ldots,\psi_\gamma}$.
We are going to show that, with $\PSI_1,\ldots,\PSI_\gamma$ chosen from $P$ and $(\RHO_{i,j})_{i,j}$ chosen from $\pi_\eps$, all mutually independent,
	\begin{align}\label{eqHotFix1}
	\Erw[J_{\PSI_1,\ldots,\PSI_\gamma}(\RHO_{i,j})_{i,j}]&=\Lambda(q\xi^\gamma)+\frac{q\xi^\gamma\eps^4(k-1)}{12}
		\brk{d(k-2)\scal{\Xi\DeltaM}\DeltaM+d^2(k-1)\scal{\Xi^2\DeltaM}\DeltaM},\\
	\Erw[R_{\PSI_1,\ldots,\PSI_\gamma}(\RHO_{i,j})_{i,j}]&=O(\eps^5)\exp(O(\gamma)),\label{eqHotFix0}
	\end{align}
whence (\ref{eqB1final}) is immediate because the Poisson distribution has sub-exponential tails.

To prove (\ref{eqHotFix1}) we apply \Lem~\ref{lem:DerivsVanish}.
Thus, we need the first and second partial derivatives of $F_{\psi_1,\ldots,\psi_\gamma}$.
To work out the first partial derivatives, let $s\in[\gamma]$, $r\in[k-1]$ and $c_1\in\Omega$.
Then
	\begin{align*}
	\frac{\partial F_{\psi_1,\ldots,\psi_\gamma}}{\partial \rho_{s,r}(c_1)}&=
		 \sum_{h \in \Omega}
		 	\left(\prod_{j \in [\gamma]\setminus \{ s\}} \sum_{\tau \in \Omega^k} \vecone \{\tau_k = h\} \psi_j (\tau) \prod_{i=1}^{k-1} \rho_{j,i}(\tau_i) \right) \left(\sum_{\tau \in \Omega^k} \vecone \{\tau_k = h, \; \tau_r = c_1\} \psi_s (\tau) \prod_{i\in [k-1]\setminus \{ r\}} \rho_{s,i}(\tau_i) \right) .
	\end{align*}
In particular, {\bf SYM} yields
	$
	\frac{\partial F_{\psi_1,\ldots,\psi_\gamma}}{\partial \rho_{s,r}(c_1)}(\bar\rho,\ldots,\bar\rho)=q\xi^\gamma$,
and thus the assumptions {\bf T1}--{\bf T2} of \Lem~\ref{lem:DerivsVanish} are satisfied.
With respect to the second derivatives, there are two cases.
First,  fix $s\in[\gamma]$, distinct $r_1,r_2\in[k-1]$ and $c_1,c_3\in\Omega$.
Let $\theta_1 : [k] \mapsto [k]$ be the permutation such that $\theta_1 (r_1) = 1,\; \theta_1 (r_2) = 2$ and $\theta(i)=i$ for all $i\neq r_1,r_2$.
Using {\bf SYM}, we obtain
	\begin{align*}
	\frac{\partial^2 F_1(\bar\rho,\ldots,\bar\rho)}{\partial \rho_{s,r_1}(c_1) \partial \rho_{s,r_2}(c_3)} 
		 &= \sum_{h \in \Omega} 
		 	\left(\prod_{j \in [\gamma]\setminus \{s\}} \sum_{\tau \in \Omega^k} \vecone \{\tau_k = h\} \psi_j (\tau) q^{1-k} \right)
			 \left(\sum_{\tau \in \Omega^k} \vecone \{\tau_k = h,\tau_{r_1} = c_1,\tau_{r_2} = c_3 \} \psi_s (\tau) q^{3-k} \right)  \\
		 &= \xi^{\gamma -1}q^{3-k} \sum_{h \in [q]} \sum_{\tau \in \Omega^k} \vecone \{\tau_k = h,\tau_{r_1} = c_1,\tau_{r_2} = c_3 \} \psi_s (\tau) \\
		  &= \xi^{\gamma -1}q^{3-k} \sum_{\tau \in \Omega^k} \vecone \{\tau_{r_1} = c_1,\tau_{r_2} = c_3 \} \psi_s (\tau) 
			 = \xi^{\gamma}q^{2} \Phi_{\psi_s^{\theta_1}}(c_1,c_3).
\end{align*}
Second, fix distinct $s,s'\in[\gamma]$ and any $r_1,r_2\in[k-1]$, $c_1,c_3\in\Omega$.
Let $\theta_2, \theta_3$ be the permutations such that $\theta_2 (k) = 2,\; \theta_2 (r_1) = 1$ and $\theta_2(i)=i$ for all $i\neq r_1,k$
and $\theta_3(k)=1, \theta_3(r_2) = 2$ and $\theta_3(i)=i$ for all $i\neq r_2,k$.
Then {\bf SYM} yields
\begin{align*}
\frac{\partial^2 F_1(\bar\rho,\ldots,\bar\rho)}{\partial \rho_{s,r_1}(c_1) \partial \rho_{s',r_2}(c_3)} 
	&= \sum_{h \in \Omega} \left(\prod_{j \in [\gamma]\setminus \{s,s' \}} \sum_{\tau \in \Omega^k} \vecone \{\tau_k = h\} \psi_j (\tau) q^{1-k} \right)\\
	&\qquad\qquad\qquad\cdot \left(\sum_{\tau \in \Omega^k} \vecone \{\tau_k = h,\tau_{r_1} = c_1\} \psi_s (\tau) q^{2-k} \right) \left(\sum_{\tau \in \Omega^k} \vecone \{\tau_k = h,\tau_{r_2} = c_3 \} \psi_{s'} (\tau) q^{2-k} \right) \\
	 &= \xi^{\gamma -2}q^{4-2k} \sum_{h \in \Omega} \left(\sum_{\tau \in \Omega^k} \vecone \{\tau_k = h,\tau_{r_1} = c_1\} \psi_s (\tau) \right) \left(\sum_{\tau \in \Omega^k} \vecone \{\tau_k = h,\tau_{r_2} = c_3 \} \psi_{s'} (\tau)  \right) \\
  &= \xi^{\gamma}q^{2} \sum_{h \in [q]} \Phipsisub[\theta_2]{s}(h,c_1) \Phipsisub[\theta_3]{s'}(c_3,h)= \xi^{\gamma}q^{2} 
  	\left(\Phi_{\psi_s^{\theta_2}}\cdot\Phi_{\psi_{s'}^{\theta_3}} \right)(c_1,c_3).
\end{align*}
Hence, \Lem~\ref{lem:DerivsVanish} gives 
	\begin{align*}
	\Erw[J_{\psi_1,\ldots,\psi_\gamma}(\RHO_{i,j})]&=
		\Lambda(\xi^\gamma q)+\frac{q\xi^\gamma\epsilon^4}{24}\bc{(k-1)(k-2)S_1+(k-1)^2S_2)},\qquad\mbox{where}\\
	S_1&=\sum_{s \in [\gamma]} \sum_{c\in [q]^4}  \left( \Phi_{\psi_s^{\theta_1}}(c_1,c_3) \Phi_{\psi_s^{\theta_1}}(c_2,c_4) +\Phi_{\psi_s^{\theta_1}}(c_2,c_3) \Phi_{\psi_s^{\theta_1}}(c_1,c_4)\right)  
				\scal{\DeltaM}{e_{c_1}\tensor e_{c_2}}\scal{\DeltaM}{e_{c_3}\tensor e_{c_4}},\\
	S_2&=\sum_{s,s' \in [\gamma]:s\neq s'} \sum_{c\in [q]^4}  \scal{\DeltaM}{e_{c_1}\tensor e_{c_2}}\scal{\DeltaM}{e_{c_3}\tensor e_{c_4}} 
		 \bigg[\bc{(\Phi_{\psi_s^{\theta_2}}\cdot\Phi_{\psi_{s'}^{\theta_3}}}(c_1,c_3)
		 		\cdot\bc{\Phi_{\psi_s^{\theta_2}}\cdot\Phi_{\psi_{s'}^{\theta_3}}}(c_2,c_4)\\
		&\qquad\qquad\qquad\qquad\qquad\qquad\qquad\qquad\qquad\qquad\qquad\qquad
			 +\bc{\Phi_{\psi_s^{\theta_2}}\cdot\Phi_{\psi_{s'}^{\theta_3}}}(c_2,c_3)\cdot\bc{\Phi_{\psi_s^{\theta_2}}\cdot\Phi_{\psi_{s'}^{\theta_3}}}(c_1,c_4)\bigg].
	\end{align*}
Further, Claim~\ref{Claim_deconstruct} 
yields
	\begin{align*}
	S_1&=2\sum_{s \in [\gamma]}\scal{\bc{\Phi_{\psi_s^{\theta_1}}\tensor\Phi_{\psi_s^{\theta_1}}}{\DeltaM}}{{\DeltaM}},\\
	S_2&=2\sum_{s,s' \in [\gamma]:s\neq s'}\scal{\bc{\bc{\Phi_{\psi_s^{\theta_2}}\cdot\Phi_{\psi_{s'}^{\theta_3}}}
		\tensor\bc{\Phi_{\psi_s^{\theta_2}}\cdot\Phi_{\psi_{s'}^{\theta_3}}}}{\DeltaM}}{{\DeltaM}}
		=2\sum_{s,s' \in [\gamma]:s\neq s'}\scal{\bc{\bc{\Phi_{\psi_s^{\theta_2}}\tensor\Phi_{\psi_s^{\theta_2}}}
		\bc{\Phi_{\psi_{s'}^{\theta_3}}\tensor\Phi_{\psi_{s'}^{\theta_3}}}}{\DeltaM}}{{\DeltaM}}.	
		\end{align*}
Therefore, since {\bf SYM} provides that the distribution $P$ is invariant under permutations,  
	\begin{align*}
	\Erw[J_{\PSI_1,\ldots,\PSI_\gamma}(\RHO_{i,j})_{i,j}]&=\Lambda(\xi^\gamma q)
		+\frac{\eps^4q\xi^\gamma(k-1)(k-2)}{12}\Erw\brk{\sum_{s=1}^{\GAMMA}\scal{(\Phi_{\PSI_s}\tensor\Phi_{\PSI_s})\DeltaM}{\DeltaM}}\\
		&\qquad\qquad\qquad+\frac{\eps^4\xi^\gamma q(k-1)^2}{12}\Erw\brk{\sum_{s,s'\in[\GAMMA]:s\neq s'}
			\scal{(\Phi_{\PSI_s}\tensor\Phi_{\PSI_s})(\Phi_{\PSI_{s'}}\tensor\Phi_{\PSI_{s'}})\DeltaM}{\DeltaM}}\\
		&=\Lambda(\xi^\gamma q)+\frac{\eps^4q\xi^\gamma(k-1)}{12}\brk{d(k-2)\scal{\Xi\DeltaM}\DeltaM+d^2(k-1)\scal{\Xi^2\DeltaM}\DeltaM},
	\end{align*}
which completes the proof of (\ref{eqHotFix1}).

Moving on to (\ref{eqHotFix0}), we write the remainder $R_{\psi_1,\ldots,\psi_\gamma}$ for $\rho_{i,j}$ in the support of $\pi_\eps$ as
	\begin{align}\label{eqHotFix10}
	R_{\psi_1,\ldots,\psi_\gamma}(\rho_{i,j})&=\sum_{h\in([\gamma]\times[k-1])^5,c\in\Omega^5}
		\frac1{5!}\frac{\partial \Lambda\circ F_{\psi_1,\ldots,\psi_\gamma}(\tilde\rho)}{\partial\rho_{h_1}(c_1)\cdots\partial\rho_{h_5}(c_5)}
			\prod_{i=1}^5(\tilde\rho_{h_i}(c_i)-q^{-1}),
	\end{align}
where $\tilde\rho=(\tilde\rho_{i,j})_{i,j}$ is a point on the line segment between the points $(\bar\rho,\ldots,\bar\rho)$ and $(\rho_{i,j})_{i,j}$.
In particular,
	$\prod_{i=1}^5(\tilde\rho_{h_i}(c_i)-q^{-1})=O(\eps^5).$
Hence, Fact \ref{Fact_Faadi} shows that
	\begin{align*}
	R_{\psi_1,\ldots,\psi_\gamma}(\rho_{i,j})&=O(\eps^5)\cdot
	\sum_{h,c}
		\sum_{\Upsilon \in \Pi (5)} \sup_{\tilde\rho}
			\Lambda^{(|\Upsilon|)}(F_{\psi_1,\ldots,\psi_\gamma}(\tilde{\rho})) \prod_{B \in \Upsilon} \frac{\partial ^{|B|}
				F_{\psi_1,\ldots,\psi_\gamma}(\tilde{\rho})}{\prod_{i \in B}\partial \rho_{h_i}(c_i)},
	\end{align*}
where $\tilde\rho$ ranges over the convex hull of the support of $\pi_\eps$.
Because all weight functions take values in the interval $(0,2)$, we find
$\prod_{B \in \Upsilon}(\partial ^{|B|}F_{\psi_1,\ldots,\psi_\gamma}(\tilde{\rho})/{\prod_{i \in B}\partial x_i})=\exp(O(\gamma))$.
In addition, 
	\begin{align*}
	\Lambda'(F_{\psi_1,\ldots,\psi_\gamma}(\tilde{\rho}))&=1+\ln F_{\psi_1,\ldots,\psi_\gamma}(\tilde{\rho})
		=O(1)\sum_{i=1}^\gamma\max_{\tau\in\Omega^k}|\ln\psi_i(\tau)|,\\
	\Lambda^{(l)}(F_{\psi_1,\ldots,\psi_\gamma}(\tilde{\rho}))&=O(F_{\psi_1,\ldots,\psi_\gamma}(\tilde{\rho})^{1-l})
		=O(1)\prod_{i=1}^\gamma\max\{\psi_i(\tau)^{1-l}:\tau\in\Omega^k\}
		&(l\geq2).
	\end{align*}
Thus, (\ref{eqBounded}) shows that $R_{\psi_1,\ldots,\psi_\gamma}(\rho_{i,j})=O(\eps^5)\exp(O(\gamma))$, which is (\ref{eqHotFix0}).
\end{proof}

\begin{claim}\label{Claim_eqB2final}
We have
	$\Erw\brk{B_2(\PSI)}=
	\Lambda(\xi)+\frac{\eps^4\xi k(k-1)}{12}\langle \Xi \DeltaM, \DeltaM \rangle +O(\epsilon^5).$
\end{claim}
\begin{proof}
To investigate $B_2(\psi)$ we apply \Lem~\ref{lem:DerivsVanish} to
	$F_{\psi}(\rho_1, \dots , \rho_{k})=
	  \sum_{\tau \in \Omega^k} \psi (\tau) \prod_{i=1}^{k} \rho_{i}(\tau_i)$.
The derivatives are
	\begin{align*}
	 \frac{\partial F_{\psi}}{\partial \rho_r(c_1)} &= \sum_{\tau \in \Omega^k} \vecone \{\tau_r = c_1\} 
	 			\psi (\tau) \prod_{i\in [k]\setminus \{ r\}} \rho_{i}(\tau_i) ,\\
	 \frac{\partial^2 F_{\psi}}{\partial \rho_{r_1}(c_1) \partial \rho_{r_2}(c_3)}(\bar\rho,\ldots,\bar\rho)
	 	& = \vecone\{r_1\neq r_2\}\sum_{\tau \in \Omega^k} \vecone \{\tau_{r_1} = c_1,  \tau_{r_2} = c_3\} \psi (\tau)q^{2-k} = q\xi 
			\Phi_{\psi^{\theta}}(c_1,c_3),
	\end{align*}
where $\theta : [k] \to [k]$ is such that $\theta (r_1) = 1,\; \theta (r_2) = 2$ and $\theta(r)=r$ for all $r\neq r_1,r_2$.
Thus, {\bf SYM} yields
	\begin{align*}
	F_{\psi}(\bar\rho,\ldots,\bar\rho)&=\xi \qquad \mbox{and} \qquad
	 \frac{\partial F_{\psi}}{\partial \rho_r(c_1)}(\bar\rho,\ldots,\bar\rho)=\xi.
	 \end{align*}
Once more we write $\Lambda\circ F_\psi=J_\psi+R_\psi$, where $J_\psi$ is the fourth Taylor polynomial as in (\ref{eqJ}).
Applying \Lem~\ref{lem:DerivsVanish}, we obtain
\begin{align*}
\Erw&\left[J_\psi(\RHO_1,\ldots,\RHO_k) \right] =\Lambda(\xi) +\\
&\frac{k(k-1)q^2\xi}{24} \sum_{c\in [q]^4}\left(  \Phi_{\psi^{\theta}}(c_1,c_3)  \Phi_{\psi^{\theta}}(c_2,c_4) + \Phi_{\psi^{\theta}}(c_2,c_3)  \Phi_{\psi^{\theta}}(c_1,c_4)\right)
   \left(\Erw\left[ \RHO(c_1) \RHO (c_2) \right]- q^{-2}\right) \left( \Erw\left[ \RHO (c_3) \RHO (c_4) \right] - q^{-2}\right).
\end{align*}
Further, Claim~\ref{Fact_EtaEtaT}  yields
$(\Erw[ \RHO(c_1) \RHO (c_2)]- q^{-2}) ( \Erw[ \RHO (c_3) \RHO (c_4)] - q^{-2})=\eps^4q^{-2}\scal{\DeltaM}{e_{c_1}\tensor e_{c_2}}
		\scal{\DeltaM}{e_{c_3}\tensor e_{c_4}}$,
whence by Claim~\ref{Claim_deconstruct},
	\begin{align}\label{eq:TaylorLambdaF1}
	\Erw\left[J_\psi(\RHO_1,\ldots,\RHO_k) \right]
	&=\Lambda(\xi)+\frac{\eps^4\xi k(k-1)}{12}\langle (\Phi_{\psi^\theta} \otimes \Phi_{\psi^\theta}) \DeltaM, \DeltaM \rangle.
	\end{align}
Furthermore, by Fact~\ref{Fact_Faadi} for any $\rho_1,\ldots,\rho_k$ in the support of $\pi_\eps$ exist
	$\tilde\rho$ on the line segment between $(\bar\rho,\ldots,\bar\rho)$ and $(\rho_1,\ldots,\rho_k)$ such that
	\begin{align*}
	R_{\psi}(\rho_1,\ldots,\rho_k)&=O(\eps^5)\cdot
	\sum_{h,c}
		\sum_{\Upsilon \in \Pi (5)} \sup_{\tilde\rho}
			\Lambda^{(|\Upsilon|)}(F_{\psi}(\tilde{\rho})) \prod_{B \in \Upsilon} \frac{\partial ^{|B|}
				F_{\psi}(\tilde{\rho})}{\prod_{i \in B}\partial \rho_{h_i}(c_i)}, 
	\end{align*}
Hence, (\ref{eqBounded}) guarantees that $\Erw[R_{\PSI}(\rho_1,\ldots,\rho_k)]=O(\eps^5)$ and thus the assertion follows from (\ref{eq:TaylorLambdaF1}).
\end{proof}

\begin{proof}[Proof of \Prop~\ref{prop:TaylorBdpi}]
Combining (\ref{eqBfinal}) with Claims \ref{Claim_eqB1final} and~\ref{Claim_eqB2final}, we obtain
	\begin{align}\label{eqBBfinal}
	\cB(d,P,\pi_\epsilon)&=
		\ln q+\frac dk\ln\xi+\frac{\eps^4 d(k-1)}{12}\brk{d(k-1)\scal{\Xi^2\DeltaM}\DeltaM-\scal{\Xi\DeltaM}\DeltaM}+O(\eps^5).
	\end{align}
Since $\scal{\Xi\DeltaM}\DeltaM=\hat\lambda$, $\scal{\Xi^2\DeltaM}\DeltaM=\hat\lambda^2$  and $\cB(d,P,\bar\pi) = \ln q+\frac dk\ln\xi$, the assertion follows from (\ref{eqBBfinal}).
\end{proof}

\section{Overlap concentration in the teacher-student model}\label{sec:ProofPreCond}

\noindent{\em Throughout this section we assume that $P$ satisfies conditions {\bf BAL}, {\bf SYM}, {\bf MIN} and {\bf POS}.}

\subsection{Outline}
{In this section we prove Proposition~\ref{prop:belowcond-unif}.
We will exhibit a connection between the overlap and the derivative $\frac\partial{\partial d}\Erw[\ln Z(\hat\G)]$ of the free energy:
if $\Erw\langle\tv{\rho_{\SIGMA_1,\SIGMA_2}-\bar\rho}\rangle_{\hat\G}$ is bounded away from $0$
for some $d<\dc$, then the derivative of the free energy is so large that the formula
	$n^{-1}\Erw[\ln Z(\hat\G)]=\ln q+\frac dk\ln\xi+o(1)$ cannot possibly hold, in contradiction to \Thm~\ref{Thm_plantedFreeEnergy}.
We begin with the following ``continuity statement'', which is a generalization of~\cite[\Lem~4.6]{CKPZ} for the Potts model: if the overlap deviates from $\bar\rho$
for some average degree $d$, then the same holds for at least a small interval of average degrees.}

\begin{lemma}\label{Lemma_mon}
For any $\eps>0$, $d>0$ there is $0<\delta=\delta(\eps,d,P)<\eps$ such that the following holds.
Assume that $m\in\cM(d)$ is a sequence such that
	\begin{equation}\label{eqLemma_mon_ass}
	\limsup_{n\to\infty}\Erw\bck{\TV{\rho_{\SIGMA_1,\SIGMA_2}-\bar\rho}}_{\hat\G(n,m)}>\eps.
	\end{equation}
Then 
	\begin{align*}
	\limsup_{n\to\infty}
		\min\cbc{\Erw\bck{\TV{\rho_{\SIGMA_1,\SIGMA_2}-\bar\rho}}_{\hat\G(n,m)}:\delta n<m-dn/k<2\delta n}
		>\delta.
	\end{align*}
\end{lemma}

\noindent
The proof of \Lem~\ref{Lemma_mon} can be found in \Sec~\ref{Sec_Lemma_mon}.
Further, in \Sec~\ref{Sec_lemma_derivatives} we derive the following asymptotic formula for the derivative of the free energy.

\begin{lemma}\label{lemma_derivatives}
Uniformly for all $d\leq\dc+1$ we have
	\begin{align}\label{eqlem:badoverlaps6}
	\frac kn\frac{\partial}{\partial d}\Erw[\ln Z(\hat\G)]&=o(1)
		+\xi^{-1}\Erw\brk{\Lambda\bc{\bck{\PSI(\SIGMA(x_{\vec i_1}),\ldots,\SIGMA(x_{\vec i_k}))}_{\hat\G}}}.
	\end{align}
with $\PSI$ chosen from $P$ independently of $\hat\G$ and $\vec i_1,\ldots,\vec i_k\in[n]$ chosen uniformly and independently.
\end{lemma}

\begin{corollary}\label{lem:badoverlaps1}
Uniformly for all $d<\dc+1$ we have
	\begin{align}
	\frac{1}{n}\frac{\partial}{\partial d}\Erw[\ln Z(\hat\G)]\ge\frac{\ln\xi}{k}+o(1).\label{eq:derivPreCond}
	\end{align}
Moreover, for any $\eps>0$ there is $\delta=\delta(\eps,P)>0$, independent of $n$ or $d$, such that
uniformly for all $d<\dc+1$,
\begin{align}
\Erw\bck{\TV{\rho_{\SIGMA,\TAU}-\bar\rho}}_{\hat\G}>\eps&\ \Rightarrow\ 
\frac{1}{n}\frac{\partial}{\partial d}\Erw[\ln Z(\hat\G)]\ge\frac{\ln\xi}{k}+\delta+o(1).\label{eq:derivPostCond}
\end{align}
\end{corollary}

\noindent
For the special case of the Potts model a result like \Cor~\ref{lem:badoverlaps1} was known~\cite[\Lem~4.10]{CKPZ}.
The proof was relatively straightforward because in the special case it is possible to write a fairly explicit formula for the expression
$\Lambda\bc{\bck{\PSI(\SIGMA(x_{\vec i_1}),\ldots,\SIGMA(x_{\vec i_k}))}_{\hat\G}}$.
Remarkably, the following proof shows that we can do without an explicit formula
thanks to a mildly tricky application of Jensen's inequality in combination with condition {\bf MIN}.

\begin{proof}[Proof of \Cor~\ref{lem:badoverlaps1}]
Since $\Lambda$ is convex, Jensen's inequality gives
	\begin{align}\label{eqlem:badoverlaps7}
	\Erw\brk{\Lambda\bc{\bck{\PSI(\SIGMA(x_{\vec i_1}),\ldots,\SIGMA(x_{\vec i_k}))}_{\hat\G}}}
		\leq\Lambda\bc{\Erw\bck{\PSI(\SIGMA(x_{\vec i_1}),\ldots,\SIGMA(x_{\vec i_k}))}_{\hat\G}}.
	\end{align}
Hence, using the Nishimori identity \eqref{eq:nishimori}  and \Cor~\ref{lem:conc_coloring}, we obtain
	\begin{align}\label{eqlem:badoverlaps8}
	\Erw\bck{\PSI(\SIGMA(x_{\vec i_1}),\ldots,\SIGMA(x_{\vec i_k}))}_{\hat\G}=
		\Erw\brk{\PSI(\hat\SIGMA_{n,\vec m}(x_{\vec i_1}),\ldots,\hat\SIGMA_{n,\vec m}(x_{\vec i_k}))}=\xi+o(1).
	\end{align}
Combining (\ref{eqlem:badoverlaps6}), \eqref{eqlem:badoverlaps7} and (\ref{eqlem:badoverlaps8}) with \Lem~\ref{lemma_derivatives} gives (\ref{eq:derivPreCond}).

To prove the second assertion we expand $\Lambda(x)$ to the second order around $\xi$ to obtain
	\begin{align}\label{eqlem:badoverlaps9}
	\Lambda(x)&=\Lambda(\xi)+(x-\xi)\Lambda'(\xi)+\frac{1}{2}(x-\xi)^2\Lambda''(\zeta_x)\quad\mbox{for some $\zeta_x$ between $\xi$ and $x$.}
	\end{align}
Since $\Lambda''(x)\geq1/2$ for all $x\in(0,2)$,  (\ref{eqlem:badoverlaps9}) and (\ref{eqlem:badoverlaps8}) yield
	\begin{align}\nonumber
	\Erw\brk{\Lambda\bc{\bck{\PSI(\SIGMA(x_{\vec i_1}),\ldots,\SIGMA(x_{\vec i_k}))}_{\hat\G}}}&\geq \Lambda(\xi)+
		\Lambda'(\xi)\brk{\Erw\bck{\PSI(\SIGMA(x_{\vec i_1}),\ldots,\SIGMA(x_{\vec i_k}))}_{\hat\G}-\xi}+
		\frac14\Erw\brk{\bc{\bck{\PSI(\SIGMA(x_{\vec i_1}),\ldots,\SIGMA(x_{\vec i_k}))}_{\hat\G}-\xi}^2}\\
	&=o(1)+\Lambda(\xi)+
		\frac14\Erw\brk{\bck{\PSI(\SIGMA(x_{\vec i_1}),\ldots,\SIGMA(x_{\vec i_k}))}_{\hat\G}^2}-\frac{\xi^2}4.
	\label{eqlem:badoverlaps10}
	\end{align}
Further, with $\SIGMA_1,\SIGMA_2$ denoting two independent samples from the Gibbs measure of $\hat\G$ we obtain
	\begin{align}\label{eqlem:badoverlaps11}
	\Erw\brk{\bck{\PSI(\SIGMA(x_{\vec i_1}),\ldots,\SIGMA(x_{\vec i_k}))}_{\hat\G}^2}&=
		\Erw\bck{\PSI(\SIGMA_1(x_{\vec i_1}),\ldots,\SIGMA_1(x_{\vec i_k}))\PSI(\SIGMA_2(x_{\vec i_1}),\ldots,\SIGMA_2(x_{\vec i_k}))}_{\hat\G}.
	\end{align}
Since $\vec i_1,\ldots,\vec i_k$ are chosen uniformly and independently of each other and of $\hat\G$ and $\PSI$, we can cast
(\ref{eqlem:badoverlaps11}) in terms of the overlap $\rho_{\SIGMA_1,\SIGMA_2}$ as
	\begin{align}\label{eqlem:badoverlaps12}
	\Erw\bck{\PSI(\SIGMA_1(x_{\vec i_1}),\ldots,\SIGMA_1(x_{\vec i_k}))\PSI(\SIGMA_2(x_{\vec i_1}),\ldots,\SIGMA_2(x_{\vec i_k}))}_{\hat\G}&=
		\sum_{\sigma,\tau\in\Omega^k}\Erw\bck{\PSI(\sigma)\PSI(\tau)\prod_{i=1}^k\rho_{\SIGMA_1,\SIGMA_2}(\sigma_i,\tau_i)}_{\hat\G}.
	\end{align}
Further, \Cor~\ref{lem:conc_coloring} and the Nishimori identity \eqref{eq:nishimori} yield
	$\Erw\bck{\TV{\rho_{\SIGMA_1}-\bar\rho}+\TV{\rho_{\SIGMA_2}-\bar\rho}}_{\hat\G}=o(1)$, whence
	\begin{align}\label{eqLemma_balanceRowCols_I}
	\Erw\brk{\sum_{\sigma\in\Omega}\bck{\abs{\sum_{\tau\in\Omega}\rho_{\SIGMA_1,\SIGMA_2}(\sigma,\tau)}+
		\abs{\sum_{\tau\in\Omega}\rho_{\SIGMA_1,\SIGMA_2}\bc{\tau,\sigma}}}_{\hat\G}}=o(1).
	\end{align}
Moreover, the function
		$\rho\in\cP(\Omega^2)\mapsto\sum_{\sigma, \tau\in \Omega^k}\Erw[\PSI(\sigma)\PSI(\tau)]  \prod_{i\in [k]} \rho(\sigma_i,\tau_i)$
is uniformly continuous.
Therefore, if $\Erw\bck{\TV{\rho_{\SIGMA,\TAU}-\bar\rho}}>\eps$, then 
Fact~\ref{Lemma_balanceRowCols}, (\ref{eqLemma_balanceRowCols_I}) and conditions {\bf MIN} and {\bf SYM}  yield $\delta=\delta(\eps)>0$ such that
	\begin{align}\label{eqlem:badoverlaps13}
	\sum_{\sigma,\tau\in\Omega^k}\Erw\bck{\PSI(\sigma)\PSI(\tau)\prod_{i=1}^k\rho_{\SIGMA_1,\SIGMA_2}\bc{\sigma_i,\tau_i}}_{\hat\G}
		>\delta+o(1)+q^{-2k}\sum_{\sigma,\tau\in\Omega^k}\Erw[\PSI(\sigma)\PSI(\tau)]=
		\xi^2+\delta+o(1).
	\end{align}
Finally, (\ref{eqlem:badoverlaps6}), (\ref{eqlem:badoverlaps10}), (\ref{eqlem:badoverlaps11}),
	(\ref{eqlem:badoverlaps12}) and (\ref{eqlem:badoverlaps13}) yield (\ref{eq:derivPostCond}).
\end{proof}

\begin{corollary}\label{Lemma_naiveLowerBound}
For all $d>0$ we have $\lim_{n\to\infty}\frac1n\Erw[\ln Z(\hat\G)]\geq \ln q+\frac dk\ln\xi$.
\end{corollary}
\begin{proof}
This follows from (\ref{eq:derivPreCond}) by integrating.
\end{proof}

\noindent
Finally, to prove Proposition \ref{prop:belowcond-unif} we combine \Lem~\ref{Lemma_mon} and \Cor~\ref{lem:badoverlaps1} to argue that if
$\Erw\bck{\TV{\rho_{\SIGMA,\TAU}-\bar\rho}}_{\hat\G}$ is bounded away from $0$ for some $d<\dc$,
then in fact for all $d$ in a small interval the derivative $\frac{1}{n}\frac{\partial}{\partial d}\Erw[\ln Z(\hat\G)]$ strictly exceeds $k^{-1}\ln\xi$.
Consequently, $n^{-1}\Erw[\ln Z(\hat\G)]$ is strictly greater than $\ln q+\frac dk\ln\xi$ for some $d<\dc$,
in contradiction to \Thm~\ref{Thm_plantedFreeEnergy}.

\begin{proof}[Proof of Proposition \ref{prop:belowcond-unif}]
Assume that there exist $D_0<\dc$ and $\eps>0$ such that
	$$\limsup_{n\to\infty}\Erw\bck{\TV{\rho_{\SIGMA,\TAU}-\bar\rho}}_{\hat\G(n,\vm_{D_0}(n))}>\eps.$$
Then \Lem~\ref{Lemma_mon} shows that there is $\delta>0$ such that
with $D_1=D_0+3\delta/2<\dc$ for infinitely many $n$ we have
	\begin{align*}
	\Erw\bck{\TV{\rho_{\SIGMA,\TAU}-\bar\rho}}_{\hat\G(n,\vm)}>\delta+o(1)\qquad\mbox{for all }D_0+4\delta/3<d<D_1.
	\end{align*}
But then Corollaries~\ref{lem:badoverlaps1} and~\ref{Lemma_naiveLowerBound} imply that for infinitely many $n$,
	\begin{align*}
	\frac1n\Erw[\ln Z(\hat\G(n,\vm_{D_1}(n)))]&=\frac1n\Erw[\ln Z(\hat\G(n,\vm_{D_0}(n)))]+
		\frac1n\int_{D_0}^{D_1}\frac{\partial}{\partial d}\Erw[\ln Z(\hat\G)]\dd d
		\geq\ln q+\frac{D_1}k\ln\xi+\Omega(1).
	\end{align*}
Consequently, $$\limsup_{n\to\infty}\frac1n\Erw[\ln Z(\hat\G(n,\vm_{D_1}))]>\ln q+\frac{D_1}{k}\ln\xi.$$
Therefore, {\Thm~\ref{Thm_plantedFreeEnergy}} yields
	$\sup_{\pi\in\Pomast}\cB(D_1,P,\pi)>\ln q+\frac{D_1}{k}\ln\xi,$
in contradiction to $D_1<\dc$.
\end{proof}

\subsection{Proof of \Lem~\ref{Lemma_mon}}\label{Sec_Lemma_mon}
The proof,
which is a non-trivial generalization of the argument for \cite[\Lem~4.6]{CKPZ} for the Potts model,
 is based on a coupling of the random factor graphs $\hat\G(n,m)$ and $\hat\G(n,m')$ with different numbers $m,m'$ of constraint nodes;
to set up the coupling we use the Nishimori identity~(\ref{eq:nishimori}). 
Thus, as a first step we need a coupling of $\hat\SIGMA_{n,m}$ and $\hat\SIGMA_{n,m'}$.

\begin{lemma}\label{Lemma_mon_cplg}
For any $\eta>0$, $d>0$ there is $\delta>0$ such that
	\begin{align}
	\limsup_{n\to\infty}\max\cbc{\dTV\bc{\hat\SIGMA_{n,m},\hat\SIGMA_{n,m'}}:|m-dn/k|+|m'-dn/k|<\delta n}<\eta.
		\label{eq:coupledD}
	\end{align}
\end{lemma}
\begin{proof}
Given $\eta>0$ pick a sufficiently small $\beta=\beta(\eta)>0$.
Let $\phi$ be the function from (\ref{eqF}).
Because the constraint nodes of $\G$ are chosen independently, for all $m\geq0$, $\sigma\in\Omega^{V_n}$ we have
	\begin{align}
	\ln\Erw[\psi_{\G(n,m)}(\sigma)]=m\ln\phi(\rho_\sigma).\label{eq_mmpRatios0}
	\end{align}
Furthermore, by \Cor~\ref{lem:conc_coloring} there exists $C>0$ such that
	\begin{equation}\label{eq_mmpRatios5}
	\pr\brk{\TV{\rho_{\hat\SIGMA_{n,m}}-\bar\rho}> C/\sqrt n}
		+\pr\brk{\TV{\rho_{\hat\SIGMA_{n,m'}}-\bar\rho}> C/\sqrt n}\leq2\beta\qquad\mbox{for all $m,m'\leq (d+1)n/k$},
	\end{equation}
which implies that
	\begin{equation}\label{eq_mmpRatios1}
	\sum_{\sigma\in\Omega^{V_n}}\vecone\{\TV{\rho_\sigma-\bar\rho}\leq C/\sqrt n\}
		\Erw[\psi_{\G(n,m)}(\sigma)]\geq(1-\beta)\Erw[Z(\G(n,m))]\qquad\mbox{for all $m\leq (d+1)n/k$}.
	\end{equation}
Applying \Lem~\ref{Lemma_F} to expand (\ref{eq_mmpRatios0}) to the second order, we obtain $C'>0$
	such that for all $m$ and all $\sigma$ satisfying $\TV{\rho_\sigma-\bar\rho}\leq C/\sqrt n$,
	\begin{align*}
	\abs{\ln\Erw[\psi_{\G(n,m)}(\sigma)]-m\bc{\ln\xi+qk(k-1)\scal{\Phi(\rho_\sigma-\bar\rho)}{\rho_\sigma-\bar\rho}/2}}\leq C'm/n^{3/2}.
	\end{align*}
Hence, choosing $\delta=\delta(\beta,C,d)>0$ small enough, we can ensure that for all $m,m'$ such that $|m-dn/k|+|m'-dn/k|\leq\delta$ and
all $\sigma$ satisfying $\TV{\rho_\sigma-\bar\rho}\leq C/\sqrt n$ the estimate
	\begin{align}\label{eq_mmpRatios3}
	\abs{\ln\Erw[\psi_{\G(n,m)}(\sigma)]-\ln\Erw[\psi_{\G(n,m')}(\sigma)]}
	\leq2\delta\bc{ nq(d+1)(k-1)\abs{\scal{\Phi(\rho_\sigma-\bar\rho)}{\rho_\sigma-\bar\rho}}/2+ C'/\sqrt n}<\beta
	\end{align}
holds.
Further, combining (\ref{eq_mmpRatios1}) and (\ref{eq_mmpRatios3}), we obtain that
	\begin{align}\label{eq_mmpRatios4}
	\abs{\ln \Erw[Z(\G(n,m))]-\ln\Erw[Z(\G(n,m'))]}&\leq\eta/4,
	\end{align}
provided that $|m-dn/k|+|m'-dn/k|\leq\delta$ and $\beta=\beta(\eta)$ was chosen small enough.
Moreover, combining (\ref{eq_mmpRatios3}) and (\ref{eq_mmpRatios4}), we conclude that
if $|m-dn/k|+|m'-dn/k|\leq\delta$ and $\TV{\rho_\sigma-\bar\rho}\leq C/\sqrt n$, then
	\begin{align}\label{eq_mmpRatios6}
	\exp(-\eta/2)\leq\frac{\pr\brk{\hat\SIGMA_{n,m}=\sigma}}{\pr\brk{\hat\SIGMA_{n,m'}=\sigma}}&=
		\frac{\Erw[\psi_{\G(n,m)}(\sigma)]\cdot\Erw[Z(\G(n,m')]}{\Erw[\psi_{\G(n,m')}(\sigma)]\cdot\Erw[Z(\G(n,m)]}\leq\exp(\eta/2).
	\end{align}
Finally, the assertion follows from (\ref{eq_mmpRatios5}) and (\ref{eq_mmpRatios6}).
\end{proof}

\begin{proof}[Proof of \Lem~\ref{Lemma_mon}]
Assume that $m\in\cM(d)$ satisfies (\ref{eqLemma_mon_ass}).
Pick $\eta=\eta(\eps)>0$ small enough, let $\delta=\delta(\eta)>0$ be the number promised by \Lem~\ref{Lemma_mon_cplg}
and assume that $n$ is a large enough number such that $|m-dn/k|<\delta n/2$ and
	\begin{equation}\label{eq_Lemma_mon_666}
	\Erw\bck{\TV{\rho_{\SIGMA_1,\SIGMA_2}-\bar\rho}}_{\hat\G(n,m)}>\eps/2.
	\end{equation}
Further, suppose that $m'>m$ is such that $|m'-dn/k|<\delta n/2$.
Then by \Lem~\ref{Lemma_mon_cplg} we can couple $\hat\SIGMA_{n,m}$ and $\hat\SIGMA_{n,m'}$ such that 
the event $\cA=\{\hat\SIGMA_{n,m}=\hat\SIGMA_{n,m'}\}$ satisfies
	\begin{equation}\label{eq_Lemma_mon_667}
	\pr\brk{\cA}>1-\eta.
	\end{equation}

We extend this to a coupling of a pair of factor graphs $\G',\G''$ such that $\G'$ is distributed as $\G^\ast(n,m',\hat\SIGMA_{n,m'})$ and 
	$\G''$ is distributed as $\G^\ast(n,m,\hat\SIGMA_{n,m})$ as follows.
First choose $\G'$ from the distribution $\G^\ast(n,m',\hat\SIGMA_{n,m'})$.
Then obtain $\G'''$ from $\G'$ by deleting a uniformly chosen set of $m'-m$ constraint nodes.
On the event $\cA$ set $\G''=\G'''$.
If $\cA$ does not occur, then choose the constraint nodes of $\G''$ independently of those of $\G'$ in such a way that $\G''$ is distributed as $\G^\ast(n,m,\hat\SIGMA_{n,m})$.
 
Now, (\ref{eq_Lemma_mon_666}) implies that with probability at least $\eps/2$ the random graph $\G''$ is such that a random sample $\TAU$ from $\mu_{\G''}$ satisfies
$\langle{\tv{\rho_{\SIGMA,\TAU}-\bar\rho}}\rangle_{\G''}\geq\eps$.
By \Cor~\ref{lem:conc_coloring} and the Nishimori identity~\eqref{eq:nishimori}, with probability $1-o(1)$ this random sample $\TAU$  is nearly balanced.
Consequently, there exists a map $G\mapsto\tau_G$ that provides a nearly balanced $\tau_{G}$ for every factor graph $G$ such that
			$\pr[\langle{\tv{\rho_{\SIGMA,\tau_{\G''}}-\bar\rho}\rangle}_{\G''}>\eps]\geq\eps/2$.
Thus, $\Erw\langle{\tv{\rho_{\SIGMA,\tau_{\G''}}-\bar\rho}\rangle}_{\G''}>\eps^2/2.$
Hence, assuming that $\eta$ was chosen small enough, we obtain from (\ref{eq:coupledD}) and the Nishimori identity~(\ref{eq:nishimori}) that
	\begin{align}\label{eqNorOvl1}
	\Erw\brk{\TV{\rho_{\hat\SIGMA_{n,m},\tau_{\G''}}-\bar\rho}|\cA}>\eps^2/3.
	\end{align}

Finally, on the event $\cA$ the factor graph $\G''=\G'''$ is obtained from $\G'$ by deleting a few random constraint nodes.
Thus, for a graph $\G'$ let $\TAU_{\G'}$ be a random assignment with distribution $\tau_{\G'''}$.
Then (\ref{eqNorOvl1}) implies
	$$\Erw\brk{\TV{\rho_{\hat\SIGMA_{n,m'},\TAU_{\G'}}-\bar\rho}|\cA}>\eps^2/3.$$
Hence, by the Nishimori identity~(\ref{eq:nishimori}) and (\ref{eq_Lemma_mon_667}),
	\begin{align}\label{eq_using_Lemma_nbalanced_CKPZ}
	\Erw\bck{\TV{\rho_{\SIGMA,\TAU_{\G'}}-\bar\rho}}_{\G'}
	=\Erw{\TV{\rho_{\hat\SIGMA_{n,m'},\TAU_{\G'}}-\bar\rho}}
	>\eps^2/6.
	\end{align}
Since by construction $\TAU_{\G'}$ is nearly balanced, the assertion follows from (\ref{eq_using_Lemma_nbalanced_CKPZ}) and \Lem~\ref{Lemma_nbalanced_CKPZ}.
\end{proof}

\subsection{Proof of \Lem~\ref{lemma_derivatives}}\label{Sec_lemma_derivatives}
We shall see shortly that calculating the derivative $\frac\partial{\partial d}\Erw[\ln Z(\hat\G)]$ basically comes down
to calculating the difference $\Erw[\ln Z(\hat\G(n,\vm+1))]-\Erw[\ln Z(\hat\G(n,\vm))]$.
We are going to perform this calculation by way of a very accurate coupling of $\hat\G(n,\vm+1)$ and $\hat\G(n,\vm)$.
A similar argument was used in~\cite{CKPZ} for the case that the set $\Psi$ of weight functions is finite.
Once more the coupling is based on the Nishimori identity (\ref{eq:nishimori}).
Thus, we begin with a coupling of the random assignments $\hat\SIGMA_{n,\vm}$ and $\hat\SIGMA_{n,\vm+1}$.
{The following is a generalization of  \cite[\Cor~3.29]{CKPZ}.}

\begin{lemma}\label{Lemma_coupling}
There exists a coupling of $\hat\SIGMA_{n,\vec m}$ and $\hat\SIGMA_{n,\vec m+1}$ such that the following holds uniformly 
for all $d\leq\dc+1$.
\begin{enumerate}[(i)]
\item With probability $1-O(n^{-1}\ln^2n)$ we have $\hat\SIGMA_{n,\vec m}=\hat\SIGMA_{n,\vec m+1}$.
\item With probability $1-O(1/n^2)$ the set $\hat\SIGMA_{n,\vec m}\triangle\hat\SIGMA_{n,\vec m+1}=
	\{x\in V_n:\hat\SIGMA_{n,\vec m}(n)\neq\hat\SIGMA_{n,\vec m+1}(x)\}$ has size at most $n^{2/3}$.
\end{enumerate}
\end{lemma}
\begin{proof}
By definition, for any $\sigma\in\Omega^{V_n}$
	\begin{align}\label{eqRough}
	\pr[\hat\SIGMA_{n,m}=\sigma]=\frac{\Erw[\PSI_{\G(n,m)}(\sigma)]}{\Erw[Z(\G(n,m))]},\qquad \pr[\hat\SIGMA_{n,m+1}=\sigma]=		\frac{\Erw[\PSI_{\G(n,m+1)}(\sigma)]}{\Erw[Z(\G(n,m+1))]}.
	\end{align}
Further, due to the independence of the constraint nodes, we obtain
	\begin{align}\label{eqRough2}
	\frac{\Erw[\PSI_{\G(n,m+1)}(\sigma)]}{\Erw[\PSI_{\G(n,m)}(\sigma)]}&=
		\frac1{n^k}\sum_{y_1,\ldots,y_k\in V_n}\Erw\brk{\PSI(\sigma(y_1),\ldots,\sigma(y_k))}.
	\end{align}
Let $\phi$ be the function from (\ref{eqF}).
Then \Lem~\ref{Lemma_Phi} and \Lem~\ref{Lemma_F} show that for $\rho\in\cP(\Omega)$,
	\begin{align}\label{eq:PsiRatio'''}
	\phi(\rho)&=\xi+O(\TV{\rho-\bar\rho}^2).
	\end{align}
Hence, expanding the r.h.s.\ of \eqref{eqRough2} to the second order, we obtain
	\begin{align}\label{eq:PsiRatio'}
	\frac1{n^k}\sum_{y_1,\ldots,y_k\in V_n}\Erw\brk{\PSI(\sigma(y_1),\ldots,\sigma(y_k))}
		=\phi(\rho_\sigma)
		&=\xi + O\left(\TV{\rho_\sigma-\bar\rho}^2\right).
	\end{align}
Moreover, let $\cN$ be the set of all $\rho\in\cP(\Omega)$ such that $n\rho(\omega)$ is an integer for every $\omega\in\Omega$.
Then
	\begin{align}
	\Erw[Z(\G(n,m))]&=\sum_{\tau\in\Omega^{V_n}}
		\bc{ n^{-k}\sum_{y_1,\ldots,y_k\in V_n}\Erw\brk{\PSI(\tau(y_1),\ldots,\tau(y_k))}}^{m}=
		\sum_{\rho\in\cN}\bink n{n\rho}\phi(\rho)^m.\label{eqRough3}
	\end{align}
Further, let $\cN'=\cbc{\rho\in\cN:\TV{\rho-\bar\rho}\leq n^{-1/2}\ln n}$.
Then (\ref{eqRough3}), Stirling's formula and \Lem s~\ref{Lemma_Phi} and~\ref{Lemma_F}  yield
	\begin{align*}
	\Erw[Z(\G(n,m))]
		&=(1+O(n^{-1}))\sum_{\rho\in\cN'}\bink n{n\rho}\phi(\rho)^m.
	\end{align*}
Of course, the corresponding formula holds for $\Erw[Z(\G(n,m+1))]$.
Hence, (\ref{eqRough2}) and (\ref{eq:PsiRatio'''}) yield
	\begin{align}\label{eqRough4}
	\frac{\Erw[Z(\G(n,m+1))]}{\Erw[Z(\G(n,m))]}&=\xi+O(n^{-1}\ln^2n).
	\end{align}
Combining (\ref{eqRough}), (\ref{eqRough2}), (\ref{eq:PsiRatio'}) and (\ref{eqRough4}), we conclude that
	\begin{equation}\label{eq:PsiRatio''}
	\pr[\hat\SIGMA_{n,\vm+1}=\sigma]=\pr[\hat\SIGMA_{n,\vm}=\sigma]\bc{1+O\left(\TV{\rho_\sigma-\bar\rho}^2+
	n^{-1}\ln^2n\right)}.
	\end{equation}
By \Cor~\ref{lem:conc_coloring} $\tv{\rho_{\SIGMA_{n,\vm}}-\bar\rho}$ is bounded by $O(n^{-1/2}\ln n)$ with probability at least $1-O(1/n)$.
Hence, (\ref{eq:PsiRatio''}) shows that $\hat\SIGMA_{n,\vec m},\hat\SIGMA_{n,\vec m+1}$ have total variation distance $O(n^{-1}\ln^2n)$,  which yields the first assertion follows.

With respect to the second, we obtain from \Cor~\ref{lem:conc_coloring} that
	$$\pr\brk{\TV{\rho_{\SIGMA_{n,\vm}}-\bar\rho}\leq n^{-1/2}\ln n}+
		\pr\brk{\TV{\rho_{\SIGMA_{n,\vm+1}}-\bar\rho}\leq n^{-1/2}\ln n}=1-O(n^{-3}).$$
Hence, if we choose the empirical distributions $\rho_{\hat\SIGMA_{n,\vm}},\rho_{\hat\SIGMA_{n,\vm+1}}$ independently, then
 $\TV{\rho_{\SIGMA_{n,\vm}}-\rho_{\SIGMA_{n,\vm+1}}}\leq2n^{-1/2}\ln n$ with probability $1-O(n^{-3})$.
Finally, we obtain the desired coupling of $\hat\SIGMA_{n,\vm}$, $\hat\SIGMA_{n,\vm+1}$ for (ii):
given $\rho$, $\rho'\in\cN$ choose a collection of pairwise disjoint sets $(S_\omega)_{\omega\in\Omega}\subset V_n$
with $|S_\omega|=n\min\{\rho(\omega),\rho'(\omega)\}$ randomly, set $\sigma(x)=\sigma'(x)=\omega$ for all $x\in S_\omega$ and
let $\sigma,\sigma'$ assign different spins to the nodes in $V_n\setminus\bigcup_{\omega\in\Omega}S_\omega$ so as to ensure that
$\rho_\sigma=\rho$ and $\rho_{\sigma'}=\rho'$.
\end{proof}

\begin{corollary}\label{Cor_coupling}
Uniformly for all $d\leq\dc+1$ the following is true.
Given the random assignment $\hat\SIGMA_{n,\vm}$ choose a constraint node
 $\vec a$ from the distribution
	\begin{align}\label{eqTeachera}
	\pr\brk{\partial\vec a=(y_1,\ldots,y_k),\psi_{\vec a}\in\cA}
			&=
		\frac{\int_{\cA}\psi(\hat\SIGMA_{n,\vm}(y_1),\ldots,\hat\SIGMA_{n,\vm}(y_k))\dd P(\psi)}
			{\sum_{z_1,\ldots,z_k\in V_n}\int_\Psi\psi(\hat\SIGMA_{n,\vm}(z_1),\ldots,\hat\SIGMA_{n,\vm}(z_k))\dd P(\psi)}
			&(y_1,\ldots,y_k\in V_n,\ \cA\subset\Psi)
	\end{align}
and choose $\G^*(n,\vec m,\hat\SIGMA_{n,\vm})$ independently.
Then
	\begin{align}\label{eqCor_coupling}
	\Erw[\ln Z(\hat\G(n,\vec m+1))]-\Erw[\ln Z(\hat\G(n,\vec m)]&=
		\Erw\brk{\ln\bck{\psi_{\vec a}(\SIGMA(\partial_1\vec a),\ldots,\SIGMA(\partial_k\vec a))}_{\G^*(n,\vec m,\hat\SIGMA_{n,\vm})}}+o(1).
	\end{align}
\end{corollary}
\begin{proof}
By the Nishimori identity (\ref{eq:nishimori}) we have
	\begin{align}
	\Erw[\ln Z(\hat\G(n,\vec m))]&=\Erw[\ln \G^*(n,\vec m,\hat\SIGMA_{n,\vm})],\label{eqCor_coupling1}\\
	\Erw[\ln Z(\hat\G(n,\vec m+1))]&=\Erw[\ln \G^*(n,\vec m+1,\hat\SIGMA_{n,\vm+1})].\label{eqCor_coupling2}
	\end{align}
To calculate the difference of the two terms on the r.h.s.\ we couple $\SIGMA'=\hat\SIGMA_{n,\vm}$ and $\SIGMA''=\hat\SIGMA_{n,\vm+1}$
	via \Lem~\ref{Lemma_coupling}.
Clearly, if $\SIGMA'=\SIGMA''$, then we can couple 
$\G'=\G^*(n,\vec m,\hat\SIGMA_{n,\vm+1})$ and $\G''=\G^*(n,\vec m+1,\hat\SIGMA_{n,\vm+1})$ such that $\G''$ is obtained from $\G'$
by adding one additional independent constraint node $\vec a=a_{\vec  m+1}$ and thus
	\begin{align*}
	\frac{Z(\G'')}{Z(\G')}&=\sum_{\tau\in\Omega^{V_n}}\psi_{\vec a}(\tau(\partial_1\vec a),\ldots,\tau(\partial_1\vec a))\frac{\psi_{\G'}(\tau)}{Z(\G')}
		=\bck{\psi_{\vec a}(\SIGMA(\partial_1\vec a),\ldots,\SIGMA(\partial_k\vec a))}_{\G'}.
	\end{align*}
Hence, by (\ref{eqBounded}) and the first part of \Lem~\ref{Lemma_coupling},
	\begin{align}\nonumber
	X&=\Erw\brk{\ln\frac{Z(\G'')}{Z(\G')}\bigg|\SIGMA'=\SIGMA''}=
		\Erw\brk{\ln\bck{\psi_{\vec a}(\SIGMA(\partial_1\vec a),\ldots,\SIGMA(\partial_k\vec a))}_{\G'}| \SIGMA'= \SIGMA''}\\
		&=\Erw\brk{\ln\bck{\psi_{\vec a}(\SIGMA(\partial_1\vec a),\ldots,\SIGMA(\partial_k\vec a))}_{\G'}}+o(1).
			\label{eqCor_coupling3}
	\end{align}

If $|\SIGMA'\triangle \SIGMA''|\leq n^{2/3}$ and $\|\rho_{\SIGMA'}-\bar\rho\|\leq n^{-1/2}\ln n$,  then by (\ref{eqBounded}) we have
	\begin{align}\label{eqNonForbiddenNeighbor1}
	\sum_{z_1,\ldots,z_k\in V_n}\Erw[\PSI(\SIGMA'(z_1),\ldots,\SIGMA'(z_k))]&\sim n^k\xi,&
		\sum_{z_1,\ldots,z_k\in V_n}\Erw[\PSI(\SIGMA''(z_1),\ldots, \SIGMA''(z_k))]&\sim n^k\xi.
	\end{align}
Further, let us write $\vec a'$ for a factor node chosen from (\ref{eqTeacher}) with respect to $\SIGMA'$ and $\vec a''$ for one chosen with respect to $\SIGMA''$.
Let $\cA$ be the event that a random factor node does not have a neighbor in $\SIGMA'\triangle \SIGMA''$.
Since $\|\rho_{\SIGMA'}-\bar\rho\|\leq n^{-1/2}\ln n$, (\ref{eqBounded}) and (\ref{eqNonForbiddenNeighbor1}) imply that
	\begin{align*}
	\pr\brk{\vec a'\not\in\cA}&=\frac{\sum_{z_1,\ldots,z_k\in V_n}
		\vecone\{\{z_1,\ldots,z_k\}\cap( \SIGMA'\triangle \SIGMA'')\neq\emptyset\}\Erw[\PSI( \SIGMA'(z_1),\ldots,\SIGMA'(z_k))]}
			{\sum_{z_1,\ldots,z_k\in V_n}\Erw[\PSI( \SIGMA'(z_1),\ldots, \SIGMA'(z_k))]}=O(| \SIGMA'\triangle \SIGMA''|/n)
			=O(n^{-1/3}),
	\end{align*}
and similarly $\pr\brk{\vec a''\not\in\cA}=O(n^{-1/3})$.
Moreover, given that $\vec a',\vec a''\in\cA$,  both factor nodes $\vec a',\vec a''$ are identically distributed.
Therefore, there is a coupling of $\vec a',\vec a''$ such that $\vec a'=\vec a''$ with probability $1-O(n^{-1/3})$.
Hence, $\G',\G''$ can be coupled such that the set $\Delta$ of constraint nodes in which both factor graphs differ has expected size $O(n^{2/3})$.
Indeed,  $\Delta$ is a binomial random variable because the constraint nodes are chosen independently.
Thus, (\ref{eqBounded})  implies
	\begin{align*}
	\abs{\Erw\brk{\ln\frac{Z(\G'')}{Z(\G'')}
		\bigg|
		| \SIGMA'\triangle \SIGMA''|\leq n^{2/3},\|\rho_{ \SIGMA'}-\bar\rho\|\leq\frac{\ln n}{\sqrt n},\Delta}}
		\leq O(\Delta)\Erw\brk{\max_{\tau\in\Omega^k}|\ln\PSI(\tau)|}=O(\Delta)
	\end{align*}
and therefore
	\begin{align}\label{eqCor_coupling4}
	X'&=\Erw\brk{\ln\frac{Z(\G'')}{Z(\G')}\bigg|0<|  \SIGMA'\triangle \SIGMA''|\leq n^{2/3},
		\|\rho_{ \SIGMA'}-\bar\rho\|\leq n^{-1/2}\ln n}=O(n^{2/3}).
	\end{align}

Finally, if either $| \SIGMA'\triangle \SIGMA''|> n^{2/3}$ or
 $\|\rho_{\SIGMA'}-\bar\rho\|> n^{-1/2}\ln n$, then we couple $\G',\G''$ by just choosing their constraint nodes independently.
Then (\ref{eqBounded}) implies
	\begin{align}\label{eqCor_coupling5}
	X''&=\Erw\brk{\ln\frac{Z(\G'')}{Z(\G')}\bigg||  \SIGMA'\triangle \SIGMA''|> n^{2/3}\mbox{ or }\|\rho_{\SIGMA'}-\bar\rho\|> n^{-1/2}\ln n}=O(n).
	\end{align}
Combining (\ref{eqCor_coupling1})--(\ref{eqCor_coupling5}) and applying \Cor~\ref{lem:conc_coloring} and \Lem~\ref{Lemma_coupling}, we obtain
	\begin{align*}
	\Erw[\ln Z(\hat\G(n,\vec m+1))]-\Erw[\ln Z(\hat\G(n,\vec m)]&=(1-o(1))X+O(n^{-1}\ln^2n)X'+O(n^{-2})X''\\
		&=\Erw\brk{\ln\bck{\psi_{\vec a}(\SIGMA(\partial_1\vec a),\ldots,\SIGMA(\partial_k\vec a))}_{\G'}}+o(1),
	\end{align*}
as claimed.
\end{proof}

\begin{proof}[Proof of \Lem~\ref{lemma_derivatives}]
The proof is a generalization of the proof of~\cite[\Lem~3.32]{CKPZ}, which dealt with the Potts model.
We begin with the well-known observation that
	\begin{align}
	\frac1n\frac{\partial}{\partial d}\Erw[\ln Z(\hat\G)]&=\frac1n\sum_{m=0}^\infty\brk{\frac{\partial}{\partial d}\pr\brk{\Po(dn/k)=m}}
		\Erw[\ln Z(\hat\G)|\vm=m]\nonumber\\
		&=\frac1k\sum_{m=0}^\infty\brk{\vecone\{m\geq1\}\pr\brk{\Po(dn/k)=m-1}+\pr\brk{\Po(dn/k)=m}}\Erw[\ln Z(\hat\G)|\vm=m]\nonumber\\
		&=\frac1k[\Erw[\ln Z(\hat\G(n,\vec m+1))]-\Erw[\ln Z(\hat\G(n,\vec m)]].
			\label{eqlem:badoverlaps1}
	\end{align}
To calculate the last term we apply \Cor~\ref{Cor_coupling}.
Let us write $\bck{\nix}=\bck{\nix}_{\G^*(n,\vec m,\hat\SIGMA_{n,m})}$ for brevity.
Expanding the logarithm on the r.h.s.\ of (\ref{eqCor_coupling}), we obtain
	\begin{align*}
	\frac1n\frac{\partial}{\partial d}\Erw[\ln Z(\hat\G)]&=o(1)-\Erw\sum_{l=1}^\infty\frac1{kl}
		\bck{1-\psi_{\vec a}(\SIGMA(\partial_1\vec a),\ldots,\SIGMA(\partial_k\vec a))}^l
	\end{align*}
(where the expectation is over the choice of $\hat\SIGMA_{n,m}$, $\G^*(n,\vec m,\hat\SIGMA_{n,m})$ and $\vec a$).
Due to (\ref{eqBounded}) and Fubini's theorem we can interchange the sum and the expectation.
Hence, writing the expectation on $\vec a$ chosen from (\ref{eqTeachera}) out, with $\PSI$ chosen from $P$ independently of everything else, we obtain
	\begin{align*}
	\frac1n\frac{\partial}{\partial d}\Erw[\ln Z(\hat\G)]&=o(1)-\sum_{l=1}^\infty
		\Erw\brk{\frac{\sum_{i_1,\ldots,i_k\in[n]}
			\PSI(\hat\SIGMA_{n,m}(x_{i_1}),\ldots,\hat\SIGMA_{n,m}(x_{i_k}))\bck{1-\PSI(\SIGMA(x_{i_1}),\ldots,\SIGMA(x_{i_k}))}^l}
				{kl\sum_{i_1,\ldots,i_k\in[n]}\int_\Psi\psi(\hat\SIGMA_{n,m}(x_{i_1}),\ldots,\hat\SIGMA_{n,m}(x_{i_k}))\dd P(\psi)}}.
	\end{align*}
Further, because $|\hat\SIGMA_{n,\vec m}^{-1}(\omega)|\sim n/q$ for all $\omega\in\Omega$ with probability at least $1-o(1)$ by \Cor~\ref{lem:conc_coloring}, 
we obtain from (\ref{eqBounded}) and {\bf SYM} that
	\begin{align}							\label{eqlem:badoverlaps3}
	\frac1n\frac{\partial}{\partial d}\Erw[\ln Z(\hat\G)]&=o(1)-\sum_{l=1}^\infty
		\sum_{i_1,\ldots,i_k\in[n]}
		\frac{1}{kl\xi n^k}
			\Erw\brk{\PSI(\hat\SIGMA_{n,m}(x_{i_1}),\ldots,\hat\SIGMA_{n,m}(x_{i_k}))\bck{1-\PSI(\SIGMA(x_{i_1}),\ldots,\SIGMA(x_{i_k}))}^l}.
	\end{align}
To evaluate the expectation on the r.h.s.\ of (\ref{eqlem:badoverlaps3}) we harness the
 Nishimori identity \eqref{eq:nishimori}, which implies the following: if $\cX:(G,\sigma)\mapsto \cX(G,\sigma)\in\RR$ is an $L^1$-function, then 
	$\Erw[\cX(\G^*(n,\vec m,\hat\SIGMA_{n,m}),\hat\SIGMA_{n,m})]=\Erw[\cX(\G^*(n,\vec m,\hat\SIGMA_{n,m}),\SIGMA_0)]$.
Applying this fact to the function $\cX(G,\sigma)=\PSI(\sigma(x_{i_1}),\ldots,\sigma(x_{i_k}))\bck{1-\PSI(\SIGMA(x_{i_1}),\ldots,\SIGMA(x_{i_k}))}_G$, we obtain
	\begin{align}			\nonumber
	\Erw\brk{\PSI(\hat\SIGMA_{n,m}(x_{i_1}),\ldots,\hat\SIGMA_{n,m}(x_{i_k}))\bck{1-\PSI(\SIGMA(x_{i_1}),\ldots,\SIGMA(x_{i_k}))}^l}\\
		&\hspace{-6cm}=\Erw\brk{\bck{1-\PSI(\SIGMA(x_{i_1}),\ldots,\SIGMA(x_{i_k}))}^l-\bck{1-\PSI(\SIGMA(x_{i_1}),\ldots,\SIGMA(x_{i_k}))}^{l+1}}.
										\label{eqlem:badoverlaps4}
	\end{align}
Plugging (\ref{eqlem:badoverlaps4}) into (\ref{eqlem:badoverlaps3}) and writing $\vec i_1,\ldots,\vec i_k$ for uniformly random indices chosen
from $[n]$  we obtain
	\begin{align}\label{eqlem:badoverlaps5}
	\frac kn\frac{\partial}{\partial d}\Erw[\ln Z(\hat\G)]&=o(1)
		-\frac1\xi\Erw\bck{1-\PSI(\SIGMA(x_{\vec i_1}),\ldots,\SIGMA(x_{\vec i_k}))}
		+\sum_{l=2}^\infty
				\frac1{l(l-1)\xi}\Erw\bck{1-\PSI(\SIGMA(x_{\vec i_1}),\ldots,\SIGMA(x_{\vec i_k}))}^l.
	\end{align}
Finally, since $\sum_{l\geq2}\frac1{l(l-1)}(1-x)^l=1-x+\Lambda(x)$, (\ref{eqlem:badoverlaps5}) yields (\ref{eqlem:badoverlaps6}).
\end{proof}

\section{Moment calculations}\label{Sec_moments}

\noindent
In this section we prove \Prop s~\ref{lem:FirstMoment} and~\ref{lem:SecondMoment}.
We begin with a very general calculation in \Sec~\ref{Sec_moments1}, from which we subsequently deduce
\Prop s~\ref{lem:FirstMoment} and~\ref{lem:SecondMoment}.

\subsection{An asymptotic formula}\label{Sec_moments1}
The following result
paves the way for the proofs of \Prop s~\ref{lem:FirstMoment} and~\ref{lem:SecondMoment}.

\begin{proposition}\label{Prop_genFirstMmt}
Assume that $P$ satisfies {\bf SYM} and that $d>0$ is such that the eigenvalues $\lambda_1\geq\cdots\geq\lambda_q$ of $\Phi$ satisfy
		\begin{equation}\label{eqProp_genFirstMmt}
		d(k-1)\max\{\lambda_2,\ldots,\lambda_q\} 
			<1.
		\end{equation}
Furthermore, assume that $\eps=\eps(n)\to0$ but $\sqrt n\eps\to\infty$ as $n\to\infty$ and let
	\begin{align*}
	Z_\eps(\G(n,m))&=Z(\G(n,m))\bck{\vecone\{\forall\omega\in\Omega:||\SIGMA^{-1}(\omega)|-n/q|<\eps n\}}_{\G(n,m)}.
	\end{align*}
Then uniformly for all $m\in\cM(d)$,
	\begin{align*}
	\Erw[Z_\eps(\G(n,m))]&\sim\frac{q^{n+\frac12}\xi^m}{\prod_{i=2}^q\sqrt{1-d(k-1)\lambda_i}}.
	\end{align*}
\end{proposition}
\begin{proof}
Let $R_{n,\eps}$ be the set of all distributions $\rho\in\cP(\Omega)$ such that
	$n\rho\in\RR^\Omega$ is an integer vector and such that
	$\|\rho-\bar\rho\|_2<\eps$ for all $\omega\in\Omega$.
Additionally, for each $\rho\in R_{n,\eps}$ let
	$Z_{\rho}(\G(n,m)) =Z(\G(n,m))\bck{\vecone\{\rho_{\SIGMA}=\rho\}}_{\G(n,m)}.$
Then
	\begin{align}\label{eqProp_genFirstMmt1}
	\Erw[Z_\eps(\G(n,m))]&=\sum_{\rho\in R_{n,\eps}}\Erw[Z_\rho(\G(n,m))].
	\end{align}
Remembering $\phi$ from (\ref{eqF}),  we claim that uniformly for all $\rho\in R_{n,\eps}$ and $m\in\cM(d)$,
	\begin{align}\label{eqProp_genFirstMmt2}
	\Erw[Z_\rho(\G(n,m))]&\sim\frac{\exp(n f_n(\rho))}{\sqrt{(2\pi n)^{q-1}\prod_{\omega\in\Omega}\rho(\omega)}},\qquad\mbox{where}&
	f_n(\rho)&=\cH(\rho)+\frac mn\ln\phi(\rho).
	\end{align}
Indeed,  because there are precisely $\bink n{n\rho}$ assignments $\sigma\in \Omega^{V_n}$ such that $\rho_\sigma=\rho$
and since the constraint nodes of $\G(n,m)$ are chosen independently, we have the exact expression
	$\Erw[Z_{\rho}(\G(n,m))]= \binom{n}{\rho n}\phi(\rho)^m$ and thus (\ref{eqProp_genFirstMmt2}) follows from Stirling's formula.
Combining (\ref{eqProp_genFirstMmt1}) and (\ref{eqProp_genFirstMmt2}), we obtain
	\begin{align}\label{eqProp_genFirstMmt3}
	\Erw[Z_\eps(\G(n,m))]&\sim
	(2\pi n)^{(1-q)/2}q^{q/2}\sum_{\rho\in R_{n,\eps}}\exp(nf_n(\rho)).
	\end{align}
In order to calculate the sum via the Laplace method, we compute the first two derivatives of $f$.
The first derivative works out to be
	$$\frac{\partial f_n}{\partial\rho(\omega)}  =-\ln(\rho(\omega))-1+ \frac mn\cdot \frac{
			\sum_{\tau \in \Omega^k} \sum^{k}_{j=1}\Erw[\PSI(\tau)]\mathbf{1}\{ \tau_j=\omega\} \prod_{i\in [k]\setminus \{j \} } 	\rho\bc{\tau_i} }  {  \sum_{\tau \in \Omega^k} \Erw[\PSI(\tau)] \prod_{i\in [k]}\rho\bc{\tau_i}  }.$$
Hence, using {\bf SYM} we see that the gradient at the point $\bar\rho$ equals
	\begin{equation}\label{eqProp_genFirstMmt4}
	D f_n(\bar\rho)=(\ln(q)-1)\vecone+\frac{km}n\Phi\vecone=(\ln(q)-1+km/n)\vecone.
	\end{equation}
Proceeding to the second derivatives, we find
\begin{align}\nonumber
\frac{\partial^2 f_n}{\partial\rho(\omega)\partial\rho(\omega')}&=-\frac{\vecone\{\omega=\omega'\}}{\rho(\omega)}
	+\frac mn\cdot
		\frac{\sum_{\tau \in \Omega^k}\sum_{j,l\in[k]:j\neq l}
			 \mathbf{1}\{ \tau_j=\omega, \ \tau_l=\omega'\}\Erw[\PSI(\tau)]\prod_{i \in [k]\setminus \{j, l \} }\rho\bc{\tau_i} }
			 {\sum_{\tau \in \Omega^k}\Erw [\PSI(\tau)] \prod_{i\in [k]}\rho\bc{\tau_i} }\\
 &\qquad-\frac{m}{n}\frac{
		\left( \sum_{\tau \in \Omega^k} \Erw [\PSI(\tau)] \sum^{k}_{j=1}\mathbf{1}\left\{ \tau_j=\omega\right\}
		 \prod_{i\neq j }\rho\bc{\tau_i} \right)
	 \left( \sum_{\tau \in \Omega^k} \Erw [\PSI(\tau)]\sum^{k}_{j=1}\mathbf{1}\left\{ \tau_j=\omega'\right\}
			 \prod_{i\neq j }\rho\bc{\tau_i} \right)}
		 { \left( \sum_{\tau \in \Omega^k}\Erw [\PSI(\tau)] \prod_{i\in [k]}\rho\bc{\tau_i}\right)^2 }. \nonumber 
\end{align}
Consequently, using {\bf SYM} we find that the Hessian at $\bar\rho$ comes out as 
	\begin{align}\label{eqProp_genFirstMmt5}
	D^2f_n(\bar\rho)&=-q(\id-(k(k-1)m/n)\Phi)+(k^2m/n)\vecone.
	\end{align}
Additionally, the third derivatives of $f$ are uniformly bounded.
Thus, combining (\ref{eqProp_genFirstMmt4}) and (\ref{eqProp_genFirstMmt5}) and observing that
$\rho-\bar\rho\perp\vecone$ for all $\rho\in R_{n,\eps}$, we see that uniformly for all $\rho\in R_{n,\eps}$,
	\begin{align}\label{eqProp_genFirstMmt6}
	f_n(\rho)&=f_n(\bar\rho)-\frac q2\scal{(\id-(k(k-1)m/n)\Phi)(\rho-\bar\rho)}{(\rho-\bar\rho)}+O(\eps^3).
	\end{align}

Since $\eps=o(1)$, plugging (\ref{eqProp_genFirstMmt6}) into (\ref{eqProp_genFirstMmt1}) we obtain uniformly for all $m\in\cM_d$,
	\begin{align}
	\Erw[Z_\eps(\G(n,m))]&\sim(2\pi n)^{(1-q)/2}q^{q/2}\sum_{\rho\in R_{n,\eps}}\exp(nf_n(\rho))\nonumber\\
		&\sim(2\pi n)^{(1-q)/2}q^{q/2}\exp(nf(\bar\rho))\sum_{\rho\in R_{n,\eps}}
			\exp\brk{-\frac{qn}2\scal{(\id-(k(k-1)m/n)\Phi)(\rho-\bar\rho)}{(\rho-\bar\rho)}}\nonumber\\
		&\sim(2\pi n)^{(1-q)/2}q^{n+q/2}\xi^m
			\sum_{\rho\in R_{n,\eps}}\exp\brk{-\frac{qn}2\scal{(\id-(k(k-1)m/n)\Phi)(\rho-\bar\rho)}{(\rho-\bar\rho)}}.
				\label{eqProp_genFirstMmt7}
	\end{align}
Further, \Lem~\ref{Lemma_Phi} shows that $\Phi$ is symmetric, there exists an orthogonal matrix $Q$ such that $\Phi=QLQ^*$, where
$L$ is the diagonal matrix whose entries are the eigenvalues $1=\lambda_1\geq\lambda_2\geq\cdots\geq\lambda_q$  of $\Phi$.
Since $\Phi$ is stochastic (once more by \Lem~\ref{Lemma_Phi}), the top eigenvalue is $\lambda_1=1$ and the corresponding eigenvector is $\vecone$.
Moreover, because all $\rho\in R_{n,\eps}$ are probability distributions on $\Omega$, we have $\rho-\bar\rho\perp\vecone$ for all $\rho\in R_{n,\eps}$.
Therefore, the set $R_{n,\eps}'=\{Q^*(\rho-\bar\rho):\rho\in R_{n,\eps}\}$ is contained in the $(q-1)$-dimensional subspace spanned
by the eigenvectors of $\Phi$ corresponding to $\lambda_2,\ldots,\lambda_q$.
Hence, because $\eps\sqrt n\to\infty$ the sum from (\ref{eqProp_genFirstMmt7}) can be approximated by a $(q-1)$-dimensional Gaussian integral
and thus uniformly for all $m\in\cM_d$, 
	\begin{align}\nonumber
	\Erw[Z_\eps(\G(n,m))]&\sim
		(2\pi/q)^{\frac{1-q}2}q^{n+\frac12}\xi^m\int_{\RR^{q-1}}\exp\brk{-\frac{q}2\sum_{i=1}^{q-1}\bc{1-k(k-1)\frac mn\lambda_{i+1}}x_i^2}\dd x
		\sim
		\frac{q^{n+\frac12}\xi^m}{\prod_{i=2}^q\sqrt{1-d(k-1)\lambda_i}} 
			\nonumber,
	\end{align}
as claimed.
\end{proof}

\begin{remark}\label{Remark_smfix}
We observe that the proof of \Prop~\ref{Prop_genFirstMmt} did not use (\ref{eqBounded}).
\end{remark}

\subsection{Proof of \Prop~\ref{lem:FirstMoment}}
In this section we assume that $P$ satisfies {\bf SYM} and {\bf BAL}.
Then \Lem~\ref{Lemma_Phi} readily shows that (\ref{eqProp_genFirstMmt}) holds for all $d>0$ and thus \Prop~\ref{Prop_genFirstMmt} applies.
Hence, to prove \Prop~\ref{lem:FirstMoment} we merely need to show that $\Erw[Z_\eps(\G(n,m))]\sim\Erw[Z(\G(n,m))]$
for a suitable $\eps(n)=o(1)$.

\begin{lemma}\label{lem:FirstMoment_2}
Assume that $P$ satisfies {\bf SYM} and {\bf BAL}, let $d>0$ and set $\eps=\eps(n)=n^{-1/3}$.
Then uniformly for all $m\in\cM(d)$ we have $\Erw[Z_\eps(\G(n,m))]\sim\Erw[Z(\G(n,m))].$
\end{lemma}
\begin{proof}
Let $R_n$ be the set of all distributions $\rho\in\cP(\Omega)$ such that $n\rho$ is an integer vector and let
$R_{n,\eps}$ be the set of all $\rho\in R_n$ such that $|\rho\bc\omega-1/q|<\eps$ for all $\omega\in\Omega$.
Let
	$\phi:\rho\in\RR^\Omega\mapsto\sum_{\tau\in\Omega^k}\Erw[\vec\psi(\tau)]\prod_{i\in[k]}\rho(\tau_i)$
(cf.\ \eqref{eqF}).
Then by the linearity of expectation and the independence of the constraint nodes of $\G(n,m)$,
	\begin{align*}
	\Erw[Z(\G(n,m))]&=\sum_{\rho\in R_n}\bink n{n\rho}\phi(\rho)^m,&
		\Erw[Z_\eps(\G(n,m))]&=\sum_{\rho\in R_{n,\eps}}\bink n{n\rho}\phi(\rho)^m.
	\end{align*}
Hence, with $\bar\rho$ denoting the uniform distribution, uniformly for all $m\in\cM(d)$,
	\begin{align*}
	\Erw[Z(\G(n,m))]-\Erw[Z_\eps(\G(n,m))]&=\sum_{\rho\in R_n\setminus R_{n,\eps}}\bink n{n\rho}\phi(\rho)^m\\
		&\leq\sum_{\rho\in R_n\setminus R_{n,\eps}}\exp\bc{n\cH(\rho)+m\ln \phi(\rho)+O(\ln n)}&&\mbox{[by Stirling]}\\
		&\leq\sum_{\rho\in R_n\setminus R_{n,\eps}}\exp\bc{n\cH(\rho)+m\ln \phi(\bar\rho)+O(\ln n)}&&\mbox{[due to {\bf BAL}]}\\
		&\leq\exp\bc{n\cH(\bar\rho)+m\ln \phi(\bar\rho)-\Omega(n^{1/3})}&&\mbox{[as $\cH(\nix)$ is strictly concave]}\\
		&=q^n\xi^m\exp(-\Omega(n^{1/3}))&&\mbox{[due to {\bf SYM}].}
	\end{align*}
Finally, \Prop~\ref{Prop_genFirstMmt} implies that $q^n\xi^m\exp(-\Omega(n^{1/3}))=o(\Erw[Z_\eps(\G(n,m))])$.
\end{proof}

\noindent
\Prop~\ref{lem:FirstMoment} is immediate from \Prop~\ref{Prop_genFirstMmt} and \Lem~\ref{lem:FirstMoment_2}.

\subsection{Proof of \Prop~\ref{lem:SecondMoment}}
Assume that $P$ satisfies {\bf SYM} and {\bf BAL} and that $d<\dc$.
In order to calculate the second moment, we employ a known construction (e.g., \cite{Victor}) of an auxiliary random factor graph model
whose first moment equals the second moment of the original model.
The spin set of this auxiliary model is the set $\Omega^\tensor=\Omega\times\Omega$ and
we denote the pairs $(s,t)\in\Omega\times\Omega$ by $s\tensor t$.
Further, for functions $\varphi,\psi:\Omega^k\to\RR$ we define
	$$\varphi\tensor\psi:(\Omega^\tensor)^k\to\RR,\qquad(\sigma_1\tensor\tau_1,\ldots,\sigma_k\tensor\tau_k)\mapsto
		\varphi(\sigma_1,\ldots,\sigma_k)\psi(\tau_1,\ldots,\tau_k).$$
Then the set of weight functions of the auxiliary model is $\Psi^\tensor=\{\psi\tensor\psi:\psi\in\Psi\}$.
Moreover, the probability distribution $P^\tensor$ on $\Psi^\tensor$  is simply the image of $P$ under the measurable map $\psi\in\Psi\mapsto\psi\tensor\psi$.
Clearly, the fact that $P$ satisfies {\bf SYM} implies that so does $P^\tensor$.
(However, $P^\tensor$ does not necessarily satisfy {\bf BAL}, and $P^\tensor$ need not satisfy the last two bounds in (\ref{eqBounded}), but these
are not needed to apply \Prop~\ref{Prop_genFirstMmt} due to Remark~\ref{Remark_smfix}.)

For any $\psi\in\Psi$ the matrix $\Phi_{\psi\tensor\psi}$ as defined in (\ref{eqPhiMatrices}) can be expressed 
in terms of the matrix $\Phi_\psi$ induced by the original weight function  as $\Phi_{\psi\tensor\psi}=\Phi_\psi\tensor\Phi_\psi$.
Hence, recalling the definitions (\ref{eqXi}) and (\ref{eqPhi}),
	\begin{align}\label{eqPhitensor}
	\Phi_{P^\tensor}&=\Erw[\Phi_{\PSI\tensor\PSI}]=\Erw[\Phi_{\PSI}\tensor\Phi_{\PSI}]=\Xi_P.
	\end{align}

\begin{proof}[Proof of \Prop~\ref{lem:SecondMoment}]
For a factor graph $G$ let $G^\tensor$ be the factor graph obtained by replacing 
the weight function $\psi_a$ by $\psi_a\tensor\psi_a$ for every factor node $a$ of $G$.
Then 
	\begin{align*}
	Z(G^\tensor)&=\sum_{\sigma\in(\Omega^{\tensor})^n}\prod_{a\in F(G)}(\psi_a\tensor\psi_a)(\sigma(\partial_1a),\ldots,\sigma(\partial_ka))\\
		&=\sum_{\sigma,\tau\in\Omega^{n}}\prod_{a\in F(G)}\psi_a(\sigma(\partial_1a),\ldots,\sigma(\partial_ka))
			\psi_a(\tau(\partial_1a),\ldots,\tau(\partial_ka))=Z(G)^2.
	\end{align*}
Hence, if $\eps=\eps(n)=o(1)$ satisfies $\eps\sqrt n\to\infty$, then
(\ref{eqPhitensor}), \Lem~\ref{Lemma_Xi}, \Prop~\ref{prop_KS} and  \Prop~\ref{Prop_genFirstMmt} yield
	\begin{align*}
	\Erw[\cZ_\eps(\G(n,m))^2]\leq\Erw[Z_\eps(\G(n,m)^\tensor)]
		\sim \frac{q^{2n+1}{\xi^{2m}}}
 	{ \prod_{\lambda\in\eig(\Xi)\setminus\{1\}}\sqrt{1-d(k-1)\lambda}},
	\end{align*}
as desired.
\end{proof}

\section{Cycle census}\label{Sec_cycles}
\noindent
{\em Throughout this section we assume that $P$ satisfies {\bf SYM} and {\bf BAL}.}

\medskip\noindent
The aim is to prove \Prop~\ref{prop:FirstCondOverFirst}.
The proof of the first assertion is rather straightforward.

\begin{lemma}\label{lemma:PoissonCycles4Uniform}
Let $d>0$.
For any $Y\in\cY$ we have $\Erw[C_Y(\G(n,m))] \sim  \kappa_Y$, uniformly for all $m\in\cM(d)$.
Moreover, if $Y_1,\ldots,Y_l\in\cY$ are pairwise {disjoint} and $y_1,\ldots,y_l\geq0$, then uniformly for all $m\in\cM(d)$,
	\begin{equation}\label{eqlemma:PoissonCycles4Uniform}
	\Pr\brk{\forall i\leq l: C_{Y_i}(\G(n,m))=y_i}\sim \prod^{l}_{t=1}\Pr[\Po(\kappa_{Y_t})=y_t].
	\end{equation}
\end{lemma}
\begin{proof}
Let $m\in\cM(d)$ be such that $m(n)$ takes the least possible value for every $n$.
Then (\ref{eqlemma:PoissonCycles4Uniform}) is immediate from Fact~\ref{lemma:cycles} and the fact that in $\G(n,m)$ the weight functions of the constraint nodes are chosen independently from $P$.
Furthermore, if $m'\in\cM(d)$ is another sequence, then the random graph $\G(n,m')$ is obtained from $\G(n,m)$ by adding at most $n^{3/4}$
random edges and with probability $1-o(1)$ none of these edges closes a cycle of bounded length.
Hence, we obtain the desired uniform rate of convergence for all sequences in $\cM(d)$.
\end{proof}

\begin{lemma}\label{lemma:PoissonCycles4Planted}
Let $d>0$.
For any $Y\in\cY$ with $\kappa_Y>0$ we have $\Erw[C_Y(\hat\G(n,m))] \sim\hat\kappa_Y$, uniformly for all $m\in\cM(d)$.
Moreover, if $Y_1,\ldots,Y_l\in\cY$ are pairwise {disjoint}, $\kappa_{Y_1},\ldots,\kappa_{Y_l}>0$ and $y_1,\ldots,y_l\geq0$, then uniformly for all $m\in\cM(d)$,
	$$\Pr\brk{\forall i\leq l: C_{Y_i}(\hat\G(n,m))=y_i}\sim \prod^{l}_{t=1}\Pr[\Po(\hat\kappa_{Y_t})=y_t].$$
\end{lemma}

\noindent
The proof is based on known arguments. We begin by calculating the expected number of dense small subgraphs of $\hat\G(n,m)$.

\begin{claim}\label{claim:NoOfInteresectingCycles}
Let $u\geq1$ be an integer and let $U(G)$ be the number of subsets $S\subset  V_n\cup F_m$ of size $|S|=u$ that span more than $2|U|$ edges.
Then $\Erw[U(\G^*(n,m,\sigma))]=O(1/n)$ uniformly for all $m\in\cM(d)$ and all $\sigma\in\Omega^{V_n}$.
\end{claim}
\begin{proof}
Fix numbers $u_1,u_2$ such that $u_1+u_2=u$ and let $S_1\subset V_n$ and $S_2\subset F_m$ be sets of size $|S_1|=u_1$, $|S_2|=u_2$.
Moreover, let $E\subset S_2\times[k]$ be a set of size $v>u_1+u_2$ and let $\cA(S_1,S_2,E)$ be the event that
for all pairs $(a,i)\in E$ we have $\partial_ia\in S_1$.
Then
	\begin{align}\label{eqclaim:NoOfInteresectingCycles_1}
	\Erw[U(\G^*(n,m,\sigma))]&\leq\sum_{u_1,u_2,S_1,S_2,E}\pr\brk{\G^*(n,m,\sigma)\in\cA(S_1,S_2,E)}.
	\end{align}
Furthermore, \eqref{eqBounded} ensures that there is a number $\alpha=\alpha(P)>0$ that does not depend on $\sigma$ such that the lower bound
	$\sum_{y_1,\ldots,y_k\in V_n}\Erw[\PSI(\sigma(y_1),\ldots,\sigma(y_k))]\geq\alpha n^k$ holds.
Therefore, (\ref{eqTeacher}) implies that for variable nodes $y_1,\ldots,y_k\in S_1$, any constraint node $a\in S_2$ and for any subset $J\subset[k]$ we have
	\begin{align}\label{eqclaim:NoOfInteresectingCycles_2}
	\pr\brk{\forall i\in J:\partial_ia=y_i}&\leq\frac{\Erw[\max_{\tau\in\Omega^k}\PSI(\tau)]n^{k-|J|}}{\alpha n^k}=O(n^{-|J|}).
	\end{align}
Since the constraint nodes are chosen independently, (\ref{eqclaim:NoOfInteresectingCycles_2}) implies that, uniformly for all $\sigma$ and all $m\in\cM(d)$,
	\begin{align}\label{eqclaim:NoOfInteresectingCycles_3}
	\pr\brk{\G^*(n,m,\sigma)\in\cA(S_1,S_2,E)}&\leq O(n^{-|E|}).
	\end{align}
Finally, given $u_1,u_2$ the number of possible sets $S_1$ is bounded by $n^{u_1}$, the number of possible $S_2$ does not exceed $m^{u_2}$
and given $v$ and $S_2$ the number of possible sets $E$ is bounded.
Thus, since $u_1+u_2<v\leq ku_2$ the assertion follows from (\ref{eqclaim:NoOfInteresectingCycles_1}) and (\ref{eqclaim:NoOfInteresectingCycles_3}).
\end{proof}

\begin{proof}[Proof of \Lem~\ref{lemma:PoissonCycles4Planted}]
Due to the Nishimori identity (\ref{eq:nishimori}) we may prove the claim for the random factor graph model $\G'=\G^*(n,m,\hat\SIGMA_{n,m})$.
Moreover, by \Cor~\ref{lem:conc_coloring} we may condition on the event that $|\hat\SIGMA^{-1}(\omega)|\sim n/q$ for all $\omega\in\Omega$,
in which case {\bf SYM} yields
	\begin{align}\label{eqCycCount2}
	\sum_{u_1,\ldots,u_k\in[n]}\Erw[\PSI(\hat\SIGMA(x_{u_1}),\ldots,\hat\SIGMA(x_{u_k}))]\sim n^k\xi.
	\end{align}

We begin by showing that for any $Y=(E_1,s_1,t_1,\ldots,E_\ell,s_\ell,t_\ell)\in\cY_\ell$ uniformly for all $m\in\cM(d)$,
	\begin{equation}\label{eqlemma:PoissonCycles4Planted_1}
	\Erw[C_Y(\G')]\sim\hat\kappa_Y.
	\end{equation}
Indeed, let $\vec i=(i_1,\ldots,i_\ell)\in[n]$ be a family of pairwise distinct indices such that $i_1<\min\{i_2,\ldots,i_\ell\}$ (cf.\ {\bf CYC1})
and let $\vec j=(j_1,\ldots,j_\ell)\in[m]$ be pairwise distinct indices {such that $j_1<\min\{j_2,\ldots,j_\ell\}$ if $\ell>1$} (cf.\ {\bf CYC2}).
Let $\cC_Y(\vec i,\vec j)$ be the event that $x_{i_1},a_{j_1},\ldots,x_{i_\ell},a_{j_\ell}$ form a cycle with signature $Y$.
Set $i_{\ell+1}=i_1$.
Then by (\ref{eqTeacher}), (\ref{eqCyclePhi}) and (\ref{eqCycCount2}) we have
	\begin{align}\nonumber
	\pr\brk{\G'\in\cC_Y(\vec i,\vec j)}&=\prod_{h=1}^\ell
		\frac{\sum_{u_1,\ldots,u_k\in[n]}\vecone\{u_{s_h}=i_h,u_{t_h}=i_{h+1}\}
			\Erw[\PSI(\hat\SIGMA(x_{u_1}),\ldots,\hat\SIGMA(x_{u_k}))\vecone\{\PSI\in E_h\}]}
			{\sum_{u_1,\ldots,u_k\in[n]}\Erw[\PSI(\hat\SIGMA(x_{u_1}),\ldots,\hat\SIGMA(x_{u_k}))]}\\
		&\sim n^{-2\ell}q^\ell\prod_{h=1}^\ell\Phi_{E_h,s_h,t_h}(\hat\SIGMA(x_{i_h}),\hat\SIGMA(x_{i_{h+1}})).\label{eqCycCount1}
	\end{align}
Summing on $\vec i,\vec j$, we get
	\begin{align*}
	\Erw[C_Y(\G')]&\sim{\frac1{2\ell}}\bcfr{mq}{n^2}^\ell\sum_{\vec i}\prod_{h=1}^l\Phi_{E_h,s_h,t_h}(\hat\SIGMA(x_{i_h}),\hat\SIGMA(x_{i_{h+1}}))
		=\kappa_Y\Tr\prod_{h=1}^\ell\Erw[\Phi_{\PSI,s_h,t_h}|E_h]=\hat\kappa_Y,
	\end{align*}
as claimed.

For integers $h_1,\ldots,h_l\geq1$ let 
	$C_{h_1,\ldots,h_l}(\G')=\prod_{i=1}^l\prod_{l=1}^{h_i}(C_{Y_i}(\G')-l+1).$
{Then due to the inclusion/exclusion argument for the joint convergence to independent Poisson variables~\cite[\Thm~1.23]{BB}}, in order to
complete the proof it suffices to show that for any $h_1,\ldots,h_l\geq1$, uniformly for all $m\in\cM(d)$,
	\begin{align}\label{eqCycCount10}
	\Erw\brk{C_{h_1,\ldots,h_l}(\G')}\sim\prod_{i=1}^{l}\hat\kappa_{Y_i}^{h_i}.
	\end{align}
Combinatorially, $C_{h_1,\ldots,h_l}(\G')$ is nothing but the total number of $(h_1+\ldots+h_l)$-tuples of cycles in $\G'$ such that the first $h_1$
cycles have signature $Y_1$, the next $h_2$ cycles have signature $Y_2$, etc.
Hence, if we define $C_{h_1,\ldots,h_l}'(\G')$ as the number of such families of pairwise vertex disjoint cycles, then Claim~\ref{claim:NoOfInteresectingCycles} yields
	\begin{align}\label{eqCycCount11}
	\Erw\brk{C_{h_1,\ldots,h_l}(\G')}=\Erw\brk{C_{h_1,\ldots,h_l}'(\G')}+o(1).
	\end{align}
Furthermore, we claim that uniformly for all $m\in\cM(d)$,
	\begin{align}\label{eqCycCount12}
	\Erw\brk{C_{h_1,\ldots,h_l}'(\G')}\sim\prod_{i=1}^{l}\hat\kappa_{Y_i}^{h_i}.
	\end{align}
Indeed, the argument that we used to prove (\ref{eqlemma:PoissonCycles4Planted_1}) easily extends to a proof of (\ref{eqCycCount12});
for if we fix index families $(\vec i_{v,w},\vec j_{v,w})_{v=1,\ldots,l,w=1,\ldots,h_s}$ that suit the signatures $Y_1,\ldots,Y_l$
such that no index from $[n]$ resp.\ $[m]$ occurs more than once, then similar steps as above reveal that
	\begin{align*}
	\pr\brk{\G'\in\bigcap_{v=1}^l\bigcap_{w=1}^{h_v}\cC_{Y_v}(\vec i_{v,w},\vec j_{v,w})}
		\sim\prod_{v=1}^l\prod_{w=1}^{h_v}\pr\brk{\G'\in\cC_{Y_v}(\vec i_{v,w},\vec j_{v,w})}
	\end{align*}
Hence, (\ref{eqCycCount12}) follows by summing on all $(\vec i_{v,w},\vec j_{v,w})_{v,w}$.
Finally, (\ref{eqCycCount10}) and (\ref{eqCycCount12})
show the dedired convergence for a single sequence $m\in\cM(d)$ and the uniformity of the rate of convergence follows from a similar argument
as in the proof of \Lem~\ref{lemma:PoissonCycles4Uniform}.
\end{proof}

\begin{proof}[Proof of {\Prop~\ref{prop:FirstCondOverFirst}}]
The claim (\ref{eqCyclePoisson}) about the cycle counts is immediate from \Lem s~\ref{lemma:PoissonCycles4Uniform} and~\ref{lemma:PoissonCycles4Planted}.
To prove the assertion about the probability of $\fS$, let us first assume that $k=2$.
Then the event $\fS$ occurs iff $C_1=C_2=0$ and thus the assertion about $\pr\brk{\G(n,m)\in\fS}$ is immediate from Fact~\ref{lemma:cycles}.
Moreover, the assertion about $\pr[\hat\G(n,m)\in\fS]$ follows from \Lem~\ref{lemma:PoissonCycles4Planted} applied to all signatures of the form
$(s_1,t_1,\Psi)$ and $(s_1,t_1,\Psi,s_2,t_2,\Psi)$.
For $k>2$ we express the event $\fS$ as
	$\fS=\cbc{C_1=0\wedge\forall 1\leq i<j\leq m:\{\partial_1a_i,\ldots,\partial_ka_i\}\neq\{\partial_1a_i,\ldots,\partial_ka_i\}}.$
In particular, $\fS$ occurs only if $C_1=0$ and therefore, by the same token as in the case $k=2$, the expressions stated in 
\Prop~\ref{prop:FirstCondOverFirst} are asymptotic upper bounds on $\pr[\G(n,m)\in\fS],\pr[\hat\G(n,m)\in\fS]$.
Finally, we notice that for $k>2$ the expected number of pairs $1\leq i<j\leq m$ such that 
$\{\partial_1a_i,\ldots,\partial_ka_i\}=\{\partial_1a_i,\ldots,\partial_ka_i\}$ is $O(1/n)$.
\end{proof}

\section{The limiting distribution}\label{Sec_totallyVexed}

\noindent{\em Throughout this section we assume that $P$ satisfies {\bf SYM} and {\bf BAL}.}

\medskip\noindent
In this section we prove  \Prop~\ref{Lemma_totallyVexed}.
Let $\PSI,\PSI_1,\PSI_2,\ldots$ be chosen independently from $P$
and for $\ell\geq0$ set	$\vY_\ell=\Tr\prod_{j=1}^\ell\Phi_{\PSI_{j}}$.
The following lemma is the main step toward the proof of (\ref{eqVarFormula}).

\begin{lemma}\label{prop:Bound4DeltaYSq}
If $d<\dc$, then $\sum_{\ell=1}^\infty\frac{(d(k-1))^\ell}{2\ell}\Erw\brk{(\vY_\ell-1)^2}=-\frac12\sum_{\lambda\in\Eig[\Xi]}\ln\bc{1-d(k-1)\lambda}.$
\end{lemma}
\begin{proof}
Let $\PHI_\ell=\prod_{j=1}^\ell\Phi_{\PSI_{j}}$.
Then 
	$$(\Tr\PHI_\ell-1)^2=(\Tr\PHI_\ell)^2-2\Tr\PHI_\ell+1=\Tr(\PHI_\ell\tensor\PHI_\ell)-2\Tr\PHI_\ell+1.$$
Hence, remembering (\ref{eqXi}) and (\ref{eqPhi}), we find
	$\Erw[(\vY_\ell-1)^2]=\Erw[(\Tr\PHI_\ell-1)^2]=\Tr(\Xi^\ell)-2\Tr(\Phi^\ell)+1.$
Furthermore,
\Lem s~\ref{Lemma_Phi} and~\ref{Lemma_Xi} yield
	$$\Tr(\Xi^\ell)=\sum_{\lambda\in\eig(\Xi)}\lambda^\ell=1+2\sum_{\lambda\in\eig(\Phi)\setminus\{1\}}\lambda^\ell+\sum_{\lambda\in\Eig[\Xi]}\lambda^\ell=-1+2\Tr(\Phi^\ell)+\sum_{\lambda\in\Eig[\Xi]}\lambda^\ell,$$
and thus
	\begin{align}\label{eqSumYell}
	\frac{(d(k-1))^\ell}{2\ell}\Erw\brk{\bc{\vY_\ell-1}^2}&=
			\sum_{\lambda\in\Eig[\Xi]}\frac{\bc{d(k-1)\lambda}^\ell}{2\ell}.
	\end{align}
As $d<\dc$ \Prop~\ref{prop_KS} yields $\max_{\lambda\in\Eig[\Xi]}|\lambda|<d(k-1)$, whence summing (\ref{eqSumYell}) on $\ell$ completes the proof.
\end{proof}

\noindent
To prove (\ref{eqVarFormula}) we need to get a handle on the discretization of the set $\Psi$ induced by the partition $\fC_r$ for $r\geq1$.
Hence, we introduce $\vY_{\ell,r}=\Tr\prod_{j=1}^\ell\Phi_{\PSI_{j}^{(r)}}.$

\begin{corollary}\label{cor:Bound4DeltaYSq}
If $d<\dc$, then
 $\sum_{\ell=1}^\infty\frac{(d(k-1))^\ell}{2\ell}\Erw[(\vY_{\ell,r}-1)^2]\leq-\frac12\sum_{\lambda\in\Eig[\Xi]}\ln\bc{1-d(k-1)\lambda}$.
\end{corollary}
\begin{proof}
By Jensen's inequality
	$\sum_{\ell=1}^\infty\frac{(d(k-1))^\ell}{2\ell}\Erw[(\vY_{\ell,r}-1)^2]\leq	\sum_{\ell=1}^\infty\frac{(d(k-1))^\ell}{2\ell}\Erw[(\vY_{\ell}-1)^2]$
and thus the assertion follows from \Lem~\ref{prop:Bound4DeltaYSq}.
\end{proof}

\noindent
We are ready to prove \eqref{eqVarFormula}.

\begin{proof}[Proof of \Prop~\ref{Lemma_totallyVexed}, part 1.]
Given $L,r$ let 
	\begin{align*}
	S_{L,r}&=\sum_{Y\in\cY_{\leq L,r}}\frac{(\kappa_Y-\hat\kappa_Y)^2}{\kappa_Y}=\sum_{\ell=1}^L\frac{(d(k-1))^\ell}{2\ell}\Erw[(\vY_{\ell,r}-1)^2],&
	S_L&=\sum_{\ell=1}^L\frac{(d(k-1))^\ell}{2\ell}\Erw[(\vY_{\ell}-1)^2].
	\end{align*}
The construction of $\fC_r$ ensures that for every fixed $\ell$, $\vY_{\ell,r}$ converges to $\vY_\ell$ almost surely as $r\to\infty$.
Hence, by \Lem~\ref{prop:Bound4DeltaYSq}, \Cor~\ref{cor:Bound4DeltaYSq} and dominated convergence,
	$$\lim_{L\to\infty}\lim_{r\to\infty}\exp(S_{L,r})=\lim_{L\to\infty}\exp(S_L)=
		\prod_{\lambda\in\Eig[\Xi]}\bc{1-d(k-1)\lambda}^{-\frac12},$$
which proves \eqref{eqVarFormula}.
\end{proof}

\noindent
In order to establish the convergence of $\cK_{\ell,r}$ to $\cK$ we use similar arguments.
We begin with the following bound.

\begin{lemma}\label{Lemma_vexed}
For every $0<d\leq\dc$ there exists $\beta>0$ such that
	$\sum_{\ell=1}^\infty\frac{(d(k-1))^\ell}{2\ell}\Erw\abs{\vecone\{\vY_\ell<\beta\}\ln \vY_\ell}<\infty.$
\end{lemma}
\begin{proof}
Pick $\beta>0$ sufficiently small and let $S=\sum_{\ell=1}^\infty(d(k-1))^\ell\Erw\abs{\vecone\{\vY_\ell<\beta\}\ln \vY_\ell}/\bc{2\ell}.$
Because by \Lem~\ref{Lemma_Phi} the matrices $\Phi_\psi$ are stochastic, we have
	\begin{align*}
	\Tr(\Phi_{\PSI_1}\cdots\Phi_{\PSI_\ell})
		&=\sum_{\sigma_1,\ldots,\sigma_\ell}\Phi_{\PSI_1}(\sigma_1,\sigma_2)\cdots\Phi_{\PSI_\ell}(\sigma_\ell,\sigma_1)
		\geq\min_{\sigma,\sigma'}\Phi_{\PSI_\ell}(\sigma,\sigma')\geq q^{1-k}\xi^{-1}\min_{\tau\in\Omega^k}\PSI_\ell(\tau).
	\end{align*}
In fact, since the trace is invariant under cyclic permutations, we obtain
	\begin{align}\label{eq_cyclic_permutations}
	\Tr(\Phi_{\PSI_1}\cdots\Phi_{\PSI_\ell})\geq q^{1-k}\xi^{-1}\max_{j\in[\ell]}\min_{\tau\in\Omega^k}\PSI_j(\tau).
	\end{align}
Since $\PSI_1,\ldots,\PSI_\ell$ are chosen independently, (\ref{eqBounded}) and (\ref{eq_cyclic_permutations}) imply that we can choose $\beta<0$ small
enough so that 
	$\Erw|\vecone\{\vY_\ell<\beta\}\ln \vY_\ell|\leq(d(k-1))^{-\ell}$ for all $\ell$, in which case the sum converges.
\end{proof}

\begin{corollary}\label{Cor_nasty}
For every $0<d<\dc$ and every $\ell,r\geq1$ we have $\Erw|\ln\vY_\ell|+\Erw|\ln\vY_{\ell,r}|<\infty$.
\end{corollary}
\begin{proof}
Because all weight functions $\psi\in\Psi$ take values in $(0,2)$, it is obvious that $\Erw\abs{\vecone\{\vY_\ell\geq\beta\}\ln \vY_\ell}<\infty$ for every $\beta<1$.
Moreover, similar steps as in the previous proof show 
$\sum_{l\geq1}\Erw\abs{\vecone\{\vY_l<\beta\}\ln \vY_l}<\infty$ for some small $0<\beta<1$.
Finally, since $x\in(0,\beta)\mapsto-\ln x$ is convex, the assertion about $|\ln\vY_{\ell,r}|$ follows from Jensen's inequality.
\end{proof}

We are going to prove that $\cK,\cK_\ell$ are well-defined by showing that they come out as the limit of the $\cK_{\ell,r}$ as $\ell,r\to\infty$.
However, a priori it may not be entirely clear that the $\cK_{\ell,r}$ are well-defined because they involve sums on random numbers $K_l$ of terms.
Let us observe that this is not a problem actually, because \Cor~\ref{Cor_nasty} implies the following.
We continue to let $(\PSI_{l,i,j})_{l,i,j}$ signify a family of independent samples from $P$.

\begin{corollary}\label{Cor_randomFubini}
For every $l\geq1,r\geq1$ the following $L^1$-limits exist:
	\begin{align*}
	\sum_{i=1}^{K_{l}}\ln{\Tr\prod_{j=1}^l\Phi_{\PSI_{l,i,j}}}&=
		\lim_{H\to\infty}\sum_{i=1}^{K_{l}\wedge H}\ln{\Tr\prod_{j=1}^l\Phi_{\PSI_{l,i,j}}},&
	\sum_{i=1}^{K_{l}}\ln\Tr\prod_{j=1}^l\Phi_{\PSI_{l,i,j}^{(r)}}&=
		\lim_{H\to\infty}\sum_{i=1}^{K_{l}\wedge H}\ln\Tr\prod_{j=1}^l\Phi_{\PSI_{l,i,j}^{(r)}}.
	\end{align*}
\end{corollary}

\begin{lemma}\label{Claim_vexed}
For every $0<d<\dc$ there exists $c=c(d,P)>0$ such that for all $r\geq1$, $L\geq1$,
	\begin{align*}
	\sum_{l=1}^L\Erw{\abs{\frac{(d(k-1))^l}{2l}\bc{1-\Tr(\Phi^l)}+\sum_{i=1}^{K_{l}}\ln{\Tr\prod_{j=1}^l\Phi_{\PSI_{l,i,j}}}}}&<c,&
	\sum_{l=1}^L\Erw{\abs{\frac{(d(k-1))^l}{2l}\bc{1-\Tr(\Phi^l)}+\sum_{i=1}^{K_{l}}\ln{\Tr\prod_{j=1}^l\Phi_{\PSI_{l,i,j}^{(r)}}}}}&<c.
	\end{align*}
\end{lemma}
\begin{proof}
Let $\kappa_l=(d(k-1))^l/\bc{2l}$, $\vec X_{l,i}=\Tr\prod_{j=1}^l\Phi_{\PSI_{l,i,j}},\vec X_{l,i}^{(r)}=\Tr\prod_{j=1}^l\Phi_{\PSI_{l,i,j}^{(r)}}$.
Then $\Erw[\vec X_{l,i}]=\Tr(\Phi^l)$ and for every $l\geq1$,
	\begin{align}\nonumber
	\Erw\abs{\frac{(d(k-1))^l}{2l}\bc{1-\Tr(\Phi^l)}+\sum_{i=1}^{K_{l}}\ln{\Tr\prod_{j=1}^l\Phi_{\PSI_{l,i,j}}}}
		&=\Erw\abs{\kappa_l\Erw[\vY_{l}-1]-\sum_{i=1}^{K_{l}}\ln \vec X_{l,i}}\\ 
		\leq\Erw\abs{\kappa_l\Erw[\vY_{l}-1]-\sum_{i=1}^{K_{l}}(\vec X_{l,i}-1)}+
			\Erw\abs{\sum_{i=1}^{K_{l}}\vec X_{l,i}-1-\ln \vec X_{l,i}}
		&\leq\bc{\Var\sum_{i=1}^{K_{l}}(\vec X_{l,i}-1)}^{1/2}
			+\Erw\abs{\sum_{i=1}^{K_{l}}\vec X_{l,i}-1-\ln \vec X_{l,i}}\label{eqClaim_vexed_0}
	\end{align}
because $\Erw[\sum_{i=1}^{K_{l}}(\vec X_{l,i}-1)]=\kappa_l\Erw[\vY_{l}-1]$ and due to Cauchy-Schwarz.
Further, because the $\PSI_{l,i,j}$ are i.i.d., for any given integer $h$ we find
	\begin{align}\label{eqClaim_vexed_1}
	\Erw\brk{\bc{\sum_{i=1}^{h}(\vec X_{l,i}-1)}^2}&=h(h-1)\Erw[\vY_{l}-1]^2+h\Erw[(\vY_{l}-1)^2].
	\end{align}
As $\Erw[K_l(K_l-1)]=\kappa_l^2$, (\ref{eqClaim_vexed_1}) implies
	\begin{align}\label{eqClaim_vexed_2}
	\Var\brk{\sum_{i=1}^{K_{l}}(\vec X_{l,i}-1)}&=\kappa_l\Erw[(\vY_{l}-1)^2].
	\end{align}

Moving on to the second summand in (\ref{eqClaim_vexed_0}), we recall that the function $x\in(0,\infty)\mapsto x-1-\ln x$ is convex and
that for any (small) $\beta>0$ there exists $u>0$ such that $x-1-\ln x\leq u(x-1)^2$ for all $x\geq\beta$.
Hence, introducing the convex function $g:x\in(0,\infty)\mapsto\max\{x-1-\ln x,u(x-1)^2\}\geq0$, we have
	\begin{align}\label{eqClaim_vexed_3}
	\Erw\abs{\sum_{i=1}^{K_{l}}\vec X_{l,i}-1-\ln \vec X_{l,i}}\leq\Erw\brk{\sum_{i=1}^{K_{l}}g(\vec X_{l,i})}
		\leq2\kappa_l\Erw[\vecone\{\vY_{l}\leq\beta\}\ln \vY_{l}]+u^2\kappa_l\Erw[(\vY_{l}-1)^2].
	\end{align}
\Lem s~\ref{prop:Bound4DeltaYSq} and~\ref{Claim_vexed} show that summing the right hand sides of (\ref{eqClaim_vexed_2}) and (\ref{eqClaim_vexed_3}) on $l$ gives a finite number. Thus, the first assertion follows from (\ref{eqClaim_vexed_0}).
With respect to the second bound, analogous steps yield
	\begin{align*}
	\Erw\abs{\frac{(d(k-1))^l}{2l}\bc{1-\Tr(\Phi^l)}+\sum_{i=1}^{K_{l}}\ln{\Tr\prod_{j=1}^l\Phi_{\PSI_{l,i,j}^{(r)}}}}
		\leq\sqrt{\Var\brk{\sum_{i=1}^{K_{l}}(\vec X_{l,i}^{(r)}-1)}}
			+\Erw\brk{\sum_{i=1}^{K_{l}}g(\vec X_{l,i}^{(r)})}
	\end{align*}
and thus the desired bound follows from Jensen's inequality.
\end{proof}

\begin{proof}[Proof of \Prop~\ref{Lemma_totallyVexed}, part 2]
\Lem~\ref{Claim_vexed} shows that the random variables $\cK_{\ell,r}$ are uniformly $L^1$-bounded.
Furthermore, the construction of $\fC^r$ guarantees that $\cK_{\ell,r}\to\cK_\ell$ almost surely for every fixed $\ell$.
Hence, $\cK_{\ell,r}$ converges to $\cK_\ell$ in the $L^1$-norm and
a second application of \Lem~\ref{Claim_vexed} shows that $\cK_\ell$ tends to $\cK$ in the $L^1$-norm.
\end{proof}

\section{The condensation threshold}\label{Sec_Thm_plantedFreeEnergy}

\noindent{\em Throughout this section we assume that $P$ satisfies {\bf SYM}, {\bf BAL} and {\bf POS}.}

\medskip\noindent
In this section we prove \Thm s~\ref{Thm_cond} and~\ref{Thm_plantedFreeEnergy}.
As a technical preparation we need a concentration inequality for the free energy of our random factor graph models.

\subsection{Concentration}
We begin with the following elementary observation.

\begin{lemma}\label{Lemma_O}
Suppose that $P$ satisfies {\bf SYM} and {\bf BAL}.
For a factor graph $G=(V,F,(\partial a)_{a\in F},(\psi_a)_{a\in F})$ define 
	\begin{align*}
	\cO(G)=\sum_{\sigma\in\Omega^k}\sum_{a\in F}\ln^2\psi_a(\sigma).
	\end{align*}
Then for every $D>0$ there exists $C=C(D,P)>0$ such that uniformly for all $m\leq Dn$, $t\geq1$ and $\sigma\in\Omega^{V_n}$ we have
		\begin{align}\label{eqO}
		\pr\brk{\cO(\G(n,m,P))>tCn}+\pr\brk{\cO(\G^*(n,m,P,\sigma))>tCn}&=t^{-3}O(n^{-2}),\\
				\label{eqO1}
		\Erw[\ln Z(\G(n,m,P))|\cO(\G(n,m,P))\leq tCn]&=\Erw[\ln Z(\G(n,m,P))]+o(1)=O(n),\\
		\Erw[\ln Z(\G^*(n,m,P,\sigma))|\cO(\G(n,m,P,\sigma))\leq tCn]&=\Erw[\ln Z(\G(n,m,P,\sigma))]+o(1)=O(n).\label{eqO2}
		\end{align}	
\end{lemma}
\begin{proof}
The bound (\ref{eqBounded}) guarantees that 	
	$\pr\brk{\max_\tau|\ln\PSI(\tau)|\geq(tn)^{3/8}}\leq t^{-3}O(n^{-3}).$
As a consequence, the probability that either $\G(n,m,P)$ or  $\G^*(n,m,P,\sigma)$ contains a constraint node $a_i$ such
that $\max_\tau|\ln\psi_{a_i}(\tau)|\geq(tn)^{3/8}$ is bounded by $t^{-3}O(n^{-2})$.
Therefore, it suffices to prove (\ref{eqO}) given  $\cA=\{\max_\tau|\ln\psi_{a_i}(\tau)|<(tn)^{3/8}\}$.
Due to (\ref{eqBounded}) the conditional expectation $\Erw[\max_\tau|\ln\PSI(\tau)|\,\big|\,\max_\tau|\ln\PSI(\tau)|<(tn)^{3/8}]$ is bounded.
Thus, the definition of the random factor graph models guarantees that uniformly for all $\sigma,m\leq Dn$,
	\begin{equation}\label{eqO3}
	\Erw[\cO(\G(n,m,P))\,|\,\cA]+\Erw[\cO(\G^*(n,m,P,\sigma))\,|\,\cA]
		=O(n).
	\end{equation}
Further, because the constraint nodes are chosen independently, Azuma's inequality implies that for any $s>1$,
	\begin{align}\label{eqO4}
	\pr\brk{\cO(\G(n,m,P))>\Erw[\cO(\G(n,m,P))\,|\,\cA]+s\big|\cA}
		&\leq2\exp\bc{-\frac{s^2}{O(t^{3/4}n^{7/4})}}.
	\end{align}
Thus, (\ref{eqO}) follows from (\ref{eqO3}) and (\ref{eqO4}) applied to $s=tCn-\Erw[\cO(\G(n,m,P))\,|\,\cA]$ with $C>0$ chosen large enough.
Finally, let either $\G'=\G(n,m,P)$ or $\G'=\G^*(n,m,P,\sigma)$.
Since $\ln Z(\G')\leq\sqrt{m\cO(\G')}$ by Cauchy-Schwarz, (\ref{eqO}) yields
	\begin{align*}
	\Erw[\vecone\{\cO(\G')>Cn\}\ln Z(\G')]&\leq\sqrt m\Erw\brk{\vecone\{\cO(\G')>Cn\}\sqrt{\cO(\G')}}
		\leq O(\sqrt m/n)=o(1),
	\end{align*}
whence (\ref{eqO1}) and (\ref{eqO2}) are immediate.
\end{proof}

\begin{lemma}\label{Lemma_Azuma}
Suppose that $P$ satisfies {\bf SYM} and {\bf BAL} and let $D>0$.
There exists $C=C(D,P)>0$ such that for any $\eps>0$ and $C'>C$ there exists $\delta>0$ such that for all $\sigma\in\Omega^{V_n}$, $m\leq Dn/k$ we have
	\begin{align*}
	\pr\brk{\abs{\ln Z(\G(n,m,P))-\Erw[\ln Z(\G(n,m,P))]}>\eps n|\cO(\G(n,m,P))\leq C'n}&\leq 2\exp\bc{-\delta n},\\
	\pr\brk{\abs{\ln Z(\G^*(n,m,P,\sigma))-\Erw[\ln Z(\G^*(n,m,P,\sigma))]}>\eps n|\cO(\G^*(n,m,P,\sigma))\leq C'n}&\leq 2\exp\bc{-\delta n}.
	\end{align*}
\end{lemma}
\begin{proof}
Let either $\G'=\G(n,m,P)$ or $\G'=\G(n,m,P,\sigma)$ and choose $c=c(\eps,C')>0$ big enough so that the following is true: if $\cO(\G')\leq C'n$, then
	\begin{equation}\label{eqO5}
	\sum_{i\in[m]}\max_\tau|\ln\psi_{a_i}(\tau)|\cdot\vecone\{\max_\tau|\ln\psi_{a_i}(\tau)|>c\}<\eps n/4.
	\end{equation}
Let $\G''$ be the factor graph obtained from $\G'$ by deleting all constraint nodes $a_i$ such that  $\max_\tau|\ln\psi_{a_i}(\tau)|>c$.
Then (\ref{eqO5}) ensures that $|\ln Z(\G')-\ln Z(\G'')|\leq\eps n/4$.
Furthermore, if $\G'''$ is obtained from $\G''$ by changing the neighborhood of some constraint node $a$ and/or its weight function,
subject merely to the condition that the new weight function $\psi$ satisfies $\max_\tau|\ln\psi_{a_i}(\tau)|\leq c$, then $|\ln Z(\G''')-\ln Z(\G'')|\leq c$.
Therefore, Azuma's inequality implies that for any $t>0$,
	\begin{equation}\label{eqO6}
	\pr\brk{|\ln Z(\G'')-\Erw\ln Z(\G'')|>t}\leq2\exp(-t^2/(2c^2m)).
	\end{equation}
Combining (\ref{eqO5}) and (\ref{eqO6}) with (\ref{eqO1}) and (\ref{eqO2}) completes the proof.
\end{proof}

\subsection{Proof of \Thm~\ref{Thm_plantedFreeEnergy}}
We recall from \Sec~\ref{Sec_outline_cycles} that $\fC_r$ is the partition of $\Psi$ obtained by chopping  $[0,2]^{\Omega^k}$ into sub-cubes with side lengths $2/r$.
Since $\fC_r$ is finite the distribution $P_r$ of $\PSI^{(r)}$ is
supported on a finite set $\Psi_r$ of weight functions $\Omega^k\to(0,2)$.

\begin{lemma}\label{Lemma_BetheConv}
For any $\alpha>0$, $D>0$ there is $r_0>0$ such that for all $d\leq D$ and all $r>r_0$ we have
	$$\sup_{\pi\in\cPcent(\Omega)}|\cB(d,P,\pi)-\cB(d,P_r,\pi)|<\alpha.$$
\end{lemma}
\begin{proof}
Let
	\begin{align*}
	B:(\psi_1,\ldots,\psi_\gamma,\rho_1,\ldots,\rho_{k\gamma})\in\Psi^\gamma\times\cP(\Omega)^\gamma\mapsto
		\frac{1}{q\xi^{\gamma}}
			\Lambda\bc{\sum_{\sigma\in\Omega}\prod_{i=1}^{\gamma}\sum_{\tau\in\Omega^k}\vecone\{\tau_{k}=\sigma\}\psi_i(\tau)\prod_{j<k}\rho_{k(i-1)+j}(\tau_j)}.
	\end{align*}
Analogously, for a fixed $r$ let
	\begin{align*}
	B_{r}:(\psi_1,\ldots,\psi_\gamma,\rho_1,\ldots,\rho_{k\gamma})\mapsto
		\frac{1}{q\xi^{\gamma}}
			\Lambda\bc{\sum_{\sigma\in\Omega}\prod_{i=1}^{\gamma}\sum_{\tau\in\Omega^k}\vecone\{\tau_{k}=\sigma\}
				\psi_i^{(r)}(\tau)\prod_{j<k}\rho_{k(i-1)+j}(\tau_j)}.
	\end{align*}
That is, we approximate $\psi_i$ by the average $\psi_i^{(r)}$ over the weight functions in the sub-cube that $\psi_i$ belongs to.
Since $\Lambda$ is continuous on $[0,\infty)$ and therefore uniformly continuous on any compact subset of $[0,\infty)$, $B_{r}\to B$ uniformly as $r\to\infty$
on the entire space $\Psi^r\times\cP(\Omega)^{k\gamma}$ for every $\gamma$.
Since the Poisson distribution has sub-exponential tails, this implies the desired convergence for the first term on the right hand side of (\ref{eqMyBethe}).
A similar argument applies to the second term.
\end{proof}

\begin{lemma}\label{Lemma_POSConv}
The distribution $P_r$ satisfies {\bf SYM} and {\bf BAL}.
Moreover, for any $\alpha>0$, $d>0$ there is $r>0$ such that the following is true for all $\pi,\pi'\in\cP_*^2(\Omega)$.
With $\MU_1,\MU_2,\ldots$ chosen from $\pi$,
			$\MU_1',\MU_2',\ldots$ chosen from $\pi'$ and $\PSI'\in\Psi$ chosen from $P_r$, all mutually independent,  we have
\begin{align}\label{eqWeakPOS}
\Erw\left[\Lambda\left(\sum_{\tau\in\Omega^k}\PSI'(\tau)\prod_{i=1}^ k\MU_i(\tau_i)\right)+(k-1)\Lambda\left(\sum_{\tau\in\Omega^k}\PSI'(\tau)\prod_{i=1}^k \MU_i'(\tau_i)\right)
	-\sum_{h=1}^k\Lambda\left(\sum_{\tau\in\Omega^k}\PSI'(\tau)\MU_h(\tau_h)
	\hspace{-2mm}\prod_{i\in[k]\setminus\{h\}}\hspace{-2mm}\MU_i'(\tau_i)\right)\right]\ge-\alpha.
\end{align}
\end{lemma}
\begin{proof}
The fact that {\bf SYM} and {\bf BAL} are satisfied is immediate from the fact that $P_r$ is a conditional expectation of $P$.
To prove (\ref{eqWeakPOS}) we  observe that
by the uniform continuity of $\Lambda$ on compact subsets of $[0,\infty)$, we can choose $r>0$ large enough so that for all $\psi\in\Psi$,
$\mu_1,\mu_1',\ldots,\mu_k,\mu_k'\in\cP(\Omega)$, 
	\begin{align*}
	\abs{\Lambda\left(\sum_{\tau\in\Omega^k}\psi^{(r)}(\tau)\prod_{i=1}^ k\mu_i(\tau_i)\right)-
		\Lambda\left(\sum_{\tau\in\Omega^k}\psi(\tau)\prod_{i=1}^ k\mu_i(\tau_i)\right)}&<\alpha/3,\\
	\abs{\Lambda\left(\sum_{\tau\in\Omega^k}\psi^{(r)}(\tau)\prod_{i=1}^k \mu_i'(\tau_i)\right)-
		\Lambda\left(\sum_{\tau\in\Omega^k}\psi(\tau)\prod_{i=1}^k \mu_i'(\tau_i)\right)}&<\alpha/3,\\
	\abs{\Lambda\left(\sum_{\tau\in\Omega^k}\psi^{(r)}(\tau)\mu_h(\tau_h)\prod_{i\in[k]\setminus\{h\}}\mu_i'(\tau_i)\right)
		-\Lambda\left(\sum_{\tau\in\Omega^k}\psi(\tau)\mu_h(\tau_h)\prod_{i\in[k]\setminus\{h\}}\mu_i'(\tau_i)\right)}&<\alpha/3.
	\end{align*}
Thus, (\ref{eqWeakPOS}) follows from the triangle inequality and the fact that $P$ satisfies {\bf POS}.
\end{proof}

\begin{lemma}\label{Lemma_EnergyConv}
For any $\alpha>0$, $d>0$ there is $r_0>0$ such that uniformly for all $r\geq r_0$ we have
	$$|\Erw[\ln Z(\hat\G(n,\vm,P))]-\Erw[\ln Z(\hat\G(n,\vm,P_r))]|<(\alpha+o(1))n.$$
\end{lemma}
\begin{proof}
By \Lem~\ref{Prop_contig} the models $\hat\G(n,\vm,P)$ and $\G^*(n,\vm,P,\SIGMA^*)$ are mutually contiguous.
Hence, \Lem~\ref{Lemma_Azuma} implies that $\Erw[\ln Z(\hat\G(n,\vm,P))]=\Erw[\ln Z(\G^*(n,\vm,P,\SIGMA^*))]+o(n)$.
Similarly, since $P_r$ satisfies {\bf SYM} and {\bf BAL} by \Lem~\ref{Lemma_POSConv},
another application of \Lem s~\ref{Prop_contig} and~\ref{Lemma_Azuma} yields
 $\Erw[\ln Z(\hat\G(n,\vm,P_r))]=\Erw[\ln Z(\G^*(n,\vm,P_r,\SIGMA^*))]+o(n)$.
Therefore, it suffices to prove that for any $\alpha>0$ for all sufficiently large $r$ we have
	\begin{equation}\label{eqLemma_EnergyConv_1}
	\max_{\sigma\in\Omega^{V_n}}\abs{\Erw[\ln Z(\G^*(n,\vm,P_r,\sigma))]-\Erw[\ln Z(\G^*(n,\vm,P,\sigma))]}\leq(\alpha+o(1))n.
	\end{equation}
In fact, since the Poisson variable $\vm$ has sub-exponential tails, (\ref{eqQuenchedAnnealed}) shows that (\ref{eqLemma_EnergyConv_1})
would follow if we could show that
	\begin{equation}\label{eqLemma_EnergyConv_2}
	\max_{\sigma\in\Omega^{V_n},m\leq 2dn/k}\abs{\Erw[\ln Z(\G^*(n,m,P_r,\sigma))]-\Erw[\ln Z(\G^*(n,m,P,\sigma))]}\leq(\alpha+o(1))n.
	\end{equation}

To prove (\ref{eqLemma_EnergyConv_2}) pick $\beta=\beta(\alpha,d,P)>0$ small enough and then $r=r(\beta)>0$ large enough.
Fix any $\sigma\in\Omega^{V_n}$ and $m\leq2dn/k$.
We couple two factor graphs $\G',\G''$ such that $\G'$ has distribution $\G^*(n,m,P,\sigma)$ and $\G''$ is distributed as $\G^*(n,m,P_r,\sigma)$ as follows.
First choose $\G'=\G^*(n,m,P,\sigma)$.
Let us write $\psi_{a_1},\ldots,\psi_{a_m}$ for the weight functions of $\G'$.
Then let $\G''$ be the factor graph where each constraint node $a_i$ is adjacent to the same variable nodes as in $\G'$ but where the corresponding weight function
is $\psi_{a_i}^{(r)}$.
It is immediate from (\ref{eqTeacher}) that $\G''$ is distributed as $\G^*(n,m,P_r,\sigma)$.

To bound $\Erw[\ln(Z(\G'')/Z(\G')]$ we observe that
	\begin{align}\nonumber
	\Erw\abs{\ln\frac{Z(\G'')}{Z(\G')}}&=\Erw\abs{\ln\sum_{\tau\in\Omega^{V_n}}
			\frac{\psi_{\G''}(\tau)}{\psi_{\G'}(\tau)}\cdot\frac{\psi_{\G'}(\tau)}{Z(\G')}}
	=\Erw\abs{\ln\bck{\prod_{i=1}^m\frac{\psi_{a_i}^{(r)}(\mathbold{\tau}(\partial_1a_i),\ldots,\mathbold{\tau}(\partial_k(a_i)))}
		{\psi_{a_i}(\mathbold{\tau}(\partial_1a_i),\ldots,\mathbold{\tau}(\partial_k(a_i)))}}_{\G'}}\\
	&\hspace{-5mm}\leq\Erw\max_{\tau\in\Omega^{V_n}}\sum_{i=1}^m
		\abs{\ln\frac{\psi_{a_i}^{(r)}(\tau(\partial_1a_i),\ldots,\tau(\partial_k(a_i)))}
		{\psi_{a_i}(\tau(\partial_1a_i),\ldots,\tau(\partial_k(a_i)))}}
		\leq\Erw\sum_{i=1}^m\max_{\tau\in\Omega^k}
		\abs{\ln\frac{\psi_{a_i}^{(r)}(\tau)}{\psi_{a_i}(\tau)}}
		\leq dn\cdot\Erw\brk{\max_{\tau\in\Omega^k}\abs{\ln\frac{\psi_{a_1}^{(r)}(\tau)}{\psi_{a_1}(\tau)}}}.
			\label{eqLemma_EnergyConv_3}
	\end{align}
Since the function $x\mapsto\ln^2x$ is strictly convex on $(0,2)$ for small $\beta$ and large $r$ we obtain from (\ref{eqTeacher}), the tail bound (\ref{eqBounded}) and Jensen's inequality that 
	\begin{align}\label{eqLemma_EnergyConv_4}
	\Erw\brk{\bc{\max_{\tau\in\Omega^k}\abs{\ln\psi_{a_1}(\tau)}+\max_{\tau\in\Omega^k}\abs{\ln\psi_{a_1}^{(r)}(\tau)}}
		\bc{\vecone\cbc{\max_{\tau\in\Omega^k}\abs{\ln\psi_{a_1}(\tau)}>\beta^{-1}}+
		\vecone\cbc{\max_{\tau\in\Omega^k}\abs{\ln\psi_{a_1}^{(r)}(\tau)}>\beta^{-1}}}}
		&<\frac{\alpha}{2d}.
	\end{align}
On the other hand,
since the map $z\in[\eul^{-1/\beta},2]\mapsto\ln z$ is uniformly continuous, we can choose a sufficiently large $r=r(\beta)$ such that
$\max_\tau|\ln(\psi_{a_1}^{(r)}(\tau)/\psi_{a_1}(\tau))|<\alpha/(2d)$
whenever $\max_{\tau\in\Omega^k}|\ln\psi_{a_1}(\tau)|,\max_{\tau\in\Omega^k}|\ln\psi_{a_1}^{(r)}(\tau)|\leq1/\beta$.
Thus, (\ref{eqLemma_EnergyConv_2}) follows from (\ref{eqLemma_EnergyConv_3}) and (\ref{eqLemma_EnergyConv_4}).
\end{proof}

\begin{proof}[Proof of \Thm~\ref{Thm_plantedFreeEnergy}]
Fix $d>0$.
Since \Lem~\ref{Lemma_POSConv} shows that $P_r$ satisfies {\bf SYM} and {\bf BAL},  \cite[\Prop~3.6]{CKPZ} implies that
	\begin{equation}\label{eqThm_plantedFreeEnergy_1}
	\limsup_{n\to\infty}n^{-1}\Erw[\ln Z(\hat\G(n,\vm,P_r))]\leq\sup_{\pi\in\cPcent(\Omega)}\cB(d,P_r,\pi).
	\end{equation}
Furthermore, \cite[\Prop~3.7]{CKPZ} implies together with equation (\ref{eqWeakPOS}) from \Lem~\ref{Lemma_POSConv} that
for any $\alpha>0$ there is $r>0$ such that
	\begin{equation}\label{eqThm_plantedFreeEnergy_2}
	\liminf_{n\to\infty}n^{-1}\Erw[\ln Z(\hat\G(n,\vm,P_r))]\geq\sup_{\pi\in\cPcent(\Omega)}\cB(d,P_r,\pi)-\alpha.
	\end{equation}
Combining (\ref{eqThm_plantedFreeEnergy_1}) and (\ref{eqThm_plantedFreeEnergy_2}) with \Lem~\ref{Lemma_BetheConv},
we conclude that for any $\alpha>0$ for all large enough $r$ we have
	\begin{equation*}
	\sup_{\pi\in\cPcent(\Omega)}\cB(d,P,\pi)-\alpha\leq\liminf_{n\to\infty}n^{-1}\Erw[\ln Z(\hat\G(n,\vm,P_r))]\leq
		\limsup_{n\to\infty}n^{-1}\Erw[\ln Z(\hat\G(n,\vm,P_r))]\leq\sup_{\pi\in\cPcent(\Omega)}\cB(d, P,\pi)+\alpha.
	\end{equation*}	
Applying  \Lem~\ref{Lemma_EnergyConv} therefore yields
	\begin{equation}\label{eqThm_plantedFreeEnergy_3}
	\lim_{n\to\infty}n^{-1}\Erw[\ln Z(\hat\G(n,\vm,P))]=\sup_{\pi\in\cPcent(\Omega)}\cB(d,P,\pi).
	\end{equation}
Moreover, since $\G^*(n,\vm,P,\SIGMA^*)$ and $\hat\G(n,\vm,P)$ are mutually contiguous by \Lem~\ref{Prop_contig},
\Lem~\ref{Lemma_Azuma} implies that $\lim_{n\to\infty}n^{-1}\Erw[\ln Z(\hat\G(n,\vm,P))]=\sup_{\pi\in\cPcent(\Omega)}\cB(d,P,\pi)$, too.
Finally, since the probability of the event $\fS$ is bounded away from $0$ by \Prop~\ref{prop:FirstCondOverFirst}, \Lem~\ref{Lemma_Azuma}
shows that
	$$\lim_{n\to\infty}n^{-1}\Erw[\ln Z(\hat\G(n,\vm,P))|\fS]=\lim_{n\to\infty}n^{-1}\Erw[\ln Z(\G^*(n,\vm,P,\SIGMA^*))|\fS]
		=\sup_{\pi\in\cPcent(\Omega)}\cB(d,P,\pi)$$ as well.
\end{proof}

\subsection{Proof of \Thm~\ref{Thm_cond}}
We begin with the observation that $\dc$ is bounded and bounded away from $0$.

\begin{lemma}\label{Lemma_dcbounds}
We have $1/(k-1)\leq\dc<\infty$.
\end{lemma}
\begin{proof}
Fix any $d<1/(k-1)$.
Then for any nearly balanced  $\sigma:V_n\to\Omega$ the expected degree of every variable node of $\G^*(n,\vm,P,\sigma)$ is $d+o(1)<1/(k-1)$.
Therefore, the well-known result on the `giant component' threshold of a random hypergraph (e.g., \cite{Schmidt}) shows
	that with probability $1-o(1)$ the random factor graph
	$\G^*(n,m,P,\sigma)$ consists of connected components of order $O(\ln n)$, all but a bounded number of which are trees.
But assumption {\bf SYM} guarantees that for every tree factor graph with $n$ variable nodes and $m$ constraint nodes the free energy is
precisely equal to $n\ln q+m\ln\xi$, as is easily verified by induction on the size of the tree.
Hence, $n^{-1}\Erw[\ln Z(\G^*(n,m,\vm,P,\sigma))]=\ln q+\frac dk\ln\xi+o(1)$ by \Lem~\ref{Lemma_Azuma}.
Since this formula holds for every nearly balanced assignment $\sigma$, we obtain $n^{-1}\Erw[\ln Z(\G^*(n,\vm,P,\SIGMA^*)]=\ln q+\frac dk\ln\xi+o(1)$.
Hence, \Thm~\ref{Thm_plantedFreeEnergy} shows that $d<\dc$ and thus $\dc\geq1/(k-1)$.

We move on to the upper bound.
Recalling that $\vec m$ has distribution $\Po(dn/k)$ and that the $\vec m$ constraint nodes in the teacher-student model are chosen independently, we obtain
	\begin{align}\label{eqLemma_dcbounds_111}
	\frac{k}{n}\frac{\partial}{\partial d}\Erw[\ln\psi_{\G^*}(\SIGMA^*)]&=
		\frac{k}{n}\frac{\partial}{\partial d}\Erw\brk{\sum_{i=1}^{\vec m}\ln\psi_{a_i}(\SIGMA^*(\partial_1a_i),\ldots,\SIGMA^*(\partial_ka_i))}
		=\Erw\brk{\ln\psi_{a_1}(\SIGMA^*(\partial_1a_1),\ldots,\SIGMA^*(\partial_ka_1))}.
	\end{align}
Further, plugging in the definition (\ref{eqTeacher}) of the teacher-student model, we can write the last term out as
	\begin{align*}
	\Erw\brk{\ln\psi_{a_1}(\SIGMA^*(\partial_1a_1),\ldots,\SIGMA^*(\partial_ka_1))}
		&=\Erw\brk{\frac{\sum_{i_1,\ldots,i_k\in[n]}\Lambda(\PSI(\SIGMA^*(x_{i_1}),\ldots,\SIGMA^*(x_{i_k})))}
			{\sum_{j_1,\ldots,j_k\in[n]}\int_\Psi\varphi(\SIGMA^*(x_{j_1}),\ldots,\SIGMA^*(x_{j_k}))\dd P(\varphi)}}.
	\end{align*}
Since the uniformly random $\SIGMA^*$ is nearly balanced with probability $1-o(1)$ as $n\to\infty$, due to {\bf SYM} and (\ref{eqBounded}) the last expression simplifies to
	\begin{align}\label{eqLemma_dcbounds_112}
	\Erw\brk{\ln\psi_{a_1}(\SIGMA^*(\partial_1a_1),\ldots,\SIGMA^*(\partial_ka_1))}
		&=o(1)+\frac{1}{\xi n^k}\sum_{i_1,\ldots,i_k\in[n]}\Erw\brk{\Lambda(\PSI(\SIGMA^*(x_{i_1}),\ldots,\SIGMA^*(x_{i_k})))}.
	\end{align}
Further, due to the third part of (\ref{eqBounded}) and because $\Lambda\bc\nix$ is strictly convex, Jensen's inequality shows that there exists an $n$-independent number $\alpha>0$ such that
	\begin{align}\label{eqLemma_dcbounds_113}
	\sum_{i_1,\ldots,i_k}\frac{\Erw\brk{\Lambda(\PSI(\SIGMA^*(x_{i_1}),\ldots,\SIGMA^*(x_{i_k})))}}{\xi n^k}&\geq
		\alpha+o(1)+\Lambda\bc{\sum_{i_1,\ldots,i_k}\frac{\Erw\brk{\PSI(\SIGMA^*(x_{i_1}),\ldots,\SIGMA^*(x_{i_k}))}}{\xi n^{k}}}=
			\alpha+\ln\xi+o(1).
	\end{align}
Combining (\ref{eqLemma_dcbounds_111})--(\ref{eqLemma_dcbounds_113}),  we find
	$\frac{\partial}{\partial d}\frac{1}{n}\Erw[\ln\psi_{\G^*}(\SIGMA^*)]\geq k^{-1}(\alpha+\ln\xi)+o(1)$.
Hence, for $d>\frac{k}\alpha\ln q$ we obtain
	\begin{align*}
	\frac1n\Erw[\ln Z(\G^*)]&\geq\frac1n\Erw[\ln\psi_{\G^*}(\SIGMA^*)]
		\geq \frac dk\bc{\alpha+\ln\xi}+o(1)>\ln q+\frac dk\ln\xi+\Omega(1).
	\end{align*}
Hence, applying \Thm~\ref{Thm_plantedFreeEnergy} and recalling  (\ref{eq:dcond}), we conclude that $\dc\leq \frac{k}\alpha\ln q<\infty$.
\end{proof}

\noindent
We derive \Thm~\ref{Thm_cond} from \Thm~\ref{Thm_plantedFreeEnergy} in two steps.
First, generalizing the argument from \cite[\Sec~3.5]{CKPZ} to the setting of infinite $\Psi$, we prove the free energy formula for $d\leq\dc$.

\begin{proof}[Proof of \Thm~\ref{Thm_cond}, part 1.]
First assume that $d<\dc$ is such that for some $\delta>0$,
	\begin{align*}
	\liminf_{n\to\infty}n^{-1}\Erw\ln Z(\G(n,\vm,P))<\ln q+\frac dk\ln\xi-3\delta.
	\end{align*}
Then there exists a sequence $m\in\cM(d)$ such that
	\begin{align*}
	\liminf_{n\to\infty}n^{-1}\Erw\ln Z(\G(n,m,P))<\ln q+\frac dk\ln\xi-2\delta.
	\end{align*}
Hence, \Lem~\ref{Lemma_Azuma} shows that for a suitably large $C>0$ and a sufficiently small $\eps>0$,
	\begin{align}\label{eq_Thm_cond_part1_1}
	\liminf_{n\to\infty}n^{-1}\ln\pr\brk{n^{-1}\ln Z(\G(n,m,P))\geq\ln q+\frac dk\ln\xi-\delta,\,\cO(\G(n,m,P))\leq Cn}&\leq-\eps.
	\end{align}
Now, with $\theta=\theta(\delta,\eps)>0$ chosen small enough, we define 
	\begin{align}\label{eq_Thm_cond_part1_2}
	Z'(G)&=Z(G)\vecone\{n^{-1}\ln Z(G)\leq\ln q+\frac dk\ln\xi+\theta,\,\cO(G)\leq Cn\}.
	\end{align}
\Thm~\ref{Thm_plantedFreeEnergy} and \Lem~\ref{Lemma_Azuma} yield
	$\pr\brk{\ln Z(\hat\G(n,m,P))\leq\ln q+\frac dk\ln\xi+\theta,\,\cO(\hat\G(n,m,P))\leq Cn}=1-o(1)$ because $d<\dc$.
Therefore, (\ref{eqnotosmm3}) and (\ref{eq:NishimoriG})  yield
	\begin{align}\nonumber
	\Erw[Z'(\G(n,m,P))]&=\Erw[Z(\G(n,m,P))]\pr\brk{n^{-1}\ln Z(\hat\G(n,m,P))\leq\ln q+\frac dk\ln\xi+\theta,\,\cO(\hat\G(n,m,P))\leq Cn}\\
		&=\exp(n(\ln q+\frac dk\ln\xi+o(1))).\label{eq_Thm_cond_part1_3}
	\end{align}
Moreover, the definition \eqref{eq_Thm_cond_part1_2} of $Z'(\G(n,m,P))$ guarantees that
	\begin{align}\label{eq_Thm_cond_part1_4}
	\Erw[Z'(\G(n,m,P))^2]&\leq\exp(2n(\ln q+\frac dk\ln\xi+\theta)).
	\end{align}
But combining (\ref{eq_Thm_cond_part1_3}) and (\ref{eq_Thm_cond_part1_4}) with the Paley-Zygmund inequality, we obtain
	\begin{align*}
	\pr\brk{n^{-1}\ln Z(\G(n,m,P))\geq \ln q+\frac dk\ln\xi-\theta}&\geq\pr\brk{Z'(\G(n,m,P))\geq\exp(n(\ln q+\frac dk\ln\xi-\theta))}\\
		&\geq\frac{\Erw[Z'(\G(n,m,P))]^2}{2\Erw[Z'(\G(n,m,P))^2]}=\exp(-2n(\theta+o(1))),
	\end{align*}
which contradicts (\ref{eq_Thm_cond_part1_1}) if $\theta$ is chosen sufficiently small.
Finally, since the probability of the event $\fS$ is bounded away from $0$ by \Prop~\ref{prop:FirstCondOverFirst},
the assertion about $\Erw[\ln Z(\hat\G(n,m,P))|\fS]$ follows from \Lem~\ref{Lemma_Azuma}.
\end{proof}

\noindent
We proceed to show that $\limsup_{n\to\infty}\frac1n\Erw[\ln Z(\G)]<\ln q+\frac dk\ln\xi$ if $d>\dc$ by
generalizing the argument from \cite[\Sec~3.5]{CKPZ} to infinite sets $\Psi$.

\begin{lemma}\label{Lemma_501}
Assume that $d>0$ is such that  $\sup_{\pi\in\cP_*^2(\Omega)}\cB(d,P,\pi)>\ln q+\frac{d}k\ln\xi+\delta$ for some $\delta>0$.
Then for every large enough $C>0$ there exists $\beta=\beta(C)>0$ such that for large enough $n$,
	\begin{align}\label{eqThm_G_501}
		\pr\brk{n^{-1}\ln Z(\G^*(n,\vec m,P,\SIGMA^*))
			\leq\ln q+\frac{d}k\ln\xi+\delta/2\bigg|\cO(\G^*(n,\vec m,P,\SIGMA^*))\leq Cn}\leq\exp(-\beta n).
	\end{align}
\end{lemma}
\begin{proof}
If $\sup_{\pi\in\cP_*^2(\Omega)}\cB(d,P,\pi)>\ln q+\frac{d}k\ln\xi+\delta$, then \Thm~\ref{Thm_plantedFreeEnergy} shows that
	\begin{align}\label{eqLemma_501_1}
	n^{-1}\Erw[\ln Z(\G^*(n,\vm,P,\SIGMA^*))]=o(1)+\sup_{\pi\in\cP_*^2(\Omega)}\cB(d,P,\pi)>\ln q+\frac{d}k\ln\xi+\delta+o(1).
	\end{align}
Fix a small enough $\alpha=\alpha(d,\delta)>0$ and an even smaller $\eta=\eta(\alpha)>0$ and
	let $\cS_\eta=\cbc{\sigma\in\Omega^{V_n}:\TV{\rho_\sigma-\bar\rho}\leq\eta}$.
Since $\SIGMA^*\in\Omega^{V_n}$ is chosen uniformly and thus $\pr[\SIGMA^*\in\cS_\eta]=1-\exp(-\Omega(n))$
while for large enough $C$ we have $\pr\brk{\cO(\G^*(n,\vec m,P,\sigma))\leq Cn}=1-o(1)$ by \Lem~\ref{Lemma_Azuma}, it suffices to prove that for all $\sigma\in\cS_\eta$,
	\begin{align}\label{eqThm_G_501_2}
		\pr\brk{n^{-1}\ln Z(\G^*(n,\vec m,P,\sigma))
			\leq\ln q+\frac{d}k\ln\xi+\delta/2\bigg|\cO(\G^*(n,\vec m,P,\sigma))\leq Cn}\leq\exp(-\beta n).
	\end{align}

To establish (\ref{eqThm_G_501_2}) we set up a coupling of $\G'=\G^*(n,\vec m,P,\sigma)$, $\G''=\G^*(n,\vec m,P,\tau)$ for any $\sigma,\tau\in\cS_\eta$.
Let us write $a_j'$ for the constraint nodes of $\G'$ and $a_j''$ for those of $\G''$.
Relabeling the variable node as necessary, we may assume without loss that $|\sigma\triangle\tau|\leq2\eta n$.
Therefore, (\ref{eqTeacher}) shows that we can couple the distribution of the neighborhoods $\partial a_j'$, $\partial a_j''$ such that, with $\eta>0$ chosen small enough,
	\begin{equation}\label{eq_Lemma_501_1}
	\pr[\partial a_j'=\partial a_j'', \partial a_j'\cap(\sigma\triangle\tau)=\emptyset]\geq1-\alpha.
	\end{equation}
Furthermore, if indeed $\partial a_j'=\partial a_j''$ and $\partial a_j'\cap(\sigma\triangle\tau)=\emptyset$,
	then by (\ref{eqTeacher}) the weight functions $\psi_{a_j'},\psi_{a_j''}$ are identically distributed and we couple such that $\psi_{a_j'}=\psi_{a_j''}$.
If, on the other hand, $\partial a_j'\neq\partial a_j''$ or $(\partial a_j'\cup\partial a_j'')\cap(\sigma\triangle\tau)\neq\emptyset$, then we choose
$\psi_{a_j'}$, $\psi_{a_j''}$ independently according to (\ref{eqTeacher}).

Since the $\vm$  constraint nodes are chosen independently, (\ref{eq_Lemma_501_1}) shows that the number  $X$ of $j\in[\vm]$ such that either $\partial a_j'\neq\partial a_j''$ or $\psi_{a_j'}\neq\psi_{a_j''}$ is binomially distributed with mean at most $\alpha n$.
Hence, $\pr\brk{X>2\alpha n}\leq\exp(-\Omega(n))$.
Furthermore, (\ref{eqBounded}) shows that the expected impact on the free energy of the $X$ constraint nodes where $\G',\G''$ differ is bounded by $cX$ for
some number $c=c(P)>0$ that does not depend on $\alpha$ or $\sigma$.
Therefore, choosing $\alpha>0$ small enough we can ensure that
	\begin{align}\label{eq_Lemma_501_3}
	\Erw\abs{\ln Z(\G')-\ln Z(\G'')}&\leq\delta n/2.
	\end{align}
Combining (\ref{eqLemma_501_1}) and (\ref{eq_Lemma_501_3}), we obtain
	\begin{align}\label{eqLemma_501_4}
	n^{-1}\Erw[\ln Z(\G^*(n,\vm,P,\sigma))]>\ln q+\frac{d}k\ln\xi+\delta/2+o(1)\qquad\mbox{for all }\sigma\in\cS_\eta.
	\end{align}
Thus, (\ref{eqThm_G_501_2}) follows from (\ref{eqLemma_501_4}) and \Lem~\ref{Lemma_Azuma}.
\end{proof}

\begin{lemma}\label{Lemma_502}
Assume that $P$ satisfies {\bf SYM} and {\bf BAL}.
For any $D>0$ the following is true uniformly for $m\leq Dn/k$.
If $\cA$ is an event such that $\pr\brk{\G^*(n,m,P,\SIGMA^*)\in\cA}\leq\exp(-\Omega(n))$, then $\pr\brk{\hat\G(n,m,P)\in\cA}\leq\exp(-\Omega(n))$.
\end{lemma}
\begin{proof}
This is immediate from the Nishimori identity \Lem~\ref{lem:nishimori} and (\ref{eqCor_strCntg666}).
\end{proof}

\begin{proof}[Proof of \Thm~\ref{Thm_cond}, part 2.]
Suppose that $d>d_{\mathrm{cond}}$.
Then there exist $d'<d$ and $\delta>0$ such that
	$$\sup_{\pi\in\cP_*^2(\Omega)}\cB(d',P,\pi)>\ln q+\frac{d'}k\ln\xi+\delta.$$
Let $\vm'=\vm_{d'}(n)$ be a $\Po(d'n/k)$-variable and consider the event $\cF=\{n^{-1}\ln Z\leq\ln q+\frac{d'}k\ln\xi+\delta/2\}$.
Then Markov's inequality and \Lem~\ref{Cor_F} yield
	\begin{align}\label{eqThm_G_50}
	\pr\brk{\G(n,\vec m',P)\in\cF}&\leq  o(1)+\sum_{m:|m-d'n/k|\leq n^{2/3}}
		\frac{\pr[\Po(d'n/k)=m]\Erw[Z(\G(n,m,P))]}{q^n\xi^{d'n/k}\exp(\delta n)}=o(1).
	\end{align}
On the other hand, \Lem~\ref{Lemma_501} shows that for large enough $C>0$,
	\begin{align}\label{eqThm_G_501}
		\pr\brk{\G^*(n,\vec m',P,\SIGMA^*)\in\cF,\,\cO(\G^*(n,\vec m',P,\SIGMA^*))\leq Cn}\leq\exp(-\Omega(n)).
	\end{align}
Now, for a factor graph $G$ obtain $G'$ by removing each constraint node with probability $1-d'/d$ independently.
Moreover, let $\cG$ be the set of all factor graphs $G$ such that $\pr[G'\in\cF]\geq1/2$, where, of course, the probability is over the removal process only.
Since the distribution of $\G(n,\vec m,P)'$ is identical to that of $\G(n,\vec m',P)$, (\ref{eqThm_G_50}) yields
	\begin{align}\label{eqThm_G_51}
	\pr\brk{\G(n,\vec m,P)'\in\cG}&=1-o(1).
	\end{align}
Similarly,  $\G^*(n,\vec m,p,\SIGMA^*)'$ and $\G^*(n,\vec m',p,\SIGMA^*)$ are identically distributed. Thus, (\ref{eqThm_G_501}) and
 \Lem~\ref{Lemma_O} imply that
	\begin{align}\label{eqThm_G_52}
	\pr\brk{\G^*(n,\vec m,P,\SIGMA^*)\in\cG,\,\cO(\G^*(n,\vec m,P,\SIGMA^*))\leq Cn}&\leq\exp(-\Omega(n)).
	\end{align}
Furthermore, (\ref{eqThm_G_52}) and \Lem~\ref{Lemma_502} yield $\chi>0$ such that
	\begin{align}\label{eqThm_G_53}
	\pr\brk{\hat\G(n,\vec m,P)\in\cG,\,\cO(\hat\G(n,\vec m,P))\leq Cn}&\leq\exp(-2\chi n).
	\end{align}

To complete the proof, assume for contradiction that 
$\limsup_{n\to\infty}n^{-1}\Erw[\ln Z(\G(n,\vm,P))]\geq\ln q+\frac dk\ln\xi$.
Then $n^{-1}\Erw[\ln Z(\G(n,\vm,P))]\geq\ln q+\frac dk\ln\xi+o(1)$ for arbitrarily large $n$.
Thus, we can apply \Lem~\ref{Lemma_Azuma} to conclude that for infinitely many $n$,
	\begin{align}\label{eqThm_G_2}
	\pr\brk{n^{-1}\ln Z(\G(n,\vm,P))<\ln q+\frac dk\ln\xi-\chi\big|\cO(\G(n,\vm,P))\leq Cn}&\leq\exp(-\Omega(n)).
	\end{align}
Combining (\ref{eqThm_G_2}) with \Lem~\ref{Lemma_O}, we see that the event
	$\cA=\{n^{-1}\ln Z<\ln q+\frac dk\ln\xi-\chi,\, \cO\leq Cn\}$
satisfies $\pr\brk{\G(n,\vm,P)\in\cA}=1-o(1)$ for arbitrarily large $n$.
But then
	\begin{align*}
	1-o(1)&=\pr\brk{\G(n,\vm,P)\in\cA\cap\cG}&&\mbox{[by  (\ref{eqThm_G_51})]}\\
	&\leq o(1)+\sum_{m:|m-dn/k|\leq n^{2/3}}\frac{\exp(\chi n+o(n))}{q^n\xi^{dn/k}}\Erw\brk{\vecone\{\G(n,m,P)\in\cA\cap\cG\}Z(\G(n,m,P)}
				&&\mbox{[by the definition of $\cA$]}\\
	&\leq o(1)+\exp(\chi n+o(1))\pr\brk{\hat\G(n,m,P)\in\cG,\,\cO(\hat\G(n,\vec m,P))\leq Cn}&&
			\mbox{[due to (\ref{eqnotosmm3}) and (\ref{eq:NishimoriG})]}\\
	&=o(1)&&\mbox{[because of (\ref{eqThm_G_53})]},
	\end{align*}
a contradiction that refutes the assumption $\limsup_{n\to\infty}n^{-1}\Erw[\ln Z(\G(n,\vm,P))]\geq\ln q+\frac dk\ln\xi$.
\end{proof}

\section{Reconstruction}\label{sec:thrm:TreeGraphEquivalence}

\noindent
Throughout this section, when there is no danger of confusion we abbreviate $\T (d, P )$ to $\T$ and
 $\T^{h} (d, P )$ to $\T^{h} $. For a rooted factor tree $T$ and any vertex $x$ in that tree, let 
$ \partial_{desc} x$ denote the children of $x$.  Also, for any factor graph $G$, any variable node 
$v$ in this graph  and any integer $\ell\geq 0$, we let $S(v, \ell)$ denote the set of  variable nodes at distance $2\ell$ from  $v$.

Given some graph $G=(V, E)$, any $M \subset V$ and an assignment  $\sigma\in \Omega^V$
let $\sigma(M)$, or $\sigma_M$ denote the assignment that $\sigma$ specifies for the set $M$
Furthermore,  let $\nu, \nu'$ be two distribution on the configuration space 
$\Omega^V$. For any $M \subset V$ we let
\[
||\nu-\nu' ||_{M}
\]
denote the total variation distance between the projections of $\nu$ and $\nu'$ on $M$.
Also, for some $\sigma\in \Omega^{V}$ we let $\nu^{\sigma_{M}}$ denote the distribution
$\nu$ conditional on that $M$ has assignment $\sigma(M)$.

For the factor tree $T$ we define the  {\em broadcasting process} which generates an assignment   
$\mathbold{\sigma}\in \Omega^V_T$ as follows:
There is  some initial  distribution $\zeta\in \cP(\Omega)$. We set $\mathbold{\sigma}(r)$ according to the distribution 
$\zeta$. Then, inductively, 
assume that we have $\mathbold{\sigma}(x)$ for some variable node  $x$.  For each  $\alpha \in \partial_{desc} x$,
independently,   the variables nodes in  $\partial \alpha$  are assigned $\tau\in \Omega^k$  with probability 
proportional to 
\begin{equation}\label{def:BroadCaseProb}
   \vecone\{\tau(j_{\alpha, x})=\mathbold{\sigma}(x) \} \psi_{\alpha}(\tau)
\end{equation}
where $\psi_{\alpha}$ is the weight function  that corresponds to $\alpha$ and  $j_{\alpha, x}$ is the position of  $x$ inside the constraint  $\psi_{\alpha}$.

\begin{lemma}\label{lemma:BroadCastingDistr}
Consider some factor tree $T$ of height $h>0$, rooted at (variable) node $r$.
Let $\mathbold{\sigma}\in \Omega^T$ be the assignment  generated by the
broadcasting process  such that the initial distribution is the uniform over $\Omega$. 

For any $\tau\in \Omega^T$, it holds that
\begin{equation}\nonumber
\Pr[\mathbold{\sigma}=\tau]=  \mu_T(\tau),
\end{equation}
where $\mu_T$ is the  Gibbs distribution specified by $T$.
\end{lemma}
\begin{proof}
Let $\mathbold{\eta}$ be distributed as in $\mu_T$. Then, we have that $\mathbold{\eta}(r)$ is distributed uniformly at random in $\Omega$.

Furthermore,  let $x\in T$ be a variable node.  Given $\mathbold{\eta}(x)$  for each 
$\alpha \in \partial_{desc} x$ the assignment  $\mathbold{\eta}(\partial \alpha)$ is independent
of the other vertices in $\partial_{desc}x$. Furthermore, for each assignment  $\tau\in \Omega^k$ 
we have $\mathbold{\eta}(\partial \alpha)=\tau$ with probability 
proportional to 
\begin{equation}\nonumber
   \vecone\{\tau(j_{\alpha, x})=\mathbold{\eta}(x) \} \psi_{\alpha}(\tau).
\end{equation}
The lemma follows by using the definition of the broadcasting process.
\end{proof}

\noindent
Consider  a sequence of factor trees $\mathcal{T}=\{T_{\ell} \}_{\ell\geq 0}$, where  
$T_h$ contains  $h$ levels of variable nodes.  Let
\begin{align*}
	\mathrm{corr}_{\mathcal{ T}}=  \lim_{\ell\to\infty} 
	\sum_{\tau\in \Omega^{S(r, 2\ell)}} \mu_{T_\ell}(\tau) \ 	||\mu^{\tau}_{T_{\ell}}-\mu ||_{\{r\}},
\end{align*}
recall that $S(r, 2\ell)$ is the set of variable nodes at distance $2\ell$ from the root $r$.
Similarly,  we define 
\begin{align*}
	\mathrm{broad}_{\mathcal{ T}}=  \lim_{\ell\to\infty}  \max_{c,c'\in \Omega^{\{r\}}}
	||\mu^c_{T_{\ell}}- \mu^{c'}_{T_{\ell}} ||_{S(r, 2\ell)}.
\end{align*}

\noindent
We study the reconstruction problem on the sequence of  factor tree $\mathcal{T}$  by means of the
 broadcasting processes  and the  quantity $\mathrm{broad}_{\mathcal{T}}$.
To be more specific,  for each $T_{\ell }\in \mathcal{T}$, rooted at $r_{\ell}$, consider two  broadcasting processes
 with some initial distribution $\zeta$ and  let $\mathbold{\sigma}_{\ell}$  and $\mathbold{\tau}_{\ell}$ be the 
 assignment s that are  generated, respectively.
 Then, the quantity $\mathrm{broad}_{\mathcal{ T}}$ expresses the $\ell_1$-distance between the distributions of the
 configurations $\mathbold{\sigma}_{\ell}(S(r_{\ell}, \ell ))$ and $\mathbold{\tau}_{\ell}(S(r_{\ell}, \ell))$, as $\ell\to \infty$, 
 conditional that $\mathbold{\sigma}_{\ell}(r_{\ell})=c$,  $\mathbold{\tau}_{\ell}(r_{\ell})=c'$,  for worst-case pair $c, c'\in \Omega$.
The following result implies that for studying reconstruction  on $\mathcal{T}$ we
can either consider  $\mathrm{broad}_{\mathcal{ T}}$,  or  $\mathrm{corr}_{\mathcal{ T}}$.

\begin{lemma}\label{lemma:CorrVsBroad}
Let $\mathcal{T}=\{T_{\ell} \}_{\ell\geq 0}$ be a sequence of factor trees,
where   $T_\ell$ contains  $\ell$ levels of variable nodes. 
Then we have that
$\mathrm{broad}_{\mathcal{ T}}=0$ if and only if $\mathrm{corr}_{\mathcal{ T}}=0$.
\end{lemma}

\begin{proof}
For some integer  $\ell>0$,  we have that
\begin{eqnarray}
	||\mu^c_{T_{\ell}}- \mu_{T_{\ell}} ||_{S(r_\ell, \ell)}
&=&\abs{ \sum_{\tau\in \Omega^{S(r_{\ell}, \ell)}}\bck{ \vecone\{ \SIGMA(S(r, \ell))=\tau\}| \SIGMA(r)=c}_{T_\ell } - 
	\bck{ \vecone\{ \SIGMA(S(r_{\ell}, \ell))=\tau\}}_{T_\ell } } \nonumber \\
&=&  q
\sum_{\tau\in \Omega^{S(r_{\ell}, \ell)}}\bck{\vecone\{ \SIGMA(S(r_{\ell}, \ell))=\tau\}} 
\abs{ \bck{ \vecone\{ \SIGMA(r_{\ell})=c \} | \SIGMA(S(r_{\ell}, \ell))=\tau}_{T_\ell } - 
	\bck{ \vecone\{ \SIGMA(r)=c}_{T_\ell } } \nonumber \\	
&=& q
\sum_{\tau\in \Omega^{S(r_{\ell}, \ell)}}\bck{\vecone\{ \SIGMA(S(r_{\ell}, \ell))=\tau\}} 
\abs{ \bck{ \vecone\{ \SIGMA(r_{\ell})=c \} | \SIGMA(S(r_{\ell}, \ell))=\tau}_{T_\ell } - q^{-1}} \nonumber \\		
&\leq &q
\sum_{\tau\in \Omega^{S(r_{\ell}, \ell)}}  \mu_{T_{\ell}}(\tau) \  ||\mu^{\tau}_{T_{\ell}}-\mu_{T_\ell} ||_{\{r_{\ell}\}}.
\end{eqnarray}
Clearly, the above implies that $\mathrm{broad}_{\mathcal{ T}}\leq q \ \mathrm{corr}_{\mathcal{ T}}$.
In turn, we get that if $\mathrm{corr}_{\mathcal{ T}}=0$, then 
 $ \mathrm{broad}_{\mathcal{ T}}=0$, as well.

We work in a similar way for the other direction. That is, 
\begin{eqnarray}
\sum_{\tau\in \Omega^{S(r_{\ell}, \ell)}}  \mu_{T_{\ell}}(\tau) \  ||\mu^{\tau}_{T_{\ell}}-\mu_{T_\ell} ||_{\{r_{\ell}\}} &=&
	\sum_{\tau\in \Omega^{S(r_{\ell}, \ell)}} \bck{\vecone\{\mathbold{\sigma}(S(r_{\ell}, \ell))=\tau\}}_{T_{\ell}}
		\sum_{s\in \Omega} \abs{ \bck{\vecone\{\mathbold{\sigma}(r_\ell)=s\} | \mathbold{\sigma}(S(r_{\ell}, \ell))=\tau}_{T_{\ell }}-q^{-1} }
\nonumber \\
&=& 
	\sum_{\tau\in \Omega^{S(r_\ell, \ell)}}  
		\sum_{s\in \Omega} \abs{ \bck{\vecone\{\mathbold{\sigma}(r_\ell)=s,\   \mathbold{\sigma}(S(r_\ell, \ell))=\tau\}}_{T_{\ell }}-
		\bck{\vecone\{\mathbold{\sigma}(r_\ell)=s\} }_{T_{\ell }}\bck{\vecone\{\mathbold{\sigma}(S(r_\ell, \ell))=\tau\}}_{T_{\ell}} } \nonumber \\
&=& 
		\sum_{s\in \Omega} 
		\bck{\vecone\{\mathbold{\sigma}(r_\ell)=s\}}_{T_{\ell }}
			\sum_{\tau\in \Omega^{S(r_\ell, \ell)}}  \abs{ \bck{\vecone\{\mathbold{\sigma}(S(r_\ell, \ell))=\tau\} | \mathbold{\sigma}(r_\ell)=s}_{T_{\ell }}-\bck{\vecone\{\mathbold{\sigma}(S(r_\ell, \ell))=\tau\}}_{T_{\ell}} } \nonumber \\
&\leq & 
		2 \max_{c,c'\in \Omega^{\{r_\ell\}}} ||\mu^c_{T_{\ell}}- \mu^{c'}_{T_{\ell}} ||_{S(r_\ell, \ell)}.
\nonumber 
\end{eqnarray}
Clearly, the above implies that
$ \mathrm{corr}_{\mathcal{ T}} \leq 2 \ \mathrm{broad}_{\mathcal{ T}}$.
In turn, we get that if  $ \mathrm{broad}_{\mathcal{ T}}=0$, then
$\mathrm{corr}_{\mathcal{ T}}=0$.
\end{proof}

\noindent
In the following result we show  that  that non-reconstruction is  monotone in the expected degree of $\T(d, P)$. 
In particular we show the following result.

\begin{lemma}\label{lemma:ReconMonotone}
For any $d_1, d_2>0$ such that $d_1\geq d_2$,  the following is true: 
If $\mathrm{corr}^\star(d_1)=0$, then  $\mathrm{corr}^\star(d_2)=0$.
\end{lemma}

\noindent
The proof of Lemma \ref{lemma:ReconMonotone} appears in Section \ref{sec:lemma:ReconMonotone}

We proceed by introducing  some further notions.
For a {\em rooted} factor graph $G$,  let ${\tt ISM}(G)$ be the isomorphism class
of rooted factor graphs to which $G$ belongs. 
Let   $\T_{\G, \ell}(v)$ be the  induced subgraph of $\G$ which includes $v$ and all variable nodes  
which are within graph distance $2\ell$ from $v$.
For $h=o(\log n)$,  $\T_{G, h}(v)$ is a tree with probability $1-o(1)$.  
In particular, there is a coupling $\rho$ of the distribution induced by
$\T_{G, h }(v)$ and $\T^{h}$ such that the following is true:
 \begin{equation}\label{eq:FactorGWTreeVsLocalBall}
\lim_{n\to \infty} \Erw_{\rho} \left[ \vecone \{ { \tt ISM}(\T_{\G, h}(v)) \neq  {\tt ISM}( \T^{h}) \} \right ]=0 \qquad
\textrm{and}\qquad
\lim_{n\to \infty} \Erw_{\rho} \left[ \vecone \{ { \tt ISM}(\T_{\G, h}(v)) \neq  {\tt ISM}( \T^{h}) \} \ |\  \fS\right ]=0. 
 \end{equation}

\noindent
For what follows, we let the event  $\mathcal{I}(v,h)=\{ \vecone \{ { \tt ISM}(\T_{\G, h}(v)) = {\tt ISM}( \T^{h}) \}$.

\begin{lemma}\label{lemma:Planting4GoodBoundary}
Let   $h = o(\log n)$. 
Consider  $(\G^*, \mathbold{\sigma}^*)$ generated according to Teacher-Student model
and some vertex $v$.
Also, consider the pair $(\T^h, \mathbold{\tau})$ such that  $\mathbold{\tau}$ is generated by a broadcasting process 
for which we assign the root $r$ the configuration ${\mathbold{\sigma}}(v)$ with probability 1.

There is a  coupling $\tilde{\lambda}$ between  $(\G^*, \mathbold{\sigma}^*)$ and $(\T^h, \mathbold{\tau})$  
such that   the following is true:
\[
\lim_{n\to \infty}
\Erw_{\tilde{\lambda}} \left[ \vecone \{ \mathcal{I}(v,h) \}  
 \sum_{\tau\in \Omega^{\T^h} }
\abs{ 
\pr[\mathbold{\sigma}^*(\T_{\G^*,h}(v))=\tau \ | \ \G^*]-
 < \vecone \{\mathbold{\sigma}=\tau \circ f\} >_{\T^{h}}} 
\right] =0,
\]
where $f$ is an isomorphism between $\T_{\G^*,h}(v)$ and $\T^h$.
The same result holds for $\G^*\in \fS$.
\end{lemma}

\noindent
The proof of Lemma \ref{lemma:Planting4GoodBoundary} appears in Section \ref{sec:lemma:Planting4GoodBoundary}.

In light of Lemma \ref{lemma:Planting4GoodBoundary} and \eqref{eq:FactorGWTreeVsLocalBall}
Theorem \ref{thrm:TreeGraphEquivalence*} is immediate.

The above result implies that in the teacher-student model, the distribution of the 
configuration of $\T_{\G^*,h}(v)$ that is specified by  $\mathbold{\sigma}^*$ is asymptotically 
the same as the distribution of the configuration  that is induced by  the broadcasting process 
on $\T_{\G^*,h}(v)$.  We use the above result with Corollary  \ref{Thm_contig} to relate
reconstruction on random factor graph $\G$ and random tree $\T$.

Now we proceed with the proof of Theorem \ref{thrm:TreeGraphEquivalence}.
In the following lemma we provide  the upper-bound  for $\dr$ and $\dr^{\star}$.

\begin{lemma}\label{lemma:ReconThrsUpperBound}
For any  $\epsilon>0$ there exists   $\dc<d<\dc+\epsilon$ such that    $\mathrm{corr}(d)>0$.
Furthermore, for any $d>\dc$ we have   $\mathrm{corr}^\star(d)$.
\end{lemma}
\begin{proof}
We consider $\mathrm{corr}(d)$. 
For any graph $G$ and two vertices $x, y$ such that $\textrm{dist}(x, y) \geq \ell  $
and any $c \in \Omega^{\{x\}}$, it is easy to see that 
\begin{equation}
||\mu^c_G-\mu_G ||_{\{ y\}}
 \leq  
 ||\mu^c_G-\mu_G ||_{\{ S(x,\ell )\}}
\end{equation}
Furthermore, working as in Lemma \ref{lemma:CorrVsBroad}, we can substitute the r.h.s. of
the above inequality and get 
\begin{equation}
||\mu^c_G-\mu_G ||_{\{ y\}}
\leq  
\sum_{\tau\in \Omega^{S(x, \ell)}}
\mu_G(\tau) \ ||\mu^{\tau}_{G}-\mu_G ||_{x}.
\nonumber
\end{equation}
For any two fixed vertices $x, y$ in  $\G$, we denote by $\mathcal{D}(x, y)$ the event that  $\textrm{dist}(x,y)\geq \ln\ln n$.
Then, for $h=\ln\ln n$ we get that
\begin{eqnarray}
\frac1{n^2}\sum_{x,y\in V_n}\Erw\TV{\mu_{\G,x,y}-\bar\rho} 
&\leq& \frac{q^{-1}}{n^2 }\sum_{x, y\in V_n }\Erw
\sum_{\tau\in \Omega^{S(x, \ell)}}
<\vecone \{ \mathbold{\sigma}(S(x, h)=\tau\} >_{\G}
\sum_{s\in \Omega}\left|  
< \vecone \{\mathbold{\sigma}(x)=s \ | \ \mathbold{\sigma}(S(x, h ))=\tau >_{\G} -
q^{-1} 
\right| \nonumber \\ &&+  \frac{q^{-1}}{n^2 }\sum_{x, y\in V_n }\Erw \ \vecone\{\mathcal{D}^c(x, y) \}. \nonumber 
\end{eqnarray}
Note  that for any two fixed vertices $x, y$ it holds that $\Pr[\mathcal{D}^c(x,y)]\leq n^{-1/2}$.
To see this, let $N_x$ be the number of vertices within distance $\ln\ln n$ from $x$.
Furthermore,  given $N_x$ each vertex  belongs to the $\ln\ln n$ neighborhood of $x$
with probability at most $N_x/n$. Then, noting that  $\Erw[N_x]=o(n^{1/100})$, we get that
\begin{equation}
\Pr[\mathcal{D}^c(x,y)]\leq \Pr[N_{x}>n^{1/3}]+{n^{1/3}}/{n} \ \leq \ n^{-1/2},
\end{equation}
where we use Markov's inequality to bound  $\Pr[N_{x}>n^{1/3}]$.
Combining all the above, we get that for any $d$ it holds that
\begin{equation}
\limsup_{n\to\infty}  \frac1{n^2}\sum_{x,y\in V_n}\Erw\TV{\mu_{\G,x,y}-\bar\rho}  \leq  \textrm{corr}(d).
\end{equation}
We conclude the part for $\textrm{corr}(d)$ by combining the above with
 \eqref{eqLongRange}. Recall that the later states that   for any $\epsilon$ there exists $\dc<d<\dc+\epsilon$
such that the l.h.s. is strictly positive.

Repeating the same arguments as above we get that  for $\dc<d<\dc+1$ it holds that
\begin{equation}
\limsup_{n\to\infty}\frac1{n^2}\sum_{x,y\in V_n}\Erw\TV{\mu_{\G^*,x,y}-\bar\rho} \leq \textrm{corr}^*(d),
\end{equation}
where $\textrm{corr}^*(d)$ is defined in \eqref{def:PlantedCorr}.

Note that the l.h.s is bounded away from zero. To see this note that if it were zero, 
then it would have implied that for $d>\dc$ we get that 
$\lim_{n\to\infty}\frac{1}{n}\Erw[\ln Z(\G^*)]=\frac{\ln\xi}{k}d+\ln q$. 
Clearly, this  is not true, e.g. see Corollary \ref{lem:badoverlaps1} and Theorem \ref{Thm_plantedFreeEnergy}. 
Then we conclude that $\textrm{corr}^*(d)>0$ for 
$\dc<d<\dc+1$. 

Using Lemma \ref{lemma:Planting4GoodBoundary}
we get that $\textrm{corr}^{\star}(d)>0$ for {\em any} $d>\dc$ as well.
To be more specific, 
note that Lemma \ref{lemma:Planting4GoodBoundary} implies the following: 
Let $\dc<d<\dc+1$. Also consider the pair $(\G^*, {\mathbold{\sigma}}^*)$ and $(\T, \mathbold{\tau})$.
For any $h=o(\log n)$,  there is a coupling between $(\T_{\G^*,h}(v), {\mathbold{\sigma}}^*(\T_{\G^*,h}(v)))$
and  $(\T, \mathbold{\tau})$ such that with probability $1-o(1)$ we have
${ \tt ISM}(\T_{\G^*, h}(v)) = {\tt ISM}( \T^{h})$, with some isomorphism $f(\cdot)$.
 Furthermore,  for every $u\in \T_{\G^*, h}(v)$ we have that $\hat{\mathbold{\sigma}}(u)=\mathbold{\tau}(f(u))$.
This coupling implies that $\textrm{corr}^*(d)=\textrm{corr}^{\star}(d)$. That is, 
for  $\dc <d<\dc+1$ we have $\textrm{corr}^{\star}(d)>0$.
Then, using the monotonicity result from Lemma \ref{lemma:ReconMonotone} we get that
  for any $d>\dc$ we  have  $\textrm{corr}^{\star}(d)>0$.
The lemma follows.
\end{proof}

\noindent
In light of Lemma \ref{lemma:ReconThrsUpperBound}, we get the first part of
Theorem \ref{thrm:TreeGraphEquivalence} by using the following result.

\begin{lemma}\label{lemma:CorrVsCorrStar}
 For any $d<\dr^{\star}$ we have  that
$\mathrm{corr}^\star(d)= \mathrm{corr}(d)=0$.
Furthermore, for $\dr^{\star}<d <\dc$ we have that
$  \mathrm{corr}^\star(d), \mathrm{corr}(d) >0$.
\end{lemma}

\noindent
The proof of Lemma \ref{lemma:CorrVsCorrStar} appears in Section \ref{sec:lemma:CorrVsCorrStar}.

As far as the the second part of Theorem \ref{thrm:TreeGraphEquivalence} is concerned
essentially it follows as a corollary from all the previous results in this section.
It is elementary to verify that
\begin{equation}\label{eq:CondSVsNoCondSReconA}
\lim_{\ell\to\infty}\limsup_{n\to\infty}
\frac1n\sum_{y\in V_n}\sum_{s\in\Omega}
		\Erw\brk{\bck{\abs{\bck{\vecone\{\SIGMA(y)=s\}\big|\nabla_{\ell}(\G,y)}_{\G}-1/q}}_{\G}|\fS}
		\leq \frac{\textrm{corr}(d)}{\Pr[\fS]}.
\end{equation}
Using Lemma \ref{lemma:PoissonCycles4Uniform} we get that $\pr[\fS]=\Omega(1)$.
Then, using Lemma \ref{lemma:CorrVsCorrStar} we get that  for any $d<\dr^{\star}$ the l.h.s. of \eqref{eq:CondSVsNoCondSReconA}
is equal to zero.
We proceed by showing that for any  $\epsilon>0$  there exists  $\dc<d<\dc+\epsilon$  such that 
\begin{equation}\label{eq:CondSVsNoCondSReconB}
\lim_{\ell\to\infty}\limsup_{n\to\infty}
\frac1n\sum_{y\in V_n}\sum_{s\in\Omega}
		\Erw\brk{\bck{\abs{\bck{\vecone\{\SIGMA(y)=s\}\big|\nabla_{\ell}(\G,y)}_{\G}-1/q}}_{\G}|\fS}>0.
\end{equation}
Using Theorem \ref{Thm_overlap} and standard arguments e.g. 
(e.g., \cite[\Sec~2]{Victor})  there is $\epsilon>0$ such that 
	\begin{align*}
	\lim_{n\to\infty}\frac1{n^2}\sum_{y_1,y_2\in V_n}\Erw \left[ 
	\left .\TV{\mu_{\G,y_1,y_2}- \bar\rho} \right| \fS\right]&>0
	&\mbox{for  $\dc<d<\dc+\epsilon$.}
	\end{align*}
Then  \eqref{eq:CondSVsNoCondSReconB} follows by working as in 
 the proof of  Lemma \ref{lemma:ReconThrsUpperBound}.
Finally, we  show that for  $\dr^{\star}<d<\dc$ we  have 
\begin{equation}\label{eq:CondSVsNoCondSReconC}
\lim_{\ell\to\infty}\limsup_{n\to\infty}
\frac1n\sum_{y\in V_n}\sum_{s\in\Omega}
		\Erw\brk{\bck{\abs{\bck{\vecone\{\SIGMA(y)=s\}\big|\nabla_{\ell}(\G,y)}_{\G}-1/q}}_{\G}\ |\ \fS}>0. 
\end{equation}
For showing the above, we work as in the second case of  Lemma \ref{lemma:CorrVsCorrStar}, i.e. we use 
Lemma \ref{lemma:Planting4GoodBoundary} and the contiguity result in Corollary \ref{Thm_contig}.
More specifically, if there is $\dr^{\star}<d<\dc$ such that the l.h.s. of \eqref{eq:CondSVsNoCondSReconC} is
zero, then Corollary \ref{Thm_contig} would imply that
\begin{equation}
\lim_{\ell\to\infty}\limsup_{n\to\infty} 
\frac1n\sum_{y\in V_n}\sum_{s\in\Omega}
		\Erw
		\brk{ \left .
		\sum_{\tau\in \Omega^{\T_{\G^*,\ell}(y)}} \mu_{G^*}(\tau) \ 
		||\mu^{\tau}_{\G^*}-\mu_{\G^*} ||_{\{y\}}
		\ \right |\ \fS}
		=0,  \nonumber
\end{equation}
recall that  $\mu_{\G^*}$ is the distribution over configurations in $\Omega^{V_n}$ that is
induced by $\mathbold{\sigma}^*$ conditional on $\G^*$.
If the above was true, then   Lemma \ref{lemma:Planting4GoodBoundary}   would imply that $\textrm{corr}^{\star}(d)=0$. Clearly this
is a contradiction due to Lemma \ref{lemma:CorrVsCorrStar}.

The theorem follows.

\subsection{Proof of Lemma \ref{lemma:ReconMonotone}}\label{sec:lemma:ReconMonotone}

\noindent
Consider  two factor trees $T_1$ and $T_2$ with roots $r_1, r_2$, respectively. We say that $T_1, T_2$
satisfy the relation $T_1\subseteq T_2$ if  there is an injective mapping  $f:V(T_1)\cup F(T_1) \to V(T_2)\cup F(T_2)$
such that the  following is true:  for every  $v\in V(T_1)$ we have $\partial_{desc} v\subseteq \partial_{desc} f(v)$,
while  every $\alpha \in F$ such that $\alpha\in \partial_{desc} v \cap  \partial_{desc} f(v)$ is assigned the same
weight function $\psi_{\alpha}$ in both trees and  $v$, $f(v)$ occupy the same position within $\psi_\alpha$.
Furthermore,  for every function node $\alpha\in F(T_1)$ we have $\partial_{desc}  \alpha=\partial_{desc} f(\alpha)$ and
every $w\in \partial_{desc}  \alpha$ occupies in $\psi_{\alpha}$ the same position as $f(w)$ in  $f(\alpha)$.

\begin{lemma}\label{lemma:monotonicity}
Consider two sequences of factor trees $\mathcal{T}_1$ and $\mathcal{T}_2$ such the the following is true:
For  $T^1_{\ell}\in \mathcal{T}_1$  and $T^2_{\ell}\in \mathcal{T}_2$ we have $T^1_{\ell} \subseteq T^2_{\ell}$, for
$\ell=1, 2, \ldots$ Then, we have that 
\[
	\mathrm{broad}_{\mathcal{ T}_1 } \leq 	\mathrm{broad}_{\mathcal{ T}_2}.
\]
\end{lemma}
\begin{proof}
For some $\ell\geq 0$,  consider $T^1_{\ell}\in \mathcal{T}_1$ and $T^2_{\ell}\in \mathcal{T}_2$.
Since  we  assumed that $T^1_{\ell}\subseteq T^2_{\ell}$, let 
$h:V(T^1_{\ell})\cup F(T^1_{\ell}) \to V(T^2_{\ell})\cup F(T^2_{\ell})$ be the mapping
that verifies that property.

For any two $s, c\in \Omega$ consider $\mathbold{\tau}_1, \mathbold{\sigma}_1$
two configurations generated by the broadcasting process on $T^1_{\ell}$ such that
$\mathbold{\tau}_1=s$ and $\mathbold{\sigma}_1=c$. Similarly, 
let $\mathbold{\tau}_2, \mathbold{\sigma}_2$
two configurations generated by the broadcasting process on $T^2_{\ell}$ such that
$\mathbold{\tau}_2=s$ and $\mathbold{\sigma}_2=c$.
Then it suffices to show the following: For any $\alpha\in [0, 1]$, if there
is a coupling $\xi_2$ for $\mathbold{\sigma}_2, \mathbold{\tau}_2$ such that
the probability that $\mathbold{\sigma}_2(S(r, 2\ell) )\neq \mathbold{\tau}_2(S(r, 2\ell))$
is equal to $\alpha$, then there exists a coupling $\xi_1$ for 
 $\mathbold{\sigma}_1, \mathbold{\tau}_1$ such that
the probability that $\mathbold{\sigma}_1(S(r, 2\ell) )\neq \mathbold{\tau}_1(S(r, 2\ell))$
is at most  $\alpha$.

From the definition of the broadcasting process, we get the following: 
Let $\mathbold{\sigma}_1$ and $\mathbold{\sigma}_2$ be two configurations 
generated by broadcasting process on $T^1_{\ell}$ and $T^2_{\ell}$, respectively,
such that  $\mathbold{\sigma}_1(r)=\mathbold{\sigma}_2(r)=c$, for some $c\in \Omega$.  
Then there is a coupling $\zeta$ for $\mathbold{\sigma}_1$, $\mathbold{\sigma}_2$ such that
for every $v\in V(T^1_{\ell})$, we have that $\mathbold{\sigma}_1(v)=\mathbold{\sigma}_2(h(v))$.

Assume that we have the coupling $\xi_2$ for $\mathbold{\sigma}_2$ and $\mathbold{\tau}_2$.
We combine couplings $\xi_2$ and $\zeta$ to get $\xi_1$. In particular we use the couplings as
follows:  First, we couple  $\mathbold{\sigma}_1$ and $\mathbold{\sigma}_2$ by using $\zeta$. 
Then, we use $\xi_2$ to couple $\mathbold{\sigma}_2$ and $\mathbold{\tau}_2$.
Finally, we use $\zeta$ to couple $\mathbold{\tau}_2$ and $\mathbold{\tau}_1$.

In the above ``chain of couplings", note that we have 
$\mathbold{\sigma}_1(S(r, 2\ell) )\neq \mathbold{\tau}_1(S(r, 2\ell))$ only if
$\mathbold{\sigma}_2(S(r, 2\ell) )\neq \mathbold{\tau}_2(S(r, 2\ell))$. This implies 
that if in $\xi_2$ the probability of the event
 $\mathbold{\sigma}_2(S(r, 2\ell) )\neq \mathbold{\tau}_2(S(r, 2\ell))$ is equal to
 $\alpha$, then in $\xi_1$ the probability of having $\mathbold{\sigma}_1(S(r, 2\ell) )\neq \mathbold{\tau}_1(S(r, 2\ell))$
 is at most  $\alpha$.
The lemma follows.
\end{proof}

\noindent
In light of Lemmas  \ref{lemma:CorrVsBroad}, \ref{lemma:monotonicity}  we get the following corollary.
\begin{corollary}
Consider two sequences of factor trees $\mathcal{T}_1$ and $\mathcal{T}_2$ such that 
for  $T^1_{\ell}\in \mathcal{T}_1$  and $T^2_{\ell}\in \mathcal{T}_2$ we have $T^1_{\ell}
\subseteq T^2_{\ell}$, for $\ell=1, 2, \ldots$, then the following is true: 
If $\mathrm{corr}_{\mathcal{ T}_2}=0$, then $\mathrm{corr}_{\mathcal{ T}_1}=0$.
\end{corollary}

\noindent
The lemma follows by using the above corollary and noting  that for any $d_1,d_2>0$ such that $d_1\geq d_2$ there is 
a standard coupling such that $\T(d_2, P)\subseteq \T(d_1, P)$.

\subsection{Proof of Lemma \ref{lemma:Planting4GoodBoundary}}\label{sec:lemma:Planting4GoodBoundary}

The case where $\G^*\in \fS$ is almost identical to the case where we don't restrict 
$\G^*$. For this reason we omit the proof of the case where $\G^*\in \fS$.

Let  the pairs $(\T_{\G^*, h}(v), \mathbold{\sigma}^*)$ and $(\T^h, \mathbold{\tau})$.
Then,  we define the relation ``$\cong$" such that  $(\T_{\G^*, h}(v), \mathbold{\sigma}^*) \cong (\T^h, \mathbold{\tau})$
if the following holds: $\T_{\G^*, h}(v)$ and $\T^h$ belong to the same isomorphism class of
rooted trees, where $\T_{\G^*, h}(v)$ is rooted at $v$ and $\T^h$ is rooted at $r$.
Furthermore,  if $f$ is an isomorphism between the two trees, then for every $u\in \T_{\G^*, h}(v)$
we have that ${\mathbold{\sigma}}^*(u)=\mathbold{\tau}(f(u))$.
We are going to show a coupling $\tilde{\lambda}$ that has the property that
\begin{equation}\label{eq:Basis:Planting4GoodBoundary}
\tilde{\lambda}\left[  (\T_{\G^*, h}(v), \mathbold{\sigma}^*) =(\T^h, \mathbold{\tau}) \right] \geq1- 2n^{-1/11}.
\end{equation}
For what follows,  we denote $f$ the isomorphism $\T_{\G^*, h}(v)$ and $\T^h$, if such exists.

Before proceeding let us state some, easy to prove results.
Recall that for an assignment $\sigma$ on $n$ vertices we denote by $\mu_\sigma=n^{-1}(|\sigma^{-1}(i)|)_{i\le q}$ its 
empirical marginal distribution.  Furthermore, it is elementary to show that
\begin{equation}\label{eq:Tail4EmpiricalDistr}
\pr\left[\| \mu_{\SIGMA^*}-q^{-1}\vecone\|>(\sqrt{n})^{-1}\ln n \right]\le  O(n^{-\ln\ln n}).
\end{equation}
Let  $\mathbold{m}$ be the number of edges in $\G^*$. Recall that $\mathbold{m}$ is a random
variable which is distributed as in Poisson with parameter  $dn/k$.  
Applying standard Chernoff's bounds for $\mathbold{m}$ we get that
$$
\pr\left[ |\mathbold{m}-dn/k |> n^{2/3} \right]\le  \exp\left( -n^{1/4}\right).
$$
We let $|\T_{\G^*, h}(v)|$ denote the number of vertices in $\T_{\G^*, h}(n)$.  
Note  that for every variable node $x\in \T_{G^*, h}(v)$, the cardinality of
$\partial_{desc} x$ is dominated by the Poisson distribution with parameter $d$.
With this observation we get that
\begin{equation}\label{eq:ExpectedSizeOfBall}
\Erw\left[ |\T_{\G^*, h}(v) | \right]\leq 2((k-1)d)^{h+1}.
\end{equation}

\noindent
The coupling $\tilde{\lambda}$ is as follows: If  ${\mathbold{\sigma}^*}$ is such that 
$\| \mu_{\SIGMA^*}-q^{-1}\vecone\|>(\sqrt{n})^{-1}\ln n$ or $|\mathbold{m}-dn/k|>n^{2/3}$ we don't couple 
$(\T_{\G^*, h}(v), \mathbold{\sigma}^*)$ and $(\T^h, \mathbold{\tau})$ at all.
Otherwise, the coupling $\tilde{\lambda}$ is defined inductively. 

First consider the coupling between ${\mathbold{\sigma}^*}(v)$ and $\mathbold{\tau}(r)$. 
Note that $f(v)=r$. Due to our assumption
about $\mu_{\mathbold{\sigma}^*}$, we can have $\tilde{\lambda}$ such that
\begin{equation}
\tilde{\lambda}(\mathbold{\sigma}^*(v)\neq \mathbold{\tau}(r))=O(  (\sqrt{n})^{-1}\ln n).
\end{equation}
The above follows by using a maximal coupling for choosing 
${\mathbold{\sigma}^*}(v), \mathbold{\tau}(r)$.

The induction step is as follows: Assume that we have exposed  partly  
$(\T_{\G^*, h}(v), \mathbold{\sigma}^*)$ and $(\T^h, \mathbold{\tau})$
and the corresponding parts agree. That is,  let  $(\T_1, \mathbold{\sigma}_1)$
and $(\T_2, \mathbold{\sigma}_2)$ be the two parts
of $(\T_{\G^*, h}(v), \hat{\mathbold{\sigma}})$ and $(\T^h, \mathbold{\tau})$,  
respectively. Our assumption is that  $(\T_1, \mathbold{\sigma}_1)\cong (\T_2, \mathbold{\sigma}_2)$.
W.l.o.g. assume that the leaves of the trees are  variable nodes.

Let $x$ be a leaf in $\T_1$  whose  descendants have not been revealed so far.  The same holds for 
$f(x)$ in $\T_2$.
Let   $\mathbold{m}_x$ be the number of hyper-edges of $G^*$ that have revealed so 
far. Recall that the number of all hyper-edges in $G^*$ is $\mathbold{m}$.
Then, it is an easy calculation to get that for  any $j$ we have 
\begin{equation}\label{eq:ConditionalDegreeOfX}
\pr\left[ |\partial_{desc} x| =j \  |\  \mathbold{m}_x\right]={\mathbold{m}-\mathbold{m}_x \choose j} \left( \frac{k}{n}\right)^{j}
\left (1-\frac{k}{n} \right)^{\mathbold{m}-\mathbold{m}_x-j}.
 \end{equation}

\noindent
If $r$ is the number of edges  of the tree we have revealed up to vertex $x$, then we have the  crude upper bound
 that $\mathbold{m}_x \leq  r$.   We have that
\begin{equation}
\pr\left[\mathbold{m}_x\geq n^{1/3} \right] \leq \pr\left[  |\T_{\G^*, h}(v)| \geq n^{1/3} \right] \leq  n^{-1/3} \Erw\left[ |\T_{\G^*, h}(v)| \right] \leq n^{-1/4},
\end{equation} 
where the third inequality follows from Markov's inequality and the last inequality follows from our
assumption that $h=o(\log n)$.
Combining the two above relations we get the following: For any  $0\leq j \leq\ln^2 n$ it holds that
\begin{eqnarray}
\pr\left[ |\partial_{desc} x| =j \right ]  &=& \pr\left[ |\partial_{desc} x| =j \  |\  \mathbold{m}_x< n^{1/3} \right]\pr[\mathbold{m}_x< n^{1/3}] +
\pr\left[ |\partial_{desc} x| =j \  |\  \mathbold{m}_x \geq n^{1/3}\right]\pr[\mathbold{m}_x\geq  n^{1/3}]  \nonumber \\
&=&  \frac{d^{j}}{j!}e^{-d} +O(n^{-3/4}).\label{eq:ChildrenGStartVsPoisson}
\end{eqnarray}
Similarly we get that  $\pr\left[ |\partial_{desc} x| >\ln^2 n\right ]=o(n^{-10}) $.

Recall that for a vertex $u\in \T^h$ we have that $|\partial_{desc} u|$ is distributed as in 
Poisson with parameter $d$. 
Using this observation we can have $\tilde{\lambda}$ such that
\begin{equation}
\tilde{\lambda}(|\partial_{desc} x| \neq |\partial_{desc} f(x)|)=O(n^{-3/4}\ln^2 n).
\end{equation}
We extend  $f$  by defining a bijection between $\partial_{desc} x$ and $\partial_{desc} f(x)$.  
From the definition of $\G^*$ we get that each $\alpha\in \partial_{desc} v$ chooses a weight 
function $\psi\in \Psi$ from a distribution which is within total variation distance   $O(n^{-1/2} \ln n)$
from $P$. Note that the term $O(n^{-1/2} \ln n)$ comes from the fact that $\mathbold{\sigma}^*$ is not
perfectly balanced, i.e. we allow some fluctuations $O(\sqrt{n}\ln n)$ on the sizes of the color classes.
For $f(\alpha)$ we have that it chooses its weight function $\psi$ with probability $P(\psi)$. 
The above observations imply that we can have $\tilde{\lambda}$ such that
\[
\tilde{\lambda} \left [  \exists \alpha\in \partial_{desc} x \ \textrm{ s.t. } \ \psi_{\alpha}\neq \psi_{f(\alpha)}  \right]  
= O(n^{-1/2}\ln^2 n),
\]
By choosing the same weight function $\psi_\alpha$  for both $\alpha$ and $f(\alpha)$ we imply that the position
of $x$ and $f(x)$ is the same in the two functions.

Finally, for every pair of constraint nodes $\alpha$ and $f(\alpha)$ for which we have chosen the weight function
$\psi_{\alpha}$ we decide on ${\mathbold{\sigma}^*}(y_i)$ and $\mathbold{\tau}(z_i)$, where $y_i\in \partial \alpha\setminus \{x\}$
and $z_i=f(x_i)$.
For each configuration  $\tau\in \Omega^k$  we have $\hat{\mathbold{\sigma}}(\partial \alpha)=\tau$ with probability 
proportional to 
\begin{equation}\nonumber
   \vecone\{\tau(j_{\alpha, x})=\mathbold{\sigma}^*(x) \} \psi_{\alpha}(\tau)+O(n^{-1/2}\ln n),
\end{equation}
where  $j_{\alpha, x}$ is the position of  $x$ inside the constraint  $\psi_{\alpha}$.
Also,  we have $\mathbold{\tau}(\partial f(\alpha))=\tau$ with probability proportional to 
\begin{equation}\nonumber
   \vecone\{\tau(j_{f(\alpha), f(x)})=\mathbold{\sigma}^*(f(x)) \} \psi_{\alpha}(\tau).
\end{equation}
From the above, it is clear that we can have $\tilde{\lambda}$ such that
$$
\tilde{\lambda}({\mathbold{\sigma}}^*(\partial \alpha) \neq \mathbold{\tau}(\partial f(\alpha))) \leq O(|\Omega|^k n^{-1/2}\ln n)
\leq O(n^{-1/2}\ln n).
$$
Let  $(\T'_1, \mathbold{\sigma}'_1)$  and $(\T'_2, \mathbold{\sigma}'_2)$ be the new parts
of of $(\T_{\G^*, h}(v), {\mathbold{\sigma}}^*)$ and $(\T^h, \mathbold{\tau})$,   after the revelation
of $\partial_{desc}x, \partial_{desc}f(x)$ and  $\partial_{desc}\alpha, \partial_{desc} f(a)$, 
for every $\alpha\in \partial_{desc}x$ and for every $f(\alpha) \in \partial_{desc}f(x)$.

Then, using all the above and a simple union bound gives that 
\[
\tilde{\lambda}\left[  (\T'_1, \mathbold{\sigma}'_1)\neq (\T'_2, \mathbold{\sigma}'_2) \ |\ |\partial_{desc} x| \right] \leq |\partial_{desc} x| n^{-1/3}.
\]
The law of total probability implies that
\begin{eqnarray}
\tilde{\lambda}\left[  (\T'_1, \mathbold{\sigma}'_1)\neq (\T'_2, \mathbold{\sigma}'_2)\right]  
&\leq& 2n^{-1/3} (\ln n)^2. \label{eq:OneStepDisagreement}
\end{eqnarray}

\noindent
Lemma \ref{lemma:Planting4GoodBoundary} follows by bounding appropriately the number of steps required for the coupling. 
Let $\mathcal{A}$ be the event that the number of steps in the coupling is more than $n^{1/10}$.
Since the number of steps of the coupling is upper bounded by the number of vertices of 
$\T_{G^*, h}(v)$, using \eqref{eq:ExpectedSizeOfBall} and Markov's inequality we get that
\begin{equation}\label{eq:CouplingStepTail}
\tilde{\lambda}(\mathcal{A}) \leq n^{-1/11}.
\end{equation}
We have that
\begin{eqnarray}
\tilde{\lambda}\left[  (\T_{\G^*, h}(v), \mathbold{\sigma}) \neq (\T^h, \mathbold{\tau}) \right] &\leq&
\tilde{\lambda}\left[  (\T_{\G^*, h}(v), \mathbold{\sigma}) \neq (\T^h, \mathbold{\tau}) | \mathcal{A}^c\right]+
\tilde{\lambda}[\mathcal{A}] \nonumber \\
&\leq & n^{1/10} \tilde{\lambda}\left[  (\T'_1, \mathbold{\sigma}'_1)\neq (\T'_2, \mathbold{\sigma}'_2)  \ |\  \mathcal{A}^c\right]+
\tilde{\lambda}[\mathcal{A}]  \qquad \qquad \mbox{[union bound]}\nonumber \\
&\leq & n^{1/10} \frac{\tilde{\lambda}\left[  (\T'_1, \mathbold{\sigma}'_1)\neq (\T'_2, \mathbold{\sigma}'_2)\right]}{\tilde{\lambda}[\mathcal{A}^c] }+
\tilde{\lambda}[\mathcal{A}]   \nonumber \\
&\leq & 2n^{1/10}  \tilde{\lambda} \left[  (\T'_1, \mathbold{\sigma}'_1)\neq (\T'_2, \mathbold{\sigma}'_2)\right] +
n^{-1/11}  \qquad\qquad\mbox{[from \eqref{eq:CouplingStepTail}]}  \nonumber \\
&\leq  &2n^{-1/11}. \nonumber
\end{eqnarray}
The above implies that \eqref{eq:Basis:Planting4GoodBoundary} is indeed true. The lemma follows.

\subsection{Proof of Lemma \ref{lemma:CorrVsCorrStar}}\label{sec:lemma:CorrVsCorrStar}
Clearly, Lemma \ref{lemma:ReconMonotone} implies
that we have  $\mathrm{corr}^\star(d)=0$ if and only if $d < \dr^{\star}$.
To see this note the following: Assume that there is  $d_0>\dr^{\star}$ such that
$\mathrm{corr}^\star(d_0)=0$.  
Then  Lemma \ref{lemma:ReconMonotone}
implies that since $d_0>\dr^{\star}$  and $\mathrm{corr}^\star(d_0)=0$, then
 we also have $\mathrm{corr}^\star(\dr^{\star})=0$, which is false.

For proving Lemma \ref{lemma:CorrVsCorrStar}, it remains to
show that $\mathrm{corr}(d)=0$ if and only if $d>\dr^{\star}$.
First we focus on showing that for $d<\dr^{\star}$ we have
\begin{equation}\label{eq:AsymptoticEquivOfCorrsA}
\mathrm{corr}(d)=0.
\end{equation}
For even integer $\ell >0$ consider the factor tree $T_{\ell} $ which contains  $\ell$ levels of variable nodes
and it is rooted at $r$.
The  configuration $\eta \in \Omega^{S(r, \ell)}$ is called  ``$(\ell, \delta)$-mixing", for some $\delta\geq 0$, if it holds  that
\begin{equation}\nonumber
||\mu^{\tau}_{T_{\ell}} -\mu ||_{\{r\}} \leq \delta.
\end{equation}
Let $\mathcal{M}(T_\ell, \ell, \delta)$ be the set of all  configurations  which are $(\ell, \delta)$-mixing 
for $T_\ell$.
The above quantity expresses the correlation between the configuration of the vertices at distance 
$2\ell$ and the root $r$ (set to vertex correlation).

Eq. \eqref{eq:AsymptoticEquivOfCorrsA} follows by showing the following result.

\begin{lemma}\label{lemma:BoundaryGVsTree}
 For $d<\dr^{\star}$ and  every $\delta>0$ there exists $\ell_0=\ell_0(\delta)$ such
that for any even $\ell \ge \ell_0$ we have
\begin{equation}\label{eq:MixingInG}
\lim_{n\to\infty}\Erw\left[ \bck{\vecone\{\mathbold{\sigma} \in \mathcal{M}(\T_{\G,h}(v), \ell, \delta ) \}}_{\G } \right] \geq 1-\delta.
\end{equation}  
\end{lemma} 
\begin{proof}
We shift our attention to considering
the teacher-student pair $(\G^*, \mathbold{\sigma}^*)$.
In light of Corollary \ref{Cor_strongCntig}, it suffices to show the following:
 For $d<\dr^{\star}$ and  every $\epsilon>0$ there exists $\ell_0=\ell_0(\epsilon)$ such
that for any $\ell \ge \ell_0$ we have
\begin{equation}\label{eq:MixingInGStar}
\lim_{n\to\infty}\pr\left[   \mathbold{\sigma}^* \notin \mathcal{M}(\T_{\G^*, h}(v), \ell, \epsilon ) \}  \right]\leq \epsilon.
\end{equation}

\noindent
In light of Lemma \ref{lemma:Planting4GoodBoundary},  for \eqref{eq:MixingInGStar} it suffices  to show the following result:
For any    $d<\dr^\star$ and any $\epsilon>0$
there exists $\ell_0=\ell_0(\epsilon)$ such that
$$
 \Erw \bck{\vecone\{\mathbold{\sigma}\notin \mathcal{M}(\T^\ell(d, P), \ell, \epsilon)  \} }_{\T^\ell}  \leq \epsilon.
$$
Clearly the above follows from the definition of $\dr^{\star}$.
\end{proof}

\noindent
From Lemma \ref{lemma:BoundaryGVsTree} we get  \eqref{eq:AsymptoticEquivOfCorrsA} by working as follows:
 Let 
\begin{eqnarray}
	\mathrm{corr}_{ v, \ell}(d)
	&=&   \Erw  \left[  
	\sum_{\tau\in \Omega^{S((v, 2\ell)}} \mu_{\G}(\tau)
||\mu^{\tau}_{\G}-\mu_{\G} ||_{\{v\}}
\right]. \nonumber
\end{eqnarray}
Furthermore, for any $\delta>0$, integer $\ell$,  for $\G$, for any vertex $v$ and 
$\mathbold{\sigma}$ distributed as in Gibbs measure,  let  $\mathcal{G}=\mathcal{G}(v, \ell, \delta)$ be the event that 
$\mathbold{\sigma} \in \mathcal{M}(\T_{\G,\ell}(v), \ell, \delta)$.
 Lemma \ref{lemma:BoundaryGVsTree} implies that for $d<\dr^{\star}$, for  every $\delta>0$ there exists 
 $\ell_0=\ell_0(\delta)$ such that for any $\ell\geq \ell_0$  the following holds:
\begin{eqnarray}
	\mathrm{corr}_{ v, \ell}
		&=&  \Erw   \left[(1- \vecone \{ \mathcal{G} \}  )
	\sum_{\tau\in \Omega^{S(v, 2\ell)}} \mu_{\G}(\tau)\ 
||\mu^{\tau}_{\G}-\mu_{\G} ||_{\{v\}}
		\right]  
		+
		 \Erw \left[ \vecone \{ \mathcal{G} \}  
		 	\sum_{\tau\in \Omega^{S(v, 2\ell)}} \mu_{\G}(\tau) \ 
||\mu^{\tau}_{\G}-\mu_{\G} ||_{\{v\}}
		\right] \nonumber  \\ 
		&\leq &  \Erw  \left[ 1-\vecone \{ \mathcal{G} \} \right]+\delta+o(1) \ \leq \  2\delta +o(1). \nonumber 
\end{eqnarray}

\noindent
Noting that  $\textrm{corr}(d)=\limsup_{\ell\to\infty}\limsup_{n\to\infty}n^{-1}\sum_{v\in V_n} \mathrm{corr}_{ v, \ell}(d)$,
we get that \eqref{eq:AsymptoticEquivOfCorrsA} is indeed true.

We conclude the proof of the Lemma \ref{lemma:CorrVsCorrStar} by showing 
 that for $d>\dr^{\star}$ we have
\begin{equation}\label{eq:AsymptoticEquivOfCorrsB}
\mathrm{corr}(d)>0.
\end{equation}

\noindent
The proof of \eqref{eq:AsymptoticEquivOfCorrsB} is by contradiction. Assume that
there exists $\dr^{\star }<d$ such that $\textrm{corr}(d)=0$,
this would entail that \eqref{eq:MixingInG} is true. Then, reversing the  arguments from the proof of Lemma 
\ref{lemma:BoundaryGVsTree}, and combining them  Corollary \ref{Cor_strongCntig}, 
we  get that for  any $\epsilon>0$ there exists $\ell_0=\ell_0(\epsilon)$ such that for any $\ell>\ell_0$
we have 
\[
 \Erw \bck{\vecone\{\mathbold{\sigma}\notin \mathcal{M}(\T^\ell(d, P), \ell, \epsilon)  \} }_{\T^\ell}  \leq \epsilon.
\]
The above  implies that $\textrm{corr}^{\star}(d)=0$. Clearly we get a contradiction since
we have shown in Lemma \ref{lemma:ReconMonotone} that for every $d>\dr^{\star}$ we have $\textrm{corr}^{\star}(d)>0$.

\bigskip\noindent{\bf Acknowledgment.}
We thank Will Perkins, Guilhem Semerjian and Nick Wormald for helpful discussions.

\end{document}